%% file: Das_diss.tex
\documentclass[12pt,oneside]{book}
\usepackage[margin=0.75in]{geometry}
\usepackage{indentfirst} %index first line of each paragraph
\usepackage{cite}
\usepackage{amsfonts}
\usepackage{amsmath}
\usepackage{amssymb}
\usepackage{mathtools}
\usepackage{graphicx}
\usepackage{enumitem}
\usepackage{bibentry}
\usepackage{lipsum}
\usepackage{subcaption}
\usepackage{caption}
\usepackage{hyperref}
\hypersetup{colorlinks ,linkcolor=blue , filecolor=magenta,      
    urlcolor=cyan}
\usepackage{mathrsfs}
\usepackage{bbm}
\usepackage{sectsty}
\usepackage{titlesec}
\usepackage{titletoc}
\usepackage{setspace}
\usepackage{physics}

\newtheorem{theorem}{Theorem}[chapter]
\newtheorem{corollary}{Corollary}[chapter]
\newtheorem{lemma}{Lemma}[chapter]
\newtheorem{proposition}{Proposition}[chapter]

\newtheorem{definition}{Definition}[chapter]

\newtheorem{remark}{Remark}[chapter]

\newcommand\blfootnote[1]{ \begingroup \renewcommand\thefootnote{}\footnote{#1} \addtocounter{footnote}{-1} \endgroup}

\def\Tr{\operatorname{Tr}}
\def\ONB{\operatorname{ONB}}

\def\sq{\operatorname{sq}}
\def\SEP{\operatorname{SEP}}
\def\Ent{\operatorname{Ent}}
\def\PPT{\operatorname{PPT}}
\def\supp{\operatorname{supp}}
\def\LOCC{\operatorname{LOCC}}
\def\T{\operatorname{T}}
\def\>{\rangle}
\def\<{\langle}
\def\({\left(}
\def\){\right)}
\def\[{\left[}
\def\]{\right]}
\def\V{\Vert}
\def\id{\operatorname{id}}
\def\d{\operatorname{d}}
\let\oldemptyset\emptyset
\let\emptyset\varnothing

\newcommand{\bb}[1]{\mathbb{#1}}
\newcommand{\mc}[1]{\mathcal{#1}}
\newcommand{\wt}[1]{\widetilde{#1}}

\newcommand{\tf}[1]{\textbf{#1}}
\newcommand{\tn}[1]{\textnormal{#1}}
\newcommand{\msc}[1]{\mathscr{#1}}
\newcommand{\bm}[1]{\mathbbm{#1}}
\newcommand{\hl}[1]{\overline{#1}}
\newenvironment{proof}[1][Proof]{\noindent\textbf{#1.} }{\ \rule{0.5em}{0.5em}}

\DeclareMathOperator*{\SumInt}{
\mathchoice
  {\ooalign{$\displaystyle\sum$\cr\hidewidth$\displaystyle\int$\hidewidth\cr}}
  {\ooalign{\raisebox{.14\height}{\scalebox{.7}{$\textstyle\sum$}}\cr\hidewidth$\textstyle\int$\hidewidth\cr}}
  {\ooalign{\raisebox{.2\height}{\scalebox{.6}{$\scriptstyle\sum$}}\cr$\scriptstyle\int$\cr}}
  {\ooalign{\raisebox{.2\height}{\scalebox{.6}{$\scriptstyle\sum$}}\cr$\scriptstyle\int$\cr}}
}

\makeatletter
\def\blfootnote{\gdef\@thefnmark{}\@footnotetext}
\makeatother

\allsectionsfont{\normalsize\singlespace}

\titlecontents{chapter}% <section-type>
[0pt]% <left>
{\addvspace{10pt}}% <above-code>
{\normalfont\chaptername\ \thecontentslabel\quad}% <numbered-entry-format>
{}% <numberless-entry-format>
{\dotfill\contentspage}% <filler-page-format>

\titlecontents{section}% <section-type>
[10pt]% <left>
{\addvspace{2pt}}% <above-code>
{\normalfont\sectiontitle\ \thecontentslabel\quad}% <numbered-entry-format>
{}% <numberless-entry-format>
{\dotfill\contentspage}% <filler-page-format>

\titleformat{\chapter}[block]
{\normalfont\large\bfseries\singlespace}{Chapter~\thechapter}{14pt}{\large} %size and all on one line, singlespaced, and remove "Chapter"

\titlespacing{\chapter}{0pt}{*1}{*.5}
\titlespacing{\section}{0pt}{*.5}{*.5}

\begin{document}
\newcommand{\Title}{BIPARTITE QUANTUM INTERACTIONS:\\ ENTANGLING AND INFORMATION PROCESSING ABILITIES}
\newcommand{\department}{The Department of Physics \& Astronomy}
\newcommand{\Author}{Siddhartha Das}

\frontmatter

%Title Page
\thispagestyle{empty}
\include{title}
\clearpage

\pagestyle{plain}
%Dedication (Do I need this?)
%\include{dedication}
%\clearpage

\include{dedication}
\clearpage

%Acknowledgements
\include{acknowledgments}

\clearpage

%\include{pub}
%\clearpage

%Table of Contents
\renewcommand*\contentsname{Table of Contents}
\setcounter{tocdepth}{1}% Allow only \section in ToC
\begin{singlespace}
\tableofcontents
%\addcontentsline{toc}{chapter}{Table of Contents} %idk why this isnt allowed
\end{singlespace}
\clearpage

%Abstract

\mainmatter

\include{abstract}

\clearpage

%Introduction
\include{intro}

\clearpage

%Preliminaries
\include{review}

\clearpage

%Bipartite quantum interactions
\include{bqi}

\clearpage

%Rate of entropy change
\include{eg}

\clearpage

%Quantum reading 
\include{qr}

\clearpage

%Private reading 
\include{pr}

\clearpage

%Secure reading
\include{sr}

\clearpage

%Conclusion
%\include{conclusion}
%\clearpage

%References
\chapter*{References}
\addcontentsline{toc}{chapter}{References} 
\begin{singlespace}
	\begingroup 
	\bibliographystyle{ieeetr}
	\renewcommand{\chapter}[2]{}% for other classes
	\bibliography{../phdbib}
	\endgroup
\end{singlespace}

\titlecontents{chapter}% <section-type>
[0pt]% <left>
{\addvspace{10pt}}% <above-code>
{\normalfont\contentspush{Appendix}\ \thecontentslabel\quad}% <numbered-entry-format>
{}% <numberless-entry-format>
{\dotfill\contentspage}% <filler-page-format>

%format appendix titles
\titleformat{\chapter}[block]
{\normalfont\large\bfseries\singlespace}{ \contentspush{Appendix }\thechapter}{14pt}{\large} %size and all on one line, singlespaced, and remove "Chapter"

\begin{appendix}
%Appendix A: Supplementary work on the optimality of a thermal state input
\include{app-hw}

\clearpage

\include{app-Rmax}

\clearpage

\include{app-eg}

\clearpage

\end{appendix}

%Vita
\addcontentsline{toc}{chapter}{Vita} 
\include{vita}

\clearpage

\end{document}

%% file: title.tex
\begin{center}
	\begin{singlespace}
	\large
	\textbf{ \Title }
	\normalsize
	\end{singlespace}
	\vfill
	
	A Dissertation \\
	
	\begin{singlespace}
	Submitted to the Graduate Faculty of the \\
	Louisiana State University and \\
	Agricultural and Mechanical College \\
	in partial fulfillment of the \\
	requirements for the degree of \\
	Doctor of Philosophy \\
	\end{singlespace}
	in \\
	\department \\
	\vfill
	
	\begin{singlespace}
	by \\
	\Author \\
	B.S. and M.S., Indian Institute of Science Education and Research Pune, 2014  \\
	December 2018
	%\vspace{1cm}
	\end{singlespace}
\end{center}

%% file: dedication.tex
\chapter*{}
I dedicate this dissertation to my family members -- sisters, Chetna Das and Hema Das, and parents, Ram Narayan Das and Kalpana Das. Their love, support, and sacrifices are immeasurable. They keep me going with a peaceful mind and content heart. In addition, I also dedicate this dissertation to all my teachers, in particular, Ishwari Shrestha and Akhilesh Kumar Das (Batul da). Their teachings and guidance still serve to shed light when I find myself lost.   

%% file: acknowledgments.tex
\chapter*{Acknowledgments}
\addcontentsline{toc}{chapter}{Acknowledgments}

I feel very fortunate and am deeply grateful to my PhD supervisor Mark M.~Wilde for all his encouragements, continuous support, and guidance which enabled me to happily complete my PhD. I can say that I got formal training in writing research articles from him. I am very thankful to him for his teachings, helping me to learn how to pursue research problems, and pen down ideas. He let me take time to think about ideas and problems to work on, gave freedom to pursue them at my own pace, and provided necessary feedback, guidelines, and help.   

I sincerely thank Jonathan P.~Dowling (Jon) for replying to my email that made me join Quantum Science and Technologies (QST) group at LSU. I am grateful to him for his continuous care, fun times, and important encouragements during my stay here. He has also played an important role in teaching me how to pen down ideas. I am very satisfied being part of QST that is largely due to Mark and Jon, and I am grateful to them for making grad life lively and eventful.    

I am grateful to Thomas Corbitt, A.~Ravi P.~Rau, and Shawn Walker for agreeing to be part of my PhD committee, and their valuable feedback during this period. I am deeply thankful to A.~Ravi P.~Rau and Ivan Agullo for always being available for discussions. I have greatly enjoyed learning from discussion and teaching of Ivan. I am very thankful to Stefan B\"auml, Francesco Buscemi, Jonathan P.~Dowling, Eneet Kaur, Sumeet Khatri, Ludovico Lami, George Siopsis, Qingle Wang, Christian Weedbrook, Mark M.~Wilde, and Andreas Winter for fruitful collaborations that I have thoroughly enjoyed and learned from. %I am grateful to Mark for allowing me to participate in these collaborations and also giving me freedom to pursue ideas that I found interesting. 

I acknowledge funding support and grants through Research Assistantship and travel supports by Mark through EDA Assistantship (LSU), Dr.~Charles Coates Memorial Fund (LSU), LINK program by the Board of Regents of the State of Louisiana, and National Science Foundation (NSF) and the Teaching Assistantship by the department. I thank support from Jon, NSF, IMS-NUS Singapore, and QCrypt organizers 2018. I thank the support of Saikat Guha in arranging my summer visit to Raytheon BBN Technologies in 2016, Stanis\l{}aw J.~Szarek for the invitation to a trimester in Institut Henri Poincar\'e, Paris during winter 2017, and Junde Wu for a summer school in Zhejiang University, Hangzhou in 2018. I am thankful to the speakers and participants of all the conferences and workshops that I have attended during the period of my PhD studies for their valuable inputs and stimulating discussions. 

I thank all QST group memebers and visitors, past and present, with whom I have interacted with. I extend my thanks to Physics and Astronomy Graduate Students Organization that got launched largely due to efforts from Kelsie Krafton, Erin Good, Alison Drefuss, and others. I also thank Arnell Nelson, Carol Duran, Claire Bullock, Stephanie Jones, Laurie Rea, Shanan Schatzle, Shemeka Law, and other administrative staffs in the Department of Physics and Astronomy for their help and support. I thank all other graduate students and faculty members in the department for making my time enjoyable here. I am thankful to Sushovit Adhikari, Amey Apte, Sukruti Bansal, Narayan Bhusal, Noah Davis, Kevin Valson Jacob, Vishal Katariya, Bhawna Kumari, Sachin Kumar, Diya Mukherjee, Jonathan P.~Olson, Aashay Patil, Haoyu Qi, Sahil Saini, Sameer Saurabh, Punya Plaban Satpathy, Kunal Sharma, Mihir Sheth, Siddharth Soni, and Chenglong You for their help and support. I thank Vishal again for his feedback on thesis.

I am grateful to Nepalese Student Association, Indian Student Association, International Student Association, and International Cultural Center for providing support and fun events to participate in.

My interest in science was led in Koshi Vidya Mandir (KVM). I got inclined to physics by the teachings of Batul da, which also led me to admire its beauty. Final push to Quantum Mechanics was largely due to lectures by Sourav Pal, late R. R. Simha, and Avinash Khare at IISER Pune. I thank my friends from KVM, Batulda Siksha Sadan, and IISER-P whose support and discussion on varying topics have been encouraging and sometimes stress buster. I feel fortunate to learn and discuss with Anil D.~Gangal during my stay in IISER Pune. I admire his grasp on wide range of topics and also his research methodology. 

Finally, I thank the care and support from my relatives and friends who have helped me. I am grateful to my uncle (mama), Alok Kumar, for his guidance. My maternal grandparents, Sushila Devi (nanaima) and Brahma Narayan Das (nanaji), and late paternal grandparents, Khaiki Devi (dadima) and Tara Prasad Das (dadaji), have always been my pillar of support, and I am deeply thankful to them.

%% file: abstract.tex
\chapter*{Abstract}
\addcontentsline{toc}{chapter}{Abstract}

The aim of this thesis is to advance the theory behind quantum information processing tasks, by deriving fundamental limits on bipartite quantum interactions and dynamics. A bipartite quantum interaction corresponds to an underlying Hamiltonian that governs the physical transformation of a two-body quantum system. Under such an interaction, the physical transformation of a bipartite quantum system may also be considered in the presence of a bath, which may be inaccessible to an observer. The goal is to determine entangling abilities of arbitrary bipartite quantum interactions. Doing so provides fundamental limitations on information processing tasks, including entanglement distillation and secret key generation, over a bipartite quantum process, which may be noisy in general. We also discuss limitations on the entropy change and its rate for dynamics of an open quantum system weakly interacting with the bath. We introduce a measure of non-unitarity to characterize the deviation of a doubly stochastic quantum process from a noiseless evolution. 

Next, we introduce information processing tasks for secure read-out of digital information encoded in read-only memory devices against adversaries of varying capabilities. The task of reading a memory device involves the identification of an interaction process between probe system, which is in known state, and the memory device. Essentially, the information is stored in the choice of channels, which are noisy quantum processes in general and are chosen from a publicly known set. Hence, it becomes pertinent to securely read memory devices against scrutiny of an adversary. In particular, for a secure read-out task called private reading when a reader is under surveillance of a passive eavesdropper, we have determined upper bounds on its performance. We do so by leveraging the fact that private reading of digital information stored in a memory device can be understood as secret key agreement via a specific kind of bipartite quantum interaction.

%% file: intro.tex
\chapter{Exploring Informational Aspects of Quantum Interactions}
\section{Motivation}
The beauty of nature is inexpressible. It is of fundamental interest to understand natural phenomena. All physical systems and processes are governed by the laws of nature. In this thesis, we limit our discussion to the domain in which the laws of nature are well described by \textit{quantum} theory. Note that the classical theory emerges from the quantum theory as an approximation; i.e., the laws of classical mechanics can be obtained from the laws of quantum mechanics by making particular choices for a quantum process and the state of a quantum system \cite{Dbook81,SCbook95}. 

A physical system that is described by the laws of quantum mechanics is called a quantum system. The state of a quantum system has its complete physical description. A physical operation is quantum when evolution of a quantum system is feasible under its action. In general, the physical description contained in the state of a quantum system cannot be accounted for classical theory~\footnote{Entanglement, superposition, uncertainty relations, contextuality, and indeterminsim~\cite{GS22a,GS22b,Hei27,BA04,HSD15,DASH13,CBTW17} are some of the interesting features of quantum theory.}~\cite{SCbook95}. While quantum effects are dominating in microscopic regime, these effects are largely unnoticeable (vanishing) in macroscopic regime. 

With the inception of the idea that information is physical~\cite{Ben01,Bek81}, there has been wide interest in studying several physical phenomena and systems from an information-theoretic perspective. Information is associated with a physical respresentation. The information content of a physical system must be finite if the region of space and the energy is finite~\cite{Bek81,Bou02}. From these observations, we can conclude that storage, processing, and transmission of information are all governed by physical laws.

The primary goal of this thesis is to explore informational aspects of bipartite quantum interactions and their capabilities to  generate entanglement, an intriguing quantum phenomenon. A physical transformation of the state of a bipartite quantum system is effected by an underlying Hamiltonian. We are interested in the most general interaction such that a bipartite quantum system may be in contact with a bath, which is inaccessible to an observer. Bipartite quantum interaction refers to an underlying Hamiltonian governing the physical evolution of an open bipartite system. Such an interaction models a non-trivial, elementary interaction in a many-body quantum system. This study is necessary also from the aspect of applications as a bipartite quantum interaction depicts a non-trivial, elementary model of a quantum network involving two parties. 

In this age of technology and intelligence, it is of pertinent interest to also inspect information processing capabilities of quantum processes. In general, the laws of quantum theory allow us to push the abilities of processing and computing information beyond the limitations imposed by the classical information theory\footnote{There can be instances when there is no advantage provided by quantum theory over classical one \cite{Par70,WH82,Die82}.}. This provides us with ample opportunities to devise new information processing, communication, and computation protocols, e.g., quantum key distribution \cite{BB84, RennerThesis}, quantum teleportation \cite{BBC+93,BFK00}, quantum sensing \cite{DRC17}, quantum algorithms for computational speed-up \cite{NC00} etc. 

Broadly speaking, quantum processes are variously referred to as quantum processes, quantum gates, quantum channels, or quantum operations\footnote{It should be clarified that all quantum channels are quantum processes; however, the converse is not true in general. Meaning and subtleties will be clear in Chapters~\ref{ch:review} and \ref{ch:eg}.}.  Evolution of the state of a spin in an Ising model due to spin-spin interaction, a photon transmitted through an optical fiber, an electron interacting with electromagnetic field, or a quantum system decohering due to interaction with surrounding-- all these quantum processes are also describable as quantum channels~\cite{NC00,Riv11,Wbook17}. 

Note that a physical system itself can act as a quantum channel (process). This is because any physical system is capable of transforming the state of a quantum system. For examples, beamsplitter, spontaneous parametric down-conversion, optical fiber, etc.~are physical systems that transform the state of incident photons. When I point a laser beam on a screen during my presentation, the  state of incident photons would be different from that of the reflected ones. In this case, the screen is a quantum channel with photons as input and output system. This process is noisy as some photons may get absorbed by the screen and not all photons will be accessible to the audience who are observers here. As we will see later, these observations are crucial in devising communication protocols. 

Let us now briefly discuss some of the key ideas -- information content, quantum processes, entanglement, and information carriers -- that led to this work. In the end, we provide an outline of this thesis by briefly discussing contents of the following chapters. 

\section{What is information?}
Does a learner gain any knowledge if something obvious is stated to her? The answer is no. Obvious, isn't it? 

Knowledge requires information. Information is related to a meaningful piece of data that can be of use to a learner. A piece of data is meaningful to the learner only if she is observant enough. Information content is understood in reference to an observer; in other words, information is observer-dependent. With gain in information, there is reduction in uncertainty about an associated event. The less favorable the occurrence of an outcome of the event is, the more information is learned upon its occurrence; the more favorable an outcome is, the less information is learned upon its occurrence. As an instance, consider that I attempt to defend my thesis in front of PhD committee members. Let us safely assume that the chance of my graduation is high given the trend. My friends will be less surprised and learn less information when the PhD committee gives a passing mark. Whereas, they will be more surprised and learn more information if I fail. 

Claude E. Shannon was the first to give a qualitative description of what information is and how it can be transmitted amid noise~\cite{Sha48}. The abstract nature of the (classical) information theory introduced by Shannon provided a unified framework to understand seemingly distinct modes of communications and information processing over classical systems. The subject area dealing with information associated to classical systems is called classical information theory. He introduced the notion of a bit as a unit of information, whose physical representation is with a classical system. Roughly speaking, a bit is a binary valued quantity in which values are orthogonal to each other, meaning that they are distinguishable. Occurrence of one value excludes the possible occurrence of the other. Conventionally, a bit is represented by a binary number that can be ``0'' or ``1". 

Consider an event, which is an information source, described by a random variable $X$ with an associated discrete alphabet $\msc{X}$ of finite size, where $x\in\msc{X}$ is called a symbol. Each symbol $x$ corresponds to an outcome of the event. The given random variable $X$ can be represented by a classical register, and suppose that we only allow classical processes to occur. Let $p_X(x)$ be a probability distribution for describing the outcomes of the event. One measure of the surprise of the symbol $x\in\msc{X}$ is 
\begin{equation}\label{eq:surp}
i(x)\coloneqq -\log_2 p_X(x).
\end{equation}
The quantity $i(x)$ is also called surprisal or information content of the symbol $x$. Formula \eqref{eq:surp} is motivated by desirable properties that a quantifier of information content should have (see~\cite{book1991cover,Wbook17}). An observer will be infinitely surprised on an outcome of a symbol with no chance of occurrence, i.e., $p_X(x)=0$. Consider an example of an event that corresponds to tossing of an unbiased coin with two sides-- heads and tails. Upon a toss, if either heads or tails shows up then an observer learns one bit of information. The expected surprisal of an event is called its Shannon entropy:
\begin{equation}
S(X)\coloneqq \mathbb{E}_X\{-\log_2{p_X(X)}\}=-\sum_{x\in\msc{X}}p(x)\log_2 p(x). 
\end{equation}

Now suppose that information carriers are quantum systems. The state of a quantum system has its complete description. Quantum information is information associated to a quantum state. Analogous to the bit of classical information theory, a unit of quantum information is a qubit\footnote{Interested readers may refer to \cite{Wbook17,Wat15,H13book,NC00} for detailed discussions on quantum information.}. A qubit is a two-level quantum system that can be in a superposition state of $\ket{0}$ and $\ket{1}$, where $\{\ket{0},\ket{1}\}$ form an orthonormal basis. A measure of the average information content of a quantum system is its von Neumann entropy~\eqref{eq:vnent} \cite{Neu32}. The mathematical framework of quantum mechanics from an information-theoretic perspective is discussed in the next chapter. 

\section{Quantum interactions and processes}
Any physical process that operates on or transforms the state of a quantum system is called a quantum process. There is an underlying Hamiltonian that gives rise to such a process. In principle, the evolution of a closed quantum system $A'$ is always given by a unitary transformation and underlying Hamiltonian acts just on $A'$. However, it is difficult to isolate a system from its surrounding, which leads to an unavoidable interaction with the bath $E'$ (environment). It is often required to deal with a many-body quantum system, which is a difficult task. A simpler case of a many-body system is a two-body system. We call the Hamiltonian responsible for the interaction between constituent quantum systems in a many-body system as multipartite quantum interaction. Bipartite unitary evolution is the simplest physical transformation considered on a many-body system. 

Suppose that the system $A'$ is uncorrelated to the bath $E'$ before the action of Hamiltonian $\hat{H}_{A'E'}$. In general, $\hat{H}_{A'E'}=\hat{H}_{A'}+\hat{H}_{E'}+\hat{H}_{\tn{int}}$, where $\hat{H}_{A'}$ and $\hat{H}_{E'}$ denote Hamiltonians acting on individual systems $A'$ and $E'$, respectively, and $\hat{H}_{\tn{int}}$ denotes a Hamiltonian giving rise to a non-trivial interaction between $A'$ and $E'$. Even though the evolution of the composite system $A'+E'$ is unitary, as $A'+E'$ is closed, transformation of the state of $A'$ is noisy in general for an observer with no access to the bath. It is possible that after the action of the Hamiltonian, the degrees of freedom of the original system changes; i.e., an observer may have access to fewer or more degrees of freedom than $A'$, even though the partial degrees of freedom of the bath system are inaccessible. The total degrees of freedom of composite system remains the same, since it is closed. Such noisy physical operations are also called \textit{quantum channels} or (noisy) gates. A unitary operation is a particular quantum channel. Quantum channels are often called quantum ``gates" in the discussion of quantum computation. 

A Hamiltonian governing the evolution of an open two-body quantum system is called a bipartite quantum interaction. These interactions give rise to quantum processes that may change correlations between interacting systems if the interaction term present in the Hamiltonian is non-zero. As discussed previously, there is a need to inspect noisy processes involving two-body systems due to the unavoidable interaction between systems of interest and the bath. For an observer who has no access to the degrees of freedom of the bath, evolution process is noisy, i.e., non-unitary in general, and it is called a bipartite quantum channel. It should be noted that the degrees of freedom of the initial and final systems may change after the action of a bipartite channel.

In an information processing task, if pairs of input and output systems, $(A',A)$ and $(B',B)$ belong to two separate observers then a bipartite channel $\mc{N}_{A'B'\to AB}$ is called a bidirectional channel. A bidirectional channel $\mc{N}_{A'B'\to AB}$ provides the simplest form of a non-trivial network setting, as it allows for an interaction between two parties. Note that a point-to-point communication protocol, which is over a channel $\mc{N}_{A'\to B}$, is a special case of a communication protocol over a bidirectional channel $\mc{N}_{A'B'\to AB}$. A bidirectional quantum channel is an elementary, non-trivial example of a quantum network.   

\section{Entanglement}
We can sometimes be more certain about the joint state of a two-body quantum system $AB$ than we can be about any one of its individual parts, $A$ or $B$. These situations occur when a given two-body quantum system is in an entangled state~\cite{S35}. A particular kind of entangled state is a ``maximally entangled" state. As a system $AB$ has two parts, $A$ and $B$, measurements on its isolated parts $A$ and $B$ are physically possible. Such measurements are referred to as local measurements or operations. If $AB$ is maximally entangled and we perform any local measurement of $A$ or $B$, then we gain no information about the preparation of the state; instead we merely generate a random bit. A famous example of a maximally entangled state is the singlet state \cite{EPR35}. Entangled states are known to be a useful resource for different information processing tasks, e.g., quantum teleportation, unconditionally secure key distribution, randomness generation, etc. 

Bipartite quantum interactions can generate correlations between two separated systems in such a way that the physical description cannot be given by local realistic, hidden variable theories. Two quantum systems need to be entangled for them to exhibit such non-local correlations. While non-local correlations between two quantum systems implies that they are entangled, the converse is not true in general. The state of a bipartite quantum system is said to be entangled when it cannot be described as a convex combination of the uncorrelated states, i.e., product states of the constituent systems. While entanglement can be uniquely defined in the case of bipartite systems, it gets complicated for multipartite systems. This is related to the fact that there is no unique way to describe non-local correlations among many-body quantum systems~(see \cite{HSD15} and references therein). 

We need quantum interactions to generate quantum correlations between separated systems. We can use these correlations to harvest information and perform computation or communication tasks that may not be achieved with classical processes and systems. 

\section{Information carriers} 
There are multiple ways to communicate and store a given message. We may use print and digital media for storage and communication of information. We may rely on our brain to store information. Methods of information storage and transfer depend on the need and accessibility of media among the users. There are continuous efforts to increase information processing abilities of memory devices; we want to decrease the physical size of memory devices while increasing their storage capacity. As we are making advancements in information technology, there is great concern for privacy. At times, we need secure methods for processing or storage of information based on the ability of adversarial scrutiny. 

We have been making use of a variety of physical sources for communication and storage of messages. For example, we can use light or sound signals to broadcast important messages to commuters for traffic control. A light signal is better as a traffic signal, as it can be seen by commuters at an appropriate place without getting disturbed by noise of vehicle horns. Sound signals with loud volume are used to alert about the type of emergency vehicles passing through roads amid a crowd. We can also use several distinguishable properties of a physical medium to communicate in different ways. Consider a known and most used medium of communication -- sound. We pronounce words to communicate a message, change pitch or tone in order to indicate the level of urgency, and whisper for privacy against an eavesdropper.  

We can also use quantum processes as information carriers by encoding a message into a sequence of quantum processes. The efficacy of such method would depend on how well can we distinguish between quantum processes that are usable for encoding. As an example, let us consider memory devices. At a fundamental level, we can always understand the storage of messages in memory devices to be in the form of quantum channels. The mechanism of reading of any memory device involves the transmission of a probe system, which is in a known state, and inspecting the transformation of its state after an interaction with the system used for encryption of a message. In principle, any quantum processes can be used for information storage and communication. The choice of quantum processes depends on the kind of quantum systems accessible by a reader. In order to build up secure information processing, we need to come up with a strategy that hides encoded messages in quantum processes against an adversary. While an authorized user can access a hidden message, the probability of an adversary being able to access a hidden message should be negligible. 

Communication of messages over quantum processes, i.e., channels,  is a topic of wide interest \cite{Wbook17,H13book}; quantum states are used to encode a message and then are transmitted over a physical medium, which are modeled as some noisy quantum processes, e.g., bosonic Gaussian channels, depolarizing channels, etc. The primary idea behind such communication protocols depends on the discrimination of quantum states. These protocols allow for quick communication between distant parties. 

The choice of information carrier as quantum states or processes depends on the need and interest of communicating parties. While mathematically there is a correspondence called the Choi--Jamio\l{}kowski isomorphism between quantum processes and states, there are distinct differences from a physical perspective. If we want an information carrier to be robust against measurements and also to be long-lived, then we cannot encode a message in terms of quantum state in general. This is because quantum states are fragile against measurement and may decohere due to unavoidable interaction with surrounding. In such situations, we may instead want robust physical systems that can be described as quantum processes required for the task. Whereas, if we want our information carrier to be securely transmitted between points over a distance, we would encrypt the message in a quantum state and transmit it over a quantum channel. Subtle issues are discussed briefly in latter chapters. 

\section{Overview of thesis}
In this section, we briefly review the main results developed and discussed in this thesis. In Chapter~\ref{ch:review}, we discuss the mathematical formulation of quantum mechanics, definitions of information measures, and important lemmas required to derive the main results discussed in latter chapters.

In Chapter~\ref{ch:bqi}, the main focus is on deriving fundamental limitations on entangling abilities of bipartite quantum interactions \cite{DBW17}. These bounds also provide limitations on the information processing abilities of a  bipartite quantum network. Entangling abilities of bipartite quantum interactions are relevant in a number of different areas of quantum physics, reaching from fundamental areas such as quantum thermodynamics and the theory of quantum measurements to other applications such as quantum computation, quantum key distribution, and other information processing protocols. A particular aspect of the study of bipartite interactions is concerned with entanglement that can be created from such interactions. In this chapter, we discuss two basic building blocks of bipartite quantum protocols, namely, the generation of maximally entangled states and secret key via bipartite quantum interactions. In particular, we provide a non-trivial, efficiently computable upper bound on the positive-partial-transpose-assisted (PPT-assisted) entanglement distillation capacity of bidirectional quantum channel, thus addressing a question that has been open since 2002. In addition, we provide an upper bound on the private capacity of bidirectional quantum channels assisted by local operations and classical communication (LOCC). 

%In Chapter~\ref{ch:bqi2}, we study limitations on the abilities to generate distillable entanglement from bipartite quantum interactions. It is in continuation of discussion in Chapter~\ref{ch:bqi}. We consider an adaptive protocol for entanglement distillation using  finite number of bidirectional channel interleaved with the assistance of PPT-preserving channels. We derive an upper bound on entanglement distillation capacity of bidirectional quantum channel. We also present definition of bipartite entanglement binding channel and discuss some of its properties. 

In Chapter~\ref{ch:eg}, we discuss limitations on quantum dynamics based on entropy change \cite{DKSW18}. It is well known in the realm of quantum mechanics and information theory that the entropy is non-decreasing for the class of doubly stochastic physical processes, also called unital processes. However, in general, the entropy does not exhibit monotonic behavior. This has restricted the use of entropy change in characterizing evolution processes. Recently, a lower bound on the entropy change was provided in \cite{BDW16}. We explore the limit that this bound places on the physical evolution of a quantum system and discuss how these limits can be used as witnesses to characterize quantum dynamics. In particular, we derive a lower limit on the rate of entropy change for memoryless quantum dynamics, and we argue that it provides a witness of non-unitality; i.e., violation of the bound would be possible only if the dynamics are non-unital. This limit on the rate of entropy change leads to definitions of several witnesses for testing memory effects in quantum dynamics. Furthermore, from the aforementioned lower bound on entropy change, we obtain a measure of non-unitarity for unital evolutions.

In Chapter~\ref{ch:read}, we discuss a general protocol for quantum reading and discuss bounds on the reading capacities \cite{DW17}. Quantum reading refers to the task of reading out classical information stored in a read-only memory device. In any such protocol, the transmitter and receiver are in the same physical location, and the goal of such a protocol is to use these devices, coupled with a quantum strategy, to read out as much information as possible from a memory device, such as a CD or DVD. As a consequence of the physical setup of quantum reading, the most natural and general definition for quantum reading capacity should allow for an adaptive operation after each call to the channel, and this is how quantum reading capacity is defined in this chapter. We  also derive several bounds on quantum reading capacity, and we introduce an environment-parametrized memory cell, delivering second-order and strong converse bounds for its quantum reading capacity. We calculate the quantum reading capacities for some exemplary memory cells, including a thermal memory cell, a qudit erasure memory cell, and a qudit depolarizing memory cell. We finally discuss an explicit example to illustrate the advantage of using an adaptive strategy in the context of zero-error quantum reading capacity.

In Chapter~\ref{ch:priv-read}, we introduce a protocol for the private reading of memory devices against an eavesdropper \cite{DBW17}. We can use this protocol for secret key agreement between two authorized parties where secret key is encoded in a memory device. The goal is to protect from the leakage of the secret key to an eavesdropper who is spying on the reader. We notice that private reading can be understood as a particular kind of secret-key-agreement protocol that employs a particular kind of bipartite interaction. We make use of this observation to derive upper bounds on private reading capacities of memory devices. 

In Chapter~\ref{ch:sr}, we introduce protocols for the secure retrieval of digital information stored in a memory device under different adversarial situations. We refer to such protocols as secure reading. Information in memory devices is encoded in terms of quantum channels selected from a particular set called a memory cell. We allow the encoder and the intended reader to share secret keys prior to the reading task is carried out. We also consider a toy model in which a message is encoded in unitary gates and show the advantage of an authorized reader who has key against an unauthorized user with no key. For more general secure reading protocol, we consider a passive adversary who has complete access to the environment, and a semi-passive adversary who can access the memory device. To illustrate these protocols, we discuss examples for the secure reading of memory devices encoded with a binary memory cell consisting of amplitude damping channels or depolarizing channels. We also briefly discuss application of a secure reading protocol for a threat level identification scheme, which is inspired by IFF: identification, friend or foe.

%% file: review.tex
\chapter{Quantum Systems, Physical Processes, and Information Measures}\label{ch:review}
In this chapter, we take an information-theoretic approach to review some of the basic concepts of quantum mechanics. We start by introducing standard notations, definitions, and important lemmas that are required for the derivation and discussions of results introduced in latter sections and chapters. We discuss the mathematical representation of the state of a quantum system and physical quantum processes, particular set of states, structure of physical processes obeying certain symmetries, measures to quantify information content in a quantum system, and the notion of bipartite entanglement measures. Finally, in Section~\ref{sec:app-ent-mes}\footnote{Section~\ref{sec:app-ent-mes} is entirely based on an unpublished work done in collaboration with Mark~M.~Wilde.}, we present results on the approximate normalization of two different entanglement measures-- entropy of entanglement\cite{Sre93} and squashed entanglement\cite{CW04}. 

\section{Bounded operators and super-operators}
In this review, the discussion is focused on finite--dimensional Hilbert spaces. See Section~\ref{sec:notation-eg} for discussion on  infinite--dimensional Hilbert spaces.

Let $\msc{B}(\mc{H})$ denote the algebra of bounded operators acting on a Hilbert space $\mc{H}$, with $\bm{1}_{\mc{H}}\in\msc{B}(\mc{H})$ denoting the identity operator and $\id$ denoting the identity super-operator\footnote{A super-operator is a linear map that acts on an operator.}. Let $\dim(\mc{H})$ denote the dimension of Hilbert space $\mc{H}$ \footnote{Note that $\dim(\mc{H})$ is equal to $+\infty$ in the case that $\mc{H}$ is a separable, infinite-dimensional Hilbert space.}. $\msc{B}(\mc{H})$ also denotes the set of all trace-class operators acting on the Hilbert space $\mc{H}$, since $\mc{H}$ is finite-dimensional. 

The Hilbert space of a system $A$ is denoted by $\mc{H}_A$ and the Hilbert space of a composite system consisting of systems $A$ and $B$ is given by the tensor product $\mc{H}_{AB}\coloneqq \mc{H}\otimes\mc{H}_B$. The notation $A^n:= A_1A_2\cdots A_n$ indicates a composite system consisting of $n$ subsystems $A$, each of which is isomorphic to the Hilbert space $\mc{H}_A$; i.e., for all $i\in [n]$, $A_i\simeq A$, where $[n]\coloneqq \{1,2,\ldots, n\}$ for $n\in\bb{N}$. Let us denote the set of all orthonormal bases of the Hilbert space $\mc{H}$ as $\ONB(\mc{H})$.

The subset of $\mc{B}(\mc{H})$ containing all positive semi-definite operators is denoted by $\msc{B}_+(\mc{H})$. We write $C\geq 0$ to indicate that $C\in\msc{B}_+(\mc{H})$, and $C\geq D$ indicates $C-D\in\msc{B}_+(\mc{H})$. 

A super-operator $\mc{M}_{A\to B}$ denotes a linear map $\mc{M}:\msc{B}(\mc{H}_A)\to\msc{B}(\mc{H}_B)$ that maps elements in $\msc{B}(\mc{H}_A)$ to elements in $\msc{B}(\mc{H}_B)$. The adjoint $\mc{M}^\dagger:\msc{B}(\mc{H}_B)\to\msc{B}(\mc{H}_A)$ of a linear map $\mc{M}:\msc{B}(\mc{H}_A)\to\msc{B}(\mc{H}_B)$ is the unique linear map that satisfies
	\begin{equation}
	\<Y_B,\mc{M}(X_A)\>=\<\mc{M}^\dag(Y_B),X_A\>, \qquad
	\forall\ X_A\in\msc{B}(\mc{H}_A), Y_B\in\msc{B}(\mc{H}_B)
	\end{equation}
	where $\<C,D\>=\Tr\{C^\dag D\}$ is the Hilbert-Schmidt inner product. An isometry $U:\mc{H}\to\mc{H}'$ is a
linear map such that $U^{\dag}U=\bm{1}_{\mathcal{H}}$ and $UU^{\dag}=\Pi_{\mc{H}'}$, where $\Pi_{\mc{H}'}$ is a projection onto a subspace of the Hilbert space $\mc{H}'$.

A super-operator $\mc{M}_{A\to B}:\msc{B}(\mc{H}_A)\rightarrow\msc{B}(\mc{H}_B)$ is called positive if it maps elements of $\msc{B}_{+}(\mc{H}_A)$ to elements of $\msc{B}_{+}(\mc{H}_B)$ and completely positive if $\id_R\otimes\mc{M}_{A\to B}$ is positive for a Hilbert space $\mc{H}_R$ of any dimension, where $\id_R$ is the identity super-operator acting on $\msc{B}(\mc{H}_R)$. A positive map $\mc{M}_{A\to B}:\msc{B}_{+}(\mc{H}_A)\to\msc{B}_{+}(\mc{H}_B)$ is called trace non-increasing if $\Tr\{\mc{M}_{A\to B}(\sigma_A)\}\leq \Tr\{\sigma_A\}$ for all $\sigma_A\in\msc{B}_+(\mc{H}_A)$, and it is called trace-preserving if $\Tr\{\mc{M}_{A\to B}(\sigma_A)\}=\Tr\{\sigma_A\}$ for all $\sigma_A\in\msc{B}_+(\mc{H}_A)$. When confusion does not arise, we omit identity operators in expressions involving multiple tensor factors, so that, for example, $\mc{M}_{A\to B}(\rho_{RA})$ is understood to mean $\id_{R}\otimes\mc{M}_{A\to B}(\rho_{RA})$.
	
A linear map $\mc{M}_{A\to B}:\msc{B}(\mc{H}_A)\rightarrow\msc{B}(\mc{H}_B)$ is called sub-unital if $\mc{M}_{A\to B}(\bm{1}_{A})\leq \bm{1}_B$, unital if $\mc{M}_{A\to B}(\bm{1}_{A})= \bm{1}_B$, and super-unital if $\mc{M}_{A\to B}(\bm{1}_{A})>\bm{1}_B$. Note that it is possible for a linear map to be neither unital, sub-unital, nor super-unital. A positive trace-preserving map can be sub-unital only if the dimension of the output Hilbert space is greater than or equal to the dimension of the input Hilbert space. A positive trace-preserving map can be super-unital only if the dimension of the output Hilbert space is less than the dimension of the input Hilbert space. Positive trace-preserving maps between two finite-dimensional Hilbert spaces of the same dimension that are sub-unital are also unital.

\section{Operator-valued functions and norms}
Let $A$ be a self-adjoint operator acting on a Hilbert space $\mc{H}$. The support $\supp(A)$ of $A$ is the span of the eigenvectors of $A$ corresponding to its non-zero eigenvalues, and the kernel of $A$ is the span of the eigenvectors of $A$ corresponding to its zero eigenvalues. There exists a spectral decomposition of $A$:
	\begin{equation}
		A=\sum_k \lambda_k \op{k},
	\end{equation}
	where $\{\lambda_k\}_k$ are the eigenvalues corresponding to an orthonormal basis of eigenvectors $\{\ket{k}\}_k$ of $A$. The projection $\Pi(A)$ onto $\supp(A)$ is then
	\begin{equation}
		\Pi(A)=\sum_{k:\lambda_k\neq 0}\op{k}.
	\end{equation}
	Let $\text{rank}(A)$ denote the rank of $A$. If $A$ is positive definite, i.e., $A>0$, then $\text{rank}(A)=\dim(\mc{H})$, $\Pi(A)=\bm{1}_{\mc{H}}$, and we say that the rank of $A$ is full. If $f$ is a real-valued function with domain $\text{Dom}(f)$, then $f(A)$ is defined as
	\begin{equation}
		f(A)=\sum_{k:\lambda_k\in\text{Dom}(f)}f(\lambda_k)\op{k}.
	\end{equation}
	
The Schatten $p$-norm of an operator $A\in\msc{B}(\mc{H})$ is defined as  
	\begin{equation}
		\Vert A\Vert_p\equiv \left( \Tr\left\{|A|^p\right\}\right)^\frac{1}{p},\label{eq:sch-norm}
	\end{equation}
	where $|A|\equiv \sqrt{A^\dagger A}$ and $p\in[1,\infty)$. If $\{\sigma_i(A)\}_i$ are the singular values of $A$, then 
	\begin{equation}
		\Vert A\Vert_p=\left[\sum_i\sigma_i(A)^p\right]^\frac{1}{p}.
	\end{equation}
	$\norm{A}_{\infty}\coloneqq\lim_{p\to\infty}\norm{A}_p$ is the largest singular value of $A$. Let $\msc{B}_p(\mc{H})$ be the subset of $\msc{B}(\mc{H})$ consisting of operators with finite Schatten $p$-norm. The Schatten $p$-norms are unitarily invariant norms.

	\begin{lemma}[H\"{o}lder's inequality~\cite{Rog88,Hol89,Bha97}]\label{thm:sch-norm}
		For all $A\in\msc{B}_p(\mc{H})$, $B\in\msc{B}_q(\mc{H})$, and $p,q\in[1,\infty]$ such that $\frac{1}{p}+\frac{1}{q}=1$, it holds that
	\begin{equation}
		\vert\<A,B\>\vert=\left\vert \Tr\left\{A^\dagger B\right\}\right\vert\leq \Vert A\Vert_p\Vert B\Vert_q.
	\end{equation}
	\end{lemma}
	
The following lemma can be found in \cite[Corollary~5.2]{W12notes}.
	\begin{lemma}\label{thm:log-con}
		Let $\mc{M}:\mc{B}_+(\mc{H}_A)\to\mc{B}_+(\mc{H}_B)$ be a linear, positive, and sub-unital map. Then, for all $\sigma_A\in\mc{B}_{+}(\mc{H}_A)$ it holds that
		\begin{equation}
			\mc{M}_{A\to B}(\log(\sigma_A))\leq\log(\mc{M}_{A\to B}(\sigma_A)). 
		\end{equation}
	\end{lemma}
	
\subsection{Derivatives of operator-valued functions}\label{app-derivative}

	Here we recall \cite[Theorem~V.3.3]{Bha97}.

	If $f$ is a continuously differentiable function on an open neighbourhood of the spectrum of some self-adjoint operator $A$, then its derivative $Df(A)$ at $A$ is a linear superoperator and its action on an operator $H$ is given by
	\begin{equation}
		Df(A)(H)=\sum_{\lambda,\eta}f^{[1]}(\lambda,\eta)P_A(\lambda)HP_A(\eta),
	\end{equation}
	where $A=\sum_\lambda \lambda P_A(\lambda)$ is the spectral decomposition of $A$ and $f^{[1]}$ is the first divided difference function. 
	
	If $t\mapsto A(t)\in\msc{B}_+(\mc{H})$ is a continuously differentiable function on an open interval in $\mathbb{R}$, with derivative $A'\coloneqq\frac{\d A}{\d t}$, then
	\begin{equation}
		 f'(A(t))\coloneqq \frac{\d}{\d t} f(A(t))= Df(A)(A'(t))=\sum_{\lambda,\eta} f^{[1]}(\lambda,\eta)P_{A(t)}(\lambda)A'(t)P_{A(t)}(\eta).\label{eq:deriv-superop}
	\end{equation}
	In particular, \eqref{eq:deriv-superop} implies the following:
	\begin{align}
		\frac{\d}{\d t}\Tr\{f(A(t))\} & =\Tr\{f'(A(t))A'(t)\},\label{eq:app-trace-der-1}\\
		\Tr\left\{B(t) f'(A(t))\right\} & =\Tr\{B(t)f'(A(t))A'(t)\},\label{eq:app-trace-der-2}
	\end{align}
	where $B(t)$ is assumed to commute with $A(t)$.	

\section{Quantum states and channels}
The state of a quantum system $A$ is represented by a density operator $\rho_A$, which is a positive semi-definite operator with unit trace. Let $\msc{D}(\mc{H}_A)$ denote the set of density operators, i.e., all elements $\rho_A\in \msc{B}_+(\mc{H}_A)$ such that $\Tr\{\rho_A\}=1$. The density operator of a composite system $AB$ is defined as $\rho_{AB}\in \msc{D}(\mc{H}_{AB})$, and the partial trace over $A$ gives the reduced density operator for the system $B$, i.e., $\Tr_A\{\rho_{AB}\}=\rho_B$ such that $\rho_B\in \msc{D}(\mc{H}_B)$. A pure state $\psi_A$ of a system $A$ is a rank-one density operator in $\msc{D}(\mc{H}_A)$, and we write it as $\psi_A=\op{\psi}_A$
for a unit vector $|\psi\>_A\in\mc{H}_A$. A purification of a density operator $\rho_A$ is a pure state $\psi^\rho_{AE}$
such that $\Tr_E\{\psi^\rho_{AE}\}=\rho_A$, where $E$ is called the purifying system. The maximally mixed state is denoted by
$\pi_A := \bm{1}_A / |A| \in\mc{D}\(\mc{H}_A\)$.

It is known that there exists a Schmidt decomposition for any bipartite quantum system in a pure state. It means that any pure state $\psi_{AB}\in\msc{D}(\mc{H}_{AB})$ can be expressed as
\begin{equation}
|\psi\>_{AB}=\sum_{i=0}^{d-1}\sqrt{p_i}|i\>_{A}|i\>_{B},
\end{equation}
such that $\{|i\>_{A}\}_i\in\ONB(\mc{H}_A)$, $\{|i\>_{B}\}_i\in\ONB(\mc{H}_{B})$, $\sum_{i=0}^{d-1}p_i=1$, and for all $i:\ 0\leq p_i\leq 1$, where $d=\min\{|A|,|B|\}$.
	
 Let $U^{\hat{H}}_{A'E'\to AE}$ be a unitary associated to a Hamiltonian $\hat{H}$, which governs the underlying interaction between an input subsystem $A'$ and a bath $E'$, to produce an output subsystem $A$ for the observer and $E$ for the bath. In general, the individual input systems $A'$ and $E'$ and the output systems $A$ and $E$ can have different dimensions. At an initial time, in the absence of an interaction Hamiltonian $\hat{H}$, the bath is in a fixed state $\tau_{E'}$ and the system $A'$ has no correlation with the bath; i.e., the state of the composite system $A'E'$ is of the form $\omega_{A'}\otimes \tau_{E'}$, where $\omega_{A'E'}$ is the joint state of the systems $A'$ and $E'$. Under the action of the interaction Hamiltonian $\hat{H}$, the state of the composite system transforms as
	\begin{equation}\label{eq2:bi-u}
		\rho_{AE}=U^{\hat{H}}(\omega_{A'}\otimes \tau_{E'})(U^{\hat{H}})^\dag.
	\end{equation}
	In the above interaction process, since the system $E$ in \eqref{eq2:bi-u} is inaccessible, the evolution of the system of interest is noisy in general. The noisy evolution of the system $A'$ under the action of the interaction Hamiltonian $\hat{H}$ is represented by a completely positive, trace-preserving (CPTP) map \cite{Sti55}, called a quantum channel:
	\begin{equation}\label{eq:unipartite-int}
		\mc{M}_{A'\to A}(\omega_{A'})=\Tr_{E}\{U^{\hat{H}}(\omega_{A'}\otimes \tau_{E'})(U^{\hat{H}})^\dag\},
	\end{equation}
	where system $E$ represents inaccessible degrees of freedom. In particular, when the Hamiltonian $\hat{H}$ is such that there is no interaction between the system $A'$ and the bath $E'$, and $A'\simeq A$, then $\mc{M}$ corresponds to a unitary evolution, i.e., $\mc{M}(\cdot)=\mc{U}^{\hat{H}}(\cdot)\coloneqq U^{\hat{H}}_{A'\to A}(\cdot)(U^{\hat{H}}_{A'\to A})^\dag$. The weakly complementary channel $\widehat{\mc{M}}_{A'\to E}$ is given by
	\begin{equation}
		\widehat{\mc{M}}_{A'\to E}(\omega_{A'})=\Tr_{B}\{U^{\hat{H}}(\omega_{A'}\otimes \tau_{E'})(U^{\hat{H}})^\dag\}.
	\end{equation}
If we suppose that the state $\tau_{E'}$ of a bath system $E'$ is pure, then $\widehat{\mc{M}}_{A'\to E}$ is called the complementary channel of $\mc{M}_{A'\to E}$.  

A completely positive, trace non-increasing map is called a quantum sub-operation.

A CPTP map $\mc{N}_{A'B'\to AB}:\msc{B}_+(\mc{H}_{A'}\otimes\mc{H}_{B'})\to \msc{B}_+(\mc{H}_{A}\otimes\mc{H}_{B})$ is called a bipartite channel. A bipartite channel $\mc{N}_{A'B'\to AB}$ is also called bidirectional channel in the setting of communication protocols when  pairs $(A',A)$ and $(B',B)$ of quantum systems are held by two spatially separated parties.  

A memory cell $\{\mc{M}^x\}_{x\in\msc{X}}$ is defined to be a set of quantum channels $\mc{M}^x$, for all $x\in\msc{X}$, where $\msc{X}$ is an alphabet, and $\mc{M}^{x}:\msc{B}_+(\mc{H}_{A'})\to\msc{B}_+(\mc{H}_{A})$ for all $x\in\msc{X}$.

A quantum instrument is a collection $\{\mathcal{M}^{x}_{A'\to A}\}_{x\in\msc{X}}$ of quantum sub-operations, such that the sum map $\sum_{x}\mathcal{M}^{x}$ is a quantum channel. The action of a quantum instrument on an input operator $\rho_{A'}$ can be described in terms of the following quantum channel:
	\begin{equation}\label{eq:q-inst}
		\rho_{A'} \mapsto \sum_{x\in\msc{X}}\mathcal{M}_{A'\to A}^{x}(\rho_A)\otimes\op{x}_X,
	\end{equation}
	where $\{|x\rangle_X\}_x\in\ONB(\mc{H}_X)$ and $X$ denotes a (classical) register that stores the classical output of the instrument.

The Choi--Jamio\l{}kowski  isomorphism represents a well known duality between channels and states. Let $\mc{M}_{A'\to A}$ be a quantum channel, and let $\left|\Upsilon\right>_{R:A'}$ denote the following maximally entangled vector:
\begin{equation}
|\Upsilon\>_{R:A'}\coloneqq \sum_{i}|i\>_{R}|i\>_{A'},
\end{equation}
where $|R|=|A'|$, and $\{|i\>_R\}_i\in\ONB(\mc{H}_{R})$ and $\{|i\>_{A'}\}_i\in\ONB(\mc{H}_{A'})$ are fixed orthonormal bases, and $R:A'$ denotes a bipartite cut. Let us extend this notation to multiple parties with a given bipartite cut as
\begin{equation}
|\Upsilon\>_{R_AR_B:A'B'}\coloneqq |\Upsilon\>_{R_A':A'}\otimes |\Upsilon\>_{R_B:B'}.
\end{equation}
A maximally entangled state $\Phi_{RA'}$ is defined for a bipartite system $R:A'$ as
\begin{equation}
\Phi_{RA'}=\frac{1}{|A'|}\op{\Upsilon}_{RA'}.
\end{equation}
 The Choi operator for a channel $\mc{M}_{A'\to A}$ is defined as
\begin{equation}
J^\mc{M}_{RA}=(\id_{R}\otimes\mc{M}_{A'\to A})\(|\Upsilon\>\<\Upsilon|_{RA'}\),
\end{equation}
where $\id_R$ denotes the identity map on $R$. For $A'\simeq A$, the following identity holds
\begin{equation}\label{eq:choi-sim}
\<\Upsilon|(\rho_{R_AA'}\otimes J^\mc{M}_{RA})|\Upsilon\>_{A':R}=\mc{M}_{A'\to A}(\rho_{R_AA'}),
\end{equation}
where $A'\simeq A$. The above identity can be understood in terms of a post-selected variant \cite{HM04} of the quantum teleportation protocol \cite{BBC+93}. Another identity that holds is
\begin{equation}
\<\Upsilon| (Q_{R_A R}\otimes \bm{1}_A)
|\Upsilon\>_{R:A}=\Tr_R\{Q_{R_AR}\},
\end{equation}
for an operator $Q_{R_AR}\in \msc{B}(\mc{H}_{R_A}\otimes\mc{H}_R)$.

\subsection{Separable and PPT: states and channels}
For a fixed basis $\{|i\>_B\}_{i\in\msc{I}}\in\ONB(\mc{H}_B)$, the partial transpose $\T_B$ on the composite system $AB$ is the following map:
\begin{equation}
\(\id_A\otimes \T_B\)(Q_{AB})=\sum_{i,j\in\msc{I}}\(\bm{1}_A\otimes |i\>\<j|_B\) Q_{AB}\( \bm{1}_A\otimes |i\>\<j|_B\),\, \label{eq:PT-1}
\end{equation}
where $Q_{AB}\in\msc{B}(\mc{H}_{A}\otimes\mc{H}_{B})$. Note that the partial transpose is self-adjoint, i.e., $\T_B=\T^\dag_B$ and is also involutory:
\begin{equation}
\T_B\circ\T_B=\bm{1}_B.
\end{equation} 
The following identity also holds:
\begin{equation}
\T_{R}(\op{\Upsilon}_{RA})=\T_{A}(\op{\Upsilon}_{RA}).
\label{eq:PT-last}
\end{equation} 

Let $\SEP(A\!:\!B)$ denote the set of all separable states $\sigma_{AB}\in\msc{D}(\mc{H}_A\otimes\mc{H}_B)$, which are states that can be written as
\begin{equation}
\sigma_{AB}=\sum_{x\in\msc{X}}p_X(x)\omega^x_A\otimes\tau^x_B,
\end{equation}
where $p_X(x)$ denotes a probability distribution corresponding to a random variable $X$ associated with an alphabet $\msc{X}$, $\omega^x_A \in \msc{D}(\mc{H}_A)$, and $\tau^x_B\in\msc{D}(\mc{H}_B)$ for all $x\in\msc{X}$. This set
is closed under the action of the partial transpose maps $\T_A$ and $\T_B$ \cite{HHH96,Per96}. Generalizing the set of separable states, we can define the set $\PPT (A\!:\!B)$ of all bipartite states $\rho_{AB}\in\msc{D}(\mc{H}_A\otimes\mc{H}_B)$ that remain positive after the action of the partial transpose $\T_B$. A state $\rho_{AB}\in\PPT(A\!:\!B)$ is also called a PPT (positive under partial transpose) state. If a state is not PPT then it is called NPT (non-positive under partial transpose). We can define an even more general set of positive semi-definite operators \cite{AdMVW02} as follows:
\begin{equation}
\PPT'(A\!:\!B)\coloneqq \{\sigma_{AB}:\ \sigma_{AB}\geq 0\land \norm{\T_B(\sigma_{AB})}_1\leq 1\}. 
\end{equation} 
We then have the containments $\SEP\subsetneq \PPT\subsetneq \PPT' $. 

\begin{lemma}[\kern-0.35em\cite{Rai99}]\label{thm:rai99}
For any $\sigma_{AB}\in\PPT'(A:B)$, the following inequality holds 
\begin{equation}
\Tr\{\Phi_{AB}\sigma_{AB}\}\leq \frac{1}{M},
\end{equation}
where $\Phi_{AB}$ is a maximally entangled state of Schmidt rank $M$, i.e., $|A|=|B|=M$. 
\end{lemma}

A bipartite quantum channel $\mc{P}_{A'B'\to AB}$ is a PPT-preserving channel if the map $\T_{B}\circ\mc{P}_{A'B'\to AB}\circ\T_{B'}$ is a quantum channel \cite{Rai99,Rai01}. A bipartite quantum channel $\mc{P}_{A'B'\to AB}$ is PPT-preserving if and only if its Choi state is a PPT state \cite{Rai01}, i.e., $\frac{J^{\mc{P}}_{R_AR_B:AB}}{ |R_A R_B|}\in \PPT(R_A A\!:\!BR_B)$, where
\begin{equation}
\frac{J^{\mc{P}}_{R_AR_B:AB}}{ |R_A R_B|} =  \mc{P}_{A'B'\to AB}(\Phi_{R_AA'}\otimes\Phi_{B'R_B}).
\end{equation}

A bipartite quantum channel $\mc{S}_{A'B'\to AB}$ is a separable channel if and only if its Choi state is a separable state \cite{CDKL01}, i.e, $\frac{J^{\mc{S}}_{R_AR_B:AB}}{ |R_A R_B|}\in \SEP(R_A A\!:\!BR_B)$, where
\begin{equation}
\frac{J^{\mc{S}}_{R_AR_B:AB}}{ |R_A R_B|} =  \mc{S}_{A'B'\to AB}(\Phi_{R_AA'}\otimes\Phi_{B'R_B}).
\end{equation}

A 1W-LOCC (one-way local operations and classical communication) channel is a separable super-operator
\begin{equation}
\mathcal{L}^{\to}_{A'B'\rightarrow AB}=\sum_{y\in\msc{Y}}\mathcal{E}
_{A'\rightarrow A}^{y}\otimes\mathcal{F}_{B'\rightarrow B}
^{y},\label{eq:LOCC-channel}
\end{equation}
where $\msc{Y}$ is an alphabet, $\{\mathcal{E}_{A\rightarrow
A^{\prime}}^{y}\}_{y\in\msc{Y}}$ is a set of CP maps such that the sum map $\sum
_{y\in\msc{Y}}\mathcal{E}_{A\rightarrow A^{\prime}}^{y}$ is trace preserving, while
$\{\mathcal{F}_{B\rightarrow B^{\prime}}^{y}\}_{y\in\msc{Y}}$ is a set of quantum channels.  Whereas, an LOCC (local communication and classical operations) channels $\mathcal{L}_{AB\rightarrow A^{\prime}B^{\prime}}$ takes the
form in \eqref{eq:LOCC-channel} such that $\{\mathcal{E}_{A\rightarrow A^{\prime}}^{y}\}_{y\in\msc{Y}}$ and $\{\mathcal{F}
_{B\rightarrow B^{\prime}}^{y}\}_{y\in\msc{Y}}$ are sets of completely positive (CP)
maps such that $\mathcal{L}_{AB\rightarrow A^{\prime}B^{\prime}}$ is trace
preserving. Thus, the LOCC channels are also separable super-operators, but the converse is not true. Note that any 1W-LOCC channel is also an LOCC channel and all LOCC channels are PPT-preserving. 

\section{Channels with symmetry}\label{sec:symmetry}
Consider a finite group $\msc{G}$ of size $|G|$. For every $g\in \msc{G}$, let $g\to U_A(g)$ and $g\to V_B(g)$ be projective unitary representations of $g$ acting on the input space $\mc{H}_A$ and the output space $\mc{H}_B$ of a quantum channel $\mc{M}_{A\to B}$, respectively. A quantum channel $\mc{M}_{A\to B}$ is covariant with respect to these representations if the following relation is satisfied \cite{Hol02,H13book}:
\begin{equation}\label{eq:cov-condition}
\forall \rho_A\in\msc{D}(\mc{H}_A) \ \text{and} \ \forall g\in \msc{G}, \ \mc{M}_{A\to B}\!\(U_A(g)\rho_A U_A^\dagger(g)\)=V_B(g)\mc{M}_{A\to B}(\rho_A)V_B^\dagger(g).
\end{equation}

\begin{definition}[Covariant channel \cite{H13book}]\label{def:covariant}
A quantum channel is covariant if it is covariant with respect to a group $\msc{G}$ which has a representation $U(g)$, for all $g\in \msc{G}$, on $\mc{H}_A$ that is a unitary one-design; i.e., the map  $\frac{1}{|G|}\sum_{g\in \msc{G}}U(g)(\cdot)U^\dagger(g)$ always outputs the maximally mixed state for all input states. 
\end{definition}

For an isometric channel $\mc{U}^\mc{M}_{A\to BE}$ extending the covariant channel $\mc{M}_{A\to B}$ defined above, there exists a unitary representation $W_E(g)$ acting on the environment Hilbert space $\mc{H}_E$ \cite{H13book}, such that
for all $g\in \msc{G}$,
\begin{equation}\label{eq:iso-covariant}
\mc{U}^\mc{M}_{A\to BE}\!\({U_A(g)\rho_AU^\dagger_A(g)}\)=\(V_B(g)\otimes W_E(g)\)\(\mc{U}^\mc{M}_{A\to BE}\(\rho_A\)\)\(V^\dagger_B(g)\otimes W^\dagger_E(g)\).
\end{equation}
We can restate this as the following lemma:
\begin{lemma}[\kern-0.35em\cite{H13book}]\label{thm:cov-hol}
Suppose that a channel $\mathcal{M}_{A\rightarrow B}$ is covariant with respect to a group $\msc{G}$. 
For an isometric extension $U_{A\rightarrow BE}^{\mathcal{M}}$ of
$\mathcal{M}_{A\rightarrow B}$, there is a set of unitaries $\{W_{E}^{g}\}_{g\in \msc{G}}$ such
that the following covariance holds for all $g \in \msc{G}$:
\begin{equation}
U_{A\rightarrow BE}^{\mathcal{M}}U_{A}^{g}=\left(  V_{B}^{g}\otimes W_{E}
^{g}\right)  U_{A\rightarrow BE}^{\mathcal{M}}.
\end{equation}
\end{lemma}

\begin{proof}
For convenience, we discuss a proof of this lemma presented in \cite[Appendix A]{DBW17}.

Given is a group $\msc{G}$  and a quantum channel $\mathcal{M}_{A\rightarrow B}$ that is covariant in the
following sense:
\begin{equation}
\mathcal{M}_{A\rightarrow B}(U_{A}^{g}\rho_{A}U_{A}^{g\dag})=V_{B}%
^{g}\mathcal{M}_{A\rightarrow B}(\rho_{A})V_{B}^{g\dag},\label{eq:cov-sym}
\end{equation}
for a set of unitaries $\{U_{A}^{g}\}_{g\in \msc{G}}$ and $\{ V_{B}^{g} \}_{g \in \msc{G}}$.

Let a Kraus representation of $\mathcal{M}_{A\rightarrow B}$ be given as%
\begin{equation}
\mathcal{M}_{A\rightarrow B}(\rho_{A})=\sum_{j}L^{j}\rho_{A}L^{j\dag}.
\end{equation}
We can rewrite \eqref{eq:cov-sym} as%
\begin{equation}
V_{B}^{g\dag}\mathcal{M}_{A\rightarrow B}(U_{A}^{g}\rho_{A}U_{A}^{g\dag}%
)V_{B}^{g}=\mathcal{M}_{A\rightarrow B}(\rho_{A}),
\end{equation}
which means that for all $g$, the following equality holds%
\begin{equation}
\sum_{j}L^{j}\rho_{A}L^{j\dag}=\sum_{j}V_{B}^{g\dag}L^{j}U_{A}^{g}\rho
_{A}\left(  V_{B}^{g\dag}L^{j}U_{A}^{g}\right)  ^{\dag}.
\end{equation}
Thus, the channel has two different Kraus representations $\{L^{j}\}_{j}$ and
$\{V_{B}^{g\dag}L^{j}U_{A}^{g}\}_{j}$, and these are necessarily related by a
unitary with matrix elements $w_{jk}^{g}$ \cite{Wbook17,Wat15}:
\begin{equation}
V_{B}^{g\dag}L^{j}U_{A}^{g}=\sum_{k}w_{jk}^{g}L^{k}.
\end{equation}
A canonical isometric extension $U_{A\rightarrow BE}^{\mathcal{M}}$ of
$\mathcal{M}_{A\rightarrow B}$ is given as%
\begin{equation}
U_{A\rightarrow BE}^{\mathcal{M}}=\sum_{j}L^{j}\otimes|j\rangle_{E},
\end{equation}
where $\{|j\rangle_{E}\}_j$ is an orthonormal basis.
Defining $W_{E}^{g}$ as the following unitary%
\begin{equation}
W_{E}^{g}|k\rangle_{E}=\sum_{j}w_{jk}^{g}|j\rangle_{E},
\end{equation}
where the states $|k\rangle_{E}$ are chosen from $\{|j\rangle_{E}\}_j$,
consider that%
\begin{align}
U_{A\rightarrow BE}^{\mathcal{M}}U_{A}^{g}  & =\sum_{j}L^{j}U_{A}^{g}%
\otimes|j\rangle_{E} =\sum_{j}V_{B}^{g}V_{B}^{g\dag}L^{j}U_{A}^{g}\otimes|j\rangle_{E} =\sum_{j}V_{B}^{g}\left[  \sum_{k}w_{jk}^{g}L^{k}\right]  \otimes
|j\rangle_{E}\nonumber\\
& =V_{B}^{g}\sum_{k}L^{k}\otimes\sum_{j}w_{jk}^{g}|j\rangle_{E}
 =V_{B}^{g}\sum_{k}L^{k}\otimes W_{E}^{g}|k\rangle_{E}
=\left(  V_{B}^{g}\otimes W_{E}^{g}\right)  U_{A\rightarrow BE}%
^{\mathcal{M}}.
\end{align}
This concludes the proof.
\end{proof}

\begin{definition}[Teleportation-simulable \cite{BDSW96,HHH99}]\label{def:tel-sim}
A channel $\mc{M}_{A\to B}$ is teleportation-simulable with an associated resource state if for all $\rho_{A}\in\msc{D}\(\mc{H}_{A}\)$ there exists a resource state $\omega_{R_AB}\in\msc{D}\(\mc{H}_{R_AB}\)$ such that 
\begin{equation}
\mc{M}_{A\to B}\(\rho_A\)=\mc{L}_{R_AA B\to B}\(\rho_{A}\otimes\omega_{R_AB}\),
\label{eq:TP-simul}
\end{equation}
where $\mc{L}_{R_AAB\to B}$ is an LOCC channel acting on $R_AA\!:\!B$.
A particular example of an LOCC channel could be  a generalized teleportation protocol \cite{Wer01}.
\end{definition}

One can find the defining equation \eqref{eq:TP-simul} explicitly stated as \cite[Eq.~(11)]{HHH99}.
 All covariant channels, as given in  Definition~\ref{def:covariant}, are teleportation-simulable with respect to a resource state $\mathcal{M}_{A\to B}(\Phi_{R_AA})$ \cite{CDP09}.

\begin{definition}[PPT-simulable \cite{KW17}]
A channel $\mc{M}_{A\to B}$ is PPT-simulable with an associated resource state if for all $\rho_{A}\in\msc{D}\(\mc{H}_{A}\)$ there exists a resource state $\omega_{R_AB}\in\msc{D}\(\mc{H}_{R_AB}\)$ such that 
\begin{equation}
\mc{M}_{A\to B}\(\rho_A\)=\mc{P}_{R_AA B\to B}\(\rho_{A}\otimes\omega_{R_AB}\),
\end{equation}
where $\mc{P}_{L_AAB\to B}$ is a PPT-preserving channel acting on $R_AA\!:\!B$, where the transposition map is with respect to the system $B$. 
\end{definition}

\begin{definition}[Jointly covariant memory cell \cite{DW17}]\label{def:cov-cell}
A set $\overline{\mc{M}}_{\msc{X}}=\{\mc{M}^x_{A\to B}\}_{x\in\msc{X}}$ of quantum channels is  jointly covariant if there exists a group $\msc{G}$ such that for all $x\in\msc{X}$, the channel $\mc{M}^x$ is a covariant channel with respect to the group $\msc{G}$ (cf., Definition~\ref{def:covariant}).
\end{definition}

\begin{remark}[\kern-0.35em\cite{DW17}]
Any jointly covariant memory cell $\overline{\mc{M}}_{\msc{X}}=\{\mc{M}^x_{A\to B}\}_{x\in\msc{X}}$ is jointly teleportation-simulable with respect to the set $\{\mc{M}^x_{A\to B}(\Phi_{R_AA})\}_{x\in\msc{X}}$ of resource states.
\end{remark}
	
\section{Entropies and information}
The von Neumann entropy of a density operator $\rho_A$ is defined as \cite{Neu32}
\begin{equation}\label{eq:vnent}
S(A)_\rho:= S(\rho_A)= -\Tr\{\rho_A\log_2\rho_A\}.
\end{equation}
The conditional quantum entropy $S(A\vert B)_\rho$ of a density operator $\rho_{AB}$ of a composite system $AB$ is defined as
\begin{equation}
S(A\vert B)_\rho \coloneqq S(AB)_\rho-S(B)_\rho.
\end{equation}
The coherent information $I(A\> B)_{\rho}$ of a density operator $\rho_{AB}$ is defined as \cite{SN96}
\begin{equation}\label{eq:coh-info}
I(A\rangle B)_{\rho} \coloneqq - S(A\vert B)_\rho = S(B)_{\rho}-S(AB)_{\rho}.
\end{equation}
The quantum relative entropy of two quantum states is a measure of their distinguishability. For $\rho\in\msc{D}(\mc{H})$ and $\sigma\in\msc{B}_+(\mc{H})$, it is defined as~\cite{Ume62} 
\begin{equation}\label{eq:rel-ent-rev}
D(\rho\V \sigma):= \left\{ 
\begin{tabular}{c c}
$\Tr\{\rho[\log_2\rho-\log_2\sigma]\}$, & $\supp(\rho)\subseteq\supp(\sigma)$\\
$+\infty$, &  otherwise.
\end{tabular} 
\right.
\end{equation}
The quantum relative entropy is non-increasing under the action of positive trace-preserving maps \cite{MR15}, which is the statement that $D(\rho\Vert\sigma)\geq D(\mc{M}(\rho)\Vert\mc{M}{(\sigma)})$ for any two density operators $\rho$ and $\sigma$ and a positive trace-preserving map $\mc{M}$ (this inequality applies to quantum channels as well \cite{Lin75}, since every completely positive map is also a positive map by definition).

The quantum mutual information $I(A;B)_\rho$ is a measure of correlation between quantum systems $A$ and $B$ in the state $\rho_{AB}$. It is defined as
\begin{align}
I(A;B)_\rho &:=\inf_{\sigma_A \in\msc{D}(\mathcal{H}_A)}D(\rho_{AB}\Vert\rho_A\otimes\sigma_B)=S(A)_\rho+S(B)_\rho-S(AB)_\rho.
\end{align}
The conditional quantum mutual information $I(A;B\vert C)_\rho$ of a tripartite density operator $\rho_{ABC}$ is defined as
\begin{align}
I(A;B\vert C)_\rho &:=S(A\vert C)_\rho+S(B\vert C)_\rho-S(AB\vert C)_\rho.
\end{align}
It is known that quantum entropy, quantum mutual information, and conditional quantum  mutual information are all non-negative quantities (see \cite{LR73,LR73b}). 

The following AFW inequality gives uniform continuity bounds for conditional entropy:
\begin{lemma}[\kern-0.35em\cite{AF04,Win16}]\label{thm:AFW}
Let $\rho_{AB},\sigma_{AB}\in\msc{D}(\mc{H}_{AB})$. Suppose that $\frac{1}{2}\left\Vert \rho_{AB}-\sigma_{AB}\right\Vert_1\leq\varepsilon$, where $\varepsilon\in\[0,1\]$. Then
\begin{equation}
\left\vert S(A|B)_\rho-S(A|B)_\sigma\right\vert \leq 2\varepsilon\log_2\dim(\mc{H}_A)+(1+\varepsilon)h_2\(\frac{\varepsilon}{1+\varepsilon}\),
\end{equation}
where $h_2(\varepsilon)$ denotes binary entropy function:
\begin{equation}\label{eq:g-2}
h_2(\varepsilon)\coloneqq -\varepsilon\log_2\varepsilon-(1-\varepsilon)\log_2(1-\varepsilon).
\end{equation} 

If system $B$ is a classical register $X$ such that $\rho_{XA}$ and $\sigma_{XA}$ are classical-quantum (cq) states of the following form:
\begin{equation}
\rho_{XA}=\sum_{x\in\mc{X}}p_X(x)|x\>\<x|_X\otimes\rho^x_A,\quad \sigma_{XA}=\sum_{x\in\mc{X}}q_X(x)|x\>\<x|_X\otimes\sigma^x_A,
\end{equation}
where $\{|x\>_X\}_{x\in\mc{X}}\in\ONB(\mc{H}_{X})$ and $\forall x\in\mc{X}:\ \rho^x_A,\sigma^x_A\in\mc{D}(\mc{H}_A)$, then
\begin{align}
\left\vert S(X|A)_\rho-S(X|A)_\sigma\right\vert &\leq \varepsilon\log_2\dim(\mc{H}_X)+g(\varepsilon),\\
\left\vert S(A|X)_\rho-S(A|X)_\sigma\right\vert &\leq \varepsilon\log_2\dim(\mc{H}_A)+g(\varepsilon).
\end{align}
\end{lemma}

\section{Generalized divergences}
A quantity is called a generalized divergence \cite{PV10,SW12} if it satisfies the following monotonicity (data-processing) inequality for all density operators $\rho\in\msc{D}(\mc{H}')$ and $\sigma\in\msc{D}(\mc{H}')$ and quantum channels $\mc{M}:\msc{B}_{+}(\mc{H}')\to \msc{B}_{+}(\mc{H})$:
\begin{equation}\label{eq:gen-div-mono}
\mathbf{D}(\rho\Vert \sigma)\geq \mathbf{D}(\mathcal{M}(\rho)\Vert \mc{M}(\sigma)).
\end{equation}
As a direct consequence of the above inequality, any generalized divergence satisfies the following two properties for an isometry $U$ and a state~$\tau$ \cite{WWY14}:
\begin{align}
\mathbf{D}(\rho\Vert \sigma) & = \mathbf{D}(U\rho U^\dag\Vert U \sigma U^\dag),\label{eq:gen-div-unitary}\\
\mathbf{D}(\rho\Vert \sigma) & = \mathbf{D}(\rho \otimes \tau \Vert \sigma \otimes \tau).\label{eq:gen-div-prod}
\end{align}
One can define a generalized mutual information for a quantum state $\rho_{AB}$ as
\begin{equation}
I_{\mathbf{D}}(A;B)_\rho :=\inf_{\sigma_B\in\msc{D}(\mc{H}_B)}\mathbf{D}(\rho_{AB}\Vert \rho_A\otimes\sigma_B).
\end{equation}

The sandwiched R\'enyi relative entropy  \cite{MDSFT13, WWY14} is denoted as $\wt{D}_\alpha(\rho\Vert\sigma)$ and defined for
$\rho\in\msc{D}(\mc{H})$, $\sigma\in\msc{B}_+(\mc{H})$, and  $\forall \alpha\in (0,1)\cup(1,\infty)$ as
\begin{equation}\label{eq:def_sre}
\wt{D}_\alpha(\rho\Vert \sigma):= \frac{1}{\alpha-1}\log_2 \Tr\left\{\left(\sigma^{\frac{1-\alpha}{2\alpha}}\rho\sigma^{\frac{1-\alpha}{2\alpha}}\right)^\alpha \right\},
\end{equation}
but it is set to $+\infty$ for $\alpha\in(1,\infty)$ if $\supp(\rho)\nsubseteq \supp(\sigma)$. The sandwiched R\'enyi relative entropy obeys the following ``monotonicity in $\alpha$'' inequality \cite{MDSFT13}:
\begin{equation}\label{eq:mono_sre}
\wt{D}_\alpha(\rho\Vert \sigma)\leq \wt{D}_\beta(\rho\Vert \sigma) \text{ if }  \alpha\leq \beta, \text{ for } \alpha,\beta\in(0,1)\cup(1,\infty).
\end{equation}
The following lemma states that the sandwiched R\'enyi relative entropy $\wt{D}_\alpha(\rho\Vert \sigma)$ is a particular generalized divergence for certain values of $\alpha$. 
\begin{lemma}[\kern-0.35em\cite{FL13,Bei13}]
Let $\mc{M}_{A'\to A}$ be a quantum channel and let $\rho_{A'}\in\msc{D}(\mc{H}_{A'})$ and $\sigma_{A'}\in \msc{B}_+(\mc{H}_{A'})$. Then,
\begin{equation}
\wt{D}_\alpha(\rho\Vert \sigma)\geq \wt{D}_\alpha(\mc{M}(\rho)\Vert \mc{M}(\sigma)), \ \forall \alpha\in \[1/2,1\)\cup (1,\infty).
\end{equation}
\end{lemma}

In the limit $\alpha\to 1$, the sandwiched R\'enyi relative entropy $\wt{D}_\alpha(\rho\Vert \sigma)$ converges to the quantum relative entropy \cite{MDSFT13,WWY14}:
\begin{equation}\label{eq:mono_renyi}
\lim_{\alpha\to 1}\wt{D}_\alpha(\rho\Vert\sigma):= D_1(\rho\Vert\sigma)=D(\rho\Vert\sigma).
\end{equation}
In the limit $\alpha\to \infty$, the sandwiched R\'enyi relative entropy $\wt{D}_\alpha(\rho\Vert\sigma)$ converges to the max-relative entropy \cite{MDSFT13}, which is defined as \cite{D09,Dat09}
\begin{equation}\label{eq:max-rel}
D_{\max}(\rho\V\sigma)=\inf\{\lambda:\ \rho \leq 2^\lambda\sigma\},
\end{equation}
and if $\supp(\rho)\nsubseteq\supp(\sigma)$ then $D_{\max}(\rho\Vert\sigma)=\infty$. 

The sandwiched  R\'enyi mutual information $\wt{I}_\alpha(R;B)_\rho$ is defined as \cite{Bei13,GW13}
\begin{equation}
\wt{I}_\alpha(R;B)_\rho:=\min_{\sigma_B}\wt{D}_\alpha(\rho_{RB}\V\rho_R\otimes\sigma_B).
\end{equation}

Another generalized divergence is the $\varepsilon$-hypothesis-testing divergence \cite{BD10,WR12},  defined as
\begin{equation}
D^\varepsilon_h\!\(\rho\Vert\sigma\):=-\log_2\inf_{\Lambda}\{\Tr\{\Lambda\sigma\}:\ 0\leq\Lambda\leq \bm{1} \wedge\Tr\{\Lambda\rho\}\geq 1-\varepsilon\},
\end{equation}
for $\varepsilon\in[0,1]$, $\rho\in\msc{D}(\mc{H})$, and $\sigma\in\msc{B}_+(\mc{H})$.

The following lemma follows directly from the statement of \cite[Theorem III.1]{CM17}.
\begin{lemma}[\kern-0.35em\cite{CM17}]\label{thm:data-tri-ineq}
Let $\rho_A\in\msc{D}(\mc{H}_A)$, and positive semidefinite operators $\sigma\in\msc{B}_+(\mc{H}_B),\sigma'\in\msc{B}_+(\mc{H}_A)$, the following inequality holds for any positive trace-preserving map $\mc{M}_{A\to B}$
\begin{equation}
D_{\max}(\mc{M}(\rho)\Vert\sigma)\leq D_{\max}(\rho\Vert \sigma')+D_{\max}(\mc{M}(\sigma')\Vert \sigma).
\end{equation}
\end{lemma}

Some other examples of generalized divergences are the trace distance and the fidelity. The trace distance between two density operators $\rho,\sigma\in\msc{D}(\mc{H})$ is equal to $\Vert \rho-\sigma\Vert_1$, where $\Vert T\Vert_1=\Tr\{\sqrt{T^\dag T}\}$. The fidelity of $\tau,\sigma\in\mc{B}_+(\mc{H})$, which is defined as $F(\tau,\sigma)=\norm{\sqrt{\tau}\sqrt{\sigma}}_1^2$~\cite{Uhl76}, is also a generalized divergence.

\begin{lemma}[Uhlmann's theorem~\cite{Uhl76}]\label{thm:ut}
The following two expressions for fidelity between two states $\rho_A$ and $\sigma_A$ are equal:
\begin{equation}
F(\rho_A,\sigma_A)=\max_{U}\vert\<\varphi^\rho|_{RA}U_{R}\otimes\mathbbm{1}_A|\varphi^\sigma_{RA}\>\vert^2=\norm{\sqrt{\rho_A}\sqrt{\sigma_A}}_1^2,
\end{equation}
where $U_R$ is a unitary operator and $\varphi^{\omega}_{RA}$ denotes purification of any $\omega_A\in\mc{D}(\mc{H}_{RA})$. 
\end{lemma}

The following well known lemma establishes relations between fidelity and trace distance.
\begin{lemma}[\kern-0.35em\cite{FV99}]\label{thm:f-t}
For any two density operators $\rho,\sigma\in\msc{D}(\mc{H})$, the following bounds hold
\begin{equation}
1-\sqrt{F(\rho,\sigma)}\leq \frac{1}{2}\norm{\rho-\sigma}_1\leq \sqrt{1-F(\rho,\sigma)}. 
\end{equation}
\end{lemma}
Another well known lemma that establishes relation between the relative entropy and trace distance is as follows.
\begin{lemma}[Pinsker's inequality~\cite{OP93}]
For any two density operators $\rho,\sigma\in\msc{D}(\mc{H})$, following bounds hold
\begin{equation}
D(\rho\Vert\sigma)\geq \frac{1}{2\ln 2}\norm{\rho-\sigma}_1^2,
\end{equation}
where $\ln$ denotes natural logarithm. 
\end{lemma}

\section{Private states and privacy test}\label{sec:rev-priv-states}

Private states \cite{HHHO05,HHHO09} are an essential notion in any discussion of secret key distillation in quantum information, and we review their basics here.

A tripartite key state $\gamma_{K_AK_BE}$ contains $\log_2 |K|$ bits of secret key, shared between systems $K_A$ and $K_B$ and protected from an eavesdropper possessing system $E$, if there exists a state $\sigma_E\in\msc{D}(\mc{H}_E)$ and a projective measurement channel $\mc{M}(\cdot)=\sum_{i}\op{i}(\cdot)\op{i}$, where $\{|i\>\}_i\in\ONB$, such that
\begin{equation}
\(\mc{M}_{K_A}\otimes\mc{M}_{K_B}\)(\gamma_{K_AK_BE})=\frac{1}{|K|}\sum_{i=0}^{K-1}\op{i}_{K_A}\otimes\op{i}_{K_B}\otimes\sigma_E.
\end{equation}
The systems $K_A$ and $K_B$ are maximally classically correlated, and the key value is uniformly random and independent of the system $E$. 

A bipartite private state $\gamma_{S_AK_AK_BS_B}$ containing $\log_2 |K|$ bits of secret key has the following form:
\begin{equation}
\gamma_{S_AK_AK_BS_B}\coloneqq U^t_{S_AK_AK_BS_B}(\Phi_{K_AK_B}\otimes\theta_{R_AR_B})(U^t_{S_AK_AK_BS_B})^\dag,
\end{equation}
where $\Phi_{K_AK_B}$ is a maximally entangled state of Schmidt rank $|K|$, $U^t_{S_AK_AK_BS_B}$ is a \textquotedblleft twisting\textquotedblright ~unitary of the form 
\begin{equation}
U^t_{S_AK_AK_BS_B}\coloneqq\sum_{i,j=0}^{K-1}\ket{i}\!\bra{i}_{K_A}\otimes\ket{j}\!\bra{j}_{K_B}\otimes U^{ij}_{S_AS_B},
\end{equation}
with each $U^{ij}_{S_AS_B}$ a unitary, and $\theta_{S_AS_B}$ is a state. The systems $S_A,S_B$ are called \textquotedblleft ~shield\textquotedblright systems because they, along with the twisting unitary, can help to protect the key in systems $K_A$ and $K_B$ from any party possessing a purification of $\gamma_{S_AK_AK_BS_B}$.

Bipartite private states and tripartite key states are equivalent \cite{HHHO05,HHHO09}. That is, for $\gamma_{S_AK_AK_BS_B}$ a bipartite private state and $\gamma_{S_AK_AK_BS_BE}\in\mc{H}_{S_A}\otimes\mc{H}_{K_A}\otimes\mc{H}_{K_B}\otimes\mc{H}_{S_B}\otimes\mc{E}$ some purification of it, $\gamma_{K_AK_BE}$ is a tripartite key state. Conversely, for any tripartite key state $\gamma_{K_AK_BE}$ and any purification $\gamma_{S_AK_AK_BS_BE}$ of it, $\gamma_{S_AK_AK_BS_B}$ is a bipartite private state. 

A state $\rho_{K_AK_BE}$ is an $\varepsilon$-approximate tripartite key state if there exists a tripartite key state $\gamma_{K_AK_BE}$ such that 
\begin{equation}
F(\rho_{K_AK_BE},\gamma_{K_AK_BE})\geq 1-\varepsilon,
\end{equation}
where $\varepsilon\in[0,1]$. Similarly, a state $\rho_{S_AK_AK_BS_B}$ is an $\varepsilon$-approximate bipartite private state if there exists a bipartite private state $\gamma_{S_AK_AK_BS_B}$ such that 
\begin{equation}
F(\rho_{S_AK_AK_BS_BE},\gamma_{S_AK_AK_BS_BE})\geq 1-\varepsilon.
\end{equation}

If $\rho_{S_AK_AK_BS_B}$ is an $\varepsilon$-approximate bipartite key state with $K$ key values, then Alice and Bob hold an $\varepsilon$-approximate tripartite key state with $|K|$ key values, and the converse is true as well \cite{HHHO05,HHHO09}.

A privacy test corresponding to $\gamma_{S_AK_AK_BS_B}$ (a $\gamma$-privacy test) is defined as the following dichotomic measurement \cite{WTB16}:
\begin{equation}
\{\Pi^\gamma_{S_AK_AK_BS_B}, \bm{1}_{S_AK_AK_BS_B}-\Pi^\gamma_{S_AK_AK_BS_B}\},
\end{equation}
where
\begin{equation}
\Pi^\gamma_{S_AK_AK_BS_B}\coloneqq U^t_{S_AK_AK_BS_B}(\Phi_{K_AK_B}\otimes \bm{1}_{S_AS_B})(U^t_{S_AK_AK_BS_B})^\dag,
\end{equation}
$\bm{1}_{S_AS_B}\in\msc{B}_+(\mc{H}_{S_AS_B})$ is the identity operator, and $U^t_{S_AK_AK_BS_B}$ is the twisting unitary discussed earlier. Let $\varepsilon\in[0,1]$ and $\rho_{S_AK_AK_BS_B}$ be an $\varepsilon$-approximate bipartite private state. The probability for $\rho_{S_AK_AK_BS_B}$ to pass the $\gamma$-privacy test is never smaller than $1-\varepsilon$ \cite{WTB16}:
\begin{equation}
\Tr\{\Pi^\gamma_{S_AK_AK_BS_B}\rho_{S_AK_AK_BS_B}\}\geq 1-\varepsilon. 
\end{equation}
For a state $\sigma_{S_AK_AK_BS_B}\in\SEP(S_AK_A\!:\!K_BS_B)$, the probability of passing any $\gamma$-privacy test is never greater than $\frac{1}{|K|}$ \cite{HHHO09}:
\begin{equation}
\Tr\{\Pi^\gamma_{S_AK_AK_BS_B}\sigma_{S_AK_AK_BS_B}\}\leq \frac{1}{|K|},
\end{equation}
where $|K|$ is the number of values that the secret key can take (i.e., $|K|=\dim(\mc{H}_{K_A})=\dim(\mc{H}_{K_B})$). These two inequalities are foundational for some of the converse bounds established in this thesis, as was the case in \cite{HHHO09,WTB16}.

\section{Entanglement measures}
Let $\Ent(A;B)_\rho$ denote an entanglement measure \cite{HHHH09} that is evaluated for a bipartite state~$\rho_{AB}$. 
The basic property of an entanglement measure is that it should be an LOCC monotone \cite{HHHH09}, i.e., non-increasing under the action of an LOCC channel. 

Entanglement distillation from a bipartite state $\rho_{AB}$ is the task of distilling a maximally entangled state $\Phi_{AB}$ of Schmidt rank $|M|$ from (asymptotically) large number of independent and identically distributed copies of $\rho_{AB}$, i.e., $\rho_{AB}^{\otimes n}$ for $n\to \infty$ via standard LOCC distillation protocols \cite{HH99,BCJ+99}. A state $\rho_{AB}$ is entanglement distillable if $\Tr\{\Phi_{AB}\rho_{AB}\}\geq \frac{1}{|M|}$ \cite{HH99,BCJ+99}. 

There are different entanglement measures based on characteristic properties of entangled states. These properties are associated to the ability of how useful these entangled states are for specific information processing tasks, such as entanglement and secret-key distillation. It is known that all entangled states are useful for distilling secret key. However, there exists class of entangled states called bound entangled states that are not entanglement distillable. 

Given such an entanglement measure $\Ent(A;B)_\rho$, one can define the entanglement $\Ent(\mc{M})$ of a channel $\mc{M}_{A\to B}$ in terms of it by optimizing over all pure, bipartite states that can be input to the channel:
\begin{equation}\label{eq:ent-mes-channel}
\Ent(\mc{M})=\sup_{\psi_{RA}} \Ent(R;B)_\omega,
\end{equation}
where $\omega_{RB}=\mc{M}_{A\to B}(\psi_{RA})$. Due to the properties of an entanglement measure and the well known Schmidt decomposition theorem, it suffices to optimize over pure states $\psi_{RA}$ such that $R\simeq A$ (i.e., one does not achieve a higher value of
$\Ent(\mc{M})$
 by optimizing over mixed states with unbounded reference system $R$). In an information-theoretic setting, the entanglement $\Ent(\mc{M})$ of a channel~$\mc{M}$ characterizes the amount of entanglement that a sender $A$ and receiver $B$ can generate by using the channel if they do not share entanglement prior to its use.

Alternatively, one can consider the amortized entanglement $\Ent_A(\mc{M})$, also called the entangling power, of a channel $\mc{M}_{A\to B}$ as the following optimization~\cite{BHLS03,KW17} (see also \cite{LHL03,CM17,DDMW17,RKB+17}):
\begin{equation}\label{eq:ent-arm}
\Ent_A(\mc{M})\coloneqq \sup_{\rho_{R_AAR_B}} \left[\Ent(R_A;BR_B)_{\tau}-\Ent(R_AA;R_B)_{\rho}\right],
\end{equation}
where $\tau_{R_ABR_B}=\mc{M}_{A\to B}(\rho_{R_AAR_B})$ and $\rho_{R_AAR_B}\in\msc{D}(\mc{H}_{R_AAR_B})$. The supremum is with respect to all states $\rho_{R_AAR_B}$ and the systems $R_A,R_B$ are finite-dimensional but could be arbitrarily large. Thus, in general, $\Ent_A(\mc{M})$ need not be computable. The amortized entanglement quantifies the net amount of entanglement that can be generated by using the channel $\mc{M}_{A\to B}$, if the sender and the receiver are allowed to begin with some initial entanglement in the form of the state $\rho_{R_AAR_B}$. That is, $\Ent(R_AA;R_B)_\rho$ quantifies the entanglement of the initial state $\rho_{R_AAR_B}$, and $\Ent(R_A;BR_B)_{\tau}$ quantifies the entanglement of the final state produced after the action of the channel. 

Recently, it was shown in \cite{KW17}, connected to related developments in \cite{LHL03,BHLS03,CM17,DDMW17,DW17}, that the amortized entanglement of a point-to-point channel $\mc{M}_{A\to B}$ serves as an upper bound on the entanglement of the final state, say $\omega_{AB}$, generated at the end of an LOCC- or PPT-assisted quantum communication protocol that uses $\mc{M}_{A\to B}$ $n$ times:
\begin{equation}
\Ent(A;B)_{\omega}\leq n\Ent_{A}(\mc{M}).
\end{equation}
Thus, the physical question of determining meaningful upper bounds on the LOCC- or PPT-assisted capacities of point-to-point channel $\mc{M}$ is equivalent to the mathematical question of whether amortization can enhance the entanglement of a given channel, i.e., whether the following equality holds for a given entanglement measure $\Ent$:
\begin{equation}
\Ent_A(\mc{M})\stackrel{?}{=} \Ent(\mc{M}). 
\end{equation}  

The Rains relative entropy of a state $\rho_{AB}$ is defined as \cite{Rai01, AdMVW02}
\begin{equation}\label{eq:rains-inf-state}
R(A;B)_\rho\coloneqq \min_{\sigma_{AB}\in \PPT'(A:B)} D (\rho_{AB}\Vert \sigma_{AB}),
\end{equation}
and it is monotone non-increasing under the action of a PPT-preserving quantum channel $\mc{P}_{A'B'\to AB}$, i.e.,
\begin{equation}
R(A';B')_\rho\geq R(A;B)_\omega,
\end{equation}
where $\omega_{AB}=\mc{P}_{A'B'\to AB}(\rho_{A'B'})$. The sandwiched Rains relative entropy of a state $\rho_{AB}$ is defined as follows \cite{TWW17}:  
\begin{equation}\label{eq:alpha-rains-inf-state}
\widetilde{R}_{\alpha}(A;B)_\rho\coloneqq \min_{\sigma_{AB}\in \PPT'(A:B)} \widetilde{D}_{\alpha} (\rho_{AB}\Vert \sigma_{AB}).
\end{equation}
The max-Rains relative entropy of a state $\rho_{AB}$ is defined as \cite{WD16b}
\begin{equation}
R_{\max}(A;B)_\rho\coloneqq \min_{\sigma_{AB}\in \PPT'(A:B)} D_{\max} (\rho_{AB}\Vert \sigma_{AB}).
\end{equation}
The max-Rains information of a quantum channel $\mc{M}_{A\to B}$ is defined as \cite{WFD17}
\begin{equation}
R_{\max}(\mc{M})\coloneqq \max_{\phi_{RA}}R_{\max} (R;B)_\omega,
\label{eq:max-Rains-channel}
\end{equation}
where $\omega_{RB}=\mc{M}_{A\to B}(\phi_{RA})$ and $\phi_{RA}$ is a pure state, with $\dim(\mc{H}_R)=\dim(\mc{H}_A)$. The amortized max-Rains information of a channel $\mc{M}_{A\to B}$, denoted as $R_{\max,A}(\mc{M})$, is defined by replacing $\Ent$ in \eqref{eq:ent-arm} with the max-Rains relative entropy $R_{\max}$ \cite{BW17}. It was shown in \cite{BW17} that amortization does not enhance the max-Rains information of an arbitrary point-to-point channel, i.e.,
\begin{equation}
R_{\max,A}(\mc{M})=R_{\max}(\mc{M}).
\end{equation}  

Recently, in \cite[Eq.~(8)]{WD16a} (see also \cite{WFD17}), the max-Rains relative entropy of a state $\rho_{AB}$ was expressed as 
\begin{equation}\label{eq:rains-w}
R_{\max}(A;B)_{\rho}=\log_2 W(A;B)_{\rho}, 
\end{equation}
where $W(A;B)_{\rho}$ is the solution to the following semi-definite program:
\begin{align}
\textnormal{minimize}\ &\ \Tr\{C_{AB}+D_{AB}\}\nonumber\\
\textnormal{subject to}\ &\ C_{AB}, D_{AB}\geq 0,\nonumber\\
   &\ \T_{B} (C_{AB}-D_{AB})\geq \rho_{AB}. \label{eq:rains-state-sdp}
\end{align}
Similarly, in \cite[Eq.~(21)]{WFD17}, the max-Rains information of a quantum channel $\mc{M}_{A\to B}$ was expressed as 
\begin{equation}\label{eq:rains-omega}
R_{\max}(\mc{M})=\log \Gamma (\mc{M}),
\end{equation}
where $\Gamma(\mc{M})$ is the solution to the following semi-definite program (SDP):
\begin{align}
\textnormal{minimize}\ &\ \norm{\Tr_B\{V_{RB}+Y_{RB}\}}_{\infty}\nonumber\\
\textnormal{subject to}\ &\ Y_{RB}, V_{RB}\geq 0,\nonumber\\
& \ \T_B(V_{RB}-Y_{RB})\geq J^\mc{M}_{RB}.\label{eq:rains-channel-sdp}
\end{align}

The sandwiched relative entropy of entanglement of a bipartite state $\rho_{AB}$ is defined as \cite{WTB16} 
\begin{equation}\label{eq:rel-ent-state}
\widetilde{E}_{\alpha}(A;B)_{\rho}\coloneqq\min_{\sigma_{AB}\in\SEP(A:B)}\widetilde{D}_{\alpha}(\rho_{AB}\Vert\sigma_{AB}).
\end{equation} 
In the limit $\alpha\to 1$, $\widetilde{E}_{\alpha}(A;B)_{\rho}$ converges to the relative entropy of entanglement \cite{VP98}, i.e.,
\begin{equation}\label{eq:rel-ent-state-1}
\lim_{\alpha\to 1}\widetilde{E}_{\alpha}(A;B)_{\rho}=E(A;B)_{\rho}\coloneqq \min_{\sigma_{AB}\in\SEP(A:B)}D(\rho_{AB}\Vert\sigma_{AB}).
\end{equation} The max-relative entropy of entanglement \cite{D09,Dat09} is defined for a bipartite state $\rho_{AB}$ as
\begin{equation}
E_{\max}(A;B)_{\rho}\coloneqq \min_{\sigma_{AB}\in\SEP(A:B)}D_{\max}(\rho_{AB}\Vert\sigma_{AB}).
\end{equation}  
The max-relative entropy of entanglement $E_{\max}(\mc{M})$ of a channel $\mc{M}_{A\to B}$ is defined as in \eqref{eq:ent-mes-channel}, by replacing $\Ent$ with $E_{\max}$ \cite{CM17}. It was shown in \cite{CM17} that amortization does not increase max-relative entropy of entanglement of a channel $\mc{M}_{A\to B}$, i.e.,
\begin{equation}
E_{\max,A}(\mc{M})=E_{\max}(\mc{M}).
\end{equation}

The squashed entanglement of a state $\rho_{AB}\in\msc{D}(\mc{H}_{AB})$ is defined as \cite{CW04} (see also \cite{Tuc99,Tuc02}):
\begin{equation}
E_{\sq}(A;B)_\rho=\frac{1}{2}\inf_{\omega_{ABE}}\left\{I(A;B|E)_\omega: \Tr_E\{\omega_{ABE}\}=\rho_{AB}\wedge \omega_{ABE}\in\mc{D}\(\mc{H}_{ABE}\)\right\}.
\end{equation}
In general, the system $E$ is finite-dimensional, but can be arbitrarily large. We can directly infer from the above definition that $E_{\sq}(B;A)_\rho=E_{\sq}(A;B)_\rho$ for any $\rho_{AB}\in\msc{D}(\mc{H}_{AB})$. We can similarly define the squashed entanglement $E_{\sq}(\mc{M})$ of a channel $\mc{M}_{A\to B}$ \cite{TGW14}, and it is known that amortization does not increase the squashed entanglement of a channel \cite{TGW14}:
\begin{equation}
E_{\sq,A}(\mc{M}) = E_{\sq}(\mc{M}).
\end{equation}

\subsection{Approximate normalization of entanglement measures}\label{sec:app-ent-mes}
Now we briefly discuss normalization properties of some entanglement measures, namely, entropy of entanglement~\cite{Sre93} and squashed entanglement~\cite{CW04}\footnote{This section is based on an unpublished work with Mark M.~Wilde.}. 

\subsubsection{Squashed Entanglement}
We know that squashed entanglement obeys the normalization property; i.e., it is equal to $\log_2d$ for a maximally entangled state $\Phi_{AB}$ of Schmidt rank $d$~\cite{CW04}. Due to the continuity of squashed entanglement~\cite{AF04}, we even know that if the state $\rho_{AB}$ is approximately close to a maximally entangled state $\Phi_{AB}$, then the squashed entanglement is near to $\log_2d$ (see also \cite[Remark 11]{CW04}). In particular, 
\begin{equation}
\frac{1}{2}\left\Vert \rho_{AB}-\Phi_{AB}\right\Vert _{1}\leq\varepsilon
\end{equation}
implies that ~\cite{AF04,Win16,Shi16}
\begin{equation}\label{eq:sqd1}
\left\vert E_{\text{sq}}\left(  A;B\right)  _{\rho}-E_{\text{sq}}\left(
A;B\right)  _{\Phi}\right\vert \leq \sqrt{2\varepsilon}\log_2d + (1+\sqrt{2\varepsilon})h_2\(\frac{\sqrt{2\varepsilon}}{1+\sqrt{2\varepsilon}}\),
\end{equation}
where $d$ is Schmidt rank of $\Phi_{AB}$ and $h_2(\cdot)$ is defined in \eqref{eq:g-2}. From \eqref{eq:sqd1}, we get
\begin{equation}
E_{\text{sq}}\left(  A;B\right)_{\rho}\geq (1-\sqrt{2\varepsilon})\log_2 d - (1+\sqrt{2\varepsilon})h_2\(\frac{\sqrt{2\varepsilon}}{1+\sqrt{2\varepsilon}}\).
\end{equation}

Our statement here is about the converse situation. We show that \begin{equation}
E_{\text{sq}}\left(  A;B\right)  _{\rho}\geq\log_2\left\vert A\right\vert
\left(  1-\varepsilon\right) \end{equation}
implies that the state $\rho_{AB}$ is near to a maximally entangled state. More precisely, we prove the following proposition. 
\begin{proposition} \label{prop:squashed}
Suppose that $\rho_{AB}\in\msc{D}(\mc{H}_{AB})$ and that 
\begin{equation}
E_{\text{sq}}\left(  A;B\right)  _{\rho}\geq\log_2\left\vert A\right\vert
\left(  1-\varepsilon\right),
\end{equation}
for some $\varepsilon\in(0,1)$. Then $\rho_{AB}$ is close to $\Phi_{AB}$ up to some local unitary $U_{B}$:
\begin{align}
\frac{1}{2}\left\Vert \rho_{AB}-U_{B}\Phi_{AB}U^\dagger_{B}\right\Vert_1\leq (2\sqrt{\varepsilon \ln|A|})^{1/2}.
\end{align}
\end{proposition}

\begin{proof}
Let us consider
\begin{equation}
E_{\text{sq}}\left(  A;B\right)  _{\rho}\geq\log_2\left\vert A\right\vert\left(  1-\varepsilon\right).
\end{equation}
Then 
\begin{equation}
\frac{1}{2}I\left(  A;B\right)  _{\rho}\geq E_{\text{sq}}\left(  A;B\right)  _{\rho} \geq \log_2\left\vert A\right\vert \left(
1-\varepsilon\right). 
\end{equation}
Let $\psi_{ABE}$ be a purification of $\rho_{AB}$. Then
\begin{align}
2\log_2\left\vert A\right\vert \left(  1-\varepsilon\right)    & \leq
I(A;B)_{\rho}\\
& =S(A)_{\rho}-S(A|B)_{\rho}\\
& =S(A)_{\rho}+S(A|E)_{\psi}.\label{eq:sqmi1}
\end{align}
From \eqref{eq:sqmi1}, we get 
\begin{align}
2\varepsilon\log_2\left\vert A\right\vert & \geq -S(A|E)+\log_2\left\vert A\right\vert+
\log_2\left\vert A\right\vert -S(A)\\
&\geq D\left( \psi_{AE} \Vert \pi_A\otimes \psi_E\right)\\
& \geq \frac{1}{2\ln 2} \left\Vert \psi_{AE}-\pi_A\otimes \psi_E\right\Vert_1^2,
\end{align}
where $\pi_A=\frac{\mathbbm{I}}{|A|}$ is the maximally mixed state. We have
\begin{align}
\left\Vert \psi_{AE}-\pi_A\otimes \psi_E\right\Vert_1\leq 2\sqrt{\varepsilon \ln|A|},
\end{align}
which implies (see Lemma~\ref{thm:f-t})
\begin{equation}
F(\psi_{AE},\pi_A\otimes \psi_E)\geq 1-2\sqrt{\varepsilon \ln|A|}.
\end{equation}
Invoking Ulhmann's theorem and then monotonicity of the fidelity under partial trace,
we can conclude that there exists some local unitary operator $U_{B}$ such that
\begin{align}
F(\rho_{AB},U_{AB}\Phi_{AB}U^\dagger_{AB})\geq 1-2\sqrt{\varepsilon \ln|A|}.\label{eq:fid-u-phi}
\end{align}
Using \eqref{eq:fid-u-phi} in Lemma~\ref{thm:f-t}, we get
\begin{align}
\norm{\rho_{AB}-U_{AB}\Phi_{B}U^\dagger_{B}}_1\leq 2\sqrt{2\sqrt{\varepsilon \ln|A|}}.
\end{align}
This completes the proof. 
\end{proof}

\subsubsection{Entropy of Entanglement}
\begin{proposition}\label{prop:eeprop}
Suppose that $\psi_{AB}$ is a pure state and that
\begin{equation}\label{eq:assump-EoE}
S(A)_{\psi}\geq\left(1-\varepsilon\right)\log_2|A|,
\end{equation}
for some $\varepsilon\in(0,1)$. Then there exists a unitary operator
$U_{B}$ such that
\begin{equation}
\frac{1}{2}\norm{U_{B}\psi_{AB}U_{B}^{\dag}-\Phi_{AB}}_{1}\leq (2\varepsilon\ln\left\vert A\right\vert)^{1/4}.
\end{equation}
\end{proposition}

\begin{proof}
Consider that our inequality is the same as
\begin{equation}
S\left(A\right)  _{\psi}-\left(1-\varepsilon\right)  \log_2\left\vert
A\right\vert \geq 0.
\end{equation}
We find that
\begin{align}
S\left(A\right)_{\psi}-\left(1-\varepsilon\right) \log_2\left\vert
A\right\vert  & =S\left(A\right)_{\psi}-\log_2\left\vert A\right\vert
+\varepsilon\log_2\left\vert A\right\vert \\
& =-D\left(\psi_{A}\Vert\pi_{A}\right)+\varepsilon\log_2\left\vert A\right\vert
\end{align}
By assuming \eqref{eq:assump-EoE}, we find that
\begin{equation}
\varepsilon\log_2\left\vert A\right\vert   \geq D\left(  \psi_{A}\Vert\pi_{A}\right) 
 \geq \frac{1}{2\ln2}\left\Vert \psi_{A}-\pi_{A}\right\Vert _{1}^{2}.
\end{equation}
By an application of Uhlmann's theorem and Lemma~\ref{thm:f-t}, we recover the statement of the theorem.
\end{proof}

%% file: bqi.tex
\chapter{Fundamental Limits on Entangling Abilities of Bipartite Quantum Interactions}\label{ch:bqi}
A bipartite quantum interaction is an underlying Hamiltonian that governs the physical evolution of an open bipartite quantum system\blfootnote{Most  of this chapter is based on \cite{DBW17}, a joint work with Stefan B\"auml and Mark M.~Wilde.}. In general, any two-body quantum system of interest can be in contact with a bath, and part of the composite system may be inaccessible to observers possessing these systems. As contact with the surrounding (bath) is unavoidable, the study of bipartite quantum interactions is pertinent. Depending on the kind of bipartite interaction and the input states, entanglement can be created, destroyed, or changed between two quantum systems \cite{PV07,HHHH09}. As entanglement is one of the fundamental and intriguing quantum phenomena~\cite{EPR35,S35}, determining the entangling abilities of bipartite quantum interactions is important. These bounds imply fundamental limitations on information processing capabilities over a bipartite quantum network. Non-trivial bounds on the entangling abilities can also serve as the benchmarks for the efficiency testing of bipartite quantum gates in Noisy Intermediate-Scale Quantum (NISQ) processors \cite{Pre18} (cf.~\cite{KDWW18}). 

It is known from quantum mechanics that a closed system evolves according to a unitary transformation \cite{Dbook81,SCbook95}. Let $U^{\hat{H}}_{A'B'E'\to ABE}$ be a unitary associated to an underlying Hamiltonian $\hat{H}$, which governs the physical evolution of the input subsystems $A'$ and $B'$ in the presence of a bath $E'$, to produce output subsystems $A$ and $B$ for the observers and $E$ for the bath.
In general, the individual input systems $A'$, $B'$, and $E'$ and the respective output systems $A$, $B$, and $E$ can have different dimensions. Initially, in the absence of an interaction Hamiltonian $\hat{H}$, the bath is taken to be in a pure state and the systems of interest have no correlation with the bath; i.e., the state of the composite system $A'B'E'$ is of the form $\omega_{A'B'}\otimes \tau_{E'}$, for some fixed state $\tau_{E'}$ of the bath. Under the action of the Hamiltonian $\hat{H}$, the state of the composite system transforms as
\begin{equation}\label{eq:bi-u}
\rho_{ABE}=U^{\hat{H}}(\omega_{A'B'}\otimes \tau_{E'})(U^{\hat{H}})^\dag.
\end{equation}
Since the system $E$ in \eqref{eq:bi-u} is inaccessible, the evolution of the systems of interest is noisy in general. The noisy evolution of the bipartite system $A'B'$ under the action of the interaction Hamiltonian $\hat{H}$ is represented by a completely positive, trace-preserving (CPTP) map \cite{Sti55}, called a bipartite quantum channel:
\begin{equation}\label{eq:bipartite-int}
\mc{N}^{\hat{H}}_{A'B'\to AB}(\omega_{A'B'})=\Tr_{E}\{U^{\hat{H}}(\omega_{A'B'}\otimes \tau_{E'})(U^{\hat{H}})^\dag\},
\end{equation}
where system $E$ represents inaccessible degrees of freedom. In particular, when the Hamiltonian $\hat{H}$ is such that there is no interaction between the composite system $A'B'$ and the bath $E'$, and $A'B'\simeq AB$, then $\mc{N}^{\hat{H}}$ corresponds to a bipartite unitary, i.e., $\mc{N}^{\hat{H}}(\cdot)=U^{\hat{H}}_{A'B'\to AB}(\cdot)(U^{\hat{H}}_{A'B'\to AB})^\dag$. 

In the setting of an information processing task, when two spatially separated observers have access to different pair of quantum systems, $(A',A)$ or $(B',B)$, then a bipartite channel $\mc{N}_{A'B'\to AB}$ is also called \textit{bidirectional channel}. 

In this chapter, we focus on two different information-processing tasks relevant for bipartite quantum interactions, the first being entanglement distillation~\cite{BBPS96,BBP+97,Rai99} and the second secret key agreement~\cite{D05,DW05,HHHO05,HHHO09}. Entanglement distillation is the task of generating a maximally entangled state, such as the singlet state, when two separated quantum systems undergo a bipartite interaction. Whereas, secret key agreement is the task of extracting maximal classical correlation between two separated systems, such that it is independent of the state of the bath system, which an eavesdropper could possess. 

In an information-theoretic setting, a bipartite interaction between classical systems was first considered in \cite{Sha61} in the context of communication; therein, a bipartite interaction was called a two-way communication channel. In the quantum domain, bipartite unitaries have been widely considered in the context of their entangling ability, applications for interactive communication tasks, and the simulation of bipartite Hamiltonians in distributed quantum computation \cite{BDEJ95,ZZF00,EJPP00,BRV00,NC00,CLP01,CDKL01,BHLS03,CLV04,JMZL17,DSW17} (see also Section~\ref{sec:rev-control-channels-1}).  These unitaries form the simplest model of non-trivial interactions in many-body quantum systems and have been used as a model of scrambling in the context of quantum chaotic systems~\cite{SS08b,HQRY16,DHW16}, as well as for the internal dynamics of a black hole~\cite{HP07} in the context of the information-loss paradox~\cite{Haw76}. More generally, \cite{CLL06} developed the model of a bipartite interaction or two-way quantum communication channel.  Bounds on the rate of entanglement generation in open quantum systems undergoing time evolution have also been discussed for particular classes of quantum dynamics~\cite{Bra07,DKSW18}.  

The maximum rate at which a particular task can be accomplished by allowing the use of a bipartite interaction a large number of times, is equal to the capacity of the interaction for the task. The entanglement generating capacity quantifies the maximum rate of entanglement that can be generated  from a bipartite interaction. Various capacities of a general bipartite unitary evolution were formalized in \cite{BHLS03}. Later, various capacities of a general two-way channel were discussed in \cite{CLL06}. The entanglement generating capacities or entangling power of bipartite unitaries for different communication protocols have been widely discussed in the  literature~\cite{ZZF00,LHL03,BHLS03,LSW09,WSM17,CY16}. Also, prior to the work of \cite{DBW17}, it was an open question to find a non-trivial, computationally efficient upper bound on the entanglement generating capacity of a bipartite quantum interaction. 

In this chapter, we determine bounds on the capacities of bipartite interactions for entanglement generation and secret key agreement. The organization of this chapter is as follows. In Section~\ref{sec:ent-dist}, we derive a strong converse upper bound on the rate at which entanglement can be distilled from a bipartite quantum interaction. This bound is given by an information quantity introduced in \cite[Section~3.1]{DBW17}, called the bidirectional max-Rains information $R^{2\to 2}_{\max}({\mc{N}})$ of a bidirectional channel $\mc{N}$. The bidirectional max-Rains information is the solution to a semi-definite program and is thus efficiently computable. In Section~\ref{sec:priv-key}, we derive a strong converse upper bound on the rate at which a secret key can be distilled from a bipartite quantum interaction. This bound is given by a related information quantity introduced in \cite[Section~4.1]{DBW17}, called the bidirectional max-relative entropy of entanglement $E^{2\to 2}_{\max}(\mc{N})$ of a bidirectional channel $\mc{N}$. In Section~\ref{sec:ent-mes-sim}, we derive upper bounds on the entanglement generation and secret key agreement capacities of bidirectional PPT- and teleportation-simulable channels, respectively. Our upper bounds on the capacities of such channels depend only on the entanglement of the resource states with which these bidirectional channels can be simulated.

\section{Bipartite interactions and controlled unitaries}\label{sec:rev-control-channels-1}
Let us consider a bipartite quantum interaction between systems $X'$ and $B'$, generated by a Hamiltonian $\hat{H}_{X'B'E'}$, where $E'$ is a bath system. Suppose that the Hamiltonian is time independent, having the following form:
\begin{equation}\label{eq:c-ham-1}
\hat{H}_{X'B'E'}\coloneqq \sum_{x\in\mc{X}}\ket{x}\!\bra{x}_{X'}\otimes \hat{H}^x_{B'E'},
\end{equation}
where $\{\vert x \rangle\}_{x\in\mc{X}}\in\ONB(\mc{H}_{X'})$ and $\hat{H}^x_{B'E'}$ is a Hamiltonian for the composite system $B'E'$. Then, the evolution of the composite system $X'B'E'$  is given by the following controlled unitary:
\begin{equation}
U_{\hat{H}}(t)\coloneqq\sum_{x\in\mc{X}}\ket{x}\!\bra{x}_{X'}\otimes \exp\!\left(-\frac{\iota}{\hslash}\hat{H}^x_{B'E'}t\right),
\end{equation}
where $t$ denotes time. Suppose that the systems $B'$ and $E'$ are not correlated before the action of Hamiltonian $\hat{H}^x_{B'E'}$ for each $x\in\mc{X}$. Then, the evolution of the system $B'$ under the interaction $\hat{H}^x_{B'E'}$ is given by a quantum channel $\mc{M}^x_{B'\to B}$ for all $x$.

For some distributed quantum computing and information processing tasks where the controlling system $X$ and input system $B'$ are jointly accessible, the following bidirectional channel is relevant:
\begin{equation}\label{eq:bi-ch-mc}
\mc{N}_{X'B'\to XB}(\cdot)\coloneqq \sum_{x\in\mc{X}}\ket{x}\!\bra{x}_X\otimes\mc{M}^x_{B'\to B}\(\bra{x}(\cdot)\ket{x}_{X'}\).
\end{equation}
In the above, $X'$ is a controlling system that determines which evolution from the set $\{\mc{M}^x\}_{x\in\mc{X}}$ takes place on input system $B'$. 
In particular, when $X'$ and $B'$ are spatially separated and the input state for the system $X'B'$ are considered to be in a product state, the noisy evolution for such constrained interaction is given by the following bidirectional channel:
\begin{equation}\label{eq:bi-ch-mcc}
\mc{N}_{X'B'\to XB}(\sigma_{X'}\otimes\rho_{B'})\coloneqq \sum_{x\in\mc{X}}\bra{x}\sigma_{X'}\ket{x}_{X'}\ket{x}\!\bra{x}_X\otimes\mc{M}^x_{B'\to B}(\rho_{B'}).
\end{equation}

\section{Entanglement distillation from bipartite quantum interactions}\label{sec:ent-dist}	
In this section, we define the bidirectional max-Rains information $R^{2\to 2}_{\max}(\mc{N})$ of a bidirectional channel $\mc{N}$ and show that it is not enhanced by amortization.  We also prove that $R^{2\to 2}_{\max}(\mc{N})$ is an upper bound on the amount of entanglement that can be distilled from a bidirectional channel $\mc{N}$. We do so by adapting to the bidirectional setting, the result  from \cite{KW17} and recent techniques developed in  \cite{CM17,RKB+17,BW17} for point-to-point quantum communication protocols.  

\begin{figure}
		\centering
		\includegraphics[scale=0.639]{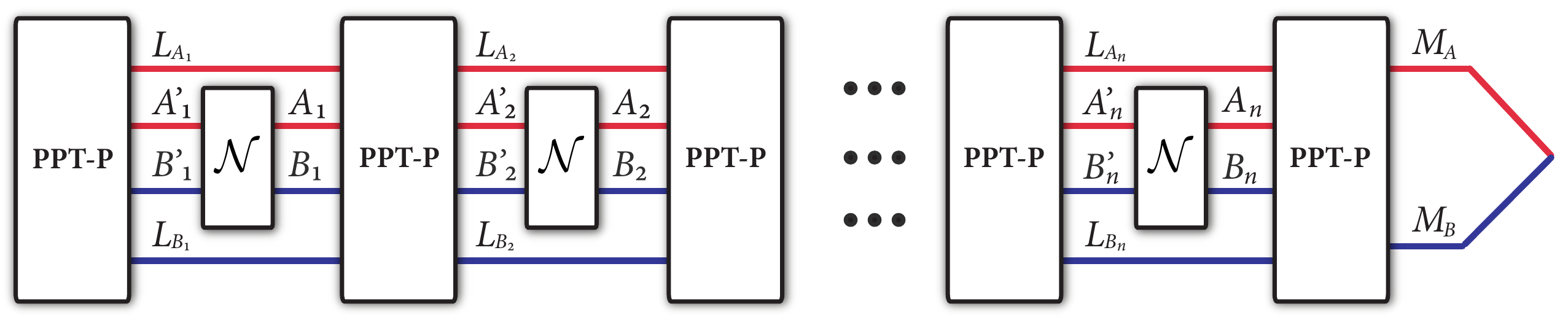}
		\caption{A protocol for PPT-assisted bidirectional quantum communication that uses a bidirectional quantum channel $\mc{N}$ $n$ times. Every channel use is interleaved by a PPT-preserving (PPT-P) channel. The goal of such a protocol is to produce an approximate maximally entangled state in the systems $M_A$ and $M_B$, where Alice possesses system $M_A$ and Bob system $M_B$.}\label{fig:bi-q-com}
	\end{figure}
%-------------------------------------------------
\subsection{Bidirectional max-Rains information} 
The following definition generalizes the  max-Rains information from \eqref{eq:max-Rains-channel}, \eqref{eq:rains-omega}, and \eqref{eq:rains-channel-sdp} to the bidirectional setting:

\begin{definition}[Bidirectional max-Rains information]\label{def:bi-max-rains}
The bidirectional max-Rains information of a bidirectional quantum channel $\mc{N}_{A'B'\to AB}$ is defined as
\begin{equation}
R_{\max}^{2\to 2}(\mc{N})\coloneqq \log \Gamma^{2\to2} (\mc{N}), \label{eq:bi-max-Rains-info}
\end{equation}
where $\Gamma^{2\to2}(\mc{N})$ is the solution to the following semi-definite program:
\begin{align}
\textnormal{minimize}\ &\ \norm{\Tr_{AB}\{V_{L_A ABL_B}+Y_{L_A ABL_B}\}}_\infty\nonumber\\
\textnormal{subject to}\ &\ V_{L_A ABL_B},Y_{L_A ABL_B}\geq 0,\nonumber\\
&\ \T_{BL_B}(V_{L_A ABL_B}-Y_{L_A ABL_B})\geq J^\mc{N}_{L_AABL_B},
\label{eq:bi-rains-channel-sdp}
\end{align} 
where $J^\mc{N}_{L_AABL_B}$ denotes the Choi operator of the bidirectional channel $\mc{N}$, such that $L_A\simeq A'$, and $L_B\simeq B'$.
\end{definition}

\begin{remark}
By employing the Lagrange multiplier method,
the bidirectional max-Rains information of a bidirectional channel $\mc{N}_{A'B'\to AB}$ can also be expressed as
\begin{equation}
R^{2\to 2}_{\max}(\mc{N})=\log \Gamma^{2\to2}(\mc{N}),
\end{equation}
where $\Gamma^{2\to2}(\mc{N})$ is solution to the following semi-definite program (SDP):
\begin{align}
\textnormal{maximize}\ &\ \Tr\{J^{\mc{N}}_{L_AABL_B}X_{L_A ABL_B}\}\nonumber\\
\textnormal{subject to}\ &\ X_{L_A ABL_B},\rho_{L_AL_B}\geq 0,\nonumber\\
&\ \Tr\{\rho_{L_A L_B}\}=1,\ -\rho_{L_A L_B}\otimes \bm{1}_{AB}\leq \T_{BL_B}( X_{L_A ABL_B})\leq \rho_{L_A L_B}\otimes \bm{1}_{AB},\label{eq:bi-rains-channel-sdp-primal}
\end{align}
such that $L_A\simeq A'$, and $L_B\simeq B'$. 
Strong duality holds by employing Slater's condition \cite{Wat15} (see also \cite{WD16a}). Thus, as indicated above, the optimal values of the primal and dual semi-definite programs, i.e., \eqref{eq:bi-rains-channel-sdp-primal} and \eqref{eq:bi-rains-channel-sdp}, respectively, are equal.
\end{remark}

The following proposition constitutes one of our main technical results, and an immediate corollary of it is that amortization does not enhance the bidirectional max-Rains information of a bidirectional quantum channel. 

\begin{proposition}[Amortization ineq.~for  bidirectional max-Rains info.]\label{prop:rains-tri-ineq}
Let $\rho_{L_A A'B'L_B}$ be an arbitrary state and let $\mc{N}_{A'B'\to AB}$ be a bidirectional channel. Then 
\begin{equation}
R_{\max}(L_A A; B L_B)_{\omega}\leq R_{\max}(L_A A';B' L_B)_{\rho}+ R^{2\to 2}_{\max}(\mc{N}),
\end{equation}
where $\omega_{L_A AB L_B}=\mc{N}_{A'B'\to AB}(\rho_{L_A A'B'L_B})$ and $R^{2\to 2}_{\max}(\mc{N})$ is the bidirectional max-Rains information of $\mc{N}_{A'B'\to AB}$. 
\end{proposition}
\begin{proof}
We adapt the proof steps of \cite[Proposition 1]{BW17} to the bidirectional setting. By removing logarithms and applying \eqref{eq:rains-w} and \eqref{eq:bi-max-Rains-info}, the desired inequality is equivalent to the following
\begin{equation}\label{eq:w-omega-ineq}
W(L_A A; B L_B)_{\omega}\leq W(L_A A';B' L_B)_{\rho}\cdot \Gamma^{2\to2}(\mc{N}),
\end{equation}
and so we aim to prove this one. Exploiting the identity in \eqref{eq:rains-state-sdp}, we find that 
\begin{equation}
W(L_A A'; B' L_B)_{\rho}=\min \Tr\{C_{L_AA'B'L_B}+D_{L_AA'B'L_B}\},
\end{equation}
subject to the constraints 
\begin{align}
C_{L_AA'B'L_B},D_{L_AA'B'L_B} &\geq 0,\\
\T_{B'L_{B}}(C_{L_AA'B'L_B}-D_{L_AA'B'L_B})&\geq \rho_{L_AA'B'L_B},
\end{align}
while the definition in \eqref{eq:bi-rains-channel-sdp} gives that 
\begin{equation}
\Gamma^{2\to2}(\mc{N})=\min \norm{\Tr_{AB}\{V_{R_A ABR_B}+Y_{R_A ABR_B}\}}_\infty,
\end{equation}
subject to the constraints 
\begin{align}
V_{R_AABR_B},Y_{R_AABR_B} &\geq 0,\\
 \T_{BR_B}(V_{R_A ABR_B}-Y_{R_A ABR_B})&\geq J^\mc{N}_{R_AABR_B}.\label{eq:choi-bi-b}
\end{align}
The identity in \eqref{eq:rains-state-sdp} implies that the left-hand side of \eqref{eq:w-omega-ineq} is equal to 
\begin{equation}
W(L_AA;BL_B)_\omega=\min \Tr\{E_{L_AABL_B}+F_{L_AABL_B}\},
\end{equation}
subject to the constraints
\begin{align}
E_{L_AABL_B},F_{L_AABL_B}&\geq 0,\label{eq:rains-sdp-ef}\\
\mc{N}_{A'B'\to AB}(\rho_{L_AA'B'L_B})&\leq \T_{BL_B}(E_{L_AABL_B}-F_{L_AABL_B})\label{eq:rains-sdp-channel-ef}.
\end{align}

Once we have these SDP formulations, we can now show that the inequality in \eqref{eq:w-omega-ineq} holds by making appropriate choices for $E_{L_AABL_B}, F_{L_AABL_B}$.  Let $C_{L_AA'B'L_B}$ and $ D_{L_AA'B'L_B}$ be optimal for $W(L_AA';B'L_B)_\rho$, and let $V_{L_AABL_B}$ and $Y_{L_AABL_B}$ be optimal for $\Gamma^{2\to2}(\mc{N})$. Let $\ket{\Upsilon}_{R_AR_B:A'B'}$ be the maximally entangled vector. Choose
\begin{align}
E_{L_AABL_B}&=\bra{\Upsilon}_{R_AR_B:A'B'}C_{L_AA'B'L_B}\otimes V_{R_AABR_B}+D_{L_AA'B'L_B}\otimes Y_{L_AABL_B}\ket{\Upsilon}_{R_AR_B:A'B'}\\
F_{L_AABL_B}&=\bra{\Upsilon}_{R_AR_B:A'B'}C_{L_AA'B'L_B}\otimes Y_{R_AABR_B}+D_{L_AA'B'L_B}\otimes V_{R_AABR_B}\ket{\Upsilon}_{R_AR_B:A'B'}.
\end{align}
The above choices can be thought of as bidirectional generalizations of those made in the proof of
\cite[Proposition 1]{BW17} (see also
\cite[Proposition 6]{WFD17}), and they can be understood roughly via \eqref{eq:choi-sim} as a post-selected teleportation of the optimal operators of $W(L_AA';B'L_B)_\rho$ through the optimal operators of $\Gamma^{2\to2}(\mc{N})$, with the optimal operators of $W(L_AA';B'L_B)_{\rho}$ being in correspondence with the Choi operator $J^\mc{N}_{R_AABR_B}$ through \eqref{eq:choi-bi-b}. Then, we have, $E_{L_AABL_B},F_{L_AABL_B}\geq 0$, because
\begin{equation}
C_{L_AA'B'L_B}, D_{L_AA'B'L_B}, Y_{R_AABR_B}, V_{R_AABR_B}\geq 0.
\end{equation}
Also, consider that
\begin{align}
&E_{L_AABL_B}-F_{L_AABL_B}\nonumber\\
&\quad = \bra{\Upsilon}_{R_AR_B:A'B'}(C_{L_AA'B'L_B}-D_{L_AA'B'L_B})\otimes (V_{R_AABR_B}- Y_{R_AABR_B})\ket{\Upsilon}_{R_AR_B:A'B'}\\
&\quad = \Tr_{R_AA'B'R_B}\{\ket{\Upsilon}\!\bra{\Upsilon}_{R_AR_B:A'B'}(C_{L_AA'B'L_B}-D_{L_AA'B'L_B})\otimes (V_{R_AABR_B}- Y_{R_AABR_B})\}.
\end{align}
 Then, using the abbreviations $E'\coloneqq~E_{L_AABL_B}$, $ F'\coloneqq~F_{L_AABL_B}$, $C'\coloneqq~C_{L_AA'B'L_B}$,  $D'\coloneqq~D_{L_AA'B'L_B}$, $V'\coloneqq~V_{R_AABR_B}$,  and $Y'\coloneqq~Y_{R_AABR_B}$, we have
\begin{align}
\T_{BL_B}(E'-F')
& = \T_{BL_B}\!\left[\Tr_{R_AA'B'R_B}\{\ket{\Upsilon}\!\bra{\Upsilon}_{R_AR_B:A'B'}(C'-D')\otimes (V'- Y')\}\right]\\
& = \T_{BL_B}\!\left[\Tr_{R_AA'B'R_B}\{\ket{\Upsilon}\!\bra{\Upsilon}_{R_AR_B:A'B'}(C'-D')\otimes(\T_{R_B}\circ\T_{R_B}) (V'- Y')\}\right]\\
& = \T_{BL_B}\!\left[\Tr_{R_AA'B'R_B}\{\T_{R_B}\ket{\Upsilon}\!\bra{\Upsilon}_{R_AR_B:A'B'} (C'-D')\otimes \T_{R_B} (V'- Y')\}\right]\\
& = \T_{BL_B}\!\left[\Tr_{R_AA'B'R_B}\{\ket{\Upsilon}\!\bra{\Upsilon}_{R_AR_B:A'B'} \T_{B'}(C'-D')\otimes \T_{R_B} (V'- Y')\}\right]\\
&=\Tr_{R_AA'B'R_B}\{\ket{\Upsilon}\!\bra{\Upsilon}_{R_AR_B:A'B'} \T_{B'L_B}(C'-D')\otimes \T_{BR_B} (V'- Y')\}\\
& \geq \bra{\Upsilon}_{R_AR_B:AB}\rho_{L_AA'B'L_B}\otimes J^\mc{N}_{R_AABR_B}\ket{\Upsilon}_{R_AR_B:AB}\\
& =\mc{N}_{A'B'\to AB}(\rho_{L_AA'B'L_B}).
\end{align}
In the above, we employed properties of the partial transpose reviewed in \eqref{eq:PT-1}--\eqref{eq:PT-last}.
Now, consider that
\begin{align}
&\Tr\{E_{L_AABL_B}+F_{L_AABL_B}\} \nonumber\\
& = \Tr\{\bra{\Upsilon}_{R_AR_B:A'B'}(C_{L_AA'B'L_B}+D_{L_AA'B'L_B})\otimes (V_{R_AABR_B}+ Y_{R_AABR_B})\ket{\Upsilon}_{R_AR_B:A'B'}\}\\
& = \Tr\{(C_{L_AA'B'L_B}+D_{L_AA'B'L_B}) T_{A'B'}(V_{A'ABB'}+ Y_{A'ABB'})\}\\
& = \Tr\{(C_{L_AA'B'L_B}+D_{L_AA'B'L_B}) T_{A'B'}(\Tr_{AB}\{V_{A'ABB'}+ Y_{A'ABB'})\}\}\\
& \leq \Tr\{(C_{L_AA'B'L_B}+D_{L_AA'B'L_B})\}\norm{ T_{A'B'}(\Tr_{AB}\{V_{A'ABB'}+ Y_{A'ABB'})\}}_\infty\\
& = \Tr\{(C_{L_AA'B'L_B}+D_{L_AA'B'L_B})\}\norm{ \Tr_{AB}\{V_{A'ABB'}+ Y_{A'ABB'}\}}_\infty\\
&=W(L_AA';B'L_B)_\rho\cdot \Gamma^{2\to2}(\mc{N}).
\end{align}
The inequality is a consequence of H\"{o}lder's inequality \cite{Bha97}. The final equality follows because the spectrum of a positive semi-definite operator is invariant under the action of a full transpose (note, in this case, $\T_{A'B'}$ is the full transpose as it acts on reduced positive semi-definite operators $V_{A'B'}$ and $Y_{A'B'}$).

Therefore, we can infer that our choices of $E_{L_AABL_B}, F_{L_AABL_B}$ are feasible for $W(L_AA;BL_B)_\omega$. Since $W(L_AA;BL_B)_\omega$ involves a minimization over all  $E_{L_AABL_B}, F_{L_AABL_B}$ satisfying \eqref{eq:rains-sdp-ef} and \eqref{eq:rains-sdp-channel-ef}, this concludes our proof of \eqref{eq:w-omega-ineq}.
\end{proof}
\bigskip

An immediate corollary of Proposition~\ref{prop:rains-tri-ineq} is the following:
\begin{corollary}\label{cor:rains-tri-ineq}
Amortization does not enhance the bidirectional max-Rains information of a bidirectional quantum channel $\mc{N}_{A'B'\to AB}$; i.e., the following inequality holds
\begin{equation}
R^{2\to 2}_{\max,A}(\mc{N})\leq R^{2\to 2}_{\max}(\mc{N}),
\end{equation}
where $R^{2\to 2}_{\max, A} (\mc{N})$ is a measure of the entangling power  of a bidirectional channel $\mc{N}$, i.e.,
\begin{equation}\label{eq:ent-locc-a}
R^{2\to 2}_{\max, A} (\mc{N})\coloneqq \sup_{\rho_{L_AA'B'L_B}\in\msc{D}(\mc{H}_{L_AA'B'L_B})} \left[R_{\max}(L_AA;BL_B)_{\sigma}-R_{\max}(L_AA';B'L_B)_{\rho}\right],
\end{equation}
and $\sigma_{L_AABL_B}\coloneqq \mc{N}_{A'B'\to AB}(\rho_{L_AA'B'L_B})$, where $L_A$ and $L_B$ can be arbitrarily large.  
\end{corollary}
\begin{proof}
The inequality $R^{2\to 2}_{\max,A}(\mc{N})\leq R^{2\to 2}_{\max}(\mc{N})$ is an immediate consequence of Proposition~\ref{prop:rains-tri-ineq}. Let $\rho_{L_AA'B'L_B}$ denote an arbitrary input state. Then from Proposition~\ref{prop:rains-tri-ineq} 
\begin{equation} \label{eq:r-lower-ineq}
R_{\max}(L_AA;BL_B)_\omega-R_{\max}(L_AA';B'L_B)_\rho\leq R^{2\to 2}_{\max}(\mc{N}),
\end{equation}
where $\omega_{L_AABL_B}=\mc{N}_{A'B'\to AB}(\rho_{L_AA'B'L_B})$. As the inequality holds for any state $\rho_{L_AA'B'L_B}$, we conclude that $R^{2\to 2}_{\max,A}(\mc{N})\leq R^{2\to 2}_{\max}(\mc{N})$.
\end{proof}

See Appendix~\ref{app:rmax} for some examples where the bidirectional max-Rains information of some channels are numerically evaluated. 

\subsection{Application to entanglement generation}\label{sec:ent-dist-protocol}
In this section, we discuss the implication of Proposition~\ref{prop:rains-tri-ineq} for PPT-assisted entanglement generation from a bidirectional channel\footnote{It is an open question whether or not NPT (non-positive under partial transpose) bound entangled states exist. However, it is known that all bipartite quantum states that are non-positive under partial transpose are distillable via some PPT-preserving channels\cite{EVWW01}. Therefore, in the standard case, the free operations allowed for the task of entanglement distillation are LOCC channels.}. Suppose that two parties Alice and Bob are connected by a bipartite quantum interaction. Suppose that the systems that Alice and Bob hold are $A'$ and $B'$, respectively. The bipartite quantum interaction between them is represented by a bidirectional quantum channel $\mc{N}_{A'B'\to AB}$, where output systems $A$ and $B$ are in possession of Alice and Bob, respectively. This kind of protocol was considered in \cite{BHLS03} when there is LOCC assistance.  
 
\subsubsection{Protocol for PPT-assisted entanglement generation}
We now discuss PPT-assisted entanglement generation protocols that make use of a bidirectional quantum channel. We do so by generalizing the point-to-point communication protocol discussed in \cite{KW17} to the bidirectional setting.  

In a PPT-assisted bidirectional protocol, as depicted in Figure~\ref{fig:bi-q-com}, Alice and Bob are spatially separated and they are allowed to undergo a bipartite quantum interaction $\mc{N}_{A'B'\to AB}$, where for a fixed basis $\{|i\>_B|j\>_{L_B}\}_{i,j}$, the partial transposition $T_{BL_B}$ is considered on systems associated to Bob. Alice holds systems labeled by $A', A$ whereas Bob holds $B',B$. They begin by performing a PPT-preserving channel $\mc{P}^{(1)}_{\oldemptyset\to L_{A_1}A_1'B_1'L_{B_1}}$, which leads to a PPT state $\rho^{(1)}_{L_{A_1}A_1'B_1'L_{B_1}}$, where $L_{A_1},L_{B_1}$ are finite-dimensional systems of arbitrary size and $A_1',B_1'$ are input systems to the first channel use. Alice and Bob send systems $A_1'$ and $B_1'$, respectively, through the first channel use, which yields the output state $\sigma^{(1)}_{L_{A_1}A_1B_1L_{B_1}}\coloneqq\mc{N}_{A_1'B_1'\to A_1B_1}(\rho^{(1)}_{L_{A_1}A_1'B_1'L_{B_1}})$. Alice and Bob then perform the PPT-preserving channel $\mc{P}^{(2)}_{L_{A_1}A_1B_1L_{B_1}\to L_{A_2}A_2'B_2'L_{B_2}}$, which leads to the state $\rho^{(2)}_{L_{A_2}A_2'B_2'L_{B_2}}\coloneqq \mc{P}^{(2)}_{L_{A_1}A_1B_1L_{B_1}\to L_{A_2}A_2'B_2'L_{B_2}}(\sigma^{(1)}_{L_{A_1}A_1B_1L_{B_1}})$. Both parties then send systems $A_2',B_2'$ through the second channel use $\mc{N}_{A_2'B_2'\to A_2B_2}$, which yields the state $\sigma^{(2)}_{L_{A_2}A_2B_2L_{B_2}}\coloneqq \mc{N}_{A_2'B_2'\to A_2B_2}(\rho^{(2)}_{L_{A_2}A_2'B_2'L_{B_2}})$. They iterate this process such that the protocol makes use of the channel $n$ times. In general, we have the following states for the $i$th use, for $i\in\{2,3,\ldots,n\}$:
\begin{align}
\rho^{(i)}_{L_{A_i}A_i'B_i'L_{B_i}} &\coloneqq \mc{P}^{(i)}_{L_{A_{i-1}}A_{i-1}B_{i-1}L_{B_{i-1}}\to L_{A_i}A_i'B_i'L_{B_i}}(\sigma^{(i-1)}_{L_{A_{i-1}}A_{i-1}B_{i-1}L_{B_{i-1}}}),\\
\sigma^{(i)}_{L_{A_i}A_iB_iL_{B_i}} &\coloneqq \mc{N}_{A_i'B_i'\to A_iB_i}(\rho^{(i)}_{L_{A_i}A_i'B_i'L_{B_i}}),
\end{align}
where $\mc{P}^{(i)}_{L_{A_{i-1}}A_{i-1}B_{i-1}L_{B_{i-1}}\to L_{A_i}A_i'B_i'L_{B_i}}$ is a PPT-preserving channel, with  the partial transposition acting on systems $B_{i-1},L_{B_{i-1}}$ associated to Bob. In the final step of the protocol, a PPT-preserving channel $\mc{P}^{(n+1)}_{L_{A_{n}}A_{n}B_{n}L_{B_{n}}\to M_AM_B}$ is applied, that generates the final state:
\begin{equation}
\omega_{M_AM_B}\coloneqq \mc{P}^{(n+1)}_{L_{A_{n}}A_{n}B_{n}L_{B_{n}}\to M_AM_B} (\sigma^{(n)}_{L_{A_n}A_n'B_n'L_{B_n}}),
\end{equation}
where $M_A$ and $M_B$ are held by Alice and Bob, respectively. 

The goal of the protocol is for Alice and Bob to distill entanglement in the end; i.e., the final state $\omega_{M_AM_B}$ should be close to a maximally entangled state $\Phi_{M_AM_B}$. For a fixed $n,\ |M|\in\mathbb{N},\ \varepsilon\in[0,1]$, the original protocol is an $(n,Q,\varepsilon)$ protocol if the channel is used $n$ times as discussed above, $|M_A|=|M_B|=|M|$, $Q\coloneqq \frac{1}{n}\log_2|M|$, and if 
\begin{align}
F(\omega_{M_AM_B},\Phi_{M_AM_B})&=\bra{\Phi}_{M_AM_B}\omega_{M_AM_B}\ket{\Phi}_{AB}\\
& \geq 1-\varepsilon,
\end{align}
where $\Phi_{M_AM_B}$ is the maximally entangled state. A rate $Q$ is said to be achievable for PPT-assisted entanglement generation if for all $\varepsilon\in(0,1]$, $\delta>0$, and sufficiently large $n$, there exists an $(n,Q-\delta,\varepsilon)$ protocol. The PPT-assisted entanglement generation capacity of a bidirectional channel $\mc{N}$, denoted as $Q^{2\to 2}_{\PPT}(\mc{N})$, is equal to the supremum of all achievable rates. Whereas, a rate $Q$ is a strong converse rate for PPT-assisted entanglement generation if for all $\varepsilon\in[0,1)$, $\delta>0$, and sufficiently large $n$, there does not exist an $(n,Q+\delta,\varepsilon)$ protocol. The strong converse PPT-assisted entanglement generation $\widetilde{Q}^{2\to 2}_{\PPT}(\mc{N})$ is equal to the infimum of all strong converse rates. A bidirectional channel $\mc{N}$ is said to obey the strong converse property for PPT-assisted entanglement generation if $Q^{2\to 2}_{\PPT}(\mc{N})=\widetilde{Q}^{2\to 2}_{\PPT}(\mc{N})$. 

Note that every LOCC channel is a PPT-preserving channel. Given this, the well-known fact that teleportation \cite{BBC+93} is an LOCC channel, and PPT-preserving channels are allowed for free in the above protocol, there is no difference between an $(n,Q,\varepsilon)$ entanglement generation protocol and an
$(n,Q,\varepsilon)$ quantum communication protocol. Thus, all of the capacities for entanglement generation are equal to those for quantum communication.

Also, we can consider the whole development discussed above for LOCC-assisted bidirectional quantum communication instead of more general PPT-assisted bidirectional quantum communication. All the notions discussed above follow when we restrict the class of assisting PPT-preserving channels allowed to be LOCC channels. It follows that the LOCC-assisted bidirectional quantum capacity $Q^{2\to 2}_{\LOCC}(\mc{N})$ and the strong converse LOCC-assisted quantum capacity $\widetilde{Q}^{2\to 2}_{\LOCC}(\mc{N})$ are bounded from above as
\begin{align}
Q^{2\to 2}_{\LOCC}(\mc{N})&\leq Q^{2\to 2}_{\PPT}(\mc{N}),\\
\widetilde{Q}^{2\to 2}_{\LOCC}(\mc{N})& \leq \widetilde{Q}^{2\to 2}_{\PPT}(\mc{N}).
\end{align} 
Also, the capacities of bidirectional quantum communication protocols without any assistance are always less than or equal to the LOCC-assisted bidirectional quantum capacities.

 The following lemma will be useful in deriving upper bounds on the bidirectional quantum capacities in the forthcoming sections, and it represents a generalization of the amortization idea to the bidirectional setting (see \cite{BHLS03} in this context).
 
\begin{lemma}\label{thm:ent-ppt-single-letter}
Let $\Ent_{\PPT}(A;B)_{\rho}$ be a bipartite entanglement measure for an arbitrary bipartite state $\rho_{AB}$. Suppose that $\Ent_{\PPT}(A;B)_{\rho}$ vanishes for all $\rho_{AB}\in \PPT(A\!:\!B)$ and is monotone non-increasing under PPT-preserving channels. Consider an $(n,M,\varepsilon)$ protocol for PPT-assisted entanglement generation over a bidirectional quantum channel $\mc{N}_{A'B'\to AB}$, as described in Section~\ref{sec:ent-dist-protocol}. Then, the following bound holds:
\begin{equation}
\Ent_{\PPT}(M_A;M_B)_\omega\leq n \Ent_{\PPT, A} (\mc{N}),
\end{equation}
where $\Ent_{\PPT, A} (\mc{N})$ is the amortized entanglement of a bidirectional channel $\mc{N}$, i.e.,
\begin{equation}\label{eq:ent-ppt-a}
\Ent_{\PPT, A} (\mc{N})\coloneqq \sup_{\rho_{L_AA'B'L_B}\in\msc{D}(\mc{H}_{L_AA'B'L_B})} \left[\Ent_{\PPT}(L_AA;BL_B)_{\sigma}-\Ent_{\PPT}(L_AA';B'L_B)_{\rho}\right],
\end{equation}
such that $\sigma_{L_AABL_B}\coloneqq \mc{N}_{A'B'\to AB}(\rho_{L_AA'B'L_B})$.
\end{lemma}

\begin{proof}
From the discussion above, as $\Ent_{\PPT}$ is monotonically non-increasing under the action of PPT-preserving channels, we get that
\begin{align}
\Ent_{\PPT}(M_A;M_B)_\omega &\leq \Ent_{\PPT}(L_{A_n}A_n;B_nL_{B_n})_{\sigma^{(n)}}\\
&= \Ent_{\PPT}(L_{A_n}A_n;B_nL_{B_n})_{\sigma^{(n)}}-\Ent_{\PPT}(L_{A_1}A'_1;B'_1L_{B_1})_{\rho^{(1)}}\\
&=\Ent_{\PPT}(L_{A_n}A_n;B_nL_{B_n})_{\sigma^{(n)}}\nonumber\\
&\quad\quad+\left[\sum_{i=2}^n \Ent_{\PPT}(L_{A_i}A'_i;B'_iL_{B_i})_{\rho^{(i)}}-\Ent_{\PPT}(L_{A_i}A'_i;B'_iL_{B_i})_{\rho^{(i)}}\right]\nonumber\\
&\quad\quad\quad - \Ent_{\PPT}(L_{A_1}A'_1;B'_1L_{B_1})_{\rho^{(1)}}\\
&\leq \sum_{i=1}^n\left[ \Ent_{\PPT}(L_{A_i}A_i;B_iL_{B_i})_{\sigma^{(i)}}-\Ent_{\PPT}(L_{A_i}A'_i;B'_iL_{B_i})_{\rho^{(i)}}\right]\\
&\leq n\Ent_{\PPT,A}(\mc{N}).
\end{align}
The first equality follows because $\rho^{(1)}_{L_{A_1}A_1'B_1'L_{B_1}}$ is a PPT state with vanishing $\Ent_{\PPT}$. The second equality follows trivially because we add and subtract the same terms. The second inequality follows because  $\Ent_{\PPT}(L_{A_i}A'_i;B'_iL_{B_i})_{\rho^{(i)}}\leq  \Ent_{\PPT}(L_{A_{i-1}}A_{i-1};B_{i-1}L_{B_{i-1}})_{\sigma^{(i-1)}}$ for all $i\in\{2,3,\ldots,n\}$, due to monotonicity of $\Ent_{\PPT}$ with respect to PPT-preserving channels. The final inequality follows by applying the definition in \eqref{eq:ent-ppt-a} to each summand. 
\end{proof}

\subsubsection{Strong converse rate for PPT-assisted entanglement generation}

We now establish the following upper bound on the entanglement generation rate $Q$ (qubits per channel use) of any $(n,Q,\varepsilon)$ PPT-assisted protocol:
\begin{theorem}\label{thm:rains-ent-dist-strong-converse}
For a fixed $n,\ |M|\in\mathbb{N},\ \varepsilon\in(0,1)$, the following bound holds for an $(n,Q,\varepsilon)$ protocol for PPT-assisted entanglement generation over a bidirectional quantum channel $\mc{N}$:
\begin{equation}\label{eq:rains-ent-dist-strong-converse}
Q\leq R^{2\to 2}_{\max}(\mc{N})+\frac{1}{n}\log_2\!\(\frac{1}{1-\varepsilon}\)
\end{equation}
such that $Q=\frac{1}{n}\log_2|M|$. 
\end{theorem}
\begin{proof}
From earlier discussion, we have that
\begin{equation}
\Tr\{\Phi_{M_AM_B}\omega_{M_AM_B}\}\geq 1-\varepsilon,
\end{equation}
while \cite[Lemma 2]{Rai99} implies that 
\begin{equation}
\forall \sigma_{M_AM_B}\in\PPT'(M_A:M_B),\ \Tr\{\Phi_{M_AM_B}\sigma_{M_AM_B}\}\leq \frac{1}{|M|}.
\end{equation}
Under an \textquotedblleft entanglement test\textquotedblright, which is a measurement with POVM $\{\Phi_{M_AM_B},\bm{1}_{M_AM_B}-\Phi_{M_AM_B}\}$, and applying the data processing inequality for the max-relative entropy, we find that 
\begin{equation}
R_{\max}(M_A;M_B)_\omega\geq \log_2[(1-\varepsilon)|M|]. \label{eq:rains-test-bound}
\end{equation}
Applying Lemma~\ref{thm:ent-ppt-single-letter} and Proposition~\ref{prop:rains-tri-ineq}, we get that
\begin{equation}
R_{\max}(M_A;M_B)_\omega \leq nR^{2\to 2}_{\max}(\mc{N}).\label{eq:rains-single-letter-proof}
\end{equation}
Combining \eqref{eq:rains-test-bound} and \eqref{eq:rains-single-letter-proof}, we get the desired inequality \eqref{eq:rains-ent-dist-strong-converse}. 
\end{proof}

\begin{remark}
The bound in \eqref{eq:rains-ent-dist-strong-converse} can also be rewritten as
\begin{equation}
1-\varepsilon \leq 2^{-n[Q-R^{2\to 2}_{\max}(\mc{N})]}.
\end{equation}
Thus, if the bidirectional communication rate $Q$ is strictly larger than the bidirectional max-Rains information $\mc{R}^{2\to 2}_{\max}(\mc{N})$, then the fidelity of the transmission $(1-\varepsilon)$ decays exponentially fast to zero in the number $n$ of channel uses. 
\end{remark}

An immediate corollary of the above remark is the following strong converse statement:
\begin{corollary}
The strong converse PPT-assisted bidirectional quantum capacity of a bidirectional channel $\mc{N}$ is bounded from above by its bidirectional max-Rains information:
\begin{equation}
\widetilde{Q}^{2\to 2}_{\PPT}(\mc{N})\leq R^{2\to 2}_{\max}(\mc{N}).
\end{equation}
\end{corollary}

\section{Secret key distillation from bipartite quantum interactions}\label{sec:priv-key}

In this section, we define the bidirectional max-relative entropy of entanglement $E^{2\to 2}_{\max}(\mc{N})$. The main goal of this section is to derive an upper bound on the rate at which secret key can be distilled from a bipartite quantum interaction. In deriving this bound, we consider private communication protocols over bidirectional quantum channels, and we make use of recent techniques developed in quantum information theory for point-to-point private communication protocols \cite{HHHO09,WTB16,CM17,KW17}. 

\subsection{Bidirectional generalized divergence of entanglement}

We define divergence based measures to quantify the ability of distilling secret key from a bipartite quantum channel. 

\begin{definition}\label{def:bi-max-sept}
The generalized divergence of entanglement from a bidirectional channel $\mc{N}_{A'B'\to AB}$ is defined as
\begin{equation}\label{eq:bi-max-sept-opt}
\tf{E}_{D}^{2\to 2}(\mc{N})=\sup_{\rho\in\SEP(L_AA':B'L_B)}\tf{E}(L_A A; B L_B)_{\omega},
\end{equation} 
where $\tf{E}(L_A A; B L_B)_{\omega}$ is a generalized divergence of entanglement of the state 
\begin{equation}
\omega_{L_A A B L_B}\coloneqq \mc{N}_{A'B'\to AB}(\rho_{L_AA'B'L_B}),
\end{equation} 
with $L_A$ and $L_B$ being arbitrarily large, 
\begin{align}\label{eq:b-max-sept-d}
\tf{E}(\hat{A}; \hat{B})_{\tau} \coloneqq \inf_{\sigma_{\hat{A}\hat{B}}\in\SEP(\hat{A}:\hat{B})}\tf{D}(\tau_{\hat{A}\hat{B}}\Vert\sigma_{\hat{A}\hat{B}}).
\end{align}
\end{definition}

The following definition generalizes a channel's max-relative entropy of entanglement from \cite{CM17} to the bidirectional setting, which we get after substituting generalized divergence $\tf{D}$ in \eqref{eq:b-max-sept-d} with the max-relative entropy $D_{\max}$:

\begin{definition}[Bidirectional max-relative entropy of entanglement]\label{def:bi-max-rel}
The bidirectional max-relative entropy of entanglement of a bidirectional channel $\mc{N}_{A'B'\to AB}$ is defined as
\begin{equation}\label{eq:bi-max-rel-opt}
E_{\max}^{2\to 2}(\mc{N})=\max_{\psi_{L_AA'} \otimes \varphi_{B'L_B}}E_{\max}(L_A A; B L_B)_{\omega},
\end{equation} 
where $\omega_{L_A A B L_B}\coloneqq \mc{N}_{A'B'\to AB}(\psi_{L_AA'} \otimes \varphi_{B'L_B})$, and $\psi_{L_AA'}\otimes\varphi_{B'L_B}\in\SEP(L_AA'\!:\!B'L_B)$ is a pure tensor-product state such that $L_A\simeq A'$, and $L_B\simeq B'$.
\end{definition}

\begin{remark}\label{rem:simplify-E-2-to-2-max}
Note that we could define $E_{\max}^{2\to 2}(\mc{N})$ to have an optimization over separable input states 
$\rho_{L_AA'B'L_B}\in \SEP(L_AA'\!:\!B'L_B)$ with finite-dimensional, but arbitrarily large auxiliary systems $L_A$ and $L_B$. However, the quasi-convexity of the max-relative entropy of entanglement \cite{D09,Dat09} and the Schmidt decomposition theorem guarantee that it suffices to restrict the optimization to be as stated in Definition~\ref{def:bi-max-rel}. 
\end{remark}

Analogous to definition of the bidirectional max-relative entropy of entanglement aforementioned, definition of the bidirectional relative entropy of entanglement $E_{D}^{2\to 2}(\mc{N})$ of an arbitrary bidirectional channel $\mc{N}$ is obtained by substituting generalized divergence in \eqref{eq:b-max-sept-d} with the relative entropy.

\begin{remark}\label{rem:bqi-sep-zero}
The bidirectional max-relative entropy of entanglement and the bidirectional relative entropy of entanglement of a bidirectional channel $\mc{N}_{A'B'\to AB}$ are both zero if and only if $\mc{N}_{A'B'\to AB}$ is a separable channel. 
\end{remark}

\begin{proposition}[Amortization ineq.~for bidirectional max-relative entropy]\label{prop:emax-tri-ineq}
Let $\rho_{L_AA'B'L_B}$ be an arbitrary state and let $\mc{N}_{A'B'\to AB}$ be a bidirectional channel. Then
\begin{equation}
E_{\max}(L_AA;BL_B)_\omega\leq  E_{\max}(L_AA';B'L_B)_\rho+E^{2\to 2}_{\max}(\mc{N}),
\end{equation}
where $\omega_{L_AABL_B}=\mc{N}_{A'B'\to AB}(\rho_{L_AA'B'L_B})$ and $E^{2\to 2}_{\max}(\mc{N})$ is the bidirectional max-relative entropy of entanglement of $\mc{N}_{A'B'\to AB}$.
\end{proposition}

\begin{proof}
Let us consider states $\sigma_{L_AA'B'L_B}'\in\SEP(L_AA'\!:\!B'L_B)$ and $\sigma_{L_AABL_B}\in\SEP(L_AA\!:\!BL_B)$, where $L_A$ and $L_B$ are finite-dimensional, but  arbitrarily large. With respect to the bipartite cut $L_AA:BL_B$, the following inequality holds
\begin{align}
E_{\max}(L_AA;BL_B)_\omega\leq D_{\max}(\mc{N}_{A'B'\to AB}(\rho_{L_AA'B'L_B})\Vert \sigma_{L_AABL_B}).
 \end{align}
Applying the data-processed triangle inequality \cite[Theorem III.1]{CM17}, we find that
\begin{align}
&D_{\max}(\mc{N}_{A'B'\to AB}(\rho_{L_AA'B'L_B})\Vert \sigma_{L_AABL_B})\nonumber\\
&\quad\quad\quad \leq D_{\max}(\rho_{L_AA'B'L_B}\Vert \sigma_{L_AA'B'L_B}')+D_{\max}(\mc{N}_{A'B'\to AB}(\sigma_{L_AA'B'L_B}')\Vert \sigma_{L_AABL_B}).
\end{align}
Since $\sigma_{L_AA'B'L_B}'$ and $\sigma_{L_AABL_B}$ are arbitrary separable states, we arrive at
\begin{equation}
E_{\max}(L_AA;BL_B)_\omega\leq  E_{\max}(L_AA';B'L_B)_\rho + E_{\max}(\mc{N}_{A'B'\to AB}(\sigma_{L_AA'B'L_B}')),
\end{equation}
where $\omega_{L_AABL_B}=\mc{N}_{A'B'\to AB}(\rho_{L_AA'B'L_B})$. This implies the desired inequality after applying the observation in Remark~\ref{rem:simplify-E-2-to-2-max}, given that 
$\sigma_{L_AA'B'L_B}' \in \SEP(L_AA'\!:\!B'L_B)$.
\end{proof}
\bigskip

An immediate consequence of Proposition~\ref{prop:emax-tri-ineq} is the following corollary:
\begin{corollary}
\label{cor:amort-max-rel-ent}
Amortization does not enhance the bidirectional max-relative entropy of entanglement of a bidirectional quantum channel $\mc{N}_{A'B'\to AB}$; and the following equality holds:
\begin{equation}
E^{2\to 2}_{\max,A}(\mc{N})=E^{2\to 2}_{\max}(\mc{N}),
\end{equation}  
where $E^{2\to 2}_{\max, A} (\mc{N})$ is the amortized entanglement of a bidirectional channel $\mc{N}$, i.e.,
\begin{equation}\label{eq:ent-locc-b}
E^{2\to 2}_{\max, A} (\mc{N})\coloneqq \sup_{\rho_{L_AA'B'L_B}\in\msc{D}(\mc{H}_{L_AA'B'L_B})} \left[E_{\max}(L_AA;BL_B)_{\sigma}-E_{\max}(L_AA';B'L_B)_{\rho}\right],
\end{equation}
and $\sigma_{L_AABL_B}\coloneqq \mc{N}_{A'B'\to AB}(\rho_{L_AA'B'L_B})$ where $L_A$ and $L_B$ can be arbitrary large.
\end{corollary}
\begin{proof}
The inequality $E^{2\to 2}_{\max,A}(\mc{N})\geq E^{2\to 2}_{\max}(\mc{N})$ always holds. The other inequality $E^{2\to 2}_{\max,A}(\mc{N})\leq E^{2\to 2}_{\max}(\mc{N})$ is an immediate consequence of Proposition~\ref{prop:emax-tri-ineq} (the argument is similar to that given in the proof of Corollary~\ref{cor:rains-tri-ineq}).
\end{proof}

\subsection{Application to secret key generation}
\subsubsection{Protocol for LOCC-assisted secret key generation}\label{sec:secret-dist-protocol}

We first discuss an LOCC-assisted secret key generation protocol that employs a bidirectional quantum channel.

In an LOCC-assisted secret key generation protocol, Alice and Bob are spatially separated and they are allowed to make use of a bipartite quantum interaction $\mc{N}_{A'B'\to AB}$, where the bipartite cut is considered between systems associated to Alice and Bob, $L_AA\!:\!L_BB$. Let $\mc{U}^\mc{N}_{A'B'\to ABE}$ be an isometric channel extending $\mc{N}_{A'B'\to AB}$:
\begin{equation}
\mc{U}^\mc{N}_{A'B'\to ABE}(\cdot)=U^{\mc{N}}_{A'B'\to ABE}(\cdot)\(U^{\mc{N}}_{A'B'\to ABE}\)^\dag,
\end{equation}
where $U^{\mc{N}}_{A'B'\to ABE}$ is an isometric extension of $\mc{N}_{A'B'\to AB}$. Let us assume that the eavesdropper Eve has access to the system $E$, also referred to as the environment, as well as a coherent copy of the classical communication exchanged between Alice and Bob. One could also consider a weaker assumption, in which the eavesdropper has access to only part of $E=E'E''$.

Alice and Bob begin by performing an LOCC channel $\mc{L}^{(1)}_{\oldemptyset\to L_{A_1}A_1'B_1'L_{B_1}}$, which leads to a state $\rho^{(1)}_{L_{A_1}A_1'B_1'L_{B_1}}\in\SEP(L_{A_1}A_1'\!:\!B_1'L_{B_1})$, where $L_{A_1},L_{B_1}$ are finite-dimensional systems of arbitrary size and $A_1',B_1'$ are input systems to the first channel use. Alice and Bob send systems $A_1'$ and $B_1'$, respectively, through the first channel use, that  outputs the state $\sigma^{(1)}_{L_{A_1}A_1B_1L_{B_1}}\coloneqq\mc{N}_{A_1'B_1'\to A_1B_1}(\rho^{(1)}_{L_{A_1}A_1'B_1'L_{B_1}})$. They then perform the LOCC channel $\mc{L}^{(2)}_{L_{A_1}A_1B_1L_{B_1}\to L_{A_2}A_2'B_2'L_{B_2}}$, which leads to the state $\rho^{(2)}_{L_{A_2}A_2'B_2'L_{B_2}}\coloneqq \mc{L}^{(2)}_{L_{A_1}A_1B_1L_{B_1}\to L_{A_2}A_2'B_2'L_{B_2}}(\sigma^{(1)}_{L_{A_1}A_1B_1L_{B_1}})$. Both parties then send systems $A_2',B_2'$ through the second channel use $\mc{N}_{A_2'B_2'\to A_2B_2}$, which yields the state $\sigma^{(2)}_{L_{A_2}A_2B_2L_{B_2}}\coloneqq \mc{N}_{A_2'B_2'\to A_2B_2}(\rho^{(2)}_{L_{A_2}A_2'B_2'L_{B_2}})$. They iterate the process such that the protocol uses the channel $n$ times. In general, we have the following states for the $i$th channel use, for $i\in[n]$:
\begin{align}
\rho^{(i)}_{L_{A_i}A_i'B_i'L_{B_i}} &\coloneqq \mc{L}^{(i)}_{L_{A_{i-1}}A_{i-1}B_{i-1}L_{B_{i-1}}\to L_{A_i}A_i'B_i'L_{B_i}}(\sigma^{(i-1)}_{L_{A_{i-1}}A_{i-1}B_{i-1}L_{B_{i-1}}}),\\
\sigma^{(i)}_{L_{A_i}A_iB_iL_{B_i}} &\coloneqq \mc{N}_{A_i'B_i'\to A_iB_i}(\rho^{(i)}_{L_{A_i}A_i'B_i'L_{B_i}}),
\end{align}
where $\mc{L}^{(i)}_{L_{A_{i-1}}A_{i-1}B_{i-1}L_{B_{i-1}}\to L_{A_i}A_i'B_i'L_{B_i}}$ is an LOCC channel corresponding to the bipartite cut $L_{A_{i-1}}A_{i-1}:B_{i-1}L_{B_{i-1}}$. In the final step of the protocol, an LOCC channel $\mc{L}^{(n+1)}_{L_{A_{n}}A_{n}B_{n}L_{B_{n}}\to K_AK_B}$ is applied, which generates the final state:
\begin{equation}
\omega_{K_AK_B}\coloneqq \mc{L}^{(n+1)}_{L_{A_{n}}A_{n}'B_{n}'L_{B_{n}}\to K_AK_B} (\sigma^{(n)}_{L_{A_n}A_n'B_n'L_{B_n}}),
\end{equation}
where the key systems $K_A$ and $K_B$ are held by Alice and Bob, respectively.

The goal of the protocol is for Alice and Bob to distill a secret key state, such that the systems $K_A$ and $K_B$ are maximally classical correlated and in tensor product with all of the systems that Eve possesses (see Section~\ref{sec:rev-priv-states} for a review of tripartite secret key states). 

\subsubsection{Purifying an LOCC-assisted secret key agreement protocol}\label{sec:priv-dist-protocol} 

As observed in \cite{HHHO05,HHHO09} and reviewed in Section~\ref{sec:rev-priv-states}, any protocol of the above form can be purified in the following sense. 

The initial state $\rho^{(1)}_{L_{A_1}A_1'B_1'L_{B_1}}\in\SEP(L_{A_1}A_1'\!:\!B_1'L_{B_1})$ is of the following form:
\begin{equation}
\rho^{(1)}_{L_{A_1}A_1'B_1'L_{B_1}}\coloneqq \sum_{y_1}p_{Y_1}(y_1)\tau^{y_1}_{L_{A_1}A_1'}\otimes\varsigma^{y_1}_{L_{B_1}B_1'}.
\end{equation}
The classical random variable $Y_1$ corresponds to a message exchanged between Alice and Bob to establish this state. It can be purified in the following way:
\begin{equation}
\vert \psi^{(1)}\rangle _{Y_1S_{A_1}L_{A_1}A_1'B_1'L_{B_1}S_{B_1}}\coloneqq \sum_{y_1}\sqrt{p_{Y_1}(y_1)}\ket{y_1}_{Y_1}\otimes\ket{\tau^{y_1}}_{S_{A_1}L_{A_1}A_1'}\otimes\ket{\varsigma^{y_1}}_{S_{B_1}L_{B_1}B_1'},
\end{equation}
where $S_{A_1}$ and $S_{B_1}$ are local ``shield" systems that in principle could be held by Alice and Bob, respectively, $\ket{\tau^{y_1}}_{S_{A_1}L_{A_1}A_1'}$ and $\ket{\varsigma^{y_1}}_{S_{B_1}L_{B_1}B_1'}$ purify $\tau^{y_1}_{L_{A_1}A_1'}$ and $\varsigma^{y_1}_{L_{B_1}B_1'}$, respectively, and Eve possesses system $Y_1$, which contains a coherent classical copy of the classical data exchanged between Alice and Bob. Each LOCC channel $\mc{L}^{(i)}_{L_{A_{i-1}}A_{i-1}B_{i-1}L_{B_{i-1}}\to L_{A_{i}}A_{i}'B_{i}'L_{B_{i}}}$ can be written in the following form \cite{Wat15}, for all $i\in \{2,3,\ldots,n\}$:
\begin{equation}\label{eq:locc}
\mc{L}^{(i)}_{L_{A_{i-1}}A_{i-1}B_{i-1}L_{B_{i-1}}\to L_{A_{i}}A_{i}'B_{i}'L_{B_{i}}}\coloneqq \sum_{y_i}\mc{E}^{y_i}_{L_{A_{i-1}}A_{i-1}\to L_{A_{i}}A_{i}'}\otimes\mc{F}^{y_i}_{B_{i-1}L_{B_{i-1}}\to B_{i}'L_{B_{i}}},
\end{equation}
where $\{\mc{E}^{y_i}_{L_{A_{i-1}}A_{i-1}\to L_{A_{i}}A_{i}'}\}_{y_i}$ and $\{\mc{F}^{y_i}_{B_{i-1}L_{B_{i-1}}\to B_{i}'L_{B_{i}}}\}_{y_i}$ are collections of completely positive, trace non-increasing maps such that the map in~\eqref{eq:locc} is trace preserving. Such an LOCC channel can be purified to an isometry in the following way:
\begin{equation}\label{eq:iso-locc}
U^{\mc{L}^{(i)}}_{L_{A_{i-1}}A_{i-1}B_{i-1}L_{B_{i-1}}\to Y_iS_{A_i}L_{A_{i}}A_{i}'B_{i}'L_{B_{i}}S_{B_i}}\coloneqq \sum_{y_i}\ket{y_i}_{Y_i}\otimes U^{\mc{E}^{y_i}}_{L_{A_{i-1}}A_{i-1}\to S_{A_i}L_{A_{i}}A_{i}'}\otimes U^{\mc{F}^{y_i}}_{B_{i-1}L_{B_{i-1}}\to B_{i}'L_{B_{i}}S_{B_i}},
\end{equation}
where $\{U^{\mc{E}^{y_i}}_{L_{A_{i-1}}A_{i-1}\to S_{A_i}L_{A_{i}}A_{i}'}\}_{y_i}$ and $\{U^{\mc{F}^{y_i}}_{B_{i-1}L_{B_{i-1}}\to B_{i}'L_{B_{i}}S_{B_i}}\}_{y_i}$ are collections of linear operators (each of which is a contraction, i.e.,  $\norm{U^{\mc{E}^{y_i}}_{L_{A_{i-1}}A_{i-1}\to S_{A_i}L_{A_{i}}A_{i}'}}_{\infty},\norm{U^{\mc{F}^{y_i}}_{B_{i-1}L_{B_{i-1}}\to B_{i}'L_{B_{i}}S_{B_i}}}_{\infty}\leq 1$ for all $y_i$) such that the linear operator $U^{\mc{L}^{(i)}}$ in \eqref{eq:iso-locc} is an isometry, the system $Y_i$ being held by Eve. The final LOCC channel can be written similarly as 
\begin{equation}
\mc{L}^{(n+1)}_{L_{A_{n}}A_{n}'B_{n}'L_{B_{n}}\to K_AK_B}\coloneqq \sum_{y_{n+1}}\mc{E}^{y_{n+1}}_{L_{A_{n}}A_{n}\to K_{A}}\otimes\mc{F}^{y_{n+1}}_{B_{n}L_{B_{n}}\to K_B},
\end{equation} 
and it can be purified to an isometry similarly as
\begin{equation}
 U^{\mc{L}^{(n+1)}}_{L_{A_{n}}A_{n}B_{n}L_{B_{n}}\to Y_{n+1} S_{A_{n+1}} K_A K_B S_{B_{n+1}}}\coloneqq \sum_{y_{n+1}}\ket{y_{n+1}}_{Y_{n+1}}\otimes U^{\mc{E}^{y_{n+1}}}_{L_{A_{n}}A_{n}\to S_{A_{n+1}}K_A}\otimes U^{\mc{F}^{y_{n+1}}}_{K_B S_{B_{n+1}}}.
\end{equation}
Furthermore, each channel use $\mc{N}_{A_i'B_i'\to A_iB_i}$, for all $i\in\{1,2,\ldots,n\}$, is purified by an isometry $U^\mc{N}_{A_i'B_i'\to A_iB_iE_i}$, such that Eve possesses the environment system $E_i$. 

At the end of the purified protocol, Alice possesses the key system $K_A$ and the shield systems $S_A\coloneqq S_{A_1}S_{A_2}\cdots S_{A_{n+1}}$, Bob possesses the key system $K_B$ and the shield systems $S_B\coloneqq S_{B_1}S_{B_2}\cdots S_{B_{n+1}}$, and Eve possesses the environment systems $E^n\coloneqq E_1E_2\cdots E_n$ as well as the coherent copies $Y^{n+1}\coloneqq Y_1Y_2\cdots Y_{n+1}$ of the classical data exchanged between Alice and Bob. The state at the end of the protocol is a pure state $\omega_{Y^{n+1}S_AK_AK_BS_BE^n}$.

For a fixed $n,\ |K|\in\mathbb{N},\ \varepsilon\in[0,1]$, the original protocol is an $(n,K,\varepsilon)$ protocol if the channel is used $n$ times as discussed above, $|K_A|=|K_B|=|K|$, and if 
\begin{align}
F(\omega_{S_AK_AK_BS_B},\gamma_{S_AK_AK_BS_B}) \geq 1-\varepsilon,
\end{align}
where $\gamma_{S_AK_AK_BS_B}$ is a bipartite private state. A rate $P$ is said to be achievable for LOCC-assisted secret key agreement if for all $\varepsilon\in(0,1]$, $\delta>0$, and sufficiently large $n$, there exists an $(n,P-\delta,\varepsilon)$ protocol. The LOCC-assisted secret-key-agreement capacity of a bidirectional channel $\mc{N}$, denoted as $P^{2\to 2}_{\LOCC}(\mc{N})$, is equal to the supremum of all achievable rates. Whereas, a rate $R$ is a strong converse rate for LOCC-assisted secret key agreement if for all $\varepsilon\in[0,1)$, $\delta>0$, and sufficiently large $n$, there does not exist an $(n,R+\delta,\varepsilon)$ protocol. The strong converse LOCC-assisted secret-key-agreement capacity $\widetilde{P}^{2\to 2}_{\LOCC}(\mc{N})$ is equal to the infimum of all strong converse rates. A bidirectional channel $\mc{N}$ is said to obey the strong converse property for LOCC-assisted secret key agreement  if $P^{2\to 2}_{\LOCC}(\mc{N})=\widetilde{P}^{2\to 2}_{\LOCC}(\mc{N})$. 

Note that the identity channel corresponding to no assistance is an LOCC channel. Therefore, we can also consider the whole development discussed above for bidirectional private communication without any assistance or feedback instead of LOCC-assisted communication. All the notions discussed above follow when we exempt the employment of any non-trivial LOCC-assistance. It follows that, unassisted bidirectional private capacity $P^{2\to 2}_{\textnormal{n-a}}(\mc{N})$ and the strong converse unassisted bidirectional private capacity $\widetilde{P}_{\textnormal{n-a}}^{2\to 2}(\mc{N})$ are bounded from above as
\begin{align}
P^{2\to 2}_{\textnormal{n-a}}(\mc{N})&\leq P^{2\to 2}_{\LOCC}(\mc{N}),\\
\widetilde{P}_{\textnormal{n-a}}^{2\to 2}(\mc{N})& \leq \widetilde{P}^{2\to 2}_{\LOCC}(\mc{N}).
\end{align} 

The following lemma will be useful in deriving upper bounds on the bidirectional secret-key-agreement capacity of a bidirectional channel. Its proof is very similar to the proof of Lemma~\ref{thm:ent-ppt-single-letter}, and so we omit it.
\begin{lemma}\label{thm:ent-locc-single-letter}
Let $\Ent_{\LOCC}(A;B)_{\rho}$ be a bipartite entanglement measure for an arbitrary bipartite state $\rho_{AB}$. Suppose that $\Ent_{\LOCC}(A;B)_{\rho}$ vanishes for all $\rho_{AB}\in \SEP(A\!:\!B)$ and is monotone non-increasing under LOCC channels. Consider an $(n,K,\varepsilon)$ protocol for LOCC-assisted secret key agreement over a bidirectional quantum channel $\mc{N}_{A'B'\to AB}$ as described in Section~\ref{sec:priv-dist-protocol}. Then the following bound holds:
\begin{equation}
\Ent_{\LOCC}(S_A K_A;K_B S_B)_{\omega}\leq n \Ent_{\LOCC, A} (\mc{N}),
\end{equation}
where $\Ent_{\LOCC, A} (\mc{N})$ is the amortized entanglement of a bidirectional channel $\mc{N}$, i.e.,
\begin{equation}\label{eq:ent-locc-c}
\Ent_{\LOCC, A} (\mc{N})\coloneqq \sup_{\rho_{L_AA'B'L_B}\in\msc{D}(\mc{H}_{L_AA'B'L_B})} \left[\Ent_{\LOCC}(L_AA;BL_B)_{\sigma}-\Ent_{\LOCC}(L_AA';B'L_B)_{\rho}\right],
\end{equation}
and $\sigma_{L_AABL_B}\coloneqq \mc{N}_{A'B'\to AB}(\rho_{L_AA'B'L_B})$.
\end{lemma}

\subsubsection{Strong converse rate for LOCC-assisted secret key agreement}
We now prove the following upper bound on the bidirectional secret key agreement rate $P=\frac{1}{n}\log_2 |K|$ (secret bits per channel use) of any $(n,P,\varepsilon)$ LOCC-assisted secret-key-agreement protocol over a bidirectional channel $\mc{N}$:

\begin{theorem}\label{thm:emax-ent-dist-strong-converse}
For a fixed $n,\ |K|\in\mathbb{N},\ \varepsilon\in(0,1)$, the following bound holds for an $(n,P,\varepsilon)$ protocol for LOCC-assisted secret key agreement over a bidirectional quantum channel $\mc{N}$:
\begin{equation}\label{eq:emax-ent-dist-strong-converse}
\frac{1}{n}\log_2K\leq E^{2\to 2}_{\max}(\mc{N})+
\frac{1}{n}\log_2\!\(\frac{1}{1-\varepsilon}\),
\end{equation}
such that $P=\frac{1}{n}\log_2|K|$.
\end{theorem}
\begin{proof}
From Section~\ref{sec:priv-dist-protocol}, the following inequality holds for an $(n,|K|,\varepsilon)$ protocol:
\begin{equation}
F(\omega_{S_AK_AK_BS_B},\gamma_{S_AK_AK_BS_B})\geq 1-\varepsilon,
\end{equation}
for some bipartite private state $\gamma_{S_AK_AK_BS_B}$ with key dimension $|K|$. From Section~\ref{sec:rev-priv-states}, $\omega_{S_AK_AK_BS_B}$ passes a $\gamma$-privacy test with probability  at least $1-\varepsilon$, whereas any $\tau_{S_AK_AK_BS_B}\in\SEP(S_AK_A:K_BS_B)$ does not pass with probability greater than $\frac{1}{|K|}$ \cite{HHHO09}. Making use of the discussion in \cite[Sections III \& IV]{CM17} (i.e., from the monotonicity of the max-relative entropy of entanglement under the $\gamma$-privacy test), it can be concluded that
\begin{equation}\label{eq:emax-test-bound}
\log_2 |K| \leq
E_{\max}(S_AK_A;K_BS_B)_{\omega}+\log_2\!\(\frac{1}{1-\varepsilon}\).
\end{equation}
Applying Lemma~\ref{thm:ent-locc-single-letter} and Corollary~\ref{cor:amort-max-rel-ent}, we get that
\begin{equation}
E_{\max}(S_AK_A;K_BS_B)_\omega \leq nE^{2\to 2}_{\max}(\mc{N}).\label{eq:emax-single-letter-proof}
\end{equation}
Combining \eqref{eq:emax-test-bound} and \eqref{eq:emax-single-letter-proof}, we get the desired inequality in \eqref{eq:emax-ent-dist-strong-converse}. 
\end{proof}

\begin{remark}
Note that Theorem~\ref{thm:emax-ent-dist-strong-converse} applies in the case that the bidirectional channel $\mc{N}_{A'B'\to AB}$ is an infinite-dimensional bipartite channel, taking input density operators acting on a separable Hilbert space to output density operators acting on a separable Hilbert space. We arrive at this conclusion because the max-relative entropy is well defined for infinite-dimensional states.
\end{remark}

\begin{remark}
The bound in \eqref{eq:emax-ent-dist-strong-converse} can also be rewritten as
\begin{equation}
1-\varepsilon \leq 2^{-n[P-E^{2\to 2}_{\max}(\mc{N})]}.
\end{equation}
Thus, if the bidirectional secret-key-agreement rate $P$ is strictly larger than the bidirectional max-relative entropy of entanglement $\mc{E}^{2\to 2}_{\max}(\mc{N})$, then the reliability and security  of the transmission $(1-\varepsilon)$ decays exponentially fast to zero in the number $n$ of channel uses. 
\end{remark}

An immediate corollary of the above remark is the following strong converse statement:
\begin{corollary}
The strong converse LOCC-assisted bidirectional secret-key-agreement capacity of a bidirectional channel $\mc{N}$ is bounded from above by its bidirectional max-relative entropy of entanglement:
\begin{equation}
\widetilde{P}^{2\to 2}_{\LOCC}(\mc{N})\leq E^{2\to 2}_{\max}(\mc{N}).
\end{equation}
\end{corollary}

\section{Entangling abilities of symmetric interactions}\label{sec:ent-mes-sim}
Interactions obeying particular symmetries have played an important role in several quantum information processing tasks in the context of quantum communication protocols \cite{BDSW96,HHH99,Hol02}, quantum computing and quantum metrology \cite{DP05,JWD+08,DM14}, and resource theories \cite{Fri15,BG15}, etc. 

In this section, we define bidirectional PPT- and teleportation-simulable channels by adapting the definitions of point-to-point PPT- and LOCC-simulable channels \cite{BDSW96,HHH99,KW17} to the bidirectional setting. Then, we derive upper bounds on the entanglement and secret-key-agreement capacities for communication protocols that employ bidirectional PPT- and teleportation-simulable channels, respectively. These bounds are generally tighter than those given in the previous section, because they exploit the symmetry inherent in bidirectional PPT- and teleportation-simulable channels.

\begin{definition}[Bidirectional PPT-simulable]\label{def:bi-ppt-sim}
A bidirectional channel $\mc{N}_{A'B'\to AB}$ is PPT-simulable
with an associated resource state 
$\theta_{\hat{S}_A\hat{S}_B}\in\msc{D}\(\mc{H}_{\hat{S}_A\hat{S}_B}\)$
if for all input states $\rho_{A'B'}\in\msc{D}\(\mc{H}_{A'B'}\)$ the following equality holds
\begin{equation}
\mc{N}_{A'B'\to AB}\(\rho_{A'B'}\)=\mc{P}_{\hat{S}_AA'B'\hat{S}_B\to AB}\(\rho_{A'B'}\otimes\theta_{\hat{S}_A\hat{S}_B}\),
\end{equation}
with $\mc{P}_{\hat{S}_AA'B'\hat{S}_B\to AB}$ being a PPT-preserving channel acting on $\hat{S}_AA'\!:\!B'\hat{S}_B$, where the partial transposition acts on the composite system $B'\hat{S}_B$.
\end{definition}    

The following definition was given in 
\cite{STM11} for the special case of bipartite unitary channels:
\begin{definition}[Bidirectional teleportation-simulable]\label{def:bi-tel-sim}
A bidirectional channel $\mc{N}_{A'B'\to AB}$ is teleportation-simulable
with associated resource state 
$\theta_{\hat{S}_A\hat{S}_B}\in\msc{D}\(\mc{H}_{\hat{S}_A\hat{S}_B}\)$
if for all input states $\rho_{A'B'}\in\msc{D}\(\mc{H}_{A'B'}\)$ the following equality holds
\begin{equation}
\mc{N}_{A'B'\to AB}\(\rho_{A'B'}\)=\mc{L}_{\hat{S}_AA'B'\hat{S}_B\to AB}\(\rho_{A'B'}\otimes\theta_{\hat{S}_A\hat{S}_B}\),
\end{equation}
where $\mc{L}_{\hat{S}_AA'B'\hat{S}_B\to AB}$ is an LOCC channel acting on $\hat{S}_AA':B'\hat{S}_B$.
\end{definition} 

Let $\msc{G}$ and $\msc{H}$ be finite groups of sizes $|G|$ and $|H|$, respectively. For $g\in \msc{G}$ and $h\in \msc{H}$, let
$g\rightarrow U_{A^{\prime}}(g)$ and $h\rightarrow V_{B^{\prime}}(h)$ be
unitary representations. Also, let $(g,h)\rightarrow W_{A}(g,h)$ and
$(g,h)\rightarrow T_{B}(g,h)$ be unitary representations. A bidirectional
quantum channel $\mathcal{N}_{A^{\prime}B^{\prime}\rightarrow AB}$ is
\textit{bicovariant} with respect to these representations if the following relation
holds for all input density operators $\rho_{A^{\prime}B^{\prime}}$ and group
elements $g\in \msc{G}$ and $h\in \msc{H}$:
\begin{equation}
\mathcal{N}_{A^{\prime}B^{\prime}\rightarrow AB}((\mathcal{U}_{A^{\prime}%
}(g)\otimes\mathcal{V}_{B^{\prime}}(h))(\rho_{A^{\prime}B^{\prime}%
}))=(\mathcal{W}_{A}(g,h)\otimes\mathcal{T}_{B}(g,h))(\mathcal{N}_{A^{\prime
}B^{\prime}\rightarrow AB}(\rho_{A^{\prime}B^{\prime}})),
\end{equation}
where $\mathcal{U}(g)(\cdot)\coloneqq U(g)(\cdot)\left(
U(g)\right)^{\dag}$, $\mathcal{V}(h)(\cdot)\coloneqq V(h)(\cdot)\left(
V(h)\right)^{\dag}$, $\mathcal{T}(g,h)(\cdot)\coloneqq T(g,h)(\cdot)\left(
T(g,h)\right)^{\dag}$, and $\mathcal{W}(g,h)(\cdot)\coloneqq W(g,h)(\cdot)\left(
W(g,h)\right)^{\dag}$ are unitary channels associated with respective unitary operators.

\begin{definition}[Bicovariant channel]\label{def:bicov}
A bidirectional channel is called bicovariant if it is bicovariant with
respect to groups that have representations as unitary one-designs, i.e., for all $\rho_{A'}\in\msc{D}(\mc{H}_{A'})$ and $\rho_{B'}\in\msc{D}(\mc{H}_{B'})$,
\begin{equation}
\frac{1}{\left\vert G\right\vert }\sum_{g}\mathcal{U}_{A^{\prime}}%
(g)(\rho_{A^{\prime}})=\pi_{A^{\prime}}\ \tn{and}\ \frac{1}{\left\vert H\right\vert
}\sum_{h}\mathcal{V}_{B^{\prime}}(h)(\rho_{B^{\prime}})=\pi_{B^{\prime}},
\end{equation}
where $\pi_{A'}$ and $\pi_{B'}$ are maximally mixed states.  
\end{definition}

An example of a bidirectional channel that is bicovariant is the
controlled-NOT (CNOT) gate \cite{BDEJ95}, for which we have the following covariances \cite{G99,GC99}:
\begin{align}
\text{CNOT}(X\otimes \bm{1})  & =(X\otimes X)\text{CNOT},\\
\text{CNOT}(Z\otimes \bm{1})  & =(Z\otimes \bm{1})\text{CNOT},\\
\text{CNOT}(Y\otimes \bm{1})  & =(Y\otimes X)\text{CNOT},\\
\text{CNOT}(\bm{1}\otimes X)  & =(\bm{1}\otimes X)\text{CNOT},\\
\text{CNOT}(\bm{1}\otimes Z)  & =(Z\otimes Z)\text{CNOT},\\
\text{CNOT}(\bm{1}\otimes Y)  & =(Z\otimes Y)\text{CNOT},
\end{align}
where $\{\bm{1},X,Y,Z\}$ is the Pauli group with the identity element $\bm{1}$. A more general example of a bicovariant channel is one that applies a CNOT with some probability and,  with the complementary probability, replaces the input with the maximally mixed state.

In \cite{GC99}, the prominent idea of gate teleportation was developed, wherein one can generate the Choi state for the CNOT gate by sending in shares of maximally entangled states and then  simulate the CNOT gate's action on any input state by using teleportation through the Choi state (see also \cite{NC97} for earlier related developments). This idea generalized the notion of teleportation simulation of channels \cite{BDSW96,HHH99} from the single-sender single-receiver setting to the bidirectional setting. After these developments, \cite{CDKL01,DBB08} generalized the idea of gate teleportation to bipartite quantum channels that are not necessarily unitary channels. 

The following result slightly generalizes the developments in \cite{GC99,CDKL01,DBB08}:
\begin{proposition}\label{prop:bicov}
If a bidirectional channel $\mathcal{N}_{A^{\prime}B^{\prime}\rightarrow AB}$
is bicovariant, Definition~\ref{def:bicov}, then it is teleportation-simulable with a resource state
$\theta_{L_{A}ABL_{B}}=\mathcal{N}_{A^{\prime}B^{\prime}\rightarrow AB}%
(\Phi_{L_{A}A^{\prime}}\otimes\Phi_{B^{\prime}L_{B}})$.
\end{proposition}
\begin{proof}
Let $\mathcal{N}_{A^{\prime}B^{\prime}\rightarrow AB}$ be a bidirectional
quantum channel. Given $\msc{G}$ and $\msc{H}$ are groups with unitary representations
$g\rightarrow\mathcal{U}_{A^{\prime}}(g)$ and $h\rightarrow V_{B^{\prime}}(h)$
and $(g,h)\rightarrow W_{A}(g,h)$ and $(g,h)\rightarrow T_{B}(g,h)$,
such that%
\begin{equation}
\frac{1}{\left\vert G\right\vert }\sum_{g}\mathcal{U}_{A^{\prime}
}(g)(X_{A^{\prime}})  =\operatorname{Tr}\{X_{A^{\prime}}\}\pi_{A^{\prime}
},\label{eq:one-design-cond}
\end{equation}
\begin{equation}
\frac{1}{\left\vert H\right\vert }\sum_{h}\mathcal{V}_{B^{\prime}
}(h)(Y_{B^{\prime}})  =\operatorname{Tr}\{Y_{B^{\prime}}\}\pi_{B^{\prime}
},
\end{equation}
\begin{equation}
\mathcal{N}_{A^{\prime}B^{\prime}\rightarrow AB}((\mathcal{U}_{A^{\prime}
}(g)\otimes\mathcal{V}_{B^{\prime}}(h))(\rho_{A^{\prime}B^{\prime}})) 
=(\mathcal{W}_{A}(g,h)\otimes\mathcal{T}_{B}(g,h))(\mathcal{N}_{A^{\prime
}B^{\prime}\rightarrow AB}(\rho_{A^{\prime}B^{\prime}})),
\end{equation}
where $X_{A^{\prime}}\in\mathcal{B}(\mathcal{H}_{A^{\prime}})$, $Y_{B^{\prime
}}\in\mathcal{B}(\mathcal{H}_{B^{\prime}})$, and $\pi$ denotes the maximally
mixed state. Consider that
\begin{equation}
\frac{1}{\left\vert G\right\vert }\sum_{g}\mathcal{U}_{A^{\prime\prime}%
}(g)(\Phi_{A^{\prime\prime}A^{\prime}})=\pi_{A^{\prime\prime}}\otimes
\pi_{A^{\prime}},\label{eq:max-ent-cov-action}%
\end{equation}
where $\Phi$ denotes a maximally entangled state and $A^{\prime\prime}$ is a
system isomorphic to $A^{\prime}$. Similarly,
\begin{equation}
\frac{1}{\left\vert H\right\vert }\sum_{h}\mathcal{V}_{B^{\prime\prime}%
}(h)(\Phi_{B^{\prime\prime}B^{\prime}})=\pi_{B^{\prime\prime}}\otimes
\pi_{B^{\prime}}.
\end{equation}
Note that in order for $\{U_{A^{\prime}}^{g}\}$ to satisfy
\eqref{eq:one-design-cond}, it is necessary that $\left\vert A^{\prime
}\right\vert ^{2}\leq\left\vert G\right\vert $ \cite{AMTW00}. Similarly, it is
necessary that $\left\vert B^{\prime}\right\vert ^{2}\leq\left\vert
H\right\vert $. Consider the POVM $\{E_{A^{\prime\prime}L_{A}}^{g}\}_{g}$,
with each element $E_{A^{\prime\prime}L_{A}}^{g}$ defined as%
\begin{equation}
E_{A^{\prime\prime}L_{A}}^{g}:=\frac{\left\vert A^{\prime}\right\vert
^{2}}{\left\vert G\right\vert }U_{A^{\prime\prime}}^{g}\Phi_{A^{\prime\prime
}L_{A}}\left(  U_{A^{\prime\prime}}^{g}\right)  ^{\dag}.
\end{equation}
It follows from the fact that $\left\vert A^{\prime}\right\vert ^{2}%
\leq\left\vert G\right\vert $ and \eqref{eq:max-ent-cov-action}\ that
$\{E_{A^{\prime\prime}L_{A}}^{g}\}_{g}$ is a valid POVM. Similarly,  let us fine
the POVM\ $\{F_{B^{\prime\prime}L_{B}}^{h}\}_{h}$ as%
\begin{equation}
F_{B^{\prime\prime}L_{B}}^{h}:=\frac{\left\vert B^{\prime}\right\vert
^{2}}{\left\vert H\right\vert }V_{B^{\prime\prime}}^{h}\Phi_{B^{\prime\prime
}L_{B}}\left(  V_{B^{\prime\prime}}^{h}\right)  ^{\dag}.
\end{equation}

The simulation of the channel $\mathcal{N}_{A^{\prime}B^{\prime}\rightarrow
AB}$ via teleportation begins with a state $\rho_{A^{\prime\prime}%
B^{\prime\prime}}$ and a shared resource $\theta_{L_{A}ABL_{B}}=\mathcal{N}%
_{A^{\prime}B^{\prime}\rightarrow AB}(\Phi_{L_{A}A^{\prime}}\otimes
\Phi_{B^{\prime}L_{B}})$. The desired outcome is for the receivers to receive
the state $\mathcal{N}_{A^{\prime}B^{\prime}\rightarrow AB}(\rho_{A^{\prime
}B^{\prime}})$ and for the protocol to work independently of the input state
$\rho_{A^{\prime}B^{\prime}}$. The first step is for the senders to locally
perform the measurement $\{E_{A^{\prime\prime}L_{A}}^{g}\otimes F_{B^{\prime
\prime}L_{B}}^{h}\}_{g,h}$ and then send the outcomes $g$\ and $h$ to the
receivers. Based on the outcomes $g$ and $h$, the receivers then perform
$W_{A}^{g,h}$ and $T_{B}^{g,h}$. The following analysis demonstrates that this
protocol works, by simplifying the form of the post-measurement state:%
\begin{align}
&  \left\vert G\right\vert \left\vert H\right\vert \operatorname{Tr}%
_{A^{\prime\prime}L_{A}B^{\prime\prime}L_{B}}\{(E_{A^{\prime\prime}L_{A}}%
^{g}\otimes F_{B^{\prime\prime}L_{B}}^{h})(\rho_{A^{\prime\prime}%
B^{\prime\prime}}\otimes\theta_{L_{A}ABL_{B}})\}\nonumber\\
&  =\left\vert A^{\prime}\right\vert ^{2}\left\vert B^{\prime}\right\vert
^{2}\operatorname{Tr}_{A^{\prime\prime}L_{A}B^{\prime\prime}L_{B}%
}\{[U_{A^{\prime\prime}}^{g}\Phi_{A^{\prime\prime}L_{A}}\left(  U_{A^{\prime
\prime}}^{g}\right)  ^{\dag}\otimes V_{B^{\prime\prime}}^{h}\Phi
_{B^{\prime\prime}L_{B}}\left(  V_{B^{\prime\prime}}^{h}\right)  ^{\dag}%
](\rho_{A^{\prime\prime}B^{\prime\prime}}\otimes\theta_{L_{A}ABL_{B}})\}\\
&  =\left\vert A^{\prime}\right\vert ^{2}\left\vert B^{\prime}\right\vert
^{2}\langle\Phi|_{A^{\prime\prime}L_{A}}\otimes\langle\Phi|_{B^{\prime\prime
}L_{B}}\left(  U_{A^{\prime\prime}}^{g}\otimes V_{B^{\prime\prime}}%
^{h}\right)  ^{\dag}(\rho_{A^{\prime\prime}B^{\prime\prime}}\otimes
\theta_{L_{A}ABL_{B}})(U_{A^{\prime\prime}}^{g}\otimes V_{B^{\prime\prime}%
}^{h})|\Phi\rangle_{A^{\prime\prime}L_{A}}\otimes|\Phi\rangle_{B^{\prime
\prime}L_{B}}\\
&  =\left\vert A^{\prime}\right\vert ^{2}\left\vert B^{\prime}\right\vert
^{2}\langle\Phi|_{A^{\prime\prime}L_{A}}\otimes\langle\Phi|_{B^{\prime\prime
}L_{B}}\left[  \left(  U_{A^{\prime\prime}}^{g}\otimes V_{B^{\prime\prime}%
}^{h}\right)  ^{\dag}\rho_{A^{\prime\prime}B^{\prime\prime}}(U_{A^{\prime
\prime}}^{g}\otimes V_{B^{\prime\prime}}^{h})\right]  \nonumber\\
&  \qquad\qquad\otimes\mathcal{N}_{A^{\prime}B^{\prime}\rightarrow AB}%
(\Phi_{L_{A}A^{\prime}}\otimes\Phi_{B^{\prime}L_{B}}))|\Phi\rangle
_{A^{\prime\prime}L_{A}}\otimes|\Phi\rangle_{B^{\prime\prime}L_{B}}\\
&  =\left\vert A^{\prime}\right\vert ^{2}\left\vert B^{\prime}\right\vert
^{2}\langle\Phi|_{A^{\prime\prime}L_{A}}\otimes\langle\Phi|_{B^{\prime\prime
}L_{B}}\left[  \left(  U_{L_{A}}^{g}\otimes V_{L_{B}}^{h}\right)  ^{\dag}%
\rho_{L_{A}L_{B}}(U_{L_{A}}^{g}\otimes V_{L_{B}}^{h})\right]  ^{\ast
}\nonumber\\
&  \qquad\qquad\mathcal{N}_{A^{\prime}B^{\prime}\rightarrow AB}(\Phi
_{L_{A}A^{\prime}}\otimes\Phi_{B^{\prime}L_{B}}))|\Phi\rangle_{A^{\prime
\prime}L_{A}}\otimes|\Phi\rangle_{B^{\prime\prime}L_{B}}%
.\label{eq:cov-tp-simul-block-1}%
\end{align}
The first three equalities follow by substitution and some rewriting. The fourth equality follows from the fact that
\begin{equation}
\langle\Phi|_{A^{\prime}A}M_{A^{\prime}}=\langle\Phi|_{A^{\prime}A}M_{A}%
^{\ast}\label{eq:ricochet-prop}%
\end{equation}
for any operator $M$ and where $\ast$ denotes the complex conjugate, taken
with respect to the basis in which $|\Phi\rangle_{A^{\prime}A}$ is defined.
Continuing, we have that%
\begin{align}
\eqref{eq:cov-tp-simul-block-1} &  =\left\vert A^{\prime}\right\vert
\left\vert B^{\prime}\right\vert \operatorname{Tr}_{L_{A}L_{B}}\left\{
\left[  \left(  U_{L_{A}}^{g}\otimes V_{L_{B}}^{h}\right)  ^{\dag}\rho
_{L_{A}L_{B}}(U_{L_{A}}^{g}\otimes V_{L_{B}}^{h})\right]  ^{\ast}%
\mathcal{N}_{A^{\prime}B^{\prime}\rightarrow AB}(\Phi_{L_{A}A^{\prime}}%
\otimes\Phi_{B^{\prime}L_{B}}))\right\}  \\
&  =\left\vert A^{\prime}\right\vert \left\vert B^{\prime}\right\vert
\operatorname{Tr}_{L_{A}L_{B}}\left\{  \mathcal{N}_{A^{\prime}B^{\prime
}\rightarrow AB}\left(  \left[  \left(  U_{A^{\prime}}^{g}\otimes
V_{B^{\prime}}^{h}\right)  ^{\dag}\rho_{A^{\prime}B^{\prime}}(U_{A^{\prime}%
}^{g}\otimes V_{B^{\prime}}^{h})\right]  ^{\dag}\left(  \Phi_{L_{A}A^{\prime}%
}\otimes\Phi_{B^{\prime}L_{B}}\right)  \right)  \right\}  \\
&  =\mathcal{N}_{A^{\prime}B^{\prime}\rightarrow AB}\left(  \left[  \left(
U_{A^{\prime}}^{g}\otimes V_{B^{\prime}}^{h}\right)  ^{\dag}\rho_{A^{\prime
}B^{\prime}}(U_{A^{\prime}}^{g}\otimes V_{B^{\prime}}^{h})\right]  ^{\dag
}\right)  \\
&  =\mathcal{N}_{A^{\prime}B^{\prime}\rightarrow AB}\left(  \left(
U_{A^{\prime}}^{g}\otimes V_{B^{\prime}}^{h}\right)  ^{\dag}\rho_{A^{\prime
}B^{\prime}}(U_{A^{\prime}}^{g}\otimes V_{B^{\prime}}^{h})\right)  \\
&  =\left(  W_{A}^{g,h}\otimes T_{B}^{g,h}\right)  ^{\dag}\mathcal{N}%
_{A^{\prime}B^{\prime}\rightarrow AB}\left(  \rho_{A^{\prime}B^{\prime}%
}\right)  (W_{A}^{g,h}\otimes T_{B}^{g,h})
\end{align}
The first equality follows because $\left\vert A\right\vert \langle
\Phi|_{A^{\prime}A}\left(  \bm{1}_{A^{\prime}}\otimes M_{AB}\right)  |\Phi
\rangle_{A^{\prime}A}=\operatorname{Tr}_{A}\{M_{AB}\}$ for any operator
$M_{AB}$. The second equality follows by applying the conjugate transpose of
\eqref{eq:ricochet-prop}. The final equality follows from the covariance
property of the channel.

Thus, if the receivers finally perform the unitaries $W_{A}^{g,h}\otimes
T_{B}^{g,h}$ upon receiving $g$ and $h$ via a classical channel from the
senders, then the output of the protocol is $\mathcal{N}_{A^{\prime}B^{\prime
}\rightarrow AB}\left(  \rho_{A^{\prime}B^{\prime}}\right)  $, so that this
protocol simulates the action of the channel $\mathcal{N}$ on the state $\rho$.
\end{proof} 
\bigskip

We now establish an upper bound on the  entanglement generation rate of any $(n,M,\varepsilon)$ PPT-assisted protocol that employs a bidirectional PPT-simulable channel.

\begin{theorem}\label{thm:rains-ent-dist-strong-converse-ppt}
For a fixed $n,\ |M|\in\mathbb{N},\ \varepsilon\in(0,1)$, the following strong converse bound holds for an $(n,Q,\varepsilon)$ protocol for PPT-assisted entanglement generation  over a bidirectional PPT-simulable quantum channel $\mc{N}$ with an associated resource state $\theta_{\hat{S}_A\hat{S}_B}$, Definition~\ref{def:bi-ppt-sim}, 
\begin{equation}\label{eq:rains-ent-dist-strong-converse-ppt}
\forall \alpha>1,\quad
Q \leq \widetilde{R}_{\alpha} (\hat{S}_A;\hat{S}_B)_{\theta} 
+\frac{\alpha}{n(\alpha-1)}\log_2 \!\(
\frac{1}{1-\varepsilon}\)
\end{equation} 
such that $Q=\frac{1}{n} \log_2 |M|$, 
where $\widetilde{R}_{\alpha}(\hat{S}_A;\hat{S}_B)_{\theta}$ is the sandwiched Rains information \eqref{eq:alpha-rains-inf-state} of the state $\theta_{\hat{S}_A\hat{S}_B}$.
\end{theorem}

\begin{proof}
The first few steps are similar to those in the proof of Theorem~\ref{thm:rains-ent-dist-strong-converse}. From Section~\ref{sec:ent-dist-protocol}, we have that
\begin{equation}
\Tr\{\Phi_{M_AM_B}\omega_{M_AM_B}\}\geq 1-\varepsilon,
\end{equation}
while \cite[Lemma 2]{Rai99} implies that 
\begin{equation}
\forall \sigma_{M_AM_B}\in\PPT'(M_A\!:\!M_B),\ \Tr\{\Phi_{M_AM_B}\sigma_{M_AM_B}\}\leq \frac{1}{|M|}.
\end{equation}
Under an \textquotedblleft entanglement test\textquotedblright, which is a measurement with POVM $\{\Phi_{M_AM_B},\bm{1}_{M_AM_B}-\Phi_{M_AM_B}\}$, and applying the data processing inequality for the sandwiched R\'enyi relative entropy, we find that, for all $\alpha>1$, 
\begin{equation}
\log_2 |M|\leq \widetilde{R}_{\alpha}(M_A;M_B)_{\omega}+\frac{\alpha}{\alpha-1}\log_2\!\(\frac{1}{1-\varepsilon}\). \label{eq:rains-test-bound-ppt}
\end{equation}
The sandwiched Rains relative entropy is monotonically non-increasing under the action of PPT-preserving channels and vanishing for a PPT state. Applying Lemma~\ref{thm:ent-ppt-single-letter}, we find that  
\begin{align}
\widetilde{R}_{\alpha}(M_A;M_B)_\omega \leq n \sup_{\rho_{L_AA'B'L_B}} \left[ \widetilde{R}_{\alpha}(L_{A}A;BL_{B})_{\mc{N}(\rho)}-\widetilde{R}_{\alpha}(L_{A}A';B'L_{B})_{\rho}\right].\label{eq:summand-ppt}
\end{align}

As stated in Definition~\ref{def:bi-ppt-sim}, a PPT-simulable bidirectional channel $\mc{N}_{A'B'\to AB}$ with an associated resource state
$\theta_{\hat{S}_A\hat{S}_B}$
is such that, for any input state $\rho'_{A'B'}$, 
\begin{equation}
\mc{N}_{A'B'\to AB}\(\rho'_{A'B'}\)=\mc{P}_{\hat{S}_AA'B'\hat{S}_B\to AB}
\(\rho'_{A'B'}\otimes\theta_{\hat{S}_A\hat{S}_B}\).
\end{equation}
Then, for any input state $\omega'_{L_AA'B'L_B}$, 
\begin{align}
&\widetilde{R}_{\alpha}(L_AA;BL_B)_{\mc{P}(\omega'\otimes\theta)}-\widetilde{R}_{\alpha}(L_AA';B'L_B)_{\omega'}\nonumber\\ 
&\leq \widetilde{R}_{\alpha}(L_A\hat{S}_AA';B'\hat{S}_BL_B)_{\omega'\otimes\theta}-\widetilde{R}_{\alpha}(L_AA';B'L_B)_{\omega'}\\
&\leq \widetilde{R}_{\alpha}(L_AA';B'L_B)_{\omega'}+\widetilde{R}_{\alpha}(\hat{S}_A;\hat{S}_B)_\theta-\widetilde{R}_{\alpha}(L_AA';B'L_B)_{\omega'}\\
&=\widetilde{R}_{\alpha}(\hat{S}_A;\hat{S}_B)_\theta. \label{eq:reduce-bound-ppt}
\end{align}
The first inequality follows from monotonicity of $\widetilde{R}_{\alpha}$ with respect to PPT-preserving channels. The second inequality follows because $\widetilde{R}_{\alpha}$ is sub-additive with respect to tensor-product states.

Applying the bound in \eqref{eq:reduce-bound-ppt} to \eqref{eq:summand-ppt}, we find that
\begin{equation}
\widetilde{R}_{\alpha}(M_A;M_B)_\omega\leq n\widetilde{R}_{\alpha}(\hat{S}_A;\hat{S}_B)_\theta.\label{eq:rains-single-letter-proof-ppt}
\end{equation}
 Combining \eqref{eq:rains-test-bound-ppt} and \eqref{eq:rains-single-letter-proof-ppt}, we get the desired inequality in \eqref{eq:rains-ent-dist-strong-converse-ppt}. 
\end{proof}

\bigskip
Now we establish an upper bound on the  secret key rate of an $(n,|K|,\varepsilon)$ secret-key-agreement protocol that employs a bidirectional teleportation-simulable channel.  

\begin{theorem}\label{thm:rel-ent-dist-strong-converse-tel}
For a fixed $n,\ |K|\in\mathbb{N},\ \varepsilon\in(0,1)$, the following strong converse bound holds for an $(n,P,\varepsilon)$ protocol for secret key agreement  over a bidirectional teleportation-simulable quantum channel $\mc{N}$ with an associated resource state $\theta_{\hat{S}_A\hat{S}_B}$:
\begin{equation}\label{eq:rel-ent-dist-strong-converse-tel}
\forall \alpha>1,\quad
P
\leq \widetilde{E}_{\alpha} (\hat{S}_A;\hat{S}_B)_{\theta}
+\frac{\alpha}{n(\alpha-1)}\log_2\!\( \frac{1}{1-\varepsilon}\)
\end{equation} 
such that $P=\frac{1}{n} \log_2 |K|$,
where $\widetilde{E}_{\alpha}(\hat{S}_A;\hat{S}_B)_{\theta}$ is the sandwiched relative entropy of entanglement \eqref{eq:rel-ent-state} of the state $\theta_{\hat{S}\hat{S}_B}$.
\end{theorem}

\begin{proof}
As stated in Definition~\ref{def:bi-ppt-sim}, a bidirectional teleportation-simulable  channel $\mc{N}_{A'B'\to AB}$ is such that, for any input state $\rho'_{A'B'}$,
\begin{equation}
\mc{N}_{A'B'\to AB}\(\rho'_{A'B'}\)=\mc{L}_{\hat{S}_AA'B'\hat{S}_B\to AB}\(\rho'_{A'B'}\otimes\theta_{\hat{S}_A\hat{S}_B}\).
\end{equation}
Then, for any input state $\omega'_{L_AA'B'L_B}$, 
\begin{align}
&\widetilde{E}_{\alpha}(L_AA;BL_B)_{\mc{L}(\omega'\otimes\theta)}-\widetilde{E}_{\alpha}(L_AA';B'L_B)_{\omega'}\nonumber\\ 
&\leq \widetilde{E}_{\alpha}(L_A\hat{S}_AA';B'\hat{S}_BL_B)_{\omega'\otimes\theta}-\widetilde{E}_{\alpha}(L_AA';B'L_B)_{\omega'}\\
&\leq \widetilde{E}_{\alpha}(L_AA';B'L_B)_{\omega'}+\widetilde{E}_{\alpha}(\hat{S}_A;\hat{S}_B)_\theta-\widetilde{E}_{\alpha}(L_AA';B'L_B)_{\omega'}\\
&=\widetilde{E}_{\alpha}(\hat{S}_A;\hat{S}_B)_\theta. \label{eq:reduce-bound-tel}
\end{align}
The first inequality follows from monotonicity of $\widetilde{E}_{\alpha}$ with respect to LOCC channels. The second inequality follows because $\widetilde{E}_{\alpha}$ is sub-additive.

From Section~\ref{sec:priv-dist-protocol}, the following inequality holds for an $(n,P,\varepsilon)$ protocol:
\begin{equation}
F(\omega_{S_AK_AK_BS_B},\gamma_{S_AK_AK_BS_B})\geq 1-\varepsilon,
\end{equation}
for some bipartite private state $\gamma_{S_AK_AK_BS_B}$ with key dimension $|K|$. From Section~\ref{sec:rev-priv-states}, $\omega_{S_AK_AK_BS_B}$ passes a $\gamma$-privacy test with probability at least $1-\varepsilon$, whereas any $\tau_{S_AK_AK_BS_B}\in\SEP(S_AK_A:K_BS_B)$ does not pass with probability greater than $\frac{1}{|K|}$ \cite{HHHO09}. Making use of the results in \cite[Section 5.2]{WTB16}, we conclude that
\begin{align}\label{eq:emax-test-bound-tel}
\log_2 |K| \leq \widetilde{E}_{\alpha}(S_AK_A;K_BS_B)_{\omega}+\frac{\alpha}{\alpha-1}\log_2\!\(\frac{1}{1-\varepsilon}\).
\end{align}
Now we can follow steps similar to those in the proof of Theorem~\ref{thm:rains-ent-dist-strong-converse-ppt} in order to arrive at \eqref{eq:rel-ent-dist-strong-converse-tel}. 
\end{proof}

\bigskip 

We can also establish the following weak converse bounds, by combining the above approach with that in \cite[Section~3.5]{KW17}:

\begin{remark}
The following weak converse bound holds for an $(n,Q,\varepsilon)$ PPT-assisted bidirectional quantum communication protocol (Section~\ref{sec:ent-dist-protocol}) that employs a bidirectional PPT-simulable quantum channel $\mc{N}$ with an associated resource state $\theta_{\hat{S}_A\hat{S}_B}$ 
\begin{equation}\label{eq:rel-ent-dist-strong-converse-ppt-2}
(1-\varepsilon)Q \leq R(\hat{S}_A;\hat{S}_B)_{\theta}+\frac{1}{n}h_2(\varepsilon),
\end{equation} 
where $R(\hat{S}_A;\hat{S}_B)_{\theta}$ is defined in \eqref{eq:rains-inf-state} and $h_2(\varepsilon)\coloneqq -\varepsilon\log_2\varepsilon-(1-\varepsilon)\log_2(1-\varepsilon)$. 
\end{remark}

\begin{remark}\label{rem:tel-sim-priv-dist}
The following weak converse bound holds for an $(n,P,\varepsilon)$ LOCC-assisted bidirectional secret key agreement protocol (see Section~\ref{sec:priv-dist-protocol}) that employs a bidirectional teleportation-simulable quantum channel $\mc{N}_{A'B'\to AB}$ with an associated resource state $\theta_{\hat{S}_A\hat{S}_B}$
\begin{equation}\label{eq:rel-ent-dist-strong-converse-tel-2}
(1-\varepsilon)P \leq E(\hat{S}_A;\hat{S}_B)_{\theta}+\frac{1}{n}h_2(\varepsilon),
\end{equation} 
where $E(\hat{S}_A;\hat{S}_B)_{\theta}$ is defined in \eqref{eq:rel-ent-state-1}. 
\end{remark}

Since every LOCC channel $\mc{L}_{\hat{S}_AA'B'\hat{S}_B\to AB}$ acting with respect to the bipartite cut $\hat{S}_AA':B'\hat{S}_B$ is also a PPT-preserving channel with the partial transposition action on $B'\hat{S}_B$, it follows that bidirectional teleportation-simulable channels are also bidirectional PPT-simulable channels. Based on Proposition~\ref{prop:bicov}, Theorem~\ref{thm:rains-ent-dist-strong-converse-ppt}, Theorem~\ref{thm:rel-ent-dist-strong-converse-tel}, and the limits $n \to \infty$ and then $\alpha\to 1$ (in this order),\footnote{One could also set $\alpha = 1+1/\sqrt{n}$ and then take the limit $n \to \infty$.} we can then  conclude the following strong converse bounds:
\begin{corollary}
\label{cor:str-conv-TP-simul}
If a bidirectional quantum channel $\mc{N}$ is bicovariant (Definition~\ref{def:bicov}), then 
\begin{align}
\widetilde{Q}^{2\to2}_{\PPT}(\mc{N})& \leq R(L_A A;BL_B)_{\theta},\\
\widetilde{P}^{2\to2}_{\LOCC}(\mc{N})&\leq E(L_AA;BL_B)_{\theta},
\end{align}
where $\theta_{L_{A}ABL_{B}}=\mathcal{N}_{A^{\prime}B^{\prime}\rightarrow AB}
(\Phi_{L_{A}A^{\prime}}\otimes\Phi_{B^{\prime}L_{B}})$, and $\widetilde{Q}^{2\to2}_{\PPT}(\mc{N})$ and $\widetilde{P}^{2\to2}_{\LOCC}(\mc{N})$ denote the strong converse PPT-assisted bidirectional quantum capacity and strong converse LOCC-assisted bidirectional secret-key-agreement capacity, respectively, of a bidirectional channel $\mc{N}$. 
\end{corollary}

\section{Conclusion}\label{sec:con-bqi}
In this chapter, we mainly focused on two different information processing tasks: entanglement distillation and secret key distillation using bipartite quantum interactions or bidirectional channels. We determined several bounds on the entanglement and secret-key-agreement capacities of bipartite quantum interactions. In deriving these bounds, we described communication protocols in the bidirectional setting, related to those discussed in \cite{BHLS03} and which generalize related point-to-point communication protocols. We defined an entanglement measure called the bidirectional max-Rains information of a bidirectional channel and showed that it is a strong converse upper bound on the PPT-assisted quantum capacity of the given bidirectional channel. We also defined a related entanglement measure called the bidirectional max-relative entropy of entanglement and showed that it is a strong converse bound on the LOCC-assisted secret-key-agreement capacity of a given bidirectional channel. When the bidirectional channels are either teleportation- or PPT-simulable, the upper bounds on the bidirectional quantum and bidirectional secret-key-agreement capacities depend only on the entanglement of an underlying resource state. If a bidirectional channel is bicovariant, then the underlying resource state can be taken to be the Choi state of the bidirectional channel.

%% file: eg.tex
\chapter{Fundamental Limits on Quantum Dynamics Based on Entropy Change}\label{ch:eg}
%\section{Introduction}
Entropy is a fundamental quantity that is of wide interest in physics and information theory \cite{Neu32,Sha48,DGM62,Bek73}\blfootnote{Most of this chapter is reproduced from [Siddhartha Das, Sumeet Khatri, George Siopsis, and Mark M.~Wilde. \textit{Journal of Mathematical Physics}, 59(1):012205, (2018)], with the permission of AIP Publishing.}. Many natural phenomena are described according to laws based on entropy, like the second law of thermodynamics \cite{Car1824,Men60,Att12}, entropic uncertainty relations in quantum mechanics and information theory \cite{Hir57,Bec75,BM75,MU88,CBTW17}, and area laws in black holes and condensed matter physics \cite{BCH73,BKLS86,Sre93,ECP10}. 
	
	No quantum system can be perfectly isolated from its environment. The interaction of a system with its environment generates correlations between the system and the environment. In realistic situations, instead of isolated systems, we must deal with open quantum systems, that is, systems whose environment is not under the control of the observer. The interaction between the system and the environment can cause a loss of information as a result of decoherence, dissipation, or decay phenomena \cite{Car09,Riv11,Wei12}. The rate of entropy change quantifies the flow of information between the system and its environment.
	
	In this chapter, we focus on the von Neumann entropy, which is defined for a system in the state $\rho$ as $S(\rho)\coloneqq-\Tr\{\rho\log\rho\}$\footnote{In this chapter, we  particularly use natural logarithm in the definition of the entropy and the relative entropy.}, and from here onwards we refer to it as the entropy. The entropy is monotonically non-decreasing under doubly-stochastic, also called unital, physical evolutions \cite{AU78,AU82}. This has restricted the use of entropy change in the characterization of quantum dynamics only to unital dynamics \cite{Str85,BP02book,CK12,MSW16}. Recently, \cite{BDW16} gave a lower bound on the entropy change for any positive trace-preserving map. Lower bounds on the entropy change have also been discussed in \cite{Str85,MSW16,AW16,ZHHH01} for certain classes of time evolution. Natural questions that arise are as follows: what are the limits placed by the bound\footnote{Specifically, we consider the bound in \cite[Theorem~1]{BDW16} as it holds for arbitrary evolution of both finite- and infinite-dimensional systems.} on the entropy change on the dynamics of a system, and can it be used to characterize evolution processes? 
	
	We delve into these questions, at first, by inspecting another pertinent question: at what rate does the entropy of a quantum system change? Although the answer is known for Markovian one-parameter semigroup dynamics of a finite-dimensional system with full-rank states \cite{Spo78}, the answer in full generality has not yet been given. In \cite{KS14}, the result of \cite{Spo78} was extended to infinite-dimensional systems with full-rank states undergoing Markovian one-parameter semigroup dynamics (cf., \cite{DPR17}). We now prove that the formula derived in \cite{Spo78} holds not only for finite-dimensional quantum systems undergoing Markovian one-parameter semigroup dynamics, but also for arbitrary dynamics of both finite- and infinite-dimensional systems with states of arbitrary rank. We then derive a lower bound on the rate of entropy change for any memoryless quantum evolution, also called a quantum Markov process. This lower bound is a witness of non-unitality in quantum Markov processes. Interestingly, this lower bound also helps us to derive witnesses for the presence of memory effects, i.e., non-Markovianity, in quantum dynamics. We compare one of our witnesses to the well-known Breuer-Laine-Piilo (BLP) measure \cite{BLP09} of non-Markovianity for two common examples. As it turns out, in one of the examples, our witness detects non-Markovianity even when the BLP measure does not, while for the other example our measure agrees with the BLP measure. We also provide bounds on the entropy change of a system. These bounds are witnesses of how non-unitary an evolution process is. We use one of these witnesses to propose a measure of non-unitarity for unital evolutions and discuss some of its properties.

The organization of the chapter is as follows. In Section~\ref{sec:notation-eg}, we introduce some definitions and facts for continuous variable systems that are not covered in Chapter~\ref{ch:review}. In Section~\ref{sec-ent_change_rate}, we discuss the explicit form (Theorem~\ref{thm:rate-oqs}) for the rate of entropy change of a system in any state undergoing arbitrary time evolution. In Section~\ref{sec:qmp}, we briefly review quantum Markov processes. We state Theorem~\ref{thm:q-Markov-ent-change-rate}, which provides a lower limit on the rate of entropy change for quantum Markov processes. We show that this lower limit provides a witness of non-unitality. We also discuss the implications of the lower limit on the rate of entropy change in the context of bosonic Gaussian dynamics (Section~\ref{sec-Gaussian}). In Section~\ref{sec:qnmp}, based on the necessary conditions for the Markovianity of quantum processes as stated in Theorem~\ref{thm:q-Markov-ent-change-rate}, we define some witnesses of non-Markovianity and also a couple of measures of non-Markovianity based on these witnesses. We apply these witnesses to two common examples of non-Markovian dynamics (Section~\ref{sec-decohere} and Section~\ref{sec-GADC}) and illustrate that they can detect non-Markovianity. In Section~\ref{sec-GADC}, we consider an example of a non-unital quantum non-Markov process whose non-Markovianity goes undetected by the BLP measure while it is detected by our witness. In Section~\ref{sec-entropy_change}, we derive an upper bound on entropy change for unital evolutions. We also show the monotonic behavior of the entropy for a wider class of operations than previously known. In Section~\ref{sec-non_unitarity}, we define a measure of non-unitarity for any unital evolution. We also discuss properties of the measure of non-unitarity. 

\section{Preliminaries}\label{sec:notation-eg}
In this section, we add few more standard notations,  definitions, and facts to the discussion in Chapter~\ref{ch:review} because of subtleties that come when dealing with continuous variable systems, which are associated to separable, infinite-dimensional Hilbert spaces.

The dimension $\dim(\mc{H})$ of the Hilbert space $\mc{H}$ is equal to $+\infty$ in the case that $\mc{H}$ is a separable, infinite-dimensional Hilbert space. The subset of $\msc{B}(\mc{H})$ containing all trace-class operators is denoted by $\msc{B}_1(\mc{H})$. Let $\mc{B}_1^{+}(\mc{H})\coloneqq\mc{B}_+(\mc{H})\cap\mc{B}_1(\mc{H})$.

The adjoint $\mc{M}^\dagger:\msc{B}(\mc{H}_B)\to\msc{B}(\mc{H}_A)$ of a linear map $\mc{M}:\msc{B}_1(\mc{H}_A)\to\msc{B}_1(\mc{H}_B)$ is the unique linear map that satisfies
	\begin{equation}\label{eq-adjoint-eg}
	\forall\ X_A\in\msc{B}_1(\mc{H}_A), Y_B\in\msc{B}(\mc{H}_B):\ \	\<Y_B,\mc{N}(X_A)\>=\<\mc{N}^\dag(Y_B),X_A\>,
	\end{equation}
	where $\<C,D\>=\Tr\{C^\dag D\}$ is the Hilbert-Schmidt inner product. 

The von Neumann entropy of a state $\rho_A$ of a quantum system $A$ is defined as 
	\begin{equation}\label{eq-entropy-eg}
		S(A)_\rho:= S(\rho_A)= -\Tr\{\rho_A\log\rho_A\},
	\end{equation}
	where $\log$ denotes the natural logarithm. In general, the state of an infinite-dimensional quantum system need not have finite entropy \cite{BV13}. For any finite-dimensional system $A$, the entropy is upper-bounded by $\log |A|$. 
	
	The quantum relative entropy of any two density operators $\rho,\sigma\in\msc{D}(\mc{H})$ is defined as \cite{Ume62,Fal70,Lin73} 
	\begin{equation}\label{eq-rel_ent_alt}
		D(\rho\V \sigma)=
		\sum_{i,j}\left|\< \phi_i\vert \psi_j\>\right|^2\left[p(i)\log\(\frac{p(i)}{q(j)}\)\right], 
	\end{equation}
	where $\rho=\sum_i p(i)\ket{\phi_i}\bra{\phi_i}$ and $\sigma=\sum_j q(j)\ket{\psi_j}\bra{\psi_j}$ are spectral decompositions of $\rho$ and $\sigma$, respectively, with both $\{\ket{\phi_i}\}_i,\{\ket{\psi_j}\}_j\in\ONB(\mc{H})$ (cf.~\eqref{eq:rel-ent-rev}). From the above definition, it is clear that $D(\rho \V \sigma) = +\infty$ if $\supp(\rho)\not\subseteq\supp(\sigma)$. 
	
	For any two positive semi-definite operators $\rho, \sigma\in\msc{B}_1^+(\mc{H})$, $D(\rho\Vert\sigma)\geq 0$ if $\Tr\{\rho\}\geq\Tr\{\sigma\}$, $D(\rho\Vert\sigma)=0$ if and only if $\rho=\sigma$, and $D(\rho\V\sigma)<0$ if $\rho<\sigma$. The quantum relative entropy is non-increasing under the action of positive trace-preserving maps \cite{MR15}, that is, $D(\rho\V\sigma)\geq D(\mc{N}(\rho)\V\mc{N}{(\sigma)})$ for any two density operators $\rho, \sigma\in\msc{D}(\mc{H})$ and positive trace-preserving map $\mc{N}:\msc{B}^1_+(\mc{H})\to \msc{B}^1_+(\mc{H}')$.
	
	We now define entropy change, which is the main focus of this chapter. 
	
	\begin{definition}[Entropy change]\label{def:ent-change}
		Let $\mc{N}:\msc{B}_1^+(\mc{H})\to\msc{B}_1^+(\mc{H}^\prime)$ be a positive trace-non-increasing map. The entropy change $\Delta S(\rho,\mc{N})$ of a system in the state $\rho\in\msc{D}(\mc{H})$ under the action of $\mc{N}$ is defined as 
		\begin{equation}
			\Delta S(\rho,\mc{N}):=S\(\mc{N}(\rho)\)-S\(\rho\)
		\end{equation}
		whenever $S(\rho)$ and $S(\mc{N}(\rho))$ are finite.
	\end{definition}
	It should be noted that $\mc{N}(\rho)$ is a sub-normalized state, i.e., $\Tr\{\mc{N}(\rho)\}\leq 1$, if $\mc{N}$ is a positive trace-non-increasing map. 
	
	It is well known that the entropy change $\Delta S(\rho,\mc{N})$ of $\rho$ is non-negative, i.e., the entropy is non-decreasing, under the action of a positive, sub-unital, and trace-preserving map $\mc{N}$ \cite{AU78,AU82} (see also \cite[Section~III]{BDW16}, \cite[Theorem~4.2.2]{Nie02}). Recently, a refined statement of this result was made in \cite{BDW16}, which is the following: 

	\begin{lemma}[Lower bound on entropy change]\label{thm:ent-lower-bound}
		Let $\mc{N}:\msc{B}_1^+(\mc{H})\to\msc{B}_1^+(\mc{H}^\prime)$ be a positive, trace-preserving map. Then, for all $\rho\in\mc{D}(\mc{H})$,
		\begin{equation}\label{eq-ent_lower_bound}
			\Delta S(\rho,\mc{N})\geq D\!\(\rho\left\Vert\mc{N}^\dag\circ\mc{N}\(\rho\)\)\right. .
		\end{equation}
	\end{lemma}
	
	\begin{proof}
		Using the definition \eqref{eq-adjoint-eg} of the adjoint, we obtain
		\begin{align}
			\Delta S(\rho,\mc{N})=S(\mc{N}(\rho))-S(\rho)&=\Tr\{\rho\log\rho\}-\Tr\left\{\mc{N}(\rho)\log\mc{N}(\rho)\right\}
			\end{align}
	Then
			\begin{align}
		\Delta S(\rho,\mc{N})&=\Tr\{\rho\log\rho\}-\Tr\left\{\rho\mc{N}^\dag\(\log\mc{N}(\rho)\)\right\} \nonumber\\
		&\geq \Tr\{\rho\log\rho\}-\Tr\left\{\rho \log\[\mc{N}^\dag\circ\mc{N}(\rho)\]\right\}\nonumber\\
			&=D\(\rho\left\Vert\mc{N}^\dag\circ\mc{N}\(\rho\)\)\right. .
		\end{align}
		The inequality follows from Lemma \ref{thm:log-con} applied to $\mathcal{N}^\dagger$, which is positive and sub-unital since $\mathcal{N}$ is positive and trace non-increasing. 
	\end{proof}

Lemma~\ref{thm:ent-lower-bound} gives a tight lower bound on the entropy change. As an example of a map saturating the inequality \eqref{eq-ent_lower_bound}, let us take the partial trace $\mc{N}_{AB\to B}=\Tr_A$, which is a quantum channel that corresponds to discarding system $A$ from the composite system $AB$. Its adjoint is $\mc{N}^\dagger(\rho_B)=\mathbbm{1}_A\otimes \rho_B$. Then, one notices that $S\(\mc{N}\(\rho_{AB}\)\)-S\(\rho_{AB}\)=S(\rho_B)-S(\rho_{AB})=D\(\rho_{AB}\Vert \mathbbm{1}_A\otimes\rho_B\)=D\(\rho_{AB}\left\Vert\mc{N}^\dag\circ\mc{N}\(\rho_{AB}\)\)\right.$. 

\section{Quantum dynamics and the rate of entropy change}\label{sec-ent_change_rate}

	In general, physical systems are dynamical and undergo evolution processes with time. An evolution process for an isolated and closed system is unitary. However, no quantum system can remain isolated from its environment. There is always an interaction between a system and its environment. The joint evolution of the system and environment is considered to be a unitary operation whereas the local evolution of the system alone can be non-unitary. This non-unitarity causes a flow of information between the system and the environment, which can change the entropy of the system. 

	For any dynamical system with associated Hilbert space $\mc{H}$, the state of the system depends on time. The time evolution of the state $\rho_t$ of the system at an instant $t$ is determined by $\frac{\d\rho_t}{\d t}$ when it is well defined\footnote{By this, we mean that each matrix element of $\rho_t$ is  differentiable with respect to $t$.}. The state $\rho_T$ at some later time $t=T$ is determined by the initial state $\rho_0$, the evolution process, and the time duration of the evolution. Since the time evolution is a physical process, the following condition holds for all $t$:
	\begin{equation}
		\Tr\left\{\dot{\rho}_t\right\}=0,
	\end{equation}
	where $\dot{\rho}_t\coloneqq\frac{\d\rho_t}{\d t}$.
	
	It is known from \cite{Spo78,Ber09} that for any finite-dimensional system the following formula for the rate of entropy change holds for any state $\rho_t$ whose kernel remains the same at all times and whose support $\Pi_t$ is differentiable:
	\begin{equation}\label{eq:incorrect-rate}
		\frac{\d}{\d t}S(\rho_t)=-\Tr\left\{\dot{\rho}_t\log\rho_t\right\}.
	\end{equation}
	The above formula has also been applied to infinite-dimensional systems for Gaussian states evolving under a quantum diffusion semigroup \cite{DPR17,KS14} whose kernels do not change in time.
	
	Here, we derive the formula \eqref{eq:incorrect-rate} for states $\rho_t$ having fewer restrictions, which generalizes the statements from \cite{Spo78,Ber09}. In particular, we show that the formula \eqref{eq:incorrect-rate} can be applied to quantum dynamical processes in which the kernel of the state changes with time, which can happen because the state has time-dependent support.

	\begin{theorem}\label{thm:rate-oqs}
		For any quantum dynamical process with $\dim(\mc{H})<+\infty$, the rate of entropy change is given by
		\begin{equation}\label{eq-pi_ent_change_rate}
			\frac{\d }{\d t}S(\rho_t)=-\Tr\left\{\dot{\rho}_t\log\rho_t\right\},
		\end{equation}
		whenever $\dot{\rho}_t$ is well defined. The above formula also holds when $\dim(\mc{H})=+\infty$ if $\dot{\rho}_t\log\rho_t$ is trace-class and the sum of the time derivative of the eigenvalues of $\rho_t$ is uniformly convergent\footnote{Uniform convergence is defined as stated in \cite[Definition~7.7]{Rud76}.} on some neighborhood of $t$, however small. 
	\end{theorem}	
	\begin{proof}
		Let $\text{Spec}(\rho_t)$ be the set of all eigenvalues of $\rho_t\in\msc{D}(\mc{H})$, including those in its kernel. Let
		\begin{equation}\label{eq-rho_t_spec}
			\rho_t=\sum_{\lambda(t)\in\text{Spec}(\rho_t)}\lambda(t)P_{\lambda}(t)
		\end{equation}
		be a spectral decomposition of $\rho_t$, where the sum of the projections $P_{\lambda}(t)$ corresponding to $\lambda(t)$ is 
		\begin{equation}
			\sum_{\lambda(t)\in\text{Spec}(\rho_t)}P_{\lambda}(t)=\mathbbm{1}_{\mc{H}}.
		\end{equation}
		The following assumptions suffice to arrive at the statement of the theorem when $\dim(\mc{H})=+\infty$. We assume that $\dot{\rho}_t$ is well defined. We further assume that $\sum_{\lambda(t)\in\text{Spec}(\rho_t)}\dot{\lambda}(t)$ is uniformly convergent on some neighborhood of $t$, and $\dot{\rho}_t\log\rho_t$ is trace-class. Note that when $\dim(\mc{H})<+\infty$, $\sum_{\lambda(t)\in\text{Spec}(\rho_t)}\dot{\lambda}(t)$ and $\dot{\rho}_t\log\rho_t$ are always uniformly convergent and trace-class, respectively.

		Now, let us define the function $s:[0,\infty)\times (-1,\infty)\rightarrow (0,\infty)$ by
		\begin{equation}\label{eq-def_s}
			s(t,h)\coloneqq \Tr\{\rho_t^{1+h}\} = \sum_{\lambda(t)\in\text{Spec}(\rho_t)}\lambda(t)^{1+h}.
		\end{equation}
		
		Noting that $\frac{\d}{\d x}a^x=a^x\log a$ for all $a>0$ and $x\in\mathbb{R}$, we get
		\begin{align}
			\frac{\d}{\d h}\rho_{t}^{h+1}=\rho_{t}^{h+1}\log\rho_{t}.
		\end{align}
		Applying \eqref{eq:app-trace-der-1} in  Section~\ref{app-derivative}, we find that
		\begin{align}
			\frac{\d}{\d t}s(t,h)&=\frac{\d}{\d t}\Tr\{\rho_{t}^{h+1}\}=\left(h+1\right)\Tr\{\rho_{t}^{h}\dot{\rho}_{t}\},\label{eq-st_1}\\
			\frac{\d}{\d h}s(t,h)&=\frac{\d}{\d h}\Tr\{\rho_t^{h+1}\}=\Tr\{\rho_{t}^{h+1}\log\rho_t\}.
		\end{align}
		Then the entropy is
		\begin{equation}
			S(\rho_t)=-\left. \frac{\d}{\d h}s(t,h)\right|_{h=0} = - \Tr\{\rho_t \log  \rho_t\} = - \sum_{\lambda(t)\in\text{Spec}(\rho_t)} \lambda(t) \log\lambda(t),
		\end{equation}
		where by definition $0\log 0=0$.
		
		Note that $\rho_t^h$ is an infinitely differentiable function of $h$, i.e., a smooth function of $h$, and a differentiable function of $t$ for all $t,h$. Also, the trace is a continuous function. Since $\frac{\d}{\d h}\frac{\d}{\d t}s(t,h)$ exists and is continuous for all $(t,h)\in[0,\infty)\times (-1,\infty)$, the following exchange of derivatives holds for all $(t,h)\in(0,\infty)\times (-1,\infty)$:
		\begin{equation}
			\frac{\d}{\d h}\left[  \frac{\d}{\d t}s(t,h)\right]=\frac{\d}{\d t}\left[  \frac{\d}{\d h}s(t,h)\right].
		\end{equation}
		This implies that
		\begin{equation}
			\left.\frac{\d}{\d h}\left[  \frac{\d}{\d t}s(t,h)\right]\right|_{h=0}=\frac{\d}{\d t}\left[ \left.\frac{\d}{\d h}s(t,h)\right|_{h=0}  \right]
		\end{equation}
		From \eqref{eq-st_1}, we notice that $\frac{\d}{\d t}s(t,h)$ is a smooth function of $h$. Therefore, the Taylor series expansion of this function in the neighborhood of $h=0$ is
		\begin{align}
			\frac{\d}{\d t}s(t,h)&=\left.\frac{\d}{\d t}s(t,h)\right|_{h=0}+\left.\frac{\d}{\d h}\left[\frac{\d}{\d t}s(t,h)\right]\right|_{h=0} h+O(h^2).
		\end{align}
		
		From \eqref{eq-def_s}, we find:
		\begin{align}
			\left.\frac{\d}{\d t}s(t,h)\right|_{h=0}&=\left.\frac{\d}{\d t}\left[\sum_{\lambda(t)\in\text{Spec}(\rho_t)}\lambda(t)^{1+h}\right]\right|_{h=0}=\sum_{\lambda(t)\in\text{Spec}(\rho_t)}\left.\frac{\d}{\d t}\left[\lambda(t)^{1+h}\right]\right|_{h=0}\\
			&=\sum_{\lambda(t)\in\text{Spec}(\rho_t)}\left[ (1+h)\lambda(t)^h\dot{\lambda}(t)\right]_{h=0}\\
			&=\sum_{\lambda(t)\neq 0}\dot{\lambda}(t)\label{eq-st_2}.
		\end{align}
		The second equality follows from \cite[Theorem~7.17]{Rud76} due to the uniform convergence of $\sum_{\lambda(t)\in\text{Spec}(\rho_t)}\dot{\lambda}(t)$ on some neighborhood of $t$. To obtain the last equality, we use the following fact: since $\lambda(t)\geq 0$ for all $t$ and $\lambda(t)$ is differentiable, if $\lambda(t^*)=0$ for some time $t=t^*\in (0,\infty)$, then $\dot{\lambda}(t^*)=0$. From \eqref{eq-st_1} and \eqref{eq-st_2}, we obtain
		\begin{equation}\label{eq-st_3}
			\Tr\{\Pi_t\dot{\rho}_t\}=\sum_{\lambda(t)\neq 0}\dot{\lambda}(t)=\frac{\d}{\d t}\sum_{\lambda(t)\neq 0}\lambda(t)=\frac{\d}{\d t}\Tr\{\rho_t\}=0,
		\end{equation}
		where $\Pi_t$ is the projection onto the support of $\rho_t$. The second equality holds because $\dot{\lambda}(t^*)=0$ whenever $\lambda(t^*)=0$ for all $\lambda(t^*)\in\text{Spec}(\rho_{t^*})$ and all $t^*\in (0,\infty)$.
		
		Employing \eqref{eq:app-trace-der-2}, we find that
		\begin{align}
			\frac{\d}{\d h}\left[  \frac{\d}{\d t}s(t,h)\right]
			& =\frac{\d}{\d h}\left[  \left(  h+1\right)  \operatorname{Tr}\{\rho_{t}^{h}%
			\dot{\rho}_{t}\}\right]  \\
			& =\operatorname{Tr}\{\rho_{t}^{h}\dot{\rho}_{t}\}+\left(  h+1\right)
			\operatorname{Tr}\left\{\left[  \rho_{t}^{h}\log\rho_{t}\right]  \dot{\rho}_{t}\right\}.
		\end{align}
		Therefore,
		\begin{align}
			-\frac{\d}{\d t}S(\rho_{t})&=\frac{\d}{\d t}\left[  \left. \frac{\d}{\d h}s(t,h)\right|_{h=0}\right]\\
			&=\left.\frac{\d}{\d h}\left[ \frac{\d}{\d t}s(t,h)\right]\right|_{h=0}\\
			&=\operatorname{Tr}\{\Pi_t\dot{\rho}_{t}\}+
			\operatorname{Tr}\left\{ \dot{\rho}_{t} \Pi_t\log\rho_{t} \right\}\\
			&=\Tr\{\dot{\rho}_t\log\rho_t\},
		\end{align}
		where to obtain the last equality we used \eqref{eq-st_3} and the fact that $\log\rho_t$ is defined on $\supp(\rho_t)$. This concludes the proof.
	\end{proof}
	
\bigskip
	
	As an immediate application of Theorem \ref{thm:rate-oqs}, consider a closed system consisting of a system of interest $A$ and a bath (environment) system $E$ in a pure state $\psi_{AE}$, for which the time evolution is given by a bipartite unitary $U_{AE}$, a special case of bipartite quantum interactions (Chapter~\ref{ch:bqi}). Under unitary evolution, the entropy of the composite system $AE$ does not change. Also, for a pure state, the entropy of the composite system is zero, and $S(\rho_{A})=S(\rho_{E})$, where $\rho_{A}$ and $\rho_E$ are the reduced states of the systems $A$ and $E$, respectively. Now, it is often of interest to determine the amount of entanglement in the reduced state $\rho_A$ of the system $A$. Several measures of entanglement have been proposed\cite{PV10}, among which the entanglement of formation\cite{BDSW96,Woo01}, the distillable entanglement\cite{BDSW96,BBP+97}, and the relative entropy of entanglement\cite{VP98,VPRK97} all reduce to the entropy $S(\rho_A)$ of the system $A$ in the case of a closed bipartite system \cite{ON02}. Thus, in this case, the rate of entropy change of the system $A$ is equal to the rate of entanglement change (with respect to the aforementioned entanglement measures) caused by unitary time evolution of the pure state of the composite system, and Theorem \ref{thm:rate-oqs} provides a general expression for this rate of entanglement change.
	
In Appendix \ref{app-rate_ent_change}, we discuss how \eqref{eq-pi_ent_change_rate} generalizes the development in \cite{Spo78,Ber09}. We consider examples of dynamical processes in which the support and/or the rank of the state change with time, but the formula \eqref{eq-pi_ent_change_rate} is still applicable according to the above theorem.

\section{Open quantum system and Markovian dynamics}\label{sec:qmp}

	The dynamics of an open quantum system can be categorized into two broad classes, quantum Markov processes and quantum non-Markov processes, based on whether the evolution process exhibits memoryless behavior or has memory effects.
	
	Here, we consider the dynamics of an open quantum system in the time interval $I=[t_1,t_2)\subset\mathbb{R}$ for $t_1<t_2$. We assume that the state $\rho_t\in\msc{D}(\mc{H})$ of the system at time $t\in I$ satisfies the following differential master equation:
	\begin{equation}\label{eq-master_equation}
		\dot{\rho}_t=\mc{L}_t(\rho_t)\quad\forall~t\in I,
	\end{equation}
	where $\mathcal{L}_t$ is called the generator \cite{Kos72}, or Liouvillian, of the dynamics and can in general be time-dependent \cite{Ali07}. A state $\rho_{\text{eq}}$ is called a fixed point, or invariant state of the dynamics, if $\dot{\rho}_{\text{eq}}=0$, or
	\begin{equation}\label{eq-fixed_point}
		\mc{L}_t(\rho_{\text{eq}})=0\quad\forall~t\in I.
	\end{equation}
	
	In general, the evolution of systems governed by the master equation \eqref{eq-master_equation} is given by the two-parameter family $\{\mc{M}_{t,s}\}_{t,s\in I}$ of maps $\mc{M}_{t,s}:\msc{B}(\mc{H})\rightarrow\msc{B}(\mc{H})$ defined by \cite{Riv11}
	\begin{equation}\label{eq-channel_TO_exp}
		\mc{M}_{t,s}=\mc{T}\exp\left[\int_s^t\mc{L}_{\tau}~\d\tau\right]~~\forall ~s,t\in I,~s\leq t,\quad \mc{M}_{t,t}=\id~~\forall ~t\in I,
	\end{equation}
	where $\mc{T}$ is the time-ordering operator, so that the state $\rho_t$ of the system at time $t$ is obtained from the state of the system at time $s\leq t$ as $\rho_t=\mc{M}_{t,s}(\rho_s)$. The maps $\{\mc{M}_{t,s}\}_{t\geq s}$ satisfy the following composition law:
	\begin{align}
	\forall ~s\leq r\leq t:\ \	\mc{M}_{t,s}&=\mc{M}_{t,r}\circ\mc{M}_{r,s},\label{eq-semi_group_time_dep}\\
	\forall ~t\in I:\ \	\mc{M}_{t,t}&=\id,
	\end{align}
	and in terms of these maps the generator $\mc{L}_t$ is given by
	\begin{equation}\label{eq-generator_time_dep_2}
		\mc{L}_t=\lim_{\varepsilon\to 0^+}\frac{\mc{M}_{t+\varepsilon,t}-\id}{\varepsilon}.
	\end{equation}
	For the maps $\{\mc{M}_{t,s}\}_{t\geq s}$ to represent physical evolution, they must be trace-preserving. This implies that for all $\rho\in\msc{D}(\mc{H})$ the generator $\mc{L}_t$ has to satisfy
	\begin{equation}\label{eq-trace_preserving}
		\Tr\left[\mc{L}_t(\rho)\right]=0\quad\forall ~t\in I.
	\end{equation}
	
	When the intermediate maps $\mc{M}_{t,r}$ and $\mc{M}_{r,s}$ are positive and trace-preserving for all $s\leq r\leq t$, the condition \eqref{eq-semi_group_time_dep} is called P-divisibility. If the intermediate maps $\mc{M}_{t,r}$ and $\mc{M}_{r,s}$ are CPTP (i.e., quantum channels) for all $s\leq r\leq t$, the condition \eqref{eq-semi_group_time_dep} is called CP-divisibility \cite{WC08,RHP14}. Based on the notion of CP-divisibility, we consider the following definition of a quantum Markov process, which was introduced in \cite{RHP10}.
	
	\begin{definition}[Quantum Markov process]\label{def:qmp}
		The dynamics of a system in a time interval $I$ are called a quantum Markov process when they are governed by \eqref{eq-master_equation} and they are CP-divisible (i.e., the intermediate maps in \eqref{eq-semi_group_time_dep} are CPTP).
	\end{definition}
	
	An important fact is that the dynamics governed by the master equation \eqref{eq-master_equation} are CP-divisible (hence Markovian) if and only if the generator $\mc{L}_t$ of the dynamics has the Lindblad form 
	\begin{equation}\label{eq-generator_time_dep}
		\mc{L}_t(\rho)=-\iota[H(t),\rho]+\sum_i\gamma_i(t)\left[A_i(t)\rho A_i^\dagger(t)-\frac{1}{2}\left\{A_i^\dagger(t) A_i(t),\rho\right\}\right],
	\end{equation}
	with $H(t)$ a self-adjoint operator and $\gamma_i(t)\geq 0$ for all $i$ and for all $t\in I$. The operators $A_i(t)$ are called Lindblad operators. In the time-independent case, this result was independently obtained by Gorini \textit{et al.} \cite{GKS76} for finite-dimensional systems and by Lindblad \cite{Lin76} for infinite-dimensional systems. For a proof of this result in the time-dependent scenario, see \cite{Riv11,CK12}. In finite dimensions, necessary and sufficient conditions for $\mc{L}_t$ to be written in Lindblad form have been given in \cite{WEC08}. It should be noted that in general, for some physical processes, $\gamma_i(t)$ can be temporarily negative for some $i$ and the overall evolution still CPTP \cite{LPB10,HCLA14}.
	
	Given the generator $\mc{L}_t$ of the dynamics \eqref{eq-master_equation} and the corresponding positive trace-preserving maps $\{\mc{M}_{s,t}\}_{s,t\in I}$, it holds that the adjoint maps $\{\mc{M}_{s,t}^\dagger\}_{s,t\in I}$ are positive and unital. Furthermore, the adjoint maps $\{\mc{M}_{s,t}^\dagger\}_{s,t\in I}$ are generated by $\mc{L}_t^\dagger$, where $\mc{L}_t^\dagger$ is the adjoint of $\mc{L}_t$. The Lindblad form \eqref{eq-generator_time_dep} of the generator $\mc{L}_t^\dagger$ is
	\begin{equation}\label{eq-generator_adjoint}
		\mc{L}_t^\dagger(X)=\iota[H(t),X]+\sum_i \gamma_i(t)\left(A_i^\dagger(t) X A_i(t)-\frac{1}{2}\left\{X,A_i^\dagger(t) A_i(t)\right\}\right)\quad \forall X\in\msc{B}(\mc{H}).
	\end{equation}
		
	Now, let us consider the rate of entropy change $\frac{\d }{\d t}S(\rho_t)$ of a system in state $\rho_t$ at time $t$ evolving under dynamics with Liouvillian $\mc{L}_t$. Theorem~\ref{thm:rate-oqs} implies the following equality:
	\begin{align}
		\frac{\d}{\d t}S(\rho_t)=-\Tr\left\{\mc{L}_t(\rho_t)\log\rho_t\right\}\quad\forall~t\in I.
	\end{align}

	We now derive a limitation on the rate of entropy change of quantum Markov processes using the lower bound in Lemma \ref{thm:ent-lower-bound} on entropy change.

	\begin{theorem}[Lower limit on the rate of entropy change]\label{thm:q-Markov-ent-change-rate}
		The rate of entropy change of any quantum Markov process (Definition \ref{def:qmp}) is lower bounded as
		\begin{equation}\label{eq:qmp-rate-lim}
		\frac{\d}{\d t}S(\rho_t)\geq -\lim_{\varepsilon\to 0^+}\frac{\d}{\d\varepsilon}\Tr\left\{\Pi_t\(\(\mc{M}_{t+\varepsilon,t}\)^\dagger\circ\mc{M}_{t+\varepsilon,t}(\rho_t)\)\right\}= -\Tr\left\{\Pi_t\mc{L}_t^\dag(\rho_t)\right\},
		\end{equation}
		where $\Pi_t$ is the projection onto the support of the state $\rho_t$ of a system.
In general,	\eqref{eq:qmp-rate-lim} also holds for dynamics that obey \eqref{eq-master_equation} and are P-divisible. 		
	\end{theorem}
	\begin{proof}
		First, since the system is governed by \eqref{eq-master_equation}, so $\rho_{t+\varepsilon}=\mc{M}_{t+\varepsilon,t}(\rho_t)$ for any $\varepsilon>0$. Also, since $\mc{M}_{t+\varepsilon,t}$ is CPTP (hence positive and trace-preserving), we can apply Lemma \ref{thm:ent-lower-bound} to get the following inequality
		\begin{equation}
			S(\mc{M}_{t+\varepsilon,t}(\rho_t))-S(\rho_t)\geq D\(\rho_t\left\Vert (\mc{M}_{t+\varepsilon,t})^\dagger\circ\mc{M}_{t+\varepsilon,t}(\rho_t)\)\right.	
		\end{equation}
		Therefore, by definition of the derivative, we obtain
		\begin{align}
			\frac{\d}{\d t}S(\rho_t)&=\lim_{\varepsilon\to 0^+}\frac{S(\rho_{t+\varepsilon})-S(\rho_t)}{\varepsilon}\\
			%&=\lim_{\varepsilon\to 0^+}\frac{S(\mc{M}_{t+\varepsilon,t}(\rho_t))-S(\rho_t)}{\varepsilon}\\
			&\geq \lim_{\varepsilon\to 0^+}\frac{1}{\varepsilon}D\(\rho_t\left\Vert\(\mc{M}_{t+\varepsilon,t}\)^\dagger\circ\mc{M}_{t+\varepsilon,t}(\rho_t)\)\right.\\
			&=\lim_{\varepsilon\to 0^+}\frac{-S(\rho_t)-\Tr\left\{\rho_t\log\left[\(\mc{M}_{t+\varepsilon,t}\)^\dagger\circ\mc{M}_{t+\varepsilon,t}(\rho_t)\right]\right\}}{\varepsilon}\label{eq-lb_pf_1}\\
			&=-\lim_{\varepsilon\to 0^+}\frac{\d}{\d\varepsilon}\Tr\left\{\rho_t\log\left[\(\mc{M}_{t+\varepsilon,t}\)^\dagger\circ\mc{M}_{t+\varepsilon,t}(\rho_t)\right]\right\}\label{eq-lb_pf_2}\\
			&=-\lim_{\varepsilon\to 0^+}\Tr\left\{\rho_t\frac{\d\left(\log\left[\(\mc{M}_{t+\varepsilon,t}\)^\dagger\circ\mc{M}_{t+\varepsilon,t}(\rho_t)\right]\right)}{\d\varepsilon}\right\}\\
			&=-\lim_{\varepsilon\to 0^+}\frac{\d}{\d\varepsilon}\Tr\left\{\Pi_t \(\mc{M}_{t+\varepsilon,t}\)^\dagger\circ\mc{M}_{t+\varepsilon,t}(\rho_t)\right\}\label{eq-lb_pf_3},
		\end{align}
		where we used the definition of the derivative to get \eqref{eq-lb_pf_2} from \eqref{eq-lb_pf_1}. From Section~\ref{app-derivative}, and noting that $\lim_{\varepsilon\to 0}\(\mc{M}_{t+\varepsilon,t}\)^\dagger\circ\mc{M}_{t+\varepsilon,t}(\rho_t)=\rho_t$, we arrive at \eqref{eq-lb_pf_3}. Then, using the definition of the adjoint and the master equation \eqref{eq-master_equation}, we get
		\begin{align}
			&-\lim_{\varepsilon\to 0^+}\frac{\d}{\d\varepsilon}\Tr\left\{\Pi_t\(\mc{M}_{t+\varepsilon,t}\)^\dagger\circ\mc{M}_{t+\varepsilon,t}(\rho_t)\right\}\nonumber\\
			&\hspace{2.5cm}=-\lim_{\varepsilon\to 0^+}\frac{\d}{\d\varepsilon}\Tr\left\{\mc{M}_{t+\varepsilon,t}(\Pi_t)\mc{M}_{t+\varepsilon,t}(\rho_t)\right\}\\
			&\hspace{2.5cm}=-\lim_{\varepsilon\to 0^+}\Tr\left\{\frac{\d}{\d\varepsilon}\left(\mc{M}_{t+\varepsilon,t}(\Pi_t)\mc{M}_{t+\varepsilon,t}(\rho_t)\right)\right\}\\
			&\hspace{2.5cm}=-\lim_{\varepsilon\to 0^+}\Tr\left\{\left(\frac{\d}{\d\varepsilon}\mc{M}_{t+\varepsilon,t}(\Pi_t)\right)\mc{M}_{t+\varepsilon,t}(\rho_t)+\mc{M}_{t+\varepsilon,t}(\Pi_t)\left(\frac{\d}{\d\varepsilon}\mc{M}_{t+\varepsilon,t}(\rho_t)\right)\right\}.
		\end{align}
		Employing \eqref{eq-generator_time_dep_2} and the fact that $\mc{M}_{t,t}=\id$ for all $t\in I$, we get
		\begin{align}
			\mc{L}_t &=\lim_{\varepsilon\to 0^+}\frac{\mc{M}_{t+\varepsilon,t}-\id}{\varepsilon}= \lim_{\varepsilon\to 0^+}\frac{\d}{\d\varepsilon}\mc{M}_{t+\varepsilon,t}.
		\end{align}
		Therefore,
		\begin{align}
			-\lim_{\varepsilon\to 0^+}\frac{\d}{\d\varepsilon}\Tr\left\{\Pi_t\(\mc{M}_{t+\varepsilon,t}\)^\dagger\circ\mc{M}_{t+\varepsilon,t}(\rho_t)\right\}&=-\Tr\left\{\mc{L}_t(\Pi_t)\rho_t+\Pi_t\mathcal{L}_t(\rho_t)\right\}\\
			&=-\Tr\left\{\Pi_t\mc{L}^\dag_t(\rho_t)\right\}, \label{eq:qmp-final-step}
		\end{align}
		where we used the fact \eqref{eq-st_3} that $\Tr\{\Pi_t\mc{L}_t(\rho_t)\}=\Tr\left\{\Pi_t\dot{\rho}_t\right\}=0$.
	\end{proof}

	\bigskip
	
	Quantum dynamics obeying \eqref{eq-master_equation} are unital in a time interval $I$ if $\mc{L}_t(\mathbbm{1})=0$ for all $t\in I$, which implies that %$\Tr\{\Pi_t\mc{L}_t(\rho_t)\}=0$ and 
	$\Tr\{\Pi_t\mc{L}^\dag_t(\rho_t)\}=0$ for any initial state $\rho_0$ and for all $t\in I$. The deviation of $\Tr\{\Pi_t\mc{L}_t^\dagger(\rho_t)\}$ from zero is therefore a \textit{witness of non-unitality} at time $t$.   

	\begin{remark}\label{rem:q-Markov-ent-change-rate}
		When $\rho_t>0$, the rate of entropy change of any quantum Markov process is lower bounded as
		\begin{equation}
			\frac{\d}{\d t}S(\rho_t)\geq -\lim_{\varepsilon\to 0}\frac{\d}{\d\varepsilon}\Tr\left\{\(\mc{M}_{t+\varepsilon,t}\)^\dagger\circ\mc{M}_{t+\varepsilon,t}(\rho_t)\right\}= -\Tr\left\{\mc{L}_t^\dag (\rho_t)\right\}.
		\end{equation}
	\end{remark}
	
	Given a quantum Markov process and a state described by a density operator $\rho_t>0$ that is not a fixed (invariant) state of the dynamics, we can make the following statements for $t\in I$ and for all $\varepsilon>0$ such that $[t,t+\varepsilon)\subset I$:
	\begin{itemize}
	\item[(i)] If $\mc{M}_{t+\varepsilon,t}$ is strictly sub-unital, i.e., $\mc{M}_{t+\varepsilon,t}\(\mathbbm{1}\)< \mathbbm{1}$, then its adjoint is trace non-increasing, which means that $\Tr\{\mc{L}^\dag_t(\rho_t)\}<0$. This implies that the rate of entropy change is strictly positive for strictly sub-unital Markovian dynamics. 
	\item[(ii)] If $\mc{M}_{t+\varepsilon,t}$ is unital, i.e., $\mc{M}_{t+\varepsilon,t}\(\mathbbm{1}\)= \mathbbm{1}$, then its adjoint is trace-preserving, which means that $\Tr\{\mc{L}^\dag_t(\rho_t)\}=0$. This implies that the rate of entropy change is non-negative for unital Markovian dynamics. 
	\item[(iii)] If $\mc{M}_{t+\varepsilon,t}$ is strictly super-unital, i.e., $\mc{M}_{t+\varepsilon,\varepsilon}\(\mathbbm{1}\)> \mathbbm{1}$, then its adjoint is trace-increasing, which means that $\Tr\{\mc{L}^\dag_t(\rho_t)\}>0$. This implies that it is possible for the rate of entropy change to be negative for strictly super-unital Markovian dynamics. 
	\end{itemize}
	
	Using the Lindblad form of $\mc{L}_t^\dagger$ in \eqref{eq-generator_adjoint}, we find that 
	\begin{equation}\label{eq-qmp_L_dag}
		\Tr\{\mc{L}_t^\dagger(\rho_t)\}=\sum_i \gamma_i(t)\left<\left[A_i(t),A_i^\dagger(t)\right]\right>_{\rho_t}
	\end{equation}
	where $\<A\>_{\rho}=\Tr\{A\rho\}$. Using this expression, the lower bound on the rate of entropy change for quantum Markov processes when the state $\rho_t>0$ is
	\begin{equation}\label{eq:qmp-rate-lim_full}
		\frac{\d}{\d t}S(\rho_t)\geq \sum_i\gamma_i(t)\left<\left[A_i^\dagger(t),A_i(t)\right]\right>_{\rho_t}.
	\end{equation}
	The inequality \eqref{eq:qmp-rate-lim_full} was first derived in \cite{BN88} and recently discussed in \cite{OCA17}.
		
\bigskip		
		
	When the generator $\mc{L}_t\equiv\mc{L}$ is time-independent and $I=[0,\infty)$, it holds that the time evolution from time $s\in I$ to time $t\in I$ is determined merely by the time difference $t-s$, that is, $\mc{M}_{t,s}=\mc{M}_{t-s,0}$ for all $t\geq s$. The evolution of the system is then determined by a one-parameter semi-group. Let $\mc{M}_{t}\coloneqq \mc{M}_{t,0}$ for all $t\geq 0$. 
	
	\begin{remark}
		If the dynamics of a system are unital and can be represented by a one-parameter semi-group $\{\mc{M}_t\}_{t\geq 0}$ of quantum channels such that the generator $\mc{L}$ is self-adjoint, then for $\rho_0>0$,
		\begin{equation}\label{eq:eg-rel-t1}
			-\Tr\{\rho_0\log\rho_{2t}\}\leq S(\rho_t)\leq -\Tr\{\rho_{2t}\log\rho_0\}.
		\end{equation}
		This follows from Lemma \ref{thm:log-con}, \eqref{eq-adjoint-eg}, and the fact that $\mc{M}_t^\dagger=\mc{M}_t$. In particular,
		\begin{align}
			S(\rho_t) = S(\mc{M}_t(\rho_0))=-\Tr\{\mc{M}_t(\rho_0)\log\mc{M}_t(\rho_0)\} &\leq -\Tr\{\mc{M}_t(\rho_0)\mc{M}_t(\log\rho_0)\}\\
			&=-\Tr\{\mc{M}_t^\dagger\circ\mc{M}_t(\rho_0)\log\rho_0\}\\
			&=-\Tr\{\rho_{2t}\log\rho_0\}.
		\end{align}
		Similarly,
		\begin{align}
			S(\rho_t) = S(\mc{M}_t(\rho_0))=-\Tr\{\mc{M}_t(\rho_0)\log\mc{M}_t(\rho_0)\}
			& = -\Tr\{\rho_0\mc{M}^\dag_t(\log\mc{M}_t(\rho_0))\} \\
			&\geq -\Tr\{\rho_0\log(\mc{M}^\dag_t\circ\mc{M}_t(\rho_0))\}\\
			&=-\Tr\{\rho_0\log\rho_{2t}\}.
		\end{align}
	\end{remark}
	
	\begin{remark}
		If the dynamics of a system are unital and can be represented by a one-parameter semi-group $\{\mc{M}_t\}_{t\geq 0}$ of quantum channels such that the generator $\mc{L}$ is self-adjoint, then the entropy change is lower bounded as
		\begin{equation}\label{eq-mul}
			S(\rho_t)-S(\rho_0)\geq D(\rho_0\left\Vert \rho_{2t})\right. .
		\end{equation}
		This follows using Lemma \ref{thm:ent-lower-bound}. Under certain assumptions, when the dynamics of a system are described by Davies maps \cite{Dav74}, the same lower bound \eqref{eq-mul} holds for the entropy change \cite{AW16}.
	\end{remark}
	
	From the above remark, notice that the entropy change in a time interval $[0,t]$ is lower bounded by the relative entropy between the initial state $\rho_0$ and the evolved state $\rho_{2t}$ after time $2t$. In the context of information theory, the relative entropy has an operational meaning as the optimal type-II error exponent (in the asymptotic limit) in asymmetric quantum hypothesis testing \cite{HP91,ON00}. The entropy change in the time interval $[0,t]$ is thus an upper bound on the optimal type-II error exponent, where $\rho_0$ is the null hypothesis and $\rho_{2t}$ is the alternate hypothesis. 
	
\begin{remark}
Consider evolution of an open bipartite quantum system $AB$ given by two-parameter family $\{\mc{M}_{t,s}\}_{t,s\in I}$ of maps $\mc{M}_{t,s}:\msc{B}(\mc{H})\rightarrow\msc{B}(\mc{H})$, where $\mc{H}=\mc{H}_{A}\otimes\mc{H}_{B}$, as defined in \eqref{eq-channel_TO_exp}. Furthermore, assume that the dynamics are CP-divisible, meaning that the intermediate maps $\mc{M}_{t,r}$ and $\mc{M}_{r,s}$ are CPTP for all $s\leq r\leq t$. In other words, we are assuming that the dynamics of the given bipartite system is a quantum Markov process. Note that entangling abilities of such bipartite interactions are limited by the bounds derived in Chapters~\ref{ch:bqi} (see also \cite{BCM14} in the context of an open quantum system). 
\end{remark}
 
\subsection{Bosonic Gaussian dynamics}\label{sec-Gaussian}
	
	Here we consider Gaussian dynamics that can be represented by the one-parameter family $\{\mc{G}_t\}_{t\geq 0}$ of phase-insensitive bosonic Gaussian channels $\mc{G}_t$ (cf.\cite{HHW09}). It is known that all phase-insensitive gauge-covariant single-mode bosonic Gaussian channels form a one-parameter semi-group \cite{GHLM10}. The Liouvillian for such Gaussian dynamics is time-independent and has the following form:
	\begin{equation}\label{eq-Gaussian_generator}
		\mc{L}=\gamma_+\mc{L}_++\gamma_-\mc{L}_-,
	\end{equation}
	where
	\begin{align}
		\mc{L}_+(\rho)&=\hat{a}^\dagger\rho\hat{a}-\frac{1}{2}\left\{\hat{a}\hat{a}^\dagger,\rho\right\},\\
		\mc{L}_-(\rho)&=\hat{a}\rho\hat{a}^\dagger-\frac{1}{2}\left\{\hat{a}^\dagger\hat{a},\rho\right\},
	\end{align}
	$\hat{a}$ is the field-mode annihilation operator of the system, and the following commutation relation holds for bosonic systems:
	\begin{equation}
		\left[\hat{a},\hat{a}^\dagger\right]=\mathbbm{1}.
	\end{equation}
	The state $\rho_t$ of the system at time $t$ is
	\begin{equation}
		\rho_t=\mc{G}_t(\rho_0)=e^{t\mc{L}}(\rho_0).
	\end{equation}
	The thermal state $\rho_{\text{th}}(N)$ with mean photon number $N$ is defined as
	\begin{equation}
		\rho_{\text{th}}(N)\coloneqq\frac{1}{N+1}\sum_{n=0}^{\infty}\left(\frac{N}{N+1}\right)^n\op{n},
	\end{equation}
	where $N\geq 0$ and $\{\ket{n}\}_{n\geq 0}$ is the orthonormal, photonic number-state basis. Using \eqref{eq-qmp_L_dag}, we get 
	\begin{align}
		-\Tr\{\mc{L}^\dagger(\rho_t)\}&=-\gamma_+\left<\left[\hat{a}^\dagger,\hat{a}\right]\right>_{\rho_t}-\gamma_-\left<\left[\hat{a},\hat{a}^\dagger\right]\right>_{\rho_t}\\
		&=\gamma_+-\gamma_-~.
	\end{align}
	Therefore, by Remark \ref{rem:q-Markov-ent-change-rate}, if $\rho_t>0$, then
	\begin{align}
		\frac{\d S(\mc{G}_t(\rho_0))}{\d t}\geq\gamma_+-\gamma_-~.
	\end{align}
	The lower bound $\gamma_+-\gamma_-$ is a witness of non-unitality. It is positive for strictly sub-unital, zero for unital, and negative for strictly super-unital dynamics. For example, when the dynamics are represented by a family $\{\mc{A}_t\}_{t\geq 0}$ of noisy amplifier channels $\mc{A}_t$ with thermal noise $\rho_{\text{th}}(N)$, then $\gamma_+=N+1$ and $\gamma_-=N$, which implies that the dynamics are strictly sub-unital. When the dynamics are represented by a family $\{\mc{B}_t\}_{t\geq 0}$ of lossy channels $\mc{B}_t$ (i.e., beamsplitters) with thermal noise $\rho_{\text{th}}(N)$, then $\gamma_+=N$, $\gamma_-=N+1$, which implies that the dynamics are strictly super-unital. When the dynamics are represented by a family $\{\mc{C}_t\}_{t\geq 0}$ of additive Gaussian noise channels $\mc{C}_t$, then $\gamma_+=\gamma_-$, which implies that the dynamics are unital. 
	
\section{Quantum non-Markovian processes}\label{sec:qnmp}

Dynamics of a quantum system that are not a quantum Markov process as stated in Definition~\ref{def:qmp} are called a quantum non-Markov process. Among these two classes of quantum dynamics, non-Markov processes are not well understood and have attracted increased focus over the past decade. Some examples of applications of quantum Markov processes are in the fields of quantum optics, semiconductors in condensed matter physics, the quantum mechanical description of Brownian motion, whereas some examples where quantum non-Markov processes have been applied are in the study of a damped harmonic oscillator, or a damped driven two-level atom \cite{Car09,Wei12,Riv11}.

	There can be several tests derived from the properties of quantum Markov processes, the satisfaction of which gives witnesses of non-Markovianity. Based on Theorem \ref{thm:q-Markov-ent-change-rate}, we mention here a few tests that will always fail for a quantum Markov process. Passing of these tests guarantees that the dynamics are non-Markovian. 

	An immediate consequence of Theorem \ref{thm:q-Markov-ent-change-rate} is that only a quantum non-Markov process can pass any of the following tests: 
	\begin{itemize}
		\item[(a)] 
			\begin{equation}\label{eq-nmw_sm}
				\frac{\d}{\d t}S(\rho_t) + \lim_{\varepsilon\to 0^+}\frac{\d}{\d\varepsilon}\Tr\left\{\Pi_t\(\(\mc{M}_{t+\varepsilon,t}\)^\dagger\circ\mc{M}_{t+\varepsilon,t}(\rho_t)\)\right\}< 0.
			\end{equation}
		\item[(b)] 
			\begin{equation}\label{eq-nmw_sl}
				\frac{\d}{\d t}S(\rho_t)+\Tr\left\{\Pi_t\mc{L}_t^\dag(\rho_t)\right\}<0.
			\end{equation}
		\item[(c)] 
			\begin{equation}
				\lim_{\varepsilon\to 0^+}\frac{\d}{\d\varepsilon}\Tr\left\{\Pi_t\(\(\mc{M}_{t+\varepsilon,t}\)^\dagger\circ\mc{M}_{t+\varepsilon,t}(\rho_t)\)\right\}\neq \Tr\left\{\Pi_t\mc{L}_t^\dag(\rho_t)\right\}. 
			\end{equation}
			\end{itemize}

If the dynamics of the system satisfy any of the above tests, then the process is non-Markovian. Based on the description of the dynamics and the state of the system, one can choose which test to apply. In the case of unital dynamics, \eqref{eq-nmw_sm} and \eqref{eq-nmw_sl} reduce to $\frac{\d}{\d t}S(\rho_t)<0$. The observation that the negativity of the rate of entropy change is a witness of non-Markovianity for random unitary processes, which are a particular kind of unital processes, was made in \cite{CM14}.

	Based on the above witnesses of non-Markovianity, we can introduce different measures of non-Markovianity for physical processes. Here, we define two measures of non-Markovianity that are based on the channel and generator representation of the dynamics of the system:

	\begin{enumerate}
		\item 
			\begin{equation}\label{eq:dual-measure-non-markov}
				\complement_{\textnormal{M}}(\mc{L}):=\max_{\rho_0}\SumInt_{\substack{t:\\\frac{\d S(\rho_t)}{\d t}+\Tr\left\{\Pi_t\mc{L}_t^\dag (\rho_t)\right\}< 0}} \left\vert\frac{\d S(\rho_t)}{\d t}+\Tr\left\{\Pi_t\mc{L}_t^\dag (\rho_t)\right\}
\right\vert.
			\end{equation}
		\item 
			\begin{equation}\label{eq:cor-measure-non-markov}
				\complement_{\textnormal{M}}(\mc{M}):=\max_{\rho_0}\SumInt_{t: f(t)<0} \left\vert f(t)
\right\vert,
			\end{equation}
			where 
			\begin{equation}\label{eq-nm_ft}
				f(t)\coloneqq \frac{\d}{\d t}S(\rho_t)+\lim_{\varepsilon\to 0^+}\frac{\d}{\d\varepsilon}\Tr\left\{\Pi_t\(\(\mc{M}_{t+\varepsilon,t}\)^\dagger\circ\mc{M}_{t+\varepsilon,t}(\rho_t)\)\right\}.
			\end{equation}
	\end{enumerate}
	In the case of unital dynamics, the above measures are equal. It should be noted that the above measures of non-Markovianity are not faithful. This is due to the fact that the statements in Theorem \ref{thm:q-Markov-ent-change-rate} do not provide sufficient conditions for the evolution to be a quantum Markov process. In other words, if the measure $\complement_{\textnormal{M}}$ \eqref{eq:dual-measure-non-markov} is non-zero, then the dynamics are non-Markovian, but if it is equal to zero, then that does not in general imply that the dynamics are Markovian.

\subsection{Examples}
	In this section, we consider two common examples of quantum non-Markov processes: pure decoherence of a qubit system (Section \ref{sec-decohere}) and a generalized amplitude damping channel (Section \ref{sec-GADC}). In order to characterize quantum dynamics, several witnesses of non-Markovianity and measures of non-Markovianity based on these witnesses have been proposed \cite{WEC08,BLP09,RHP10,LWS10,LPP11,LFS12,ZSMWN13,LLW13,LPP13,CM14,HCLA14,HSK15,PGD+16}. Many of these measures are based on the fact that certain quantities are monotone under Markovian dynamics, such as the trace distance between states \cite{BLP09}, entanglement measures \cite{RHP10,LPP11,LFS12}, Fisher information and Bures distance \cite{LWS10,ZSMWN13,LLW13}, and the volume of states \cite{LPP13}. Among these measures, the one proposed in \cite{RHP10}  based on the Choi representation of dynamics is both necessary and sufficient. The measure proposed in \cite{HCLA14} is also necessary and sufficient and is based on the values of the decay rates $\gamma_i(t)$ appearing in the Lindblad form \eqref{eq-generator_time_dep} of the Liouvillian of the dynamics.
	
	Here, we compare our measures of non-Markovianity with the widely-considered Breuer-Laine-Piilo (BLP) measure of non-Markovianity \cite{BLP09}. This is a measure of non-Markovianity defined using the trace distance and is based on the fact that the trace distance is monotonically non-increasing under quantum channels. Breuer et al.~\cite{BLP09} in 2009 defined Markovianity using CP-divisibility. BLP measure uses the trace distance and exploits the fact that it is monotonically non-increasing under quantum channels. Violation of this monotonicity is thus an indication of non-Markovianity. Specifically, for a given set $\{\mc{M}_{t,s}\}_{s,t\geq 0}$ of completely positive and trace-preserving maps, their measure is
			\begin{equation}\label{eq-Breuer_measure}
				N=\max_{\rho_1(0),\rho_2(0)}\int_{\sigma>0}\sigma(t,\rho_1(0),\rho_2(0))~\d t,
			\end{equation}
			where $\sigma(t,\rho_1(0),\rho_2(0))=\frac{\d}{\d t}\frac{1}{2}\norm{\rho_1(t)-\rho_2(t)}_1$ and $\rho_1(t)=\mc{M}_{t,0}\rho_1(0)$, $\rho_2(t)=\mc{M}_{t,0}\rho_2(0)$.
			
Our measure agrees with the BLP measure in the case of pure decoherence of a qubit. In the case of the generalized amplitude damping channel, our witness is able to detect non-Markovianity even when the BLP measure does not.

\subsubsection{Pure decoherence of a qubit system}\label{sec-decohere}

	Consider a two-level system with ground state $\ket{-}$ and excited state $\ket{+}$. This qubit system is allowed to interact with a bosonic environment that is a reservoir of field modes. The time evolution of the qubit system is given by
	\begin{equation}
		\frac{\d\rho_t}{\d t}=-\iota[H(t),\rho_t]+\gamma(t)\left[\sigma_{-}\rho_t\sigma_+-\frac{1}{2}\{\sigma_+\sigma_-,\rho_t\}\right],
	\end{equation}
	where $\sigma_+=\ket{+}\bra{-}$, $\sigma_-=\ket{-}\bra{+}$ and $t\geq 0$. If $H(t)=0$, then the system undergoes pure decoherence and the Liouvillian reduces to
	\begin{equation}
		\mc{L}_t(\rho_t)=\frac{\gamma(t)}{2}\(\sigma_z\rho_t\sigma_z-\rho_t\),
	\end{equation}
	where $\sigma_z=[\sigma_+,\sigma_-]$. The decoherence rate is given by $\gamma(t)$, and it can be determined by the spectral density of the reservoir \cite{BLP09}.
	
	We can verify that $\Tr\{\Pi_t\mc{L}_t^\dagger(\rho_t)\}=0$ for all $t\geq 0$ and any initial state $\rho_0$. This implies that the dynamics are unital for all $t\geq 0$. In this case, for $t>0$, our witness \eqref{eq-nmw_sl} reduces to $\frac{\d}{\d t}S(\rho_t)<0$. For qubit systems undergoing the given unital evolution, it holds that $\rho_t>0$ for all $t>0$, and thus for $t>0$ our measures \eqref{eq:dual-measure-non-markov} and \eqref{eq:cor-measure-non-markov} are equal and reduce to the measure in \cite[Eq.~(15)]{HSK15}, which was based on the fact that the rate of entropy change is non-negative for unital quantum channels. As stated therein, these measures of non-Markovianity are positive and agree with those obtained by the BLP measure \cite[Eq.~(11)]{BLP09}.

\subsubsection{Generalized amplitude damping channel}\label{sec-GADC}

	In this example, we consider non-unital dynamics that can be represented as a family of generalized amplitude damping channels $\{\mc{M}_t\}_{t\geq 0}$ on a two-level system \cite{LLW13}. These channels have Kraus operators \cite{Fuj04}
	\begin{equation}\label{eq-GADC}
		\begin{aligned}
		M_t^1&=\sqrt{p_t}\begin{pmatrix} 1&0\\0&\sqrt{\eta_t}\end{pmatrix}\\
		M_t^2&=\sqrt{p_t}\begin{pmatrix} 0&\sqrt{1-\eta_t}\\0&0\end{pmatrix}\\
		M_t^3&=\sqrt{1-p_t}\begin{pmatrix} \sqrt{\eta_t}&0\\0&1\end{pmatrix}\\
		M_t^4&=\sqrt{1-p_t}\begin{pmatrix} 0&0\\\sqrt{1-\eta_t}&0\end{pmatrix},
		\end{aligned}
	\end{equation}
	where $p_t=\cos^2(\omega t)$, $\omega\in\mathbb{R}$, and $\eta_t=e^{-t}$. Then, for all $t\geq 0$, $\mc{M}_t(\rho)=\sum_{i=1}^4 M_t^i\rho (M_t^i)^\dagger$. $\mc{M}_t$ is unital if and only if $p_t=\frac{1}{2}$ or $\eta_t=1$. When $\eta_t=1$, $\mc{M}_t=\id$ for all $\omega$.
	
	It was shown in\cite{LLW13} that the BLP measure \cite{BLP09} does not capture the non-Markovianity of the dynamics given by \eqref{eq-GADC}.
	
	Let the initial state $\rho_0$ be maximally mixed, that is, $\rho_0=\frac{1}{2}\mathbbm{1}$. The evolution of this state under $\mc{M}_t$ is then
	\begin{equation}
		\rho_t\coloneqq\mc{M}_t(\rho_0)=\frac{1}{2}\begin{pmatrix} 1+W_t&0\\0&1-W_t\end{pmatrix},
	\end{equation}
	where $W_t=(2p_t-1)(1-\eta_t)=\cos(2\omega t)(1-e^{-t})$. Note that $\rho_t>0$ for all $t\geq 0$. The evolution of these states for an $\varepsilon >0 $ time interval is
	\begin{equation}
		\rho_{t+\varepsilon}=\mc{M}_\varepsilon(\rho_t)=\frac{1}{2}\begin{pmatrix} 1+W_\varepsilon+\eta_\varepsilon W_t&0\\0&1-W_\varepsilon-\eta_\varepsilon W_t\end{pmatrix}
	\end{equation}
	To check whether or not the given dynamics are non-Markovian, we apply the test in \eqref{eq-nmw_sm}. First, we evaluate
	\begin{equation}
		\mc{M}_\varepsilon^\dagger\circ\mc{M}_\varepsilon(\rho_t)=\frac{1}{2}\begin{pmatrix} a_t&0\\0&b_t\end{pmatrix},
	\end{equation}
	where
	\begin{align}
		a_t&\coloneqq p_t(1+W_\varepsilon+\eta_\varepsilon W_t)+(1-p_t)\eta_t(1+W_\varepsilon+\eta_\varepsilon W_t)+(1-p_t)(1-\eta_t)(1-W_\varepsilon+\eta_\varepsilon W_t)\\
		b_t&\coloneqq p_t\eta_t(1-W_\varepsilon+\eta_\varepsilon W_t)+p_t(1-\eta_t)(1+W_\varepsilon+\eta_\varepsilon W_t)+(1-p_t)(1-W_\varepsilon+\eta_\varepsilon W_t).
	\end{align}
	Then,
	\begin{equation}
		\lim_{\varepsilon\to 0^+}\frac{\d}{\d\varepsilon}\Tr\left\{\mc{M}_\varepsilon^\dagger\circ\mc{M}_\varepsilon(\rho_t)\right\}=W_t.
	\end{equation}
	We note that the deviation of $W_t$ from zero is a witness of non-unitality. For a unital process, for any initial state $\rho_0$ and for all time $t$, we have $\lim_{\varepsilon\to 0^+}\frac{\d}{\d\varepsilon}\Tr\left\{\Pi_t\mc{M}_\varepsilon^\dagger\circ\mc{M}_\varepsilon(\rho_t)\right\}=0$. For a non-unital process, there will exist some initial state such that for some time $t$, $\lim_{\varepsilon\to 0^+}\frac{\d}{\d\varepsilon}\Tr\left\{\Pi_t\mc{M}_\varepsilon^\dagger\circ\mc{M}_\varepsilon(\rho_t)\right\}\neq 0$.
	Next, we evaluate the entropy of the state $\rho_t$ to be
	\begin{equation}
		S(\rho_t)=-\frac{1}{2}\left[(1+W_t)\log\(\frac{1+W_t}{2}\)+(1-W_t)\log\(\frac{1-W_t}{2}\)\right].
	\end{equation}
	This implies that the rate of entropy change is:
	\begin{align}
		\frac{\d S(\rho_t)}{\d t}&=-\frac{1}{2}\frac{\d W_t}{\d t}\log\left[\frac{1+W_t}{2}\right]+\frac{1}{2}\frac{\d W_t}{\d t}\log\left[\frac{1-W_t}{2}\right]\\
		&=\frac{1}{2}\frac{\d W_t}{\d t}\log\left[\frac{1-W_t}{1+W_t}\right],
	\end{align}
	where
	\begin{equation}
		\frac{\d W_t}{\d t}=-2\omega\sin(2\omega t)(1-e^{-t})+\cos(2\omega t)e^{-t}.
	\end{equation}
	Therefore, the test in \eqref{eq-nmw_sm} reduces to
	\begin{equation}\label{eq:f-non-markov}
		f(t)= -\frac{1}{2}\frac{\d W_t}{\d t}\log\left[\frac{1+W_t}{2}\right]+\frac{1}{2}\frac{\d W_t}{\d t}\log\left[\frac{1-W_t}{2}\right]+W_t<0,
	\end{equation}
	where $f$ is defined in \eqref{eq-nm_ft}. For values of $\omega$ such that the dynamics are non-unital, we find that $f$ can be negative in several time intervals; for example, see Fig.~\ref{fig:nu-nm} for the case $\omega=5$.
	
	\begin{figure}
		\centering
		\includegraphics[scale=0.6]{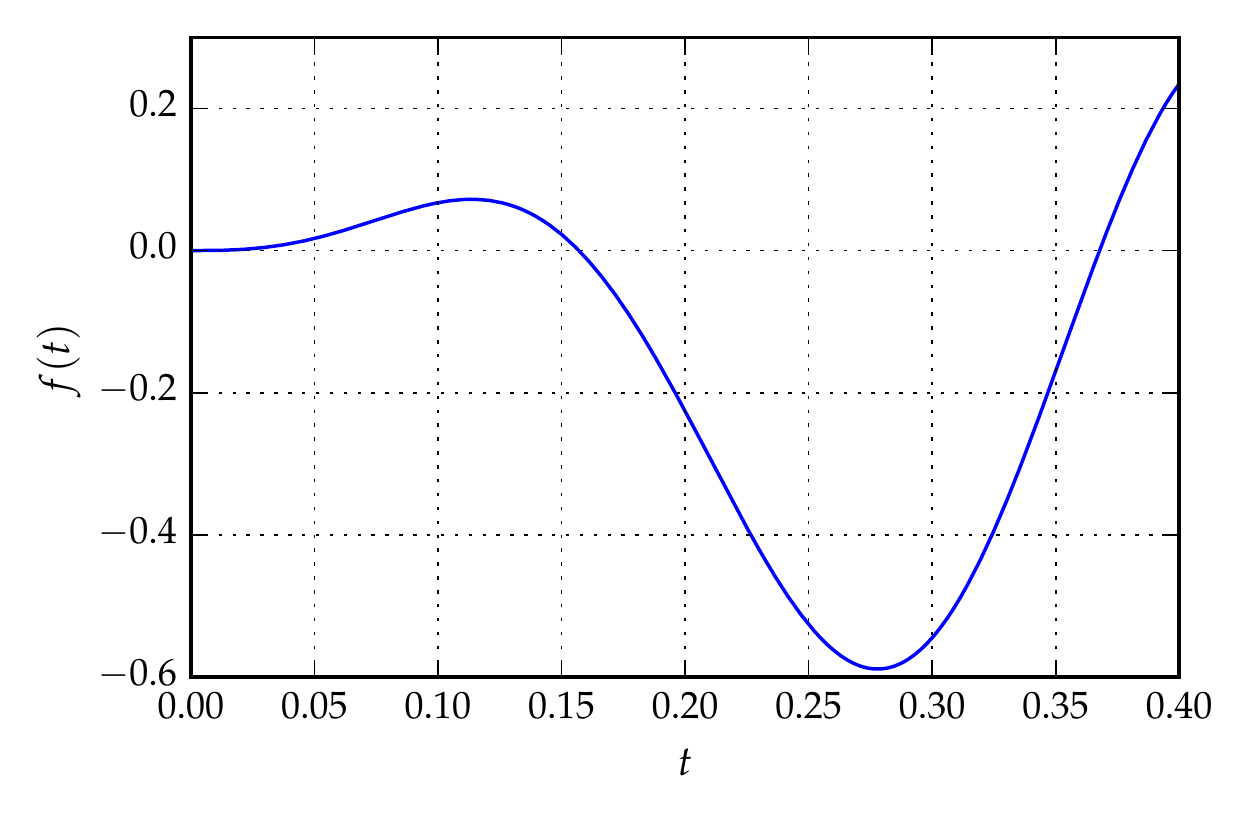}
		\caption{Negative values of $f$, as given in \eqref{eq:f-non-markov}, indicate non-Markovianity for $\omega=5$.}\label{fig:nu-nm}
	\end{figure}
	
\section{Bounds on entropy change}\label{sec-entropy_change}

	In this section, we derive bounds on how much the entropy of a system can change as a function of the initial state of the system and the evolution it undergoes.
		
	\begin{lemma}\label{thm:ent-lower-bound-2}
		Let $\mc{M}:\msc{B}_1^+(\mc{H})\to\msc{B}_1^+(\mc{H}^\prime)$ be a positive, trace-non-increasing map. Then, for all $\rho\in\msc{D}(\mc{H})$ such that $\mc{M}(\rho)>0$,
		\begin{equation}\label{eq-ent_lower_bound-2}
			\Delta S(\rho,\mc{M})\geq D\(\rho\left\Vert\mc{M}^\dag\circ\mc{M}\(\rho\)\)\right. .
		\end{equation}
	\end{lemma}
	\begin{proof}
		Using the definition \eqref{eq-adjoint-eg} of the adjoint, one obtains
		\begin{align}
			\Delta S(\rho,\mc{M})=S(\mc{M}(\rho))-S(\rho)&=\Tr\{\rho\log\rho\}-\Tr\left\{\mc{M}(\rho)\log\mc{M}(\rho)\right\}\nonumber\\
			&=\Tr\{\rho\log\rho\}-\Tr\left\{\rho\mc{M}^\dag\(\log\mc{M}(\rho)\)\right\}\nonumber\\
			&\geq \Tr\{\rho\log\rho\}-\Tr\left\{\rho \log\[\mc{M}^\dag\circ\mc{M}(\rho)\]\right\}\nonumber\\
			&=D\(\rho\left\Vert\mc{M}^\dag\circ\mc{M}\(\rho\)\)\right. .
		\end{align}
		The inequality follows from Lemma~\ref{thm:log-con} applied to $\mathcal{M}^\dagger$, which is positive and sub-unital since $\mathcal{M}$ is positive and trace non-increasing. 
	\end{proof}
	
	\bigskip
	
	Note that for a quantum channel $\mc{M}$, $\Delta S(\rho,\mc{M})=0$ for all $\rho$ if and only if $\rho=\mc{M}^\dag\circ\mc{M}(\rho)$, which is true if and only if $\mc{M}$ is a unitary operation \cite[Theorem~2.1]{NS06},\cite[Theorem~3.4.1]{Riv11}. We use this fact to provide a measure of non-unitarity in Section \ref{sec-non_unitarity}.
	
	As an application of the lower bound in Lemma~\ref{thm:ent-lower-bound}, let us suppose that a quantum channel $\mc{E}_{A\to B}$ can be simulated as follows
	\begin{equation}\label{eq:env-par}
\forall~\rho_A\in\msc{D}(\mc{H}_A):\		\mc{E}_{A\to B}(\rho_A)=\mc{F}_{AC\to B}(\rho_A\otimes\theta_C),
	\end{equation}
	for a fixed (interaction) channel $\mc{F}_{AC\to B}$ and a fixed ancillary state $\theta_C$. By applying Lemma~\ref{thm:ent-lower-bound} to $\mc{F}$ and the state $\rho_A\otimes\theta_C$, we obtain
	\begin{align}
	\Delta S(\rho_A,\mc{E})&=S\(\mc{F}(\rho_A\otimes\theta_C)\)-S\(\rho_A\) \\
	&\geq S\(\rho_A\otimes\theta_C\)-S\(\rho_A\) +D\(\rho_A\otimes\theta_C\left\Vert \mc{F}^\dagger\circ\mc{F}(\rho_A\otimes\theta_C)\)\right. \\
	%&=S\(\rho_A\)+S\(\theta_C\)-S\(\rho_A\)+D\(\rho_A\otimes\theta_C\left\Vert \mc{F}^\dagger\circ\mc{F}(\rho_A\otimes\theta_C)\)\right.\\
	&=S\(\theta_C\)+D\(\rho_A\otimes\theta_C\left\Vert \mc{F}^\dagger\circ\mc{F}(\rho_A\otimes\theta_C)\)\right. .
	\end{align}
	Equality holds, i.e., $\Delta S(\rho,\mc{E})=S\(\theta_C\)$, if and only if the interaction channel $\mc{F}$ is a unitary interaction. If $\mc{F}$ is a sub-unital channel, then $\Delta S(\rho,\mc{E})\geq S\(\theta_C\)$ because the relative entropy term is non-negative. This result is of relevance in the context of quantum channels obeying certain symmetries (see Section~\ref{sec:symmetry}).

	\begin{lemma}\label{thm:ent-upper-bound}
		Let $\mc{M}:\msc{B}_+(\mc{H})\to\msc{B}_+(\mc{H}^\prime)$ be a sub-unital channel. Then, for all $\rho\in\msc{D}(\mc{H})$ such that $\rho>0$,
	\begin{equation}\label{eq:ent-upper-bound}
		\Delta S(\rho,\mc{M})\leq \Tr\left\{\[\rho-\mc{M}^\dag\circ\mc{M}(\rho)\]\log\rho\right\}.
	\end{equation}
	This also holds for any positive sub-unital map satisfying the above conditions. 
	\end{lemma}
	\begin{proof}
		By applying Lemma \ref{thm:log-con} to $\mc{M}$, we get
		\begin{align}
			\Delta S(\rho,\mc{M})&=\Tr\{\rho\log\rho\}-\Tr\left\{\mc{M}(\rho)\log\mc{M}(\rho)\right\}\nonumber\\
		&\leq \Tr\{\rho\log\rho\}-\Tr\left\{\mc{M}\(\rho\) \mc{M}\(\log\rho\)\right\}\nonumber\\
			&=\Tr\left\{\[\rho-\mc{M}^\dag\circ\mc{M}(\rho)\]\log\rho\right\}.
		\end{align}
		This concludes the proof.
	\end{proof}

	\bigskip

	By applying H\"{o}lder's inequality (Lemma~\ref{thm:sch-norm}) to this upper bound, we obtain the following.
	\begin{corollary}\label{thm:sub-u-ent-change}
		Let $\mc{M}:\msc{B}_+(\mc{H})\to\msc{B}_+(\mc{H}^\prime)$ be a sub-unital channel. Then, for all $\rho\in\msc{D}(\mc{H})$ such that $\rho>0$,
		\begin{equation}
			\Delta S(\rho,\mc{M})\leq \left\Vert\rho-\mc{M}^\dag\circ\mc{M}(\rho)\right\Vert_1\left\Vert\log\rho\right\Vert_\infty .
		\end{equation}
	\end{corollary}
	
	\bigskip

	Now, assume $\mc{M}$ to be a sub-unital quantum sub-operation, then as a consequence of Lemma \ref{thm:ent-lower-bound} and Corollary \ref{thm:sub-u-ent-change}, we have, for all states $\rho>0$ such that $\mc{M}(\rho)>0$ and the entropies $S(\rho)$ and $S(\mc{M}(\rho))$ are finite,
	\begin{equation}\label{eq-ent_change_channel}
		D\(\rho\left\Vert\mc{M}^\dag\circ\mc{M}\(\rho\)\) \leq S(\mc{M}(\rho))-S(\rho)\leq \left\Vert\rho-\mc{M}^\dag\circ\mc{M}(\rho)\right\Vert_1\left\Vert\log\rho\right\Vert_\infty \right..
	\end{equation}
	It is interesting to note that \eqref{eq-ent_change_channel} implies
	\begin{equation}
		\norm{\rho-\mc{M}^\dag\circ\mc{M}(\rho)}_1 \geq \frac{1}{\norm{\log\rho}_\infty} D\(\rho\left\Vert\mc{M}^\dag\circ\mc{M}\(\rho\)\)\right.
	\end{equation}
	for a sub-unital quantum sub-operation $\mc{M}$ and a state $\rho>0$ such that $\mc{M}(\rho)>0$. This inequality has the reverse form of Pinsker's inequality \cite{HOT81}, which in this case is
	\begin{equation}\label{eq-pinsker}
		D\(\rho\left\Vert\mc{M}^\dag\circ\mc{M}\(\rho\)\)\right.\geq \frac{1}{2} \norm{\rho-\mc{M}^\dag\circ\mc{M}(\rho)}_1^2.
    \end{equation}
    In general, the relationship between relative entropy and different distance measures, including trace distance, has been studied in \cite{BR96,AE05,AE11}.
    
    \section{Measure of non-unitarity}\label{sec-non_unitarity}
	
	In this section, we introduce a measure of non-unitarity for any unital quantum channel that is inspired by the discussion at the end of Section \ref{sec-entropy_change}. A measure of unitarity for channels $\mc{M}:\msc{D}(\mc{H}_A)\to\msc{D}(\mc{H}_A)$, where $\mc{H}_A$ is finite-dimensional, was defined in \cite{WGH15}. A related measure for non-isometricity for sub-unital channels was introduced in \cite{BDW16}. A measure of non-unitarity for a unital channel is a quantity that gives the distinguishability between a given unital channel with respect to any unitary operation. It quantifies the deviation of a given unital evolution from a unitary evolution. These measures are relevant in the context of cryptographic applications \cite{PLSW04,Aub09} and randomized benchmarking \cite{WGH15}. 
	
	\begin{figure}
		\centering
		\includegraphics[scale=0.6]{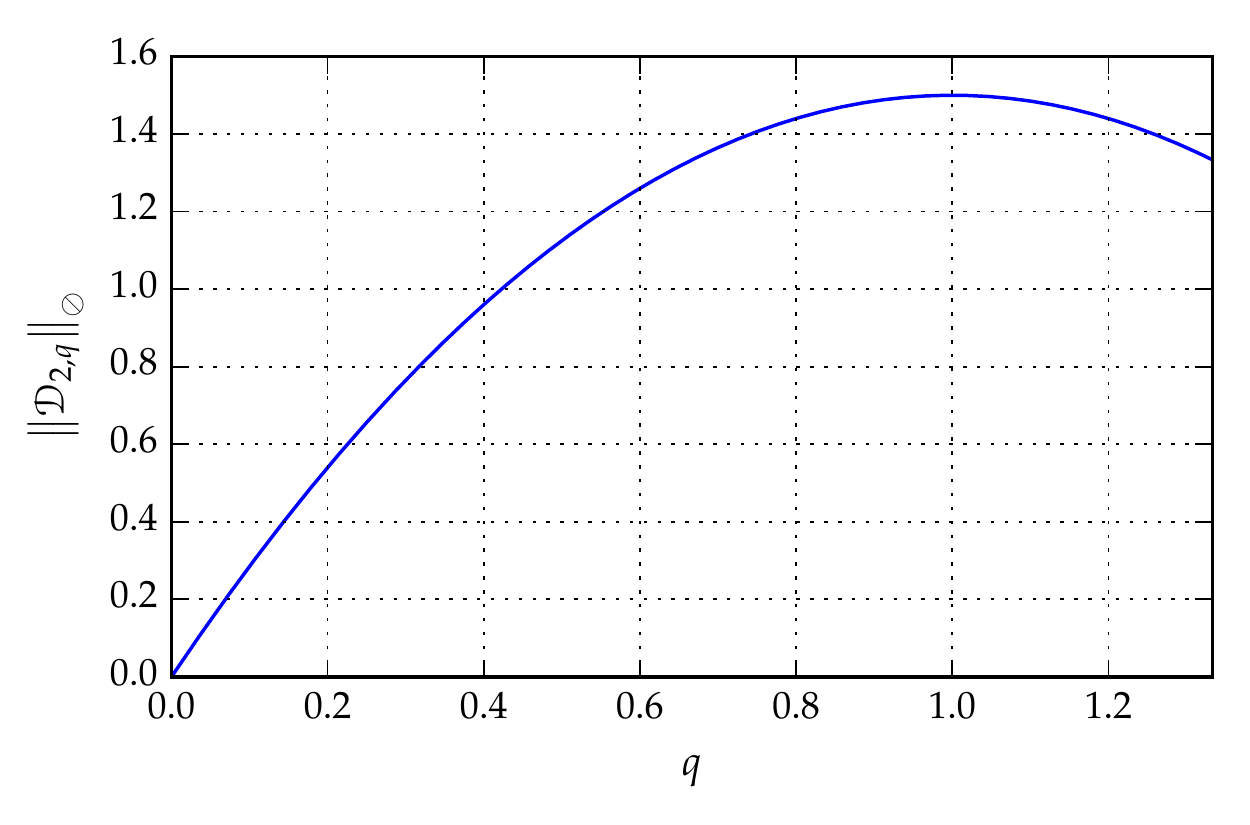}
		\caption{The measure $\norm{\mc{D}_{2,q}}_{\oslash}$ of non-unitarity for the qubit depolarizing channel $\mc{D}_{2,q}$ as a function of the parameter $q\in\left[0,\frac{4}{3}\right]$.}\label{fig:non_unitary}
	\end{figure}

	It is known that any unitary evolution is reversible. The adjoint of a unitary operator is also a unitary operator, and a unitary operator and its adjoint are the inverse of each other. These are the distinct properties of any unitary operation. Let $U_{A\to B}$ denote a unitary operator, where $\dim(\mc{H}_A)=\dim(\mc{H}_B)$. Then a necessary condition for the unitarity of $U_{A\to B}$ is that $\(U_{A\to B}\)^\dag U_{A\to B}=\mathbbm{1}_{A}$. The unitary evolution $\mc{U}_{A\to B}$ of a quantum state $\rho_A$ is given by
	\begin{equation}\label{eq-unitary_channel}
		\mc{U}_{A\to B}(\rho_A)= U_{A\to B}(\rho_A) \(U_{A\to B}\)^\dag.
	\end{equation}
	From the reversibility property of a unitary evolution, it holds that $\(\mc{U}_{A\to B}\)^\dag\circ\mc{U}_{A\to B}=\id_A$. It is clear that $\(\mc{U}_{A\to B}\)^\dag$ is also a unitary evolution, and $\(\mc{U}_{A\to B}\)^\dag$ and $\mc{U}_{A\to B}$ are the inverse of each other. 

	Contingent upon the above observation, it is to be noted that a measure of non-unitarity for a unital channel $\mc{M}_{A\to B}$ should quantify the deviation of $\(\mc{M}_{A\to B}\)^\dag\circ \mc{M}_{A\to B}$ from $\id_A$ and is desired to be a non-negative quantity. We make use of the trace distance, which gives a distinguishability measure between two positive semi-definite operators and appears in the upper bound\footnote{Notice that the lower bound on the entropy change can also be used to arrive at the measure in terms of trace distance by employing Pinsker's inequality \eqref{eq-pinsker}.} on entropy change for a unital channel (Section \ref{sec-entropy_change}), to define a measure of non-unitarity for a unital channel called the diamond norm of non-unitarity.    

	\begin{definition}[Diamond norm of non-unitarity]
		The diamond norm of non-unitarity of a unital channel $\mc{M}_{A\to B}$ is a measure that quantifies the deviation of a given unital evolution from a unitary evolution and is defined as
		\begin{equation}\label{eq-od}
			\Vert\mc{M}\Vert_{\oslash}=\norm{\id-\mc{M}^\dagger\circ\mc{M}}_{\diamond},
		\end{equation}
		where the diamond norm $\norm{\cdot}_{\diamond}$ \cite{Kit97} of a Hermiticity-preserving map $\mc{M}$ is defined as
		\begin{equation}\label{eq-diamond_norm}
			\norm{\mc{M}}_{\diamond}=\max_{\rho_{RA}\in\msc{D}(\mathcal{H}_{RA})}\norm{(\id\otimes \mc{M})(\rho_{RA})}_1.
		\end{equation}
		In other words,
		\begin{equation}\label{eq-ot}
			\norm{\mc{M}}_{\oslash}=\max_{\rho_{RA}\in\msc{D}(\mathcal{H}_{RA})}\norm{(\id\otimes(\id-\mc{M}^\dagger\circ\mc{M}))(\rho_{RA})}_1.
		\end{equation}
	\end{definition}

	The diamond norm of non-unitarity of any unital channel $\mc{M}$ has the following properties:
	\begin{enumerate}
		\item $\Vert\mc{M}\Vert_{\oslash}\geq 0$.\label{it-ol}
		\item $\Vert\mc{M}\Vert_{\oslash}=0$ if and only if $\mc{M}^\dag\circ\mc{M}=\id$, i.e., the unital channel $\mc{M}$ is unitary. \label{it-ou}
		\item In \eqref{eq-ot}, it suffices to take $\rho_{RA}$ to be rank one and to let $\dim(\mc{H}_{R})=\dim(\mc{H}_A)$. \label{it-or}
		\item $\norm{\mc{M}}_{\oslash}\leq 2$.\label{it-o2}
	\end{enumerate}
	
	Noticing that $\mc{M}^\dagger\circ\mc{M}:\msc{D}(\mc{H}_A)\to\msc{D}(\mc{H}_A)$ is a quantum channel, properties \ref{it-ol}, \ref{it-or}, and \ref{it-o2} are direct consequences of the properties of the diamond norm \cite{Wat15}. For property \ref{it-or}, the reference system $R$ has to be comparable with the channel input system $A$, following from the Schmidt decomposition. So $\mc{H}_R$ should be countably infinite if $\mc{H}_A$ is.  Property \ref{it-ou} follows from \cite[Theorem~2.1]{NS06},\cite[Theorem~3.4.1]{Riv11}.  

	The diamond norm has an operational interpretation in terms of channel discrimination \cite{Wbook17,Wat15} (see also \cite{Hel69,Hel76book} for state discrimination). Specifically, the optimal success probability $p_{\text{succ}}(\mc{M}_1,\mc{M}_2)$ of distinguishing between two channels $\mc{M}_1$ and $\mc{M}_2$ is 
	\begin{equation}\label{eq-channel_discrimination}
		p_{\text{succ}}(\mc{M}_1,\mc{M}_2)\coloneqq\frac{1}{2}\left(1+\frac{1}{2}\norm{\mc{M}_1-\mc{M}_2}_{\diamond}\right).
	\end{equation}
	It follows that the optimal success probability of distinguishing between the identity channel and $\mc{M}^\dagger\circ\mc{M}$ is
	\begin{align}
		p_{\text{succ}}(\id,\mc{M}^\dagger\circ\mc{M})&=\frac{1}{2}\left(1+\frac{1}{2}\norm{\id-\mc{M}^\dagger\circ\mc{M}}_{\diamond}\right)\\
		&=\frac{1}{2}\left(1+\frac{1}{2}\Vert\mc{M}\Vert_{\oslash}\right).
	\end{align}
	%\begin{definition}[$\varepsilon$-unitary]	A unital channel $\mc{N}$ is called $\varepsilon$-unitary if $\norm{\mc{N}}_{\oslash}\leq\varepsilon$. \end{definition}
	
	\begin{proposition}\label{prop-unitarity}
		Let $\mc{M}:\msc{D}(\mc{H})\to\msc{D}(\mc{H})$ be a unital channel. If there exists a unitary operator $U\in\msc{B}(\mc{H})$ such that
		\begin{equation}
			\norm{\mc{M}-\mc{U}}_{\diamond}\leq \delta,
		\end{equation}
		where $\mc{U}:\msc{D}(\mc{H})\to\msc{D}(\mc{H})$ is the unitary evolution \eqref{eq-unitary_channel} associated with $U$, then $\norm{\mc{M}}_{\oslash}\leq \sqrt{2\delta}+\delta$.
	\end{proposition}
	\begin{proof}
		The following relations hold
		\begin{align}
			\norm{\id-\mc{M}^\dagger\circ\mc{M}}_{\diamond}&=\norm{\id-\mc{M}^\dagger\circ\mc{U}+\mc{M}^\dagger\circ\mc{U}-\mc{M}^\dagger\circ\mc{M}}_{\diamond}\\
			&\leq\norm{\id-\mc{M}^\dagger\circ\mc{U}}_{\diamond}+\norm{\mc{M}^\dagger\circ(\mc{U}-\mc{M})}_{\diamond}\\
			&\leq\norm{\id-\mc{M}^\dagger\circ\mc{U}}_{\diamond}+\delta.\label{eq-oslash_bound}
		\end{align}
		To obtain these inequalities, we have used the following properties of the diamond norm \cite{Wat15}:
		\begin{enumerate}
			\item Triangle inequality: $\norm{\mc{M}_1+\mc{M}_2}_{\diamond}\leq\norm{\mc{M}_1}_{\diamond}+\norm{\mc{M}_2}_{\diamond}$.
			\item Sub-multiplicativity: $\norm{\mc{M}_1\circ\mc{M}_2}_\diamond\leq\norm{\mc{M}_1}_{\diamond}\norm{\mc{M}_2}_{\diamond}$.
			\item For all channels $\mc{M}$, $\norm{\mc{M}}_{\diamond}=1$.
		\end{enumerate}
		In particular, to use the third fact, we notice that $\mc{M}^\dagger$ is a channel since $\mc{M}$ is unital. We have also made use of an assumption that $\norm{\mc{U}-\mc{M}}_{\diamond}\leq\delta$.
		
		Now, from the assumption $\norm{\mc{U}-\mc{M}}_{\diamond}\leq\delta$, it follows by unitary invariance of the diamond norm that
		\begin{equation}
			\norm{\id-\mc{U}^\dagger\circ\mc{M}}_{\diamond}\leq\delta.
		\end{equation}
		By the operational interpretation of the diamond distance, this means that the success probability of distinguishing the channel $\mc{U}^\dagger\circ\mc{M}$ from the identity channel, using any scheme whatsoever, cannot exceed $p_{\text{succ}}(\id,\mc{U}^\dagger\circ\mc{M})$ as defined in \eqref{eq-channel_discrimination}. In other words, the success probability cannot exceed $\frac{1}{2}\left(1+\frac{1}{2}\delta\right)$. One such scheme is to send in a bipartite state $\ket{\psi}_{RA}$ on a reference system $R$ and the system $A$ on which the channel acts and perform the measurement defined by the positive operator-valued measure (POVM) $\{\op{\psi}_{RA},\mathbbm{1}_{RA}-\op{\psi}_{RA}\}$. If the outcome of the measurement is $\op{\psi}_{RA}$, then one guesses that the channel is the identity channel, and if the outcome of the measurement is $\mathbbm{1}_{RA}-\op{\psi}_{RA}$ then one guesses that the channel is $\mc{U}^\dagger\circ\mc{M}$. The success probability of this scheme is
		\begin{align}
			& \frac{1}{2}\left[  \Tr\{\op{\psi}_{RA}
\id_{RA}(\op{\psi}_{RA})\}+\Tr\{\left(\mathbbm{1}_{RA}-\op{\psi}_{RA}\right)\left[\id_{R}\otimes(\mathcal{U}^{\dag}\circ\mathcal{M})_{A}\right](\op{\psi}_{RA})\}\right]  \nonumber\\
		&	\qquad =\frac{1}{2}\left[  2-\bra{\psi}_{RA}\left[\id_{R}\otimes(\mathcal{U}^{\dag}\circ\mathcal{M})_{A}\right](\op{\psi}_{RA})\ket{\psi}_{RA}\right]  .
		\end{align}
		By employing the above, we get
		\begin{align}
			\frac{1}{2}\left[2-\bra{\psi}_{RA}\left[\id_R\otimes(\mc{U}^\dagger\circ\mc{M})_A\right](\op{\psi}_{RA})\ket{\psi}_{RA}\right]&\leq\frac{1}{2}\left(1+\frac{1}{2}\delta\right)\\
			%\Leftrightarrow 2-\bra{\psi}_{RA}\left[\id_R\otimes(\mc{U}^\dagger\circ\mc{M})_A\right](\op{\psi}_{RA})\ket{\psi}_{RA}&\leq 1+\frac{1}{2}\delta\\
			%\Leftrightarrow 1-\bra{\psi}\left[\id_R\otimes(\mc{U}^\dagger\circ\mc{M})_A\right](\op{\psi}_{RA})\ket{\psi}_{RA}&\leq\frac{1}{2}\delta\\
			\Leftrightarrow \bra{\psi}_{RA}\left[\id_R\otimes(\mc{U}^\dagger\circ\mc{M})_A\right](\op{\psi}_{RA})\ket{\psi}_{RA}&\geq 1-\frac{1}{2}\delta.
		\end{align}
		By employing the definition of the channel adjoint, we find that
		\begin{align}
			&\bra{\psi}_{RA}\left[\id_R\otimes(\mc{U}^\dagger\circ\mc{M})_A\right](\op{\psi}_{RA})\ket{\psi}_{RA} \nonumber \\
			&\qquad =\bra{\psi}_{RA}\left[\id_R\otimes (\mc{M}^\dagger\circ\mc{U})_A\right](\op{\psi}_{RA})\ket{\psi}_{RA}\geq 1-\frac{1}{2}\delta.
		\end{align}
		This holds for all input states, so we can conclude that the following inequality holds:
		\begin{equation}\label{eq-ent_fid}
			\min_{\psi_{RA}}\bra{\psi}_{RA}\left[\id_R\otimes(\mc{M}^\dagger\circ\mc{U})_A\right](\op{\psi}_{RA})\ket{\psi}_{RA}\geq 1-\frac{1}{2}\delta.
		\end{equation}
		Now, by the definition \eqref{eq-diamond_norm} of the diamond norm, and the fact that it suffices to take the maximization in the definition of the diamond norm over only pure states, we have
		\begin{equation}
			\norm{\id-\mc{M}^\dagger\circ\mc{U}}_{\diamond}=\max_{\psi_{RA}}\norm{\left[\id_R\otimes(\id-\mc{M}^\dagger\circ\mc{U})_A\right](\ket{\psi}\bra{\psi}_{RA})}_1.
		\end{equation}
		By the Fuchs-van de Graaf inequality \cite{FV99}, we obtain
		\begin{align}
			&\norm{\left[\id_R\otimes(\id-\mc{M}^\dagger\circ\mc{U})_A\right](\ket{\psi}\bra{\psi}_{RA})}_1\\
			&\qquad =\norm{\op{\psi}_{RA}-\left[\id_R\otimes(\mc{M}^\dagger\circ\mc{U})_A\right](\op{\psi}_{RA})}_1\\
			&\qquad \leq 2\sqrt{1-\bra{\psi}_{RA}\left[\id_R\otimes (\mc{M}^\dagger\circ\mc{U})_A\right](\op{\psi}_{RA})\ket{\psi}_{RA}}.
		\end{align}
		It follows that
		\begin{equation}
			\norm{\id-\mc{M}^\dagger\circ\mc{U}}_{\diamond}\leq 2\sqrt{1-\min_{\psi_{RA}}\bra{\psi}_{RA}\left[\id_R\otimes(\mc{M}^\dagger\circ\mc{U})_A\right](\op{\psi}_{RA})\ket{\psi}_{RA}}.
		\end{equation}
		Using \eqref{eq-ent_fid}, we therefore obtain
		\begin{equation}
			\norm{\id-\mc{M}^\dagger\circ\mc{U}}_{\diamond}\leq 2\sqrt{\frac{1}{2}\delta}=\sqrt{2\delta}.
		\end{equation}
		Finally, from \eqref{eq-oslash_bound} we arrive at
		\begin{equation}
			\norm{\id-\mc{M}^\dagger\circ\mc{M}}_{\diamond}\leq \sqrt{2\delta}+\delta,
		\end{equation}
		as required.
	\end{proof}
	
	\bigskip

	We can also make qualitative argument for the converse statement to suggest that if the diamond norm of non-unitarity $\norm{M}_{\oslash}$ of a unital channel $\mc{M}$ is less than $\delta$, then $\mc{M}$ is close to some unitary evolution (channel) for small $\delta$ \footnote{A concrete proof of the converse statement of Proposition~\ref{prop-unitarity} has been derived in an unpublished work with Sumeet Khatri, Mark M.~Wilde, and Elton Y.~Zhu.}.  Consider that $\norm{\id-\mathcal{M}^\dagger\circ\mathcal{M}}_{\diamond}\leq\delta$. Then, using tools of channel discrimination in the same way as in the proof of Proposition~\ref{prop-unitarity}, we obtain
\begin{equation}
	\forall~\ket{\psi}_{RA}:~\Tr\left[\op{\psi}_{RA}(\id_R\otimes\mathcal{M}^\dagger\circ\mathcal{M})(\op{\psi}_{RA})\right]\geq 1-\frac{1}{2}\delta,
\end{equation}
which implies that
\begin{equation}\label{eq:non-u-ineq}
	\Tr\left[(\id_R\otimes\mathcal{M})(\op{\psi}_{RA})^2\right]\geq 1-\frac{1}{2}\delta.
\end{equation}
We know that $\Tr\{\rho^2\}=1$ for a density operator $\rho$ if and only if it is a pure state, and any deviation of $\Tr\{\rho^2\}$ from unit shows how mixed the state is. Hence, the above inequality~\eqref{eq:non-u-ineq} implies that less noise is introduced in the system $A$ if the unital channel is close to some unitary process.   
\bigskip

	We now quantify the non-unitarity of the qudit depolarizing channel $\mc{D}_{d,q}$ defined as \cite{BSST99}
	\begin{equation}
		\mc{D}_{d,q}(\rho)= (1-q)\rho+q\frac{1}{d}\mathbbm{1} \quad\forall \rho\in\msc{D}(\mc{H}_A),
	\end{equation}
	where $\dim(\mc{H}_A)=d$ and $q\in\left[0,\frac{d^2}{d^2-1}\right]$. The input state $\rho$ remains invariant with probability $1-\(1-\frac{1}{d^2}\)q$ under the action of $\mc{D}_{d,q}$.  
	
	\begin{proposition}
		For the depolarizing channel $\mc{D}_{d,q}$, the diamond norm of non-unitarity is 
		\begin{equation}
			\norm{\mc{D}_{d,q}}_{\oslash}=2q(2-q)\left(1-\frac{1}{d^2}\right).
		\end{equation}
	\end{proposition}
	
	\begin{proof}
		The result follows directly from \cite[Section~V.A]{MGE12}, but here we provide an alternative proof argument that holds for more general classes of channels.
		
		The depolarizing channel is self-adjoint, that is, $\mc{D}_{d,q}^{\dagger}=\mc{D}_{d,q}$ for all $q$, which means that $\mc{D}_{d,q}^\dagger\circ\mc{D}_{d,q}=\mc{D}_{d,q}^2=\mc{D}_{d,2q-q^2}$. 
	Therefore,
	\begin{equation}\label{eq:norm-d-q}
		\norm{\mc{D}_{d,q}}_{\oslash}=\norm{\id-\mc{D}_{d,q}^2}_{\diamond}=\left|2q-q^2\right|\max_{\psi_{A'A}}\norm{\psi_{A'A}-\psi_{A'}\otimes\frac{\mathbbm{1}}{d}}_1,
	\end{equation}
	where $\psi_{A'A}=\op{\psi}_{A'A}$ is a pure state and $\dim(\mc{H}_{A'})=\dim(\mc{H}_A)=d$. 

	The identity channel and the depolarizing channel are jointly teleportation-simulable \cite[Definition~6]{DW17} with respect to resource states, which in this case are the respective Choi states (because these channels are also jointly covariant (Definition~\ref{def:cov-cell}, also see \cite[Definitions~7 \& 12]{DW17}). It is known that the trace distance is monotonically non-increasing under the action of a quantum channel. Therefore, we can conclude from the form \cite[Eq.~(3.2)]{DW17} of the action of jointly teleportation-simulable channels that the diamond norm between any two jointly teleportation-simulable channels is upper bounded by the trace distance between the associated resource states. 
	
	Since $\dim(\mc{H}_A)$ is finite, the maximally entangled state $\ket{\Phi}_{A'A}\coloneqq\frac{1}{\sqrt{d}}\sum_{i=1}^d\ket{i}\ket{i}$, where $\{\ket{i}\}_{i=1}^d\in\ONB(\mc{H}_A)$, is an optimal state in \eqref{eq:norm-d-q}. It is known that
	\begin{equation}
		\frac{\mathbbm{1}}{d}\otimes\frac{\mathbbm{1}}{d}=\frac{1}{d^2}\sum_{x=0}^{d^2-1}\sigma_A^x\Phi_{A'A}\sigma_A^x,
	\end{equation}
	where $\{\sigma_A^x\Phi_{A'A}\sigma_A^x\}_{x=0}^{d^2-1}\in\ONB(\mc{H}_{A'A})$ and $\{\sigma^x\}_{x=0}^{d^2-1}$ forms the Heisenberg-Weyl group (see Appendix~\ref{app:qudit}). We denote the identity element in $\{\sigma^x\}_{x=0}^{d^2-1}$ by $\sigma^0$. Using this, we get
	\begin{align}\label{eq-Dq}
		\norm{\mc{D}_{d,q}}_{\oslash}&=(2q-q^2)\norm{\Phi_{A'A}-\frac{\mathbbm{1}}{d}\otimes\frac{\mathbbm{1}}{d}}_1\\
		&=(2q-q^2)\norm{\(1-\frac{1}{d^2}\)\Phi_{A'A}-\frac{1}{d^2}\sum_{x=1}^{d^2-1}\sigma_A^x\Phi_{A'A}\sigma_A^x}_1\\
		&=(2q-q^2)\left[\(1-\frac{1}{d^2}\)+\frac{d^2-1}{d^2}\right]\\
		&=2q(2-q)\(1-\frac{1}{d^2}\).
	\end{align}
	Hence, we can conclude that $\norm{\mc{D}_{d,q}}_{\oslash}=2q(2-q)\(1-\frac{1}{d^2}\)$. See Fig.~\ref{fig:non_unitary} for a plot of $\norm{\mc{D}_{2,q}}_{\oslash}$ as a function of $q$.
	\end{proof}

	\section{Conclusion}\label{sec:conclusion-eg}
	In this chapter, we discussed the rate of entropy change of a system undergoing time evolution for arbitrary states and proved that the formula derived in \cite{Spo78} holds for both finite- and infinite-dimensional systems undergoing arbitrary dynamics with states of arbitrary rank. We derived a lower limit on the rate of entropy change for anarbitrary quantum Markov process. We discussed the implications of this lower limit in the context of bosonic Gaussian dynamics. From this lower limit, we also obtained several witnesses of non-Markovianity, which we used in two common examples of non-Markovian dynamics. Interestingly, our witness turned out to be useful in detecting non-Markovianity for given non-unital process, which could not be detected using BLP measure. We generalized the class of operations for which the entropy exhibits monotonic behavior. We also defined a measure of non-unitarity based on bounds on the entropy change, discussed its properties, and evaluated it for the depolarizing channel.

%% file: qr.tex
\chapter{Reading of Memory Devices: General Protocol and Bounds}\label{ch:read}
%\section{Introduction}
One of the primary goals of quantum information theory is to identify limitations on information processing when constrained by the laws of quantum mechanics\blfootnote{This chapter is entirely based on \cite{DW17}, a joint work with Mark M.~Wilde.}. In general, quantum  information theory uses tools that are universally applicable to the processing of arbitrary quantum systems, which include quantum optical systems, superconducting systems, trapped ions, etc.~\cite{NC00}. The abstract approach to quantum information allows us to explore how to use the principles of quantum mechanics  for communication or computation tasks, some of which would not be possible without quantum mechanics. 

In \cite{BRV00}, a communication protocol was introduced in which a classical message is encoded in a set of unitary operations, and later on, one can read out the information stored in the unitary operations by calling them. Over a decade after \cite{BRV00} was published, this communication model was generalized and studied under the name \textquotedblleft quantum reading\textquotedblright\ in \cite{Pir11}, and it was  applied to the setting of an optical read-only memory device. An optical read-only memory device is one of the prototypical examples of quantum reading, and for this reason, quantum reading had been mainly considered in the context of optical realizations like CD-ROMs and DVDs \cite{Pir11,PLG+11,GDN+11,GW12,WGTL12}. In this case, classical bits are encoded in the reflectivity and phase of memory cells, which can be modeled as a collection of pure-loss bosonic channels. More generally and abstractly, a memory cell is a collection of quantum channels, from which an encoder can select to form codewords for the encoding of a classical message. Each quantum channel in a codeword, representing one part of the stored information, is read only once. In subsequent works \cite{PLG+11,LP16}, the model was extended to a memory cell consisting of arbitrary quantum channels. In a quantum reading strategy, one exploits entangled states and collective measurements to help read out a classical message  stored in a read-only memory device. In many cases, one can achieve performance better than what can be achieved when using a classical strategy \cite{Pir11}. 

Some early developments in quantum reading \cite{Pir11} were based on a direct application of developments in quantum channel discrimination \cite{Kit97,DPP01,Aci01,WY06,TEGGLMPS08,DFY09,CMW14,HHLW10,DGLL16}. 
However, the past few years have seen some progress in quantum reading: there have been developments in defining protocols for quantum reading (including limited definitions of reading capacity and zero-error reading capacity), giving upper bounds on the rates for classical information readout, achievable rates for memory cells consisting of a particular class of bosonic channels, and details of a quantum measurement that can achieve non-trivial rates for memory cells consisting of a certain class of bosonic channels \cite{Pir11,PLG+11,GDN+11,GW12,WGTL12,LP16}. The information-theoretic study of quantum reading is based on considerations coming from quantum Shannon theory, and the most abstract and general way to define the encoding of a classical message in a quantum reading protocol is as mentioned above, a sequence of quantum channels chosen from a given memory cell.

Hitherto, all prior works on quantum reading considered decoding protocols of the following form: A reader possessing a transmitter system entangled with an idler system sends the transmitter system through the coded sequence of quantum channels. Finally, the reader decodes the message by performing a collective measurement on the joint state of the output system and the idler system.  

However, the above approach neglects an important consideration: \textit{in a quantum reading protocol, the transmitter and receiver are in the same physical location}. We can thus refer to both devices as a single device called a transceiver. As a consequence of this physical setup, the most  general  and natural definition for quantum reading capacity should allow for the transceiver to perform an adaptive operation after each call to the memory, and this is how  quantum reading capacity was defined in \cite{DW17}.

In general, an adaptive strategy can have a significant advantage over a non-adaptive strategy in the context of quantum channel discrimination \cite{HHLW10}. Furthermore, a quantum channel discrimination protocol employing a non-adaptive strategy is a special case of one that uses an adaptive strategy. Since quantum reading bears close connections to quantum channel discrimination, we suspect that adaptive operations could help to increase  quantum reading capacity in some cases, and this is one contribution of \cite{DW17}.

 It is to be noted that the physical setup of quantum reading is rather different from that considered in a typical communication problem, in which the sender and receiver are in different physical locations. In this latter case, allowing for adaptive operations represents a different physical model and  is thus considered as a different kind of capacity, typically called a feedback-assisted capacity. However, as advocated above, the physical setup of quantum reading necessitates that there should be no such distinction between capacities: the quantum reading capacity should be defined as it is here, in such a way as to allow for adaptive operations.

Another point of concern with prior work on quantum reading is as follows: so far, all bounds on the quantum reading rate have been derived in the usual setting of quantum Shannon theory, in which the number of uses of the channels tends to infinity (also called the i.i.d.~setting, where i.i.d.~stands for ``independent and identically distributed''). However, it is important for practical purposes to determine rates for quantum reading in the non-asymptotic scenario, i.e., for a finite number of quantum channel uses and a given error probability for decoding. The information-theoretic analysis in the non-asymptotic case is motivated by the fact that in practical scenarios, we have only finite resources at our disposal \cite{RennerThesis,DR09,Tbook15}.

The main focus of this chapter is to address some of the concerns mentioned above by giving the most general and natural definition for a quantum reading protocol and quantum reading capacity. We also establish bounds on the rates of quantum reading for wider classes of memory cells in both the asymptotic and non-asymptotic cases. First, we define a quantum reading protocol and quantum reading capacity in the most general setting possible by allowing for adaptive strategies. We give weak-converse, single-letter bounds on the rates of quantum reading protocols that employ either adaptive or non-adaptive strategies for arbitrary memory cells. We also introduce a particular class of memory cell, which we call an
\textit{environment-parametrized}  (see Section~\ref{sec:channel-symmetry-qr} for definitions), for which stronger statements can be made for the rates and capacities in the non-asymptotic situation of a finite number of uses of the channels. It should be noted that a particular kind of environment-parametrized memory cell consists of a collection of channels that are jointly teleportation simulable. Many channels of interest obey these symmetries: some examples are erasure,  dephasing, thermal, noisy amplifier, and Pauli channels \cite{BDSW96,DP05,JWD+08,Mul12,PLOB15,WTB16,TW16}. Here we determine strong converse and  second-order bounds on the quantum reading capacities of environment-parametrized memory cells. Based on an example from \cite[Section 3]{HHLW10}, we show in  Section~\ref{sec:zero-error} that there exists a memory cell for which its zero-error reading capacity with adaptive operations is at least $\tfrac12$, but its zero-error reading capacity without adaptive operations is equal to zero. This example emphasizes how reading capacity should be defined in such a way as to allow for adaptive operations, as stressed in this chapter.

The organization of this chapter is as follows. In Section~\ref{sec:channel-symmetry-qr}, we briefly review the set up of quantum reading protocol as a particular instance of communication protocol over bipartite quantum interaction. We then introduce two of the aforementioned classes of memory cells. In Section~\ref{sec:q-r-p}, we define a quantum reading protocol and quantum reading capacity in the most general and natural way. Section~\ref{sec:converse bounds} contains main results, which were briefly summarized in the previous paragraph. In Section~\ref{sec:example}, we calculate quantum reading capacities for a thermal memory cell and for a class of jointly covariant memory cells, including a qudit erasure memory cell and a qudit depolarizing memory cell. In Section~\ref{sec:zero-error}, we provide an example to illustrate the advantage of adaptive operations over non-adaptive operations in the context of zero-error quantum reading capacity.  In the final section of the chapter, we conclude and shed some light on possible future work.

\section{Memory cells with symmetry}\label{sec:channel-symmetry-qr}
In this section, we first review controlled channels as a particular instance of bipartite quantum interactions. This observation leads to the realization that a (quantum) reading protocol is a particular instance of information processing or communication task that uses bipartite quantum interactions of specific forms. Next, we define a broad class of memory cell called environment-parametrized memory cell that is a set of quantum channels obeying certain symmetries. 

Throughout this chapter, let $\msc{X}$ denote an alphabet of size $|X|$, where $|X|$ is finite.
\subsection{Bipartite interaction and quantum reading} 
Consider a bipartite quantum interaction between systems $X'$ and $B'$, generated by a Hamiltonian $\hat{H}_{X'B'E'}$, where $E'$ is a bath system, as given by \eqref{eq:c-ham-1}.

For some distributed quantum computing and information processing tasks where the controlling system (register) $X$ and input system $B'$ are jointly accessible, the following bidirectional channel is relevant:
\begin{equation}\label{eq:bi-ch-mc-1}
\mc{B}_{X'B'\to XB}(\cdot)\coloneqq \sum_{x\in\msc{X}}\op{x}_X\otimes\mc{N}^x_{B'\to B}\(\bra{x}(\cdot)\ket{x}_{X'}\).
\end{equation}
In the above, $X'$ is a controlling system that determines which evolution from the set $\{\mc{N}^x\}_{x\in\msc{X}}$ takes place on input system $B'$. 
In particular, when $X'$ and $B'$ are spatially separated and the input states for the system $X'B'$ are considered to be in product state, the noisy evolution for such constrained interactions is given by the following bidirectional channel:
\begin{equation}\label{eq:bi-ch-mc-2}
\mc{B}_{X'B'\to XB}(\sigma_{X'}\otimes\rho_{B'})\coloneqq \sum_{x\in\msc{X}}\bra{x}\sigma_{X'}\ket{x}_{X'}\op{x}_X\otimes\mc{N}^x_{B'\to B}(\rho_{B'}).
\end{equation}

This kind of bipartite interaction is in one-to-one correspondence with the notion of a memory cell from the context of quantum reading \cite{BRV00,Pir11}. There, a memory cell is a collection $\{\mc{N}^x_{B'\to B}\}_{x\in\msc{X}}$ of quantum channels. One party chooses which channel is applied to another party's input system $B'$ by selecting a classical letter $x\in\msc{X}$. Clearly, the description in
\eqref{eq:bi-ch-mc-1} is a fully quantum description of this process, and thus one notices that quantum reading can be understood as the use of a particular kind of bipartite interaction. 

\subsection{Environment-parametrized memory cells}
A collection of channels $\{\mc{E}^x_{B'\to B}\}_{x\in\msc{X}}$ is called environment-parametrized if  
there exists a set of ancillary states $\{\theta^x_E\}_{x\in\msc{X}}$, and
$\mc{F}_{B'E\to B}$ a quantum channel, called an interaction channel, such that $\mc{E}^x_{B'\to B}$ can be realized as follows:
\begin{equation}
\mc{E}^x_{B'\to B}(X_{B'}):=\mc{F}_{B'E\to B}(X_{B'}\otimes\theta^x_E),
\end{equation}
for all $X_{B'}\in\msc{B}_1(\mc{H}_{B'})$. This notion is related to the notion of programmable channels, used in the context of quantum computation \cite{DP05} (see Section~\ref{sec:symmetry}). 

\begin{remark}
We notice from Definition~\ref{def:tel-sim} that a teleportation-simulable channel is a particular kind of environment-parametrized channel in which a resource state
$\omega_{RB}$ is the ancillary state and LOCC channel
$\mc{L}_{R{B'}B\to B}$ is the interaction channel.  
\end{remark}

We now define a broad class of sets of quantum channels that we call environment-parametrized memory cells, and we discuss two classes of sets of quantum channels that are particular kinds of environment-parametrized memory cells. 

\begin{definition}[Environment-parametrized memory cell]\label{def:env-cell}
A set $\msc{E}_{\msc{X}}=\{\mc{N}_{{B'}\to B}^x\}_{x\in\msc{X}}$ of quantum channels is an  environment-parametrized memory cell if there exists a set $\{\theta^x_{E}\}_{x\in\msc{X}}$ of ancillary states and a fixed interaction channel
$\mc{F}_{{B'}E\to B}$
such that for all input states $\rho_{B'}$ and
$\forall x\in\msc{X}$
\begin{equation}\label{eq:env-cell}
\mc{N}^x_{{B'}\to B}(\rho_{B'}) =\mc{F}_{{B'}E\to B}(\rho_{B'}\otimes\theta^x_E).
\end{equation}
\end{definition}

\begin{definition}[Jointly teleportation-simulable memory cell]\label{def:tel-cell}
A set $\msc{T}_{\msc{X}}=\{\mc{T}^x_{{B'}\to B}\}_{x\in \msc{X}}$ of quantum channels is a jointly teleportation-simulable memory cell if there exists a set $\{\omega^x_{RB}\}_{x\in\msc{X}}$ of resource states and an LOCC channel $\mc{L}_{{B'} RB\to B}$ such that, for all input states $\rho_{{B'}}$ and
$\forall x\in\msc{X}$
\begin{equation}\label{eq:tel-cell}
 \mc{T}^x_{{B'}\to B}(\rho_{B'})=\mc{L}_{{B'}RB\to B}(\rho_{B'}\otimes\omega^x_{RB}),
\end{equation}  
 where the LOCC channel input is with respect to the bipartition $RB'\!:\!B$. 
\end{definition}

\begin{definition}[Jointly covariant memory cell]\label{def:j-cov-cell}
A set $\msc{T}_{\msc{X}}=\{\mc{T}^x_{{B'}\to B}\}_{x\in\msc{X}}$ of quantum channels is  jointly covariant if there exists a group $\msc{G}$ such that for all $x\in\msc{X}$, the channel $\mc{T}^x$ is a covariant channel with respect to the group $\msc{G}$ (Definition~\ref{def:covariant}).
\end{definition}

\begin{proposition}
Any jointly covariant memory cell $\msc{T}_{\msc{X}}=\{\mc{T}^x_{{B'}\to B}\}_{x\in\msc{X}}$ is jointly teleportation-simulable with respect to a set $\{\mc{T}^x_{{B'}\to B}(\Phi_{R{B'}})\}_{x\in\msc{X}}$ of resource states.
\end{proposition}
\begin{proof}
For a jointly covariant memory cell with respect to a group $\msc{G}$, all the channels $\mc{T}^x_{{B'}\to B}$ are jointly teleportation-simulable with respect to the resource states $\mc{T}^x_{{B'}\to B}(\Phi_{R{B'}})$, which are respective Choi states, by using a fixed POVM $\{E^g_{B'' R}\}_{g\in \msc{G}}$, similar to that defined in \cite[Equation (A.4), Appendix A]{WTB16}. See \cite[Appendix A]{WTB16} for an explicit proof.
\end{proof}

\begin{remark}\label{rem:env-tel}
Any jointly teleportation-simulable memory cell is environment-parametrized, an observation that is a direct consequence of definitions. This implies that all jointly covariant memory cells are also environment-parametrized. 
\end{remark}

\section{Quantum reading protocols and quantum reading capacity}
\label{sec:q-r-p}
In a quantum reading protocol, we consider an encoder and a reader (transceiver). An encoder is one who encodes a message onto a physical memory device that is delivered to Bob, a receiver, whose task it is to read the message. Bob is also referred to as the reader. The quantum reading task comprises the estimation of a message encoded in the form of a sequence of quantum channels chosen from a given set $\{\mc{N}^x_{{B'}\to B}\}_{x\in\msc{X}}$ of quantum channels called a memory cell, where $\msc{X}$ is a finite alphabet. In the most general setting considered here, the reader can use an adaptive strategy for quantum reading.

Both the encoder and the reader agree upon a memory cell $\msc{S}_{\msc{X}}=\{\mathcal{N}^x_{B'\to B}\}_{x\in\msc{X}}$ before executing the reading protocol. Consider a classical message set $\msc{M}=\{1,2,\ldots,|M|\}$, and let $M$ be an associated system denoting a classical register for the message. The encoder encodes a message $m\in\msc{M}$ using a sequence
$x^n(m)=(x_1(m),x_2(m),\ldots, x_n(m))$
of length $n$, where $x_i(m)\in\msc{X}$ for all $ i\in\{1,2,\ldots,n\}$. Each sequence identifies with a corresponding codeword formed from quantum channels chosen from the memory cell $\msc{S}_{\msc{X}}$:
\begin{equation}
\left(\mathcal{N}^{x_1(m)}_{B'_1\to B_1}, \mathcal{N}^{x_2(m)}_{B'_2\to B_2},\ldots,\mathcal{N}^{x_n(m)}_{B'_n\to B_n}\right).
\end{equation}
Each quantum channel in a codeword, each of which represents one part of the stored information, is only read once.

\begin{figure}[ptb]
%\begin{center}
\includegraphics[width=\linewidth]{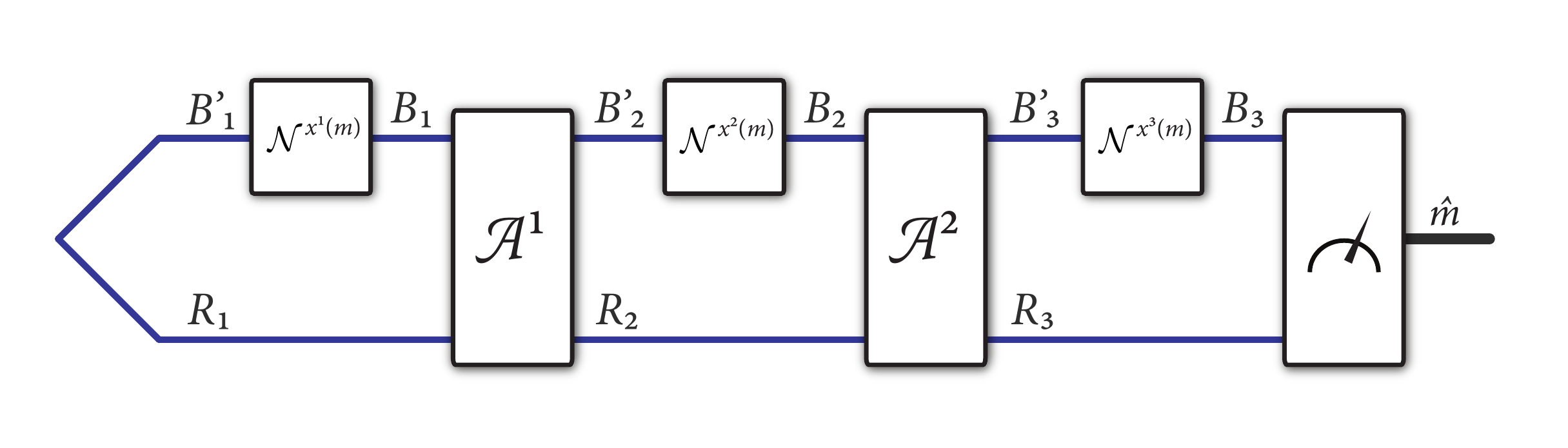}
%\end{center}
\caption{The figure depicts a quantum reading protocol that calls a memory cell three times to decode the message $m$ as $\hat{m}$. See the discussion in Section~\ref{sec:q-r-p} for a detailed description of a quantum reading protocol.}
\label{fig:reading-protocol}%
\end{figure}

An adaptive decoding strategy $\mathcal{J}_{\msc{S}_\msc{X}}$ makes $n$ calls to the memory cell $\msc{S}_{\msc{X}}$. It is specified in terms of a transmitter state $\rho_{R_1B'_1}$, a set of adaptive, interleaved channels $\{\mc{A}^i_{R_iB_i\to R_{i+1}B'_{i+1}}\}_{i=1}^{n-1}$, and a final quantum measurement $\{\Lambda^{\hat{m}}_{R_nB_n}\}_{\hat{m}\in\msc{M}}$ that outputs an estimate $\hat{m}$ of the message $m$. The strategy begins with Bob preparing the input state $\rho_{R_1B'_1}$ and sending the $B'_1$ system into the channel $\mc{N}^{x_1(m)}_{B'_1\to B_1}$. The channel outputs the system $B_1$, which is available to Bob. She adjoins the system $B_1$ to the system $R_1$ and applies the channel $\mc{A}^1_{R_1B_1\to R_2B'_2}$. The channel $\mc{A}^i_{R_iB_i\to R_{i+1}B'_{i+1}}$ is called adaptive because it can take an action conditioned on the information in the system $B_i$, which itself might contain partial information about the message $m$. Then, he sends the system $B'_2$ into the second use of the channel $\mc{N}^{x_2(m)}_{B'_2\to B_2}$, which outputs a system $B_2$. The process of successively using the channels interleaved by the adaptive channels  continues $n-2$ more times, which results in the final output systems $R_n$ and $B_n$ with Bob. Next, he performs a measurement $\{\Lambda^{\hat{m}}_{R_nB_n}\}_{\hat{m}\in\msc{M}}$ on the output state $\rho_{R_n B_n}$, and the measurement outputs an estimate $\hat{m}$ of the original message~$m$. See Figure~\ref{fig:reading-protocol} for a depiction of a quantum reading protocol.

It is apparent that a non-adaptive strategy is a special case of an adaptive strategy in which the reader does not perform any adaptive channels and instead uses $\rho_{RB^{'n}}$ as the transmitter state with each $B'_i$ system passing through the corresponding channel $\mc{N}^{x_i(m)}_{B'_i\to B_i}$ and $R$ being an idler system. The final step in such a non-adaptive strategy is to perform a decoding measurement on the  joint system $RB^n$. 

As we argued previously, it is natural to consider the use of an adaptive strategy for a quantum reading protocol because the channel input and output systems are in the same physical location. In a quantum reading protocol, the reader assumes the role of both the transmitter and receiver. 

\begin{definition}[Quantum reading protocol]\label{def:QR}
An $(n,R,\varepsilon)$ quantum reading protocol for a memory cell $\msc{S}_{\msc{X}}$ is defined by an encoding map $\mc{E}_{\tn{enc}}:\msc{M}\to \msc{X}^{\times n}$ and an adaptive strategy $\mathcal{J}_{\msc{S}_\msc{X}}$ with measurement $\{\Lambda_{R_n B_n}^{\hat{m}}\}_{\hat{m}\in\msc{M}}$. The protocol is such that the average success probability is at least $1-\varepsilon$, where $\varepsilon\in(0,1)$:
\begin{multline}
1- \varepsilon \leq 1 - p_{\operatorname{err}} \coloneqq\\
\frac{1}{|M|} 
\sum_{m}\Tr\left\{\Lambda^{(m)}_{R_nB_n}\(\mc{N}^{x_n(m)}_{B'_n\to B_n}\circ\mc{A}^{{n-1}}_{R_{n-1}B_{n-1}\to R_nB'_n}\circ\cdots\circ\mc{A}^{1}_{R_1B_1\to R_2B'_2}\circ\mc{N}^{x_1(m)}_{B'_1\to B_1}\)\(\rho_{R_1{B'}_1}\)\right\}.
\end{multline} 
The rate $R$ of a given $(n,R,\varepsilon)$ quantum reading protocol is equal to the number of bits read per channel use:
\begin{equation}
R\coloneqq \frac{1}{n}\log_2|M|.
\end{equation}
\end{definition}

To arrive at a definition of quantum reading capacity, we demand that there exists a sequence of reading protocols, indexed by $n$, for which the error probability $p_{\operatorname{err}}\to 0$ as $n\to \infty$ at a fixed rate~$R$.

\begin{definition}[Achievable rate]
A rate $R$ is called achievable if $\forall \varepsilon\in (0,1)$, $\delta>0$, and sufficiently large $n$, there exists an $(n,R-\delta,\varepsilon)$ code. 
\end{definition}

\begin{definition}[Quantum reading capacity]\label{def:capacity}
The quantum reading capacity $\mc{C}(\msc{S}_{\msc{X}})$ of a memory cell $\msc{S}_{\msc{X}}$ is defined as the supremum of all achievable rates $R$. 
\end{definition}

\section{Fundamental limits on quantum reading capacities}\label{sec:converse bounds}
In this section, we establish second-order  and strong converse bounds for any environment-parametrized memory cell. We also establish general weak converse (upper) bounds on various reading capacities. 

\subsection{Converse bounds for environment-parametrized memory cells}
In this section, we derive upper bounds on the performance of quantum reading of environment-parametrized memory cells. 

\begin{figure}
\begin{center}
\includegraphics[
width=\linewidth
]{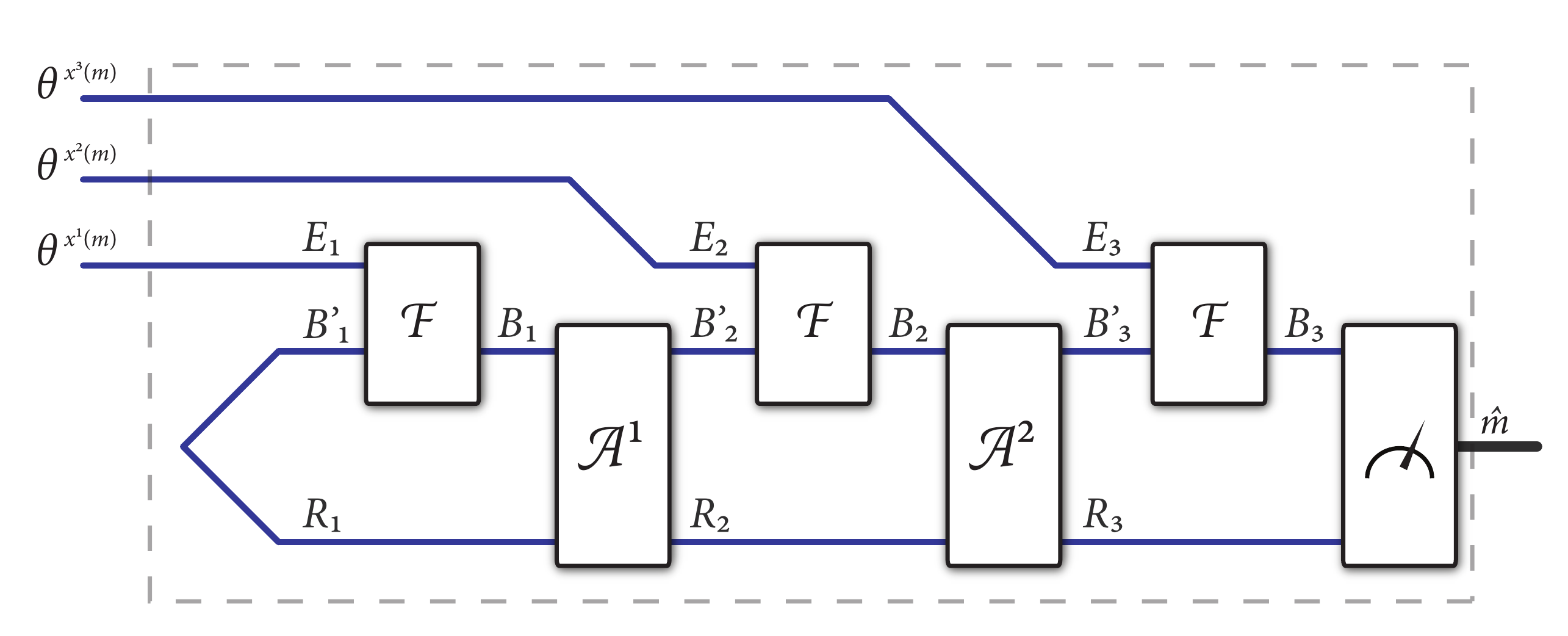}
\end{center}
\caption{The figure depicts how a quantum reading protocol of an environment-parametrized memory cell can be rewritten as a protocol that tries to decode the message $m$ from the ancillary states $\theta^{x^n(m)}_{E^n}$. All of the operations inside the dashed lines can be understood as a measurement on the states $\theta^{x^n(m)}_{E^n}$.}
\label{fig:env-param-reading-protocol}
\end{figure}

To begin with, let us consider an $(n,R,\varepsilon)$ quantum reading protocol of an environment-parametrized memory cell $\msc{E}_\msc{X}=\{\mc{N}^x\}_{x\in\msc{X}}$ (Definition~\ref{def:env-cell}). The structure of reading protocols involving adaptive channels simplifies immensely for  memory cells that are teleportation-simulable and more generally environment-parametrized. This is a direct consequence of the symmetry obeyed by the channels in the cell. For such memory cells, a quantum reading protocol can be simulated by one in which every channel use is replaced by the encoder preparing the ancillary state $\theta^{x_i(m)}_{E}$ from \eqref{eq:env-cell} and then interacting the channel input with the interaction channel $\mc{F}_{B'E\to B}$. Critically, each interaction channel
$\mc{F}_{B'E\to B}$ is independent of the message $m\in\msc{M}$. Let 
\begin{equation}
\theta^{x^n(m)}_{E^n}:= \bigotimes_{i=1}^n\theta^{x_i(m)}_{E}
\end{equation}
denote the ancillary state needed for the simulation of all $n$ of the channel uses in the quantum reading protocol. This leads to the translation of a general quantum reading protocol to one in which all of the rounds of adaptive channels can be delayed until the very end of the protocol, such that the resulting protocol is a non-adaptive quantum reading protocol.

The following proposition, holding for any environment-parametrized memory cell, is a direct consequence of observations made in \cite[Section V]{BDSW96}, \cite{HHH99}, \cite[Theorem 14 \& Remark 11]{Mul12}, and \cite{DM14}. We thus omit a detailed proof, but Figure~\ref{fig:env-param-reading-protocol} clarifies the main idea: any quantum reading protocol of an environment-parametrized memory cell can be rewritten as in Figure~\ref{fig:env-param-reading-protocol}. Inspecting the figure, one can notice that the protocol can be understood as a non-adaptive decoding of the ancillary states  $\theta^{x^n(m)}_{E^n}$,
with the decoding measurement constrained to contain the interaction channel $\mc{F}_{B'E\to B}$ interleaved between arbitrary adaptive channels. Thus, Proposition~\ref{thm:sim-red} establishes that an adaptive strategy used for decoding an environment-parametrized memory cell can be reduced to a particular non-adaptive decoding of the ancillary states $\theta^{x^n(m)}_{E^n}$. 

\begin{proposition}[Adaptive-to-non-adaptive reduction]\label{thm:sim-red}
Let $\msc{E}_\msc{X}=\{\mc{N}^x_{B'\to B}\}_{x\in\msc{X}}$ be an environment-parametrized memory cell with an associated set of ancillary states $\{\theta^x_{E}\}_{x\in\mc{X}}$ and a fixed interaction channel $\mc{F}_{B'E\to B}$, as given in Definition~\ref{def:env-cell}. Then any quantum reading protocol as stated in Definition~\ref{def:QR}, which uses an adaptive strategy $\mathcal{J}_{\msc{E}_\msc{X}}$, can be simulated as a non-adaptive quantum reading protocol, in the following sense:
\begin{multline}
\Tr\left\{\Lambda^{\hat{m}}_{E_nB_n}\(\mc{N}^{x_n(m)}_{B'_n\to B_n}\circ\mc{A}^{{n-1}}_{E_{n-1}B_{n-1}\to E_nB'_n}\circ\cdots\circ\mc{A}^{1}_{E_1B_1\to E_2B'_2}\circ\mc{N}^{x_1(m)}_{B'_1\to B_1}\)\(\rho_{E_1B'_1}\)\right\}\\
=\Tr\left\{\Gamma^{\hat{m}}_{E^n}\(\bigotimes_{i=1}^n\theta^{x_i(m)}_{E}\)\right\}, \label{eq:tel-sim-red}
\end{multline}
for some POVM $\{\Gamma^{\hat{m}}_{E^n}\}_{\hat{m}\in \msc{M}}$ that depends on $\mathcal{J}_{\msc{E}_\msc{X}}$.
\end{proposition}

Using the observation in Proposition~\ref{thm:sim-red}, we now show how to arrive at upper bounds on the performance of any reading protocol that uses an environment-parametrized memory cell.

Our proof strategy is to employ a generalized divergence to make a comparison between the states involved in the actual reading protocol and one in which the memory cell is fixed as $\hat{\msc{E}}\coloneqq\{\mc{P}_{B'\to B}\}$, containing only a single channel with environment state $\hat{\theta}_{E}$ and interaction channel $\mc{F}_{B'E\to B}$. The latter reading protocol contains no information about the message $m$.
Observe that the augmented memory cell $\{\msc{E}_\msc{X},\hat{\msc{E}}\}$ is environment-parametrized.

One of the main steps that we use in our proof is as follows. Consider the following states:\begin{align}
\sigma_{M\hat{M}} & =\sum_{m\in\msc{M},\hat{m}\in\msc{M}}\frac{1}{|M|}\op{m}_M\otimes p_{\hat{M}|M}\(\hat{m}|m\)\op{\hat{m}}_{\hat{M}},\label{eq:out-useful}\\
\tau_{M\hat{M}} & =\sum_{m\in\msc{M}}
\frac{1}{|M|}
\op{m}_M\otimes \hat{\tau}_{\hat{M}},\label{eq:out-useless}
\end{align}
where $p_{\hat{M}|M}(\hat{m}|m)$ is a distribution that results after the final decoding step of an $(n,R,\varepsilon)$ quantum reading protocol, while  $\hat{\tau}_{\hat{M}}$
is a fixed state. By applying
the comparator test $\{\Pi_{M\hat{M}},\bm{1}_{M\hat{M}}-\Pi_{M\hat{M}}\}$, defined by
\begin{equation}
\Pi_{M\hat{M}} := \sum_{m}\op{m}_M\otimes\op{m}_{\hat{M}},\label{eq:comparator-test}
\end{equation}
and using definitions, we arrive at the following inequalities that hold for an arbitrary
$(n,R,\varepsilon)$ quantum reading protocol:
\begin{equation}
\Tr\{\Pi_{M\hat{M}}\sigma_{M\hat{M}}\}\geq 1-\varepsilon,\quad \Tr\{\Pi_{M\hat{M}}\tau_{M\hat{M}}\}= \frac{1}{|M|}.
\end{equation}
Then by applying the definition of the $\varepsilon$-hypothesis-testing divergence, we arrive at the following bound, which is a critical first step to establish second-order and strong converse bounds:
\begin{equation}\label{eq:eps-div-M}
D^\varepsilon_h\!\(\sigma_{M\hat{M}}\Vert\tau_{M\hat{M}}\)\geq \log_2|M|.
\end{equation}

In the converse proof that follows, the main idea for arriving at an upper bound on performance is to make a comparison between the case in which the message $m$ is encoded in a sequence of quantum channels and the case in which it is not.

\subsubsection{Second-order asymptotics and strong converse}
In this section, we derive second-order asymptotics and strong converse bounds for environment-parametrized memory cells. We begin by deriving a relation between the quantum reading rate and the hypothesis testing divergence.  

\begin{lemma}\label{thm:hyp-test-bound}
The following bound holds for an $(n,R,\varepsilon)$ reading protocol that uses an environment-parametrized memory cell (Definition~\ref{def:env-cell}):
\begin{multline}
\log_2 |M| = nR \leq \\
\sup_{p_{X^n}}\inf_{\hat{\theta}} D_h^{\varepsilon}\!\(\sum_{x^n\in\msc{X}^n}p_{X^n}(x^n)\op{x^n}_{X^n}\otimes \theta^{x^n}_{E^n}\left\Vert \sum_{x^n\in\msc{X}^n}p_{X^n}(x^n)\op{x^n}_{X^n}\otimes\hat{\theta}_{E}^{\otimes n}\)\right. .
\end{multline}
\end{lemma}

\begin{proof}
Proof begins by applying the observation from Proposition~\ref{thm:sim-red}, which allows reduction of any adaptive protocol to a non-adaptive one. If the encoder chooses the message $m$ uniformly at random and places it in a system $M$, the output state in \eqref{eq:out-useful} after Bob's decoding measurement in the actual protocol is
\begin{equation}
\sigma_{M\hat{M}} = \sum_{m,\hat{m}}\frac{1}{|M|}\op{m}_M\otimes \Tr\!\left\{\Gamma_{E^n}^{\hat{m}}\theta^{x^n(m)}_{E^n}\right\}\op{\hat{m}}_{\hat{M}},
\end{equation}
where
\begin{equation}
\theta^{x^n(m)}_{E^n} \equiv 
\bigotimes_{i=1}^n\theta^{x_i(m)}_{E}.
\end{equation}
The success probability $p_{\operatorname{succ}}\coloneqq 1-p_{\operatorname{err}}$ is defined as
\begin{equation}
p_{\operatorname{succ}}\coloneqq \frac{1}{|M|}\sum_{m\in\msc{M}}\Tr\!\left\{\Gamma_{E^n}^{m}\theta^{x^n(m)}_{E^n}\right\}.
\end{equation} 
The output state in \eqref{eq:out-useless} after Bob's decoding measurement in a reading protocol that uses the memory cell $\hat{\msc{E}}$ is 
\begin{equation}
\tau_{M\hat{M}} = \sum_m\frac{1}{|M|}\op{m}_M\otimes\sum_{\hat{m}}\Tr\!\left\{\Gamma_{E^n}^{\hat{m}}\hat{\theta}^{\otimes n}_{E}\right\}\op{\hat{m}}_{\hat{M}}.
\end{equation}
Then a generalized divergence can be bounded as follows:
\begin{align}
& \mathbf{D}\!\(
\{p_{\operatorname{succ}}, 
1-p_{\operatorname{succ}}\} \middle \Vert 
\left\{ 1/|M|, 1-1/|M|\right\}
\) \nonumber\\
&\leq \mathbf{D}\!\(\sigma_{M\hat{M}} \Vert \tau_{M\hat{M}}\) \nonumber\\
& \leq  \mathbf{D}\!\(\sum_{m}\frac{1}{M}\op{m}_M\otimes \theta^{x^n(m)}_{E^n} \left\Vert \sum_m\frac{1}{M}\op{m}_M\otimes\hat{\theta}^{\otimes n}_{E}\)\right.
\end{align}
The first inequality follows from applying the comparator test in \eqref{eq:comparator-test} to
$\sigma_{M\hat{M}}$ and $\tau_{M\hat{M}}$.
The second inequality follows from the data-processing inequality in \eqref{eq:gen-div-mono} as the final measurement is a quantum channel.
Since the above bound holds for all $\hat{\theta}_{E}$, it can be concluded that
\begin{multline}
\mathbf{D}\!\(
\{p_{\operatorname{succ}}, 
1-p_{\operatorname{succ}}\} \middle \Vert 
\left\{ 1/|M|, 1-1/|M|\right\}
\) \leq \\
\inf_{\hat{\theta}}\mathbf{D}\!\(\sum_{m}\frac{1}{|M|}\op{m}_M\otimes \theta^{x^n(m)}_{E^n} \left\Vert \sum_m\frac{1}{|M|}\op{m}_M\otimes\hat{\theta}^{\otimes n}_{E}\)\right.
\end{multline}
Now optimizing over all input distributions, we arrive at the following general bound:
\begin{multline}
\mathbf{D}\!\(
\{p_{\operatorname{succ}}, 
1-p_{\operatorname{succ}}\} \middle \Vert 
\left\{ 1/|M|, 1-1/|M|\right\}
\)\leq \\
\sup_{p_{X^n}}\inf_{\hat{\theta}} \mathbf{D}\!\(\sum_{x^n\in\msc{X}^n}p_{X^n}(x^n)\op{x^n}_{X^n}\otimes \theta^{x^n}_{E^n}\left\Vert \sum_{x^n\in\msc{X}^n}p_{X^n}(x^n)\op{x^n}_{X^n}\otimes\hat{\theta}_{E}^{\otimes n}\)\right.\label{eq:gen-div-bound-JP-cells} ,
\end{multline}
where $x^n:= x_1x_2\cdots x_n$ and $\theta_{E^n}^{x^n}=\bigotimes_{i=1}^n\theta_{E}^{x_i}$.
Observe that the lower bound contains the relevant performance parameters such as success probability and number of messages, while the upper bound is an information quantity, depending exclusively on the memory cell $\msc{E}_\msc{X}$. 

Substituting the hypothesis testing divergence in the above and applying \eqref{eq:eps-div-M}, we obtain the following bound for an $(n,R,\varepsilon)$ reading protocol that uses an environment-parametrized memory cell:
\begin{multline}
\log_2 |M| = nR \leq \\
\sup_{p_{X^n}}\inf_{\hat{\theta}} D_h^{\varepsilon}\!\(\sum_{x^n\in\msc{X}^n}p_{X^n}(x^n)\op{x^n}_{X^n}\otimes \theta^{x^n}_{E^n}\left\Vert \sum_{x^n\in\msc{X}^n}p_{X^n}(x^n)\op{x^n}_{X^n}\otimes\hat{\theta}_{E}^{\otimes n}\)\right.
\end{multline}
This concludes our proof.
\end{proof}

A direct consequence of Lemma~\ref{thm:hyp-test-bound} and \cite[Theorem 4]{TT15} is the following proposition:
\begin{proposition}\label{thm:env-cell-sec-order}
For an $(n,R,\varepsilon)$ quantum reading protocol for an environment-parametrized memory cell $\msc{E}_{\msc{X}}=\{\mc{N}^x\}_{x\in\msc{X}}$ (Definition~\ref{def:env-cell}), the following inequality holds 
\begin{equation}
R\leq  \max_{p_X}I(X;E)_\theta +\sqrt{\frac{V_\varepsilon(\mc{E}_{\msc{X}})}{n}}\Phi^{-1}(\varepsilon)+O\!\(\frac{\log n}{n}\),
\end{equation}
where 
\begin{equation}\label{eq:cq-theta}
\theta_{XE}=\sum_{x\in\msc{X}}p_X(x)\op{x}_X\otimes\theta^x_E,
\end{equation}
$\Phi^{-1}(\varepsilon)$ is the inverse of the cumulative distribution function\footnote{The cumulative distribution function corresponding to the standard normal random variable is defined as
\begin{equation}
\Phi(a)\coloneqq \int_{-\infty}^a\frac{1}{\sqrt{2\pi}}\exp\!\(-\frac{1}{2}x^2\)\ dx.
\end{equation}
Its inverse is also useful for us and is defined as $\Phi^{-1}(a):=\sup\left\{a\in\mathbb{R}|\Phi(a)\leq \varepsilon\right\}$, which reduces to the usual inverse for $\varepsilon\in(0,1)$.}, and
\begin{equation}
V_\varepsilon(\mc{E}_{\msc{X}})= \left\{ 
\begin{tabular}{c c}
$\min_{p_X\in P({\msc{E}_{\msc{X}}})} V(\theta_{XE}\V\theta_X\otimes\theta_{E} )$,\ \ \ \ \  & $\varepsilon \in  (0,1/2]$\\
$\max_{p_X\in P({\msc{E}_{\msc{X}}})} V(\theta_{XE}\V\theta_X\otimes\theta_{E} )$,\ \ \ \ \  & $\varepsilon \in (1/2,1)$
\end{tabular}
\right\},
\end{equation}
where $V(\rho\Vert\sigma)$ denotes the variance between $\rho,\sigma\in\msc{D}(\mc{H})$ and $P({\mc{E}_{\msc{X}}})$ denotes a set $\{p_X\}$ of probability distributions that achieve the maximum in 
$\max_{p_X}I(X;E)_\theta$.
\end{proposition}

\begin{proposition}\label{thm:env-cell-strong}
The success probability $p_{\operatorname{succ}}$ of any $(n,R,\varepsilon)$ quantum reading protocol for an environment-parametrized memory cell $\msc{E}_\msc{X}$ (Definition~\ref{def:env-cell}) is bounded
from above as 
\begin{equation}
p_{\operatorname{succ}}\leq 2^{-n\sup_{\alpha>1}\(1-\frac{1}{\alpha}\) \(R-\wt{I}_\alpha(\msc{E}_\msc{X})\)},
\end{equation}
where
\begin{equation}
\wt{I}_\alpha(\msc{E}_\msc{X})=\max_{p_X}\wt{I}_\alpha(X;E)_\theta,
\end{equation}
for $\theta_{XE}$ as defined in \eqref{eq:cq-theta}.
\end{proposition}
\begin{proof}
A proof follows by combining the bound in \eqref{eq:gen-div-bound-JP-cells} with the main result of \cite{WWY14} (see also
\cite{DW15}
for arguments about extending the range of $\alpha$ from $(1,2]$ to $(1,\infty)$). 
\end{proof}

\begin{theorem}\label{thm:env-cell-qrc}
The quantum reading capacity of any environment-parametrized memory cell $\msc{E}_{\msc{X}}=\{\mc{N}^x_{B'\to B}\}_{x\in\msc{X}}$ (Definition~\ref{def:env-cell}) is bounded from above as
\begin{equation}
\mc{C}(\msc{E}_{\msc{X}})\leq\max_{p_X}I(X;E)_\theta,
\end{equation} 
where  $\theta_{XE}$ is defined in \eqref{eq:cq-theta}.
\end{theorem}
\begin{proof}
The statement follows from Proposition~\ref{thm:env-cell-sec-order}, by taking the limit $n\to \infty$. Alternatively, 
the statement can also be concluded from Definition~\ref{def:capacity} and Proposition~\ref{thm:env-cell-strong}, by taking the limit $\alpha \to 1$.
\end{proof}

\bigskip
Direct consequences of the above theorems and Remark~\ref{rem:env-tel} are the following corollaries: 
\begin{corollary}\label{thm:rate-2nd}
For any $(n,R,\varepsilon)$ quantum reading protocol and jointly teleportation-simulable memory cell $\msc{T}_{\msc{X}}$ (Definition~\ref{def:tel-sim}) with
associated  resource states $\{\omega^x_{RB}\}_{x\in\msc{X}}$, the reading rate $R$ is bounded from above as
\begin{equation}
R\leq\max_{p_X}I(X;RB)_{\omega}+\sqrt{\frac{V_\varepsilon(\msc{T}_{\msc{X}})}{n}}\Phi^{-1}(\varepsilon)+\mc{O}\(\frac{\log n}{n}\),
\end{equation}
where
\begin{equation}
\omega_{XRB}\coloneqq \sum_{x\in\msc{X}}p_X(x)\op{m}_X\otimes\omega^x_{RB}\label{eq:cq-resource-state}
\end{equation}
and
\begin{equation}
V_\varepsilon(\msc{T}_{\msc{X}})\coloneqq \left\{ 
\begin{tabular}{c c}
$\min_{p_X\in P(\msc{T}_{\msc{X}})} V\(\omega_{XRB}\V\omega_X\otimes\omega_{RB} \)$,\ \ \ \ \  & $\varepsilon \in  (0,1/2]$\\
$\max_{p_X\in P(\msc{T}_{\msc{X}})} V\(\omega_{XRB}\V\omega_X\otimes\omega_{RB} \)$,\ \ \ \ \  & $\varepsilon \in (1/2,1)$
\end{tabular}
\right\}.
\end{equation}
In the above, $P(\msc{T}_\msc{X})$ denotes a set $\{p_X\}$ of probability distributions that are optimal for $\max_{p_X}I(X;RB)_{\omega}$.
\end{corollary}

\begin{corollary}
The quantum reading capacity of any jointly teleportation-simulable memory cell $\msc{T}_{\msc{X}}=\{\mc{T}^x_{B'\to B}\}_{x\in\msc{X}}$ associated with a set $\{\omega^x_{RB}\}$ of resource states is bounded from above as
\begin{equation}
\mc{C}(\msc{T}_{\msc{X}})\leq\max_{p_X}I(X;RB)_\omega,
\end{equation} 
where
\begin{equation}\label{eq:omega-XRB}
\omega_{XRB}=\sum_{x\in\mc{X}}p_X(x)\op{x}_X\otimes\omega^x_{RB}.
\end{equation}
\end{corollary}

The capacity bounds given above are tight for a wide variety of channels, as clarified in the following remark:

\begin{remark}\label{rem:cov-cell-achievability}
The quantum reading capacity is achieved for a jointly teleportation-simulable memory cell 
$\msc{T}_{\msc{X}}=\{\mc{T}^x_{B'\to B}\}_{x\in\msc{X}}$
when, for all $x\in\msc{X}$, $\omega^x_{RB}$ is equal to the Choi state of the channel $\mc{T}^x_{B'\to B}$. More finely, the upper bound in Corollary~\ref{thm:rate-2nd} is achieved in such a case by invoking \cite[Theorem 4]{TT15}.
\end{remark}

\subsection{Weak converse bound for a non-adaptive reading protocol}

In this section, we establish a general weak converse when the strategy employed is non-adaptive.
Consider a state $\rho_{MRB'^{n}}$ of the form
\begin{equation}
\rho_{MRB'^n}=\frac{1}{|M|}\sum_{m_{\msc{M}}}\op{m}_M\otimes\rho_{RB'^n}.
\end{equation}
Suppose that $\rho_{RB'^n}$ is purified by the pure state
$\psi_{RSB'^n}$.
Bob passes the transmitter state $\rho_{RB'^n}$ through a codeword sequence $\mathcal{N}^{x^n(m)}_{B'^n\to B^n}\coloneqq \bigotimes_{i=1}^n\mc{N}^{x_i(m)}_{B'_i\to B_i}$, where the choice $m$ depends on the classical value~$m\in\msc{M}$ in the register $M$. Let $\mathcal{U}^{\mc{N}^{x^n(m)}}_{B'^n\to B^n E^n}\coloneqq \bigotimes_{i=1}^n\mathcal{U}^{\mc{N}^{x_i(m)}}_{B'_i\to B_i E_i} $, where $\mathcal{U}^{\mc{N}^{x_i(m)}}_{B'_i\to B_i E_i}$ denotes an isometric quantum channel extending ${\mc{N}}^{x_i(m)}_{B'_i\to B_i}$, for all $i\in [n]$. After the isometric channel acts, the overall state is as follows:
\begin{equation}
\sigma _{MRSB^nE^n}=
\frac{1}{|M|}\sum_{m}\op{m}_M\otimes
\mathcal{U}^{\mc{N}^{x^n(m)}}_{B'^n\to B^n E^n}\left(\psi_{RSB'^n}\right).
\end{equation}
Let $\sigma^\prime_{M\hat{M}}=\mc{D}_{RB^n\to \hat{M}}\(\sigma_{MRB^n}\)$ be the output state at the end of protocol after the decoding measurement $\mc{D}$ is performed by Bob. 
Let $\overline{\Phi}_{M\hat{M}}$ denote the maximally classically correlated state:
\begin{equation}\label{eq:cor-state}
\overline{\Phi}_{M\hat{M}}:=\frac{1}{|M|}\sum_{m\in\msc{M}}\op{m}_M\otimes\op{m}_{\hat{M}}.
\end{equation}

\begin{proposition}
\label{thm:weak-converse-non-adaptive}
The non-adaptive reading capacity of any quantum memory cell $\msc{S}_{\msc{X}}=\{\mc{N}^x\}_{\mc{X}}$ is upper bounded as
\begin{equation}
\mc{C}_{\textnormal{non-adaptive}}(\msc{S}_{\msc{X}})\leq \sup_{p_X,\phi_{RB'}}I(XR;B)_\tau,
\end{equation}
where
\begin{align}
\tau_{XRB}&=\sum_{x}p_X(x)\op{x}_X\otimes\mc{N}^x_{B'\to B}(\phi_{R B'}),
\end{align}
and it suffices for $\phi_{R B'}$ to be a pure state such that $\dim(\mc{H}_R) = \dim(\mc{H}_B')$. 
\end{proposition}
\begin{proof}
For any $(n,R,\varepsilon)$ quantum reading protocol using a non-adaptive strategy, one has
\begin{equation}
\frac{1}{2}\left\Vert \overline{\Phi}_{M\hat{M}}-\sigma^\prime_{M\hat{M}}\right\Vert_1\leq \varepsilon. 
\end{equation}
Then consider the following chain of inequalities:
\begin{align}
\log_2|M|&=I(M;\hat{M})_{\overline{\Phi}}\\
& \leq I(M;\hat{M})_{\sigma^\prime}+f(n,\varepsilon)\\
& \leq I(M;RSB^n)_\sigma + f(n,\varepsilon)\\
& = I(M;RS)_\sigma +I(M;B^n|RS)_\sigma + f(n,\varepsilon)\\
&=I(M;B^n|RS)_\sigma + f(n,\varepsilon)\\
&=S(B^n|RS)_\sigma-S(B^n|RSM)_\sigma + f(n,\varepsilon)\\
&=S(B^n|RS)_\sigma+S(B^n|E^nM)_\sigma + f(n,\varepsilon) \label{eq:converse-1st-block-last-step}
\end{align}
The first inequality follows from the uniform continuity of conditional entropy \cite{AF04,Win16}, where $f(n,\varepsilon)$ is a function of $n$ and the error probability $\varepsilon$ such that $\lim_{\varepsilon\to 0}\lim_{n\to \infty}\frac{f(n,\varepsilon)}{n}=0$. The second inequality follows from data processing. The second equality follows from the chain rule for the mutual information. The third equality follows because the reduced state of systems $M$ and $RS$ is a product state. The fifth equality follows from the duality of the conditional entropy. 
Continuing, it follows that
\begin{align}
\eqref{eq:converse-1st-block-last-step} &\leq \sum_{i=1}^n \left[S(B_i|RS)_\sigma +S(B_i|E_iM)_\sigma\right]+ f(n,\varepsilon)\\
&= \sum_{i=1}^n\left[ S(B_i|RS)_\sigma -S(B_i|RSB'_{[n]\setminus \{i\}}M)_\sigma \right]+ f(n,\varepsilon)\\
&=\sum_{i=1}^nI(MB'_{[n]\setminus \{i\}};B_i|RS)_\sigma + f(n,\varepsilon)\\
&\leq \sum_{i=1}^nI(MB'_{[n]\setminus \{i\}}RS;B_i)_\sigma+ f(n,\varepsilon)\\
&= nI(MR^\prime;B|Q)_\sigma+ f(n,\varepsilon)\\
&\leq n\sup_{p_X,\phi_{\wt{R} B'}}I(X\wt{R};B)_\tau+ f(n,\varepsilon).
\end{align}
The first inequality follows from subadditivity of quantum entropy. The final inequality follows because the average can never exceed the maximum.
In the above, $B'_{[n]\setminus \{i\} }$ denotes the joint system $B'_1{B'}_2\cdots B'_{i-1}{B'}_{i+1}\cdots B'_n$, such that  system $B'_i$ is excluded. Furthermore,
\begin{equation}
\sigma_{MQR^\prime B} =\frac{1}{M}\frac{1}{n}\sum_{m=1}^{|M|}\sum_{i=1}^n\op{m}_M\otimes\op{i}_Q\otimes\mc{N}^{x_i(m)}_{B'_i\to B_i}(\sigma_{R S B'_i B'_{[n]\setminus{i}}}) ,
\end{equation}
where we have introduced an auxiliary classical register $Q$, and  $R^\prime \coloneqq RSB'_{[n]\setminus i}$. Also,
\begin{equation}\label{eq:cq-phi}
\tau_{X\wt{R} B}=\sum_{x}p_X(x)\op{x}_X\otimes\mc{N}^x(\phi_{\wt{R} B'}).
\end{equation}
 
Now we argue that it is sufficient to take $\phi_{\wt{R}B'}$ to be a pure state. Suppose that $\phi_{\hat{R}B'}$ is a mixed state and let $R^{\prime\prime}$ be a purifying system for it. Then by the data-processing inequality, it follows that
\begin{equation}
I(X\hat{R};B)_\tau \leq I(X\hat{R}R^{\prime\prime};B)_\tau,
\end{equation}
where $\tau_{X\hat{R}R''B}$ is a state of the form in \eqref{eq:cq-phi}.
The statement in the theorem about the dimension of the reference system  follows from the Schmidt decomposition and the fact that the reference system purifies the system $B'$ being input to the channel.
\end{proof}

\subsection{Weak converse bound for a quantum reading protocol}

Now we establish a general weak converse bound for the quantum reading capacity of an arbitrary memory cell.

\begin{theorem}\label{thm:q-r-capacity}
The quantum reading capacity of a quantum memory cell $\msc{S}_{\msc{X}}=\{\mc{N}^x\}_{\msc{X}}$ is bounded from above as
\begin{equation}
\mathcal{C}(\msc{S}_{\msc{X}})\leq \sup_{\rho_{XR{B'}}}\left[I(X;B|R)_\omega-I(X;{B'}|R)_\rho\right],
\end{equation}
where
\begin{align}
\omega_{XRB}&\coloneqq \sum_{x\in\msc{x}}p_X(x)\op{x}_X\otimes\mc{N}^x_{{B'}\to B}(\rho^x_{R{B'}}),\\
\rho_{XR{B'}}&=\sum_{x\in\msc{X}}p_X(x)\op{x}_X\otimes\rho^x_{R{B'}},
\end{align}
and $\dim(\mc{H}_R)$ can be unbounded. 
\end{theorem}

\begin{remark}
It should be noted that the upper bound
$\sup_{\rho_{XR{B'}}}\left[I(X;B|R)_\omega-I(X;{B'}|R)_\rho\right]$
is non-negative. A particular choice of
the input state $\rho_{XR{B'}}$ is $\rho_{XR{B'}} =
\sum_{x\in\msc{X}}p_X(x)\op{x}_X\otimes\rho_{R{B'}}$. Then in this case,
\begin{equation}
I(X;B|R)_\omega-I(X;{B'}|R)_\rho = I(X;RB)_\omega-I(X;R{B'})_\rho
= I(X;RB)_\omega \geq 0,
\end{equation}
with $\omega_{XRB}=\sum_{x\in\msc{X}}p_X(x)\op{x}_X\otimes\mc{N}^x_{{B'}\to B}(\rho_{R{B'}})$. Thus, we can conclude that
\begin{equation}
\sup_{\rho_{XR{B'}}}\left[I(X;B|R)_\omega-I(X;{B'}|R)_\rho\right] \geq 0.
\end{equation}
\end{remark}

\begin{proof}[Proof of Theorem~\ref{thm:q-r-capacity}]
For any $(n,R,\varepsilon)$ quantum reading protocol as stated in Definition~\ref{def:QR}, we have
\begin{equation}
\frac{1}{2}\left\Vert  \overline{\Phi}_{M\hat{M}}-\sigma^\prime_{M\hat{M}}\right\Vert_1\leq\varepsilon,
\end{equation}
where $\overline{\Phi}_{M\hat{M}}$ is a maximally classically correlated state \eqref{eq:cor-state} and 
\begin{equation}
\sigma^\prime_{M\hat{M}}=\mc{D}_{R_nB_n\to\hat{M}}\(\sigma^n_{MR_nB_n}\)
\end{equation}
is the output state at the end of the protocol after Bob performs the final decoding measurement. 
The input state before the $i$th call of the channel is denoted as
\begin{multline}
\rho^{i}_{M R_i B'_i}
=\frac{1}{|M|}\sum_{m\in\mc{M}}\op{m}_M\otimes\mc{A}^{(i-1)}_{R_{i-1}B_{i-1}\to R_iB'_i}\circ\mc{N}^{x_{i-1}(m)}_{B'_{i-1}\to B_{i-1}}\circ\cdots\\
\cdots\circ\mc{N}^{x_2(m)}_{B'_2\to B_2}\circ\mc{A}^{(1)}_{R_{1}B_{1}\to R_2B'_2}\circ\mc{N}^{x_{1}(m)}_{B'_{1}\to B_{1}}(\rho_{R_1B'_1})\label{eq:state-i-pre-use} ,
\end{multline}
and the output state after the $i^{\tn{th}}$ call of the channel is denoted as
\begin{multline}
\omega^{i}_{M R_i B_i}
=\frac{1}{|M|}\sum_{m\in\mc{M}}\op{m}_M\otimes\mc{N}^{x_i(m)}_{B'_i\to B_i}\circ\mc{A}^{(i-1)}_{R_{i-1}B_{i-1}\to R_iB'_i}\circ\mc{N}^{x_{i-1}(m)}_{B'_{i-1}\to B_{i-1}}\circ\cdots\\
\cdots\circ\mc{N}^{x_2(m)}_{B'_2\to B_2}\circ\mc{A}^{(1)}_{R_{1}B_{1}\to R_2B'_2}\circ\mc{N}^{x_{1}(m)}_{B'_{1}\to B_{1}}(\rho_{R_1B'_1}).\label{eq:state-i-use}
\end{multline}
 
The initial part of our proof follows steps similar to those in the proof of Proposition~\ref{thm:weak-converse-non-adaptive}.
\begin{align}
\log_2|M| &= I(M;\hat{M})_{\overline{\Phi}}\\
& \leq I(M;\hat{M})_{\sigma^\prime}+f(n,\varepsilon)\\
& \leq I(M;R_n B_n)_{\sigma^n}+f(n,\varepsilon)\\
& = I(M;R_nB_n)_{\omega^n}-I(M;R_1B'_1)_{\rho^1}+f(n,\varepsilon)\\
&= I(M;R_nB_n)_{\omega^n}-I(M;R_nB'_n)_{\rho^n}+I(M;R_nB'_n)_{\rho^n} -I(M;R_{n-1}B'_{n-1})_{\rho^{n-1}}\nonumber\\ &\qquad+I(M;R_{n-1}B'_{n-1})_{\rho^{n-1}}-\cdots-I(M;R_2B'_2)_{\rho^2}\nonumber\\
& \qquad+I(M;R_2B'_2)_{\rho^2}-I(M;R_1B'_1)_{\rho^1}+f(n,\varepsilon)\\
&\leq I(M;R_nB_n)_{\omega^n}-I(M;R_nB'_n)_{\rho^n}+I(M;R_{n-1}B_{n-1})_{\omega^{n-1}} -I(M;R_{n-1}B'_{n-1})_{\rho^{n-1}}\nonumber\\ &\qquad+I(M;R_{n-2}B_{n-2})_{\omega^{n-2}}-\cdots-I(M;R_2B'_2)_{\rho^2}\nonumber \\
& \qquad +I(M;R_1B_1)_{\omega^1}-I(M;R_1B'_1)_{\rho^1}+f(n,\varepsilon) \label{eq:first-block-wk-conv-adaptive}
\end{align}
The second equality follows because the state $\rho^1$ is product between systems $M$ and $R_1B'_1$. The third equality follows by  adding and subtracting  equal information quantities.   
The third inequality follows from the data-processing inequality: mutual information is non-increasing under the local action of quantum channels. Continuing, it follows that
\begin{align}
\eqref{eq:first-block-wk-conv-adaptive} &= \sum_{i=1}^n\left[I(M;R_iB_i)_{\omega^i}-I(M;R_iB'_i)_{\rho^i}\right]+f(n,\varepsilon)\\
& = \sum_{i=1}^n\left[I(M;B_i|R_i)_{\omega^i}-I(M;B'_i|R_i)_{\rho^i}\right]+f(n,\varepsilon)\\
& = n \left[I(M;B|RQ)_{\overline{\omega}}-I(M;{B'}|RQ)_{\overline{\rho}}\right]+f(n,\varepsilon)\\
&\leq n\sup_{\rho_{XR{B'}}}\left[I(X;B|R)_\omega-I(X;{B'}|R)_\rho\right]+f(n,\varepsilon),
\end{align}
The second equality follows from the chain rule for conditional mutual information. The third equality follows by defining the following states:
\begin{align}
\overline{\omega}_{QMRB} & =
\sum_{i=1}^n \frac{1}{n} \op{i}_Q \otimes 
\omega^{i}_{M R_i B_i}, \\ 
\overline{\rho}_{QMR{B'}} & =
\sum_{i=1}^n \frac{1}{n} \op{i}_Q \otimes 
\rho^{i}_{M R_i B'_i}.
\end{align}

The final inequality follows by defining the following states:
\begin{align}
\omega_{XRB}&=\sum_{x}p_X(x)\op{x}_X\otimes\mc{N}^x_{{B'}\to B}(\rho^x_{R{B'}}),\\
\rho_{XR{B'}}&=\sum_{x}p_X(x)\op{x}_X\otimes\rho^x_{R{B'}},
\end{align}
and realizing that the states $\overline{\omega}_{QMRB}$ and $\overline{\rho}_{QMR{B'}}$ are particular examples of the states $\omega_{XRB}$ and $\rho_{XR{B'}}$, respectively, with the identifications $M \to X$ and $QR \to R$.
Putting everything together, we get
\begin{equation}
\frac{1}{n}\log_2|M| \leq 
\sup_{\rho_{XR{B'}}}\left[I(X;B|R)_\omega-I(X;{B'}|R)_\rho\right]+\frac{1}{n} f(n,\varepsilon)
\end{equation}
Taking the limit as $n \to \infty$ and then as $\varepsilon \to 0$ concludes the proof.
\end{proof}

\bigskip
Now we develop a general upper bound on the energy-constrained quantum reading capacity of a beamsplitter memory cell $\msc{B}_{\msc{X}}=\{\mc{B}^{x}\}_{x\in\msc{X}}$, where $x\in\msc{X}$ represents the transmissivity $\eta$ and phase $\phi$ of the beamsplitter $\mc{B}^x$~\cite[Eqns.~(5)--(6)]{KMN+05} (see Section~\ref{sec-Gaussian}). This bound has implications for the reading protocols considered in \cite{WGTL12}. 

Let $\hat{O}$ denote the familiar $a^\dag a$ number observable
and let $N_S \in [0,\infty)$. The energy-constrained reading capacity
$\mc{C}(\msc{B}_{\mc{X}},\hat{O}, N_S)$
of a beamsplitter memory cell $\msc{B}_{\msc{X}}$ is defined in the obvious way, such that the average input to each call of the memory is bounded from above by $N_S \geq 0$. This definition implies that the function to optimize in the  capacity upper bound has the following constraint: for any input ensemble $\{p_X(x),\rho^x_{R{B'}}\}$,
\begin{equation}
\Tr\left\{\hat{O}\int p_X(x)\rho^x_{{B'}}\right\}\leq N_{S}.
\end{equation}
Since the energy of the output state of $\mc{B}^x$ does not depend on the phase $\phi$, the dependence of $x$ on $\phi$ is dropped and $x=\eta$ is taken for the discussion. 
For a memory cell $\msc{B}_{\msc{X}}$, the energy of the output state is constrained as
\begin{align}
\Tr\left\{\sum_{x\in\msc{X}}p_X(x)\mc{B}^x(\rho^x_{B'})\hat{O}\right\} &=\sum_{x\in\msc{X}}p_X(x)\Tr\left\{\mc{B}^x(\rho^x_{B'})\hat{O}\right\}\\
&=\sum_{x\in\msc{X}}p_X(x)\eta \Tr\left\{\rho^x_{B'}\hat{O}\right\}\\
&\leq N_S,
\end{align}
where the second equality holds because the transmissivity of each $\mc{B}^x$ is $\eta \in [0,1]$.

Based on the above discussion, the following theorem can be stated.
\begin{corollary}
The energy-constrained reading capacity of a beamsplitter memory cell $\msc{B}_{\msc{X}}=\{\mc{B}^x\}_{x\in\msc{X}}$ is bounded from above as
\begin{equation}
\mc{C}(\msc{B}_{\msc{X}},\hat{O}, N_S)\leq 2 g(N_S),
\end{equation}
where $\theta^{N_S}$ is a thermal state~\eqref{eq:thermal-state} such that $\Tr\{ \hat{O} \theta^{N_S}\} = N_S$ and $g(y):=(y+1)\log_2(y+1)-y\log_2 y$.
\end{corollary}

\begin{proof}
From a straightforward extension of  Theorem~\ref{thm:q-r-capacity}, which takes into account the energy constraint, we find that 
\begin{align}
\mc{C}(\msc{B}_{\msc{X}},\hat{O}, N_S)
&\leq \sup_{\{p_X(x),\rho^x_{R{B'}}\}\ :\  \mathbb{E}_X\{\Tr\{\hat{O} \rho^X_{{B'}}\}\} \leq N_S }I(X;B|R)_\omega - I(X;{B'}|R)_\rho\\
& \leq \sup_{\{p_X(x),\rho^x_{R{B'}}\}\ :\  \mathbb{E}_X\{\Tr\{\hat{O} \rho^X_{{B'}}\}\} \leq N_S }I(X;B|R)_\rho\\
&\leq \sup_{\{p_X(x),\rho^x_{R{B'}}\}
\ :\  \mathbb{E}_X\{\Tr\{\hat{O} \rho^X_{{B'}}\}\} \leq N_S} 2S(B)_\rho\\
&\leq 2S(\theta^{N_S})\\
&= 2g(N_S).
\end{align}
The first inequality follows from the extension of Theorem~\ref{thm:q-r-capacity}. The second inequality follows from non-negativity of the conditional quantum mutual information. The third inequality follows from a standard entropy bound for the conditional quantum mutual information. 
The fourth inequality follows because the thermal state of mean energy $N_S$ has the maximum entropy under a fixed energy constraint (see, e.g., \cite{Carlen09}). The final equality follows because the observable $\hat{O}$ is the familiar $a^\dag a$ number observable, for which the entropy of its thermal state of mean photon number $N_S$ is given by $g(N_S)$.
\end{proof}

\begin{remark}
It follows  that $\mc{C}_{\textnormal{non-adaptive}}(\msc{B}_{\msc{X}},\hat{O}, N_S)\leq
2g(N_S)$ because
$\mc{C}_{\textnormal{non-adaptive}}(\msc{B}_{\msc{X}},\hat{O}, N_S)\leq  \mc{C}(\msc{B}_{\msc{X}},\hat{O}, N_S)$ by the definition of the energy-constrained quantum reading capacity of a memory cell~$\msc{B}_{\msc{X}}$. 
\end{remark}

\section{Examples of environment-parametrized memory cells}
\label{sec:example}

In this section, we calculate the quantum reading capacities of several environment-parametrized memory cells, including a thermal memory cell, and a jointly covariant memory cell formed from a channel $\mc{N}$ and a group $\msc{G}$ with respect to which $\mc{N}$ is covariant (Definition~\ref{def:j-cov-cell}). Examples of such a jointly covariant memory cell include qudit erasure and depolarizing memory cells formed respectively from erasure and depolarizing channels. 

\subsection{Jointly covariant memory cell: $\msc{N}^{\textnormal{cov}}_{\msc{G}}$}

Now we show that the quantum reading capacity of a memory cell $\msc{N}^{\textnormal{cov}}_{\msc{G}}$ (see Definition~\ref{def:N-G-cell} below) is equal to the entanglement-assisted classical capacity of the underlying channel $\mc{N}$. This result makes use of the fact that the entanglement-assisted classical capacity of a covariant channel $\mc{T}$ is equal to $I(R;B)_{\mc{T}(\Phi)}$ \cite{BSST99,BSST02}. Furthermore, we use this result to evaluate the quantum reading capacity of a qudit erasure memory cell (Definition~\ref{def:qudit-erasure}) and a qudit depolarizing memory cell (Definition~\ref{def:qudit-dep}). 

\begin{definition}[$\msc{N}^{\textnormal{cov}}_{\msc{G}}$]\label{def:N-G-cell}
Let $\mc{N}$ be a covariant channel (Definition~\ref{def:covariant}) with respect to a group $\msc{G}$. The memory cell $\msc{N}^{\textnormal{cov}}_{\msc{G}}$ is defined as
\begin{equation}
\msc{N}^{\textnormal{cov}}_{\msc{G}}=\left\{\mc{N}_{{B'}\to B}\circ\mc{U}^g_{B'}\right\}_{g\in {\msc{G}}},
\end{equation}
where $\mc{U}_{B'}^g\coloneqq U_{B'}(g)(\cdot)U^\dag_{B'}(g)$.  It follows from~\eqref{eq:cov-condition} that 
\begin{equation}\label{eq:con-cov-2}
\mc{N}_{{B'}\to B}\circ\mc{U}^g_{B'}= \mc{V}^g_B\circ\mc{N}_{{B'}\to B},
\end{equation}
where $\mc{V}_B^g:=V_B(g)(\cdot)V^\dag_B(g)$. It also follows that $\msc{N}^{\textnormal{cov}}_{\msc{G}}$ is a jointly covariant memory cell.
\end{definition}

\begin{theorem}\label{thm:covariant-to-EA-cap}
The quantum reading capacity $\mc{C}(\msc{N}^\textnormal{cov}_{\msc{G}})$ of the jointly covariant memory cell $\msc{N}^{\textnormal{cov}}_{\msc{G}}=\left\{\mc{N}_{{B'}\to B}\circ\mc{U}^g_{B'}\right\}_{g\in \msc{G}}$ (Definition~\ref{def:N-G-cell}), is equal to the entanglement-assisted classical capacity of~$\mc{N}$:
\begin{equation}
\mc{C}(\msc{N}^\textnormal{cov}_{\msc{G}}) = I(R;B)_{\mc{N}(\Phi)},
\end{equation}
where $\mc{N}(\Phi)\coloneqq\mc{N}_{{B'}\to B}(\Phi_{R{B'}})$ is the Choi state of the underlying channel $\mc{N}$. 
\end{theorem}

\begin{proof}
 Proof here consists of two parts: the converse part and the achievability part. We first show the converse part:
\begin{equation}
\mc{C}\(\msc{N}^\textnormal{cov}_{\msc{G}}\) \leq I(R;B)_{\mc{N}(\Phi)}. 
\end{equation}
From Remark~\ref{rem:cov-cell-achievability}, we can conclude that the quantum reading capacity of $\msc{N}^\textnormal{cov}_{\msc{G}}$ is as follows:
\begin{equation}
\mc{C}\(\msc{N}^\textnormal{cov}_{\msc{G}}\)=\max_{p_G}I(G;RB)_{\omega},
\end{equation}
where
\begin{equation}
\omega_{GRB}\coloneqq \sum_{g\in {\msc{G}}}p_G(g)\op{g}_G\otimes\omega^{g}_{RB},
\end{equation}
such that $\{|g\>\}_{g\in {\msc{G}}}\in\ONB(\mc{H}_{G})$ and 
\begin{equation}
\forall g\in {\msc{G}}:\ \omega^g_{RB}=(\mc{N}_{{B'}\to B}\circ\mc{U}^g_{B'})(\Phi_{R{B'}}).
\end{equation}
Let us consider $p_G$ to be fixed. Then
\begin{align}
I(G;RB)_{\omega}&=S\!\(\sum_{g\in {\msc{G}}}p_G(g)\omega^g_{RB}\)-\sum_{g\in {\msc{G}}}p_G(g)S(\omega^{g}_{RB})\\
&=S\!\(\sum_{g\in {\msc{G}}}p_G(g)(\mc{V}^g_B\circ\mc{N}_{{B'}\to B})(\Phi_{R{B'}})\)-\sum_{g\in {\msc{G}}}p_G(g)H((\mc{V}^g_B\circ\mc{N}_{{B'}\to B})(\Phi_{R{B'}}))\label{eq:observe-cov-uni}\\
%&=H\(\sum_{g\in G}p_G(g)\mc{V}^g_B\circ\mc{N}_{{B'}\to B}(\Phi_{R{B'}})\)-H\(\mc{N}_{{B'}\to B}(\Phi_{R{B'}})\)\\
&=\sum_{g^\prime\in {\msc{G}}}\frac{1}{|G|}S\!\(\sum_{g\in {\msc{G}}}p_G(g)(\mc{V}^{g^\prime}_B\circ\mc{V}^g_B\circ\mc{N}_{{B'}\to B})(\Phi_{R{B'}})\)-S(\mc{N}_{{B'}\to B}(\Phi_{R{B'}}))\\
&\leq S\!\(\frac{1}{|G|}\sum_{g,g^\prime\in {\msc{G}}}p_G(g)(\mc{V}^{g^\prime}_B\circ\mc{V}^g_B\circ\mc{N}_{{B'}\to B})(\Phi_{R{B'}})\)-S(\mc{N}_{{B'}\to B}(\Phi_{R{B'}}))\\
&=S\(\mc{N}_{{B'}\to B}\(\frac{1}{|G|}\sum_{g^\prime\in {\msc{G}}}\mc{U}^{g^\prime}_{B'}\(\sum_{g\in {\msc{G}}}p_G(g)\mc{U}^g_{B'}\(\Phi_{R{B'}}\)\)\)\)-S(\mc{N}_{{B'}\to B}(\Phi_{R{B'}}))\\
&=S\(\mc{N}_{{B'}\to B}(\pi_R\otimes\pi_{B'})\)-S(\mc{N}_{{B'}\to B}(\Phi_{R{B'}}))\\
&=S\(\pi_R\)+S\(\mc{N}_{{B'}\to B}(\pi_B)\)-S(\mc{N}_{{B'}\to B}(\Phi_{R{B'}}))\\
&=I(R;B)_{\mc{N}(\Phi)}.
\end{align}
The second equality follows from~\eqref{eq:con-cov-2}. The third equality follows because entropy is invariant with respect to unitary or isometric channels. The first inequality follows from the concavity of entropy. The fourth equality follows from~\eqref{eq:con-cov-2}. The fifth equality follows from Definition~\ref{def:covariant}. The sixth equality follows because entropy is additive for product states. 
Since the above upper bound holds for any $p_G$, it follows that
\begin{equation}\label{eq:cov-con}
\mc{C}\(\msc{N}^\textnormal{cov}_{\msc{G}}\)=\max_{p_G}I(G;RB)_{\omega}\leq I(R;B)_{\mc{N}(\Phi)}.
\end{equation} 

To prove the achievability part, we take $p_G$ to be a uniform distribution, i.e., $p_G\sim \frac{1}{|G|}$. Putting $p_G\sim\frac{1}{|G|}$ in \eqref{eq:observe-cov-uni}, we obtain the following lower bound
\begin{equation}\label{eq:cov-ach}
\mc{C}\(\msc{N}^\textnormal{cov}_{\msc{G}}\) \geq I(G;RB)_\omega = I(R;B)_{\mc{N}(\Phi)}.
\end{equation}
Thus, from \eqref{eq:cov-con} and \eqref{eq:cov-ach}, we conclude the statement of the theorem:
$\mc{C}\(\msc{N}^\textnormal{cov}_{\msc{G}}\)=I(R;B)_{\mc{N}(\Phi)}.$
\end{proof}

\bigskip 
Now we state two corollaries, which are direct consequences of the above theorem. These corollaries establish the quantum reading capacities for jointly covariant memory cells formed from the erasure channel and depolarizing channel with respect to the Heisenberg--Weyl group $\mathbf{H}$, as discussed below (see Appendix~\ref{app:qudit} for some basic notations and definitions related to qudit systems).

\begin{definition}[Qudit erasure memory cell]\label{def:qudit-erasure}
The qudit erasure memory cell $\msc{Q}^q_{\msc{X}}=\left\{\mc{Q}^{q,x}_{{B'}\to B}\right\}_{x\in\msc{X}}$, where the size of $\msc{X}$ is $|X|=d^2$, consists of the following qudit channels:
\begin{equation}
\mc{Q}^{q,x}(\cdot)=\mc{Q}^q(\sigma^x(\cdot)\(\sigma^x\)^\dag) ,
\end{equation}
where $\mc{Q}^q$ is a qudit erasure channel \cite{GBP97}:
\begin{equation}
\mc{Q}^q(\rho_{B'})=(1-q)\rho+q\op{e}
\end{equation}
such that $q\in [0,1]$, $\dim(\mc{H}_{B'})=d$, $\op{e}$ is an erasure state orthogonal to the support of all possible input states $\rho$, and
$\forall x\in\msc{X}: \sigma^x\in\mathbf{H}$ are the Heisenberg--Weyl operators as given in \eqref{eq:HW-op}. Observe that $\msc{Q}^q_{\msc{X}}$ is jointly covariant with respect to the Heisenberg--Weyl group $\mathbf{H}$ because the qudit erasure channel $\mc{Q}^q$ is covariant with respect to $\mathbf{H}$.
\end{definition}

\begin{definition}[Qudit depolarizing memory cell]\label{def:qudit-dep}
The qudit depolarizing memory cell $\msc{D}^q_{\msc{X}}=\left\{\mc{D}^{q,x}_{{B'}\to B}\right\}_{x\in\msc{X}}$, where $\msc{X}$ is of size $|\mc{X}|=d^2$, consists of qudit channels
\begin{equation}
\mc{D}^{q,x}(\cdot)=\mc{D}^q\(\sigma^x(\cdot)\(\sigma^x\)^\dag\)
\end{equation}
where $\mc{D}^q$ is a qudit depolarizing channel:
\begin{equation}
\mc{D}^q(\rho)=(1-q)\rho+q\pi,
\end{equation}
where $q\in\[0,\frac{d^2}{d^2-1}\]$, $\dim(\mc{H}_{B'})=d$ and 
$\forall x\in\msc{X}: \sigma^x\in\mathbf{H}$
are the Heisenberg--Weyl operators as given in \eqref{eq:HW-op}.  
Observe that $\msc{D}^q_{\msc{X}}$ is jointly covariant with respect to the Heisenberg--Weyl group $\mathbf{H}$ because the qudit depolarizing channel $\mc{D}^q$ is covariant with respect to $\mathbf{H}$.
\end{definition}

As a consequence of Theorem~\ref{thm:covariant-to-EA-cap}, one immediately finds the quantum reading capacities of the above memory cells:

\begin{corollary}\label{cor:e-cell-r}
The quantum reading capacity $\mc{C}(\msc{Q}^q_{\msc{X}})$ of the qudit erasure memory cell $\msc{Q}^q_{\msc{X}}$ (Definition~\ref{def:qudit-erasure}) is equal to the entanglement-assisted classical capacity of the erasure channel $\mc{Q}^q$ \cite{BSST99}:
\begin{equation}
\mc{C}(\msc{Q}^q_{\msc{X}})=2(1-q)\log_2d.
\end{equation}
\end{corollary}

\begin{corollary}
The quantum reading capacity $\mc{C}(\msc{D}^q_{\msc{X}})$ of the qudit depolarizing memory cell $\msc{D}^q_{\msc{X}}$ (Definition~\ref{def:qudit-dep}) is equal to the entanglement-assisted classical capacity of the depolarizing channel $\mc{D}^q$ \cite{BSST99}:
\begin{equation}
\mc{C}(\msc{D}^q_{\msc{X}})=2\log_2d+\(1-q+\frac{q}{d^2}\)\log_2\!\(1-q+\frac{q}{d^2}\)+(d^2-1)\frac{q}{d^2}\log_2\!\(\frac{q}{d^2}\).
\end{equation}
\end{corollary}

\subsection{A thermal memory cell}
Now we discuss an example of a thermal memory cell $\hat{\msc{E}}_{\msc{X},\eta} = \{ \mc{E}^{x,\eta}\}_x$, which is an environment-parametrized memory cell consisting of thermal channels $\mc{E}^{x,\eta}$ with known transmissivity parameter $\eta\in[0,1]$ and unknown excess noise $x$. Let $\hat{a},\hat{b},\hat{e},\hat{e}^\prime$ be the respective field-mode annihilation operators for Bob's input, Bob's output, the environment's input, and the environment's output of these channels. The interaction channel in this case is a fixed bipartite unitary $U_{{B'}E\to BE^\prime}$ corresponding to a beamsplitter interaction, defined from the following Heisenberg input-output relations:
\begin{align}
\hat{b}&=\sqrt{\eta}\hat{a}+\sqrt{1-\eta}\hat{e},\\
\hat{e}^\prime&=-\sqrt{1-\eta}\hat{a}+\sqrt{\eta}\hat{e}.
\end{align}
The environmental mode $\hat{e}$ of a thermal channel $\mc{E}^{x,\eta}$ is prepared in a thermal state $\theta^x:=\theta(N_B=x)$ of mean photon number $N_B\geq 0$: 
\begin{equation}\label{eq:thermal-state}
\theta(N_B):=\frac{1}{N_B+1}\sum_{k=0}^\infty\(\frac{N_B}{N_B+1}\)^k\op{k},
\end{equation}
where $\{|k\>\}_{k\in\mathbb{N}}$ is the orthonormal, photonic number-state basis. Parameter $x$ is the excess noise of the thermal channel $\mc{E}^{x,\eta}$. Note that for $x=0$, $\theta^x$ reduces to a vacuum state and the channel $\mc{E}^{x,\eta}$ is called the pure-loss channel (see Section~\ref{sec-Gaussian}).  

\begin{proposition}
The quantum reading capacity $\mc{C}(\hat{\msc{E}}_{\msc{X},\eta})$ of the thermal memory cell $\hat{\msc{E}}_{\msc{X},\eta} = \{ \mc{E}^{x,\eta}\}_x$ (as described above) is equal to
\begin{equation}
\mc{C}(\hat{\msc{E}}_{\msc{X},\eta}) = \max_{p_X}
\left[S(\overline{\theta}) - \int dx \ p_X(x) S(\theta^x)\right],
\end{equation}
where $p_X$ is a probability distribution for the parameter $x$ and $\overline{\theta} \coloneqq 
\int dx \ p_X(x) \theta^x$.
\end{proposition}

\begin{proof}
We begin by proving the achievability part, which corresponds to the inequality
\begin{equation}
\mc{C}(\hat{\msc{E}}_{\msc{X},\eta}) \geq I(X;E)_\theta,
\end{equation}
where $\theta_{XE}\coloneqq \int dx \ p_X(x)\op{x}_X\otimes\theta^x_E$. The main idea for the achievability part builds on the results of \cite[Eqns.~(38)--(48)]{TW16}.

The two-mode squeezed vacuum state is equivalent to a purification of the thermal state in \eqref{eq:thermal-state} and is defined as
\begin{equation}
\left|\phi^\textnormal{TMS}(N_S)\right>_{R{B'}}\coloneqq \frac{1}{\sqrt{N_S+1}}\sum_{k=0}^\infty\[\frac{N_S}{N_S+1}\]^{\frac{k}{2}}|k\>_R|k\>_{B'}.
\end{equation}
When sending the ${B'}$ system of this state through the channel $\mc{E}^{x,\eta}_{ {B'}\to B}$, the output state is as follows:
\begin{align}
\omega^{x,\eta}_{RB}(N_S)&\coloneqq (\id_R\otimes\mc{E}^{x,\eta}_{ {B'}\to B})\(\phi^\textnormal{TMS}_{R{B'}}(N_S)\)\label{eq:omega-x}\\
&=\Tr_{E^\prime}\left\{U_{{B'}E\to BE^\prime}\(\phi_{R{B'}}(N_S)\otimes\theta^x_E\)\(U_{{B'}E\to BE^\prime}\)^\dag\right\}\label{eq:theta-x-thermal},
\end{align}
and the average output state is as follows, when the channel $\mc{E}^{x,\eta}_{ {B'}\to B}$ being applied is chosen with probability
$p_X(x)$:
\begin{align}
\sum_{x\in\msc{X}}p_X(x)\omega^{x,\eta}_{RB}(N_S)&=\sum_{x\in\msc{X}}p_X(x)\Tr_{E^\prime}\left\{U_{{B'}E\to BE^\prime}\(\phi_{R{B'}}(N_S)\otimes\theta^x_E\)\(U_{{B'}E\to BE^\prime}\)^\dag\right\}\\
&=\Tr_{E^\prime}\left\{U_{{B'}E\to BE^\prime}\(\phi_{R{B'}}(N_S)\otimes\sum_{x\in\msc{X}}p_X(x)\theta^x_E\)\(U_{{B'}E\to BE^\prime}\)^\dag\right\}.
\end{align}
Consider the following classical--quantum state:
\begin{equation}
\omega^{\eta}_{XRB}(N_S)\coloneqq \sum_{x\in\msc{X}}p_X(x)\op{x}_X\otimes\omega^{x,\eta}_{RB},
\end{equation}
and 
\begin{equation}
I(X;RB)_{\omega^{\eta}(N_S)}=\sum_{x\in\msc{X}}p_X(x)D\(\omega^{x,\eta}_{RB}(N_S)\left\Vert\sum_{x\in\msc{X}}p_X(x)\omega^{x,\eta}_{RB}(N_S)\)\right. .
\end{equation}

The Wigner characteristic function covariance matrix \cite{adesso14} for $\omega^{x,\eta}_{RB}(N_S)$ in \eqref{eq:omega-x} is as follows:
\begin{equation}
V_{\omega^{x,\eta}(N_S)}=\begin{bmatrix}
a & c & 0 & 0 \\
c & b & 0 & 0 \\
0 & 0 & a & -c \\
0 & 0 & -c & b \\
\end{bmatrix},
\end{equation}
where
\begin{align}
a=\eta N_S+\(1-\eta\)x+\frac{1}{2},\quad\,
b=N_S+\frac{1}{2},\quad\,
c&=\sqrt{\eta N_S(N_S+1)}\ . 
\end{align}

Now consider the following symplectic transformation \cite{TW16}:
\begin{equation}
S^{\eta}(N_S)=\begin{bmatrix}
\gamma_+ & -\gamma_- & 0 & 0 \\
-\gamma_- & \gamma_+ & 0 & 0 \\
0 & 0 & \gamma_+ & \gamma_- \\
0 & 0 & \gamma_- & \gamma_+ \\
\end{bmatrix} ,
\end{equation}
where
\begin{align}
\gamma_+=\sqrt{\frac{1+N_S}{1+(1-\eta)N_S}},\quad\quad\quad
\gamma_-=\sqrt{\frac{\eta N_S}{1+(1-\eta)N_S}}\ .
\end{align}

Action of the symplectic matrix $S^{\eta}(N_S)$ on the covariance matrix $V_{\omega^{x,\eta}(N_S)}$ gives
\begin{align}
\hat{V}_{\omega^{x,\eta}(N_S)}:=S^{\eta}(N_S)V_{\omega^{x,\eta}(N_S)}\(S^{\eta}(N_S)\)^\textnormal{T}=\begin{bmatrix}
a_s & -c_s & 0 & 0 \\
-c_s & b_s & 0 & 0 \\
0 & 0 & a_s & c_s \\
0 & 0 & c_s & b_s \\
\end{bmatrix},
\end{align}
where
\begin{align}
a_s&=x+\frac{1}{2}+\mc{O}\(\frac{1}{N_S}\),\\
b_s&=\(1-\eta\)N_S+\eta x+\frac{1}{2}+\mc{O}\(\frac{1}{N_S}\),\\
c_s&=\sqrt{\eta}x+\mc{O}\(\frac{1}{N_S}\)\ . 
\end{align}
Thus, by applying this transformation to $\omega^{x,\eta}(N_S)$ and tracing out the second mode, we are left with a state that becomes indistinguishable from a thermal state of mean photon number $x$ in the limit as $N_S \to \infty$. Note that this occurs independent of the value of the transmissivity~$\eta$.

The symplectic transformation $S^{\eta}(N_S)$ can be realized by a two-mode squeezer, which corresponds to a unitary transformation acting on the tensor-product Hilbert space. Letting the unitary transformation be of the form $W_{RB\to EB}$, then $\hat{V}_{\omega^{x,\eta}(N_S)}$ represents the covariance matrix of the state $\omega^{x,\eta}_{EB}(N_S)$. 

We use the formula for fidelity between two thermal states \cite[Equation 34]{TW16} and the relation between trace norm and fidelity \cite[Theorem 9.3.1]{Wbook17} to conclude that
\begin{align}
\lim_{N_S\to\infty}\left\Vert \omega^{x,\eta}_E(N_S)-\theta^x_E\right\Vert_1&\leq \lim_{N_S\to\infty}\sqrt{1-F\(\omega^{x}_E(N_S),\theta^x_E\)}=0.\end{align}
From the convexity of trace norm, we obtain
\begin{align}
\left\Vert \sum_{x\in\msc{X}}p_X(x)\omega^x_E(N_S)-\sum_{x\in\msc{X}}p_X(x)\theta^x_E\right\Vert_1&\leq \sum_{x\in\msc{X}}p_X(x)\left\Vert \omega^{x}_E(N_S)-\theta^x_E\right\Vert_1,
\end{align}
which in turn implies that
\begin{align}
 \lim_{N_S\to\infty}\left\Vert \sum_{x\in\msc{X}}p_X(x)\omega^x_E(N_S)-\sum_{x\in\msc{X}}p_X(x)\theta^x_E\right\Vert_1=0.
\end{align}

Invoking the result of \cite[Equation 28]{TW16} and the lower semi-continuity of relative entropy, one gets
\begin{equation}
\lim_{N_S\to\infty} D\(\omega^{x,\eta}_{RB}(N_S)\left\Vert\sum_{x\in\msc{X}}p_X(x)\omega^{x,\eta}_{RB}(N_S)\)\right.=D\(\theta^x_{E}\left\Vert \sum_{x\in\msc{X}}p_X(x)\theta^x_E\)\right..
\end{equation}

Thus, from the above relations, we obtain the following result
\begin{equation}
\lim_{N_S\to \infty}I(X;RB)_{\omega^{\eta}(N_S)}=I(X;E)_{\theta},
\end{equation}
where 
\begin{equation}
\theta_{XE}=\sum_{x\in\msc{X}}p_X(x)\op{x}_X\otimes\theta^x_E ,
\end{equation}
for $\theta^x_E$ defined in \eqref{eq:theta-x-thermal}. This shows that $I(X;E)_{\theta}$ is an achievable rate for any $p_X$. 

The converse part of the proof, which corresponds to the inequality 
\begin{equation}
\mc{C}(\hat{\msc{E}}_{\mc{X},\eta})\leq \max_{p_X}I(X;E)_\theta,
\end{equation}
follows directly from Theorem~\ref{thm:env-cell-qrc}.
\end{proof}

\section{Zero-error quantum reading capacity}\label{sec:zero-error}

In an $(n,R,\varepsilon)$ quantum reading protocol (Definition~\ref{def:QR}) for a memory cell $\msc{S}_{\msc{X}}=\{\mc{M}^x_{{B'}\to B}\}_{x\in\msc{X}}$, one can demand the error probability to vanish, i.e., $\varepsilon=0$. In this section, we define zero-error quantum reading protocols and the zero-error quantum reading capacity for any memory cell. We provide an explicit example of a memory cell for which a quantum reading protocol using an adaptive strategy has a clear advantage over a quantum reading protocol that uses a non-adaptive strategy. 

\begin{definition}[Zero-error quantum reading protocol]
A zero-error quantum reading protocol
of a memory cell $\msc{S}_{\msc{X}}$ is a particular 
$(n,R,\varepsilon)$ quantum reading protocol for which $\varepsilon=0$.
\end{definition}

\begin{definition}[Zero-error quantum reading capacity]
The zero-error quantum reading capacity
$\mc{Z}(\msc{S}_{\msc{X}})$ of a memory cell $\msc{S}_{\msc{X}}$ is defined as the largest rate $R$ such that there exists a zero-error reading protocol. 
\end{definition}

A zero-error non-adaptive quantum reading protocol of a memory cell is a special case of a zero-error quantum reading protocol in which the reader uses a non-adaptive strategy to decode the message. 
 
\subsection{Advantage of an adaptive strategy over a non-adaptive strategy}

Now we employ the main example from
\cite{HHLW10} to  illustrate the advantage of an adaptive zero-error quantum reading protocol over a non-adaptive zero-error quantum reading protocol. 

Let us consider a memory cell $\msc{B}_{\msc{X}}=\{\mc{M}_{{B'}\to B}^x\}_{x\in\msc{X}}$, $\msc{X}=\{1,2\}$, consisting of the following quantum channels that map two qubits to a single qubit,  acting as
\begin{equation}
\mc{M}^x(\cdot)=\sum_{j=1}^5A^x_j(\cdot)\(A^x_j\)^\dag, \ x\in\msc{X} ,
\end{equation}
where
\begin{align}
A^1_1=|0\>\<00|,& \ \ A^1_2=|0\>\<01|,& A^1_3=|0\>\<10|,\ \ & \ \ A^1_4=\frac{1}{\sqrt{2}}|0\>\<11|,& A^1_5=\frac{1}{\sqrt{2}}|1\>\<11|,\nonumber \\
A^2_1=|+\>\<00|,& \ \ A^2_2=|+\>\<01|,& A^2_3=|1\>\<1+|,\ \ & \ \ A^2_4=\frac{1}{\sqrt{2}}|0\>\<1-|,& A^2_5=\frac{1}{\sqrt{2}}|1\>\<1-|,
\end{align}
and the standard bases for the channel inputs and outputs are $\{|00\>,|01\>,|10\>,|11\>\}$ and $\{|0\>,|1\>\}$, respectively. 
 
It follows from \cite{HHLW10,DFY09} that it is possible to discriminate perfectly these two channels using an adaptive strategy that makes two calls to the unknown channel
$\mc{M}^x$. This implies that the encoder can encode two classical messages (one bit) into two uses of the quantum channels from $\msc{B}_{\msc{X}}$ such that Bob can perfectly read the message, i.e., with zero error. Thus, it can be concluded that the zero-error quantum reading capacity
of $\msc{B}_{\msc{X}}$
is  bounded
from below by $\frac{1}{2}$ (one bit per two channel uses).

Closely following the arguments of \cite[Section 4]{HHLW10}, we can show that non-adaptive strategies can never realize perfect discrimination of the sequences $\mc{M}_{{B'}^n\to B^n}^{x^n}$ and $\mc{M}_{{B'}^n\to B^n}^{y^n}$, for any finite number $n$ of channel uses if $x^n\neq y^n$.  Equivalently,
\begin{equation}
\textnormal{for}\ x^n\neq y^n: \Vert \mc{M}_{{B'}^n\to B^n}^{x^n}-\mc{M}_{{B'}^n\to B^n}^{y^n}\Vert_{\diamond}< 2 \ \forall n\in\mathbb{N}\label{eq:ch-discrimination}
\end{equation} 
where $\Vert \cdot\Vert_{\diamond}$ is the diamond norm \eqref{eq-diamond_norm}. Thus, the zero-error non-adaptive quantum reading capacity of $\msc{B}_{\msc{X}}$ is equal to zero. 

To prove the above claim, we proceed with a proof by contradiction along the lines of that given in \cite[Section 4]{HHLW10}. We need to show that: for any finite $n\in\mathbb{N}$, if $x^n\neq y^n$, then there does not exist any state $\sigma_{R{B'}^n}$ such that the two sequences $\mc{M}_{{B'}^n\to B^n}^{x^n}$ and $\mc{M}_{{B'}^n\to B^n}^{y^n}$ can be perfectly discriminated. Note that perfect discrimination is possible if and only if 
\begin{align}\label{eq:ch-dis}
\Tr\left\{\(\id_R \otimes \mc{M}_{{B'}^n\to B^n}^{x^n}\)(\sigma_{R{B'}^n})\(\id_R \otimes\mc{M}_{{B'}^n\to B^n}^{y^n} \)(\sigma_{R{B'}^n})\right\}=0.
\end{align}
Now assume that there exists a $\sigma_{R{B'}^n}$ such that \eqref{eq:ch-dis} holds. Convexity then implies that \eqref{eq:ch-dis} holds  for some pure state $\psi_{R{B'}^n}$. Then, by carefully following the steps from \cite[Section 4]{HHLW10}, \eqref{eq:ch-dis} implies that 
for any set of complex coefficients $\{\alpha^{x,y}_{j,k}\in\mathbb{C}:1\leq j,k\leq 5,\ x,y\in\msc{X}\}$  
\begin{equation}
\langle\psi|_{R{B'}^n} \left[ \bm{1}_R \otimes \sum_{1\leq j,k\leq 5\,:\,i\in[n]}\alpha^{x_1,y_1}_{j_1,k_1}\cdots\ \alpha^{x_n,y_n}_{j_n,k_n}\({B'}^{y_1}_{j_1}\)^\dag {B'}^{x_1}_{k_1}\otimes\cdots\otimes\({B'}^{y_n}_{j_n}\)^\dag {B'}^{x_n}_{k_n}\right]|\psi\rangle_{R{B'}^n}= 0.\label{eq:for-a-contra}
\end{equation}
%The intermediate steps are same as those discussed in \cite[Section 4]{HHLW10}.
Let us choose the coefficients $\{\alpha^{x,y}_{j,k}\in\mathbb{C}:1\leq j,k\leq 5,\ x,y\in\msc{X}\}$ as follows: 
\begin{align}
\left\{ 
\begin{tabular}{c c}
for\ $x\neq y$: &$\alpha^{x,y}_{1,1}=\alpha^{x,y}_{2,2}=\sqrt{2}$, $\alpha^{x,y}_{3,5}=\alpha^{x,y}_{4,3}=1,$ $\alpha^{x,y}_{4,4}=-2\sqrt{2}$, otherwise $ \alpha^{x,y}_{j,k}=0$,\\
for\ $x=y$: &  $\alpha^{x,y}_{j,k}=\delta_{j,k}$
\end{tabular} 
\right.
\end{align}
where, if $j=k$ then $\delta_{j,k}=1$, else $\delta_{j,k}=0$.

For the above choice of the coefficients, it follows  that 
\begin{equation*}
\bm{1}_R \otimes \sum_{1\leq j,k\leq 5\,:\,i\in[n]}\alpha^{x_1,y_1}_{j_1,k_1}\cdots\ \alpha^{x_n,y_n}_{j_n,k_n}\(A^{y_1}_{j_1}\)^\dag A^{x_1}_{k_1}\otimes\cdots\otimes\(A^{y_n}_{j_n}\)^\dag A^{x_n}_{k_n}=I_R \otimes P^{x_1,y_1}\otimes\cdots\otimes P^{x_n,y_n}
\end{equation*}  
where 
\begin{equation*}
\textnormal{for}\ i\in[n]:\ P^{x_i,y_i}= \left\{ 
\begin{tabular}{c c}
$P>0$, & $x_i\neq y_i$\\
$I>0$, &  otherwise,
\end{tabular} 
\right.
\end{equation*}
and $P=\op{00}+\op{01}+\op{11}+\op{1-}$. Observe that the operator
$\bm{1}_R \otimes P^{x_1,y_1}\otimes\cdots\otimes P^{x_n,y_n}$ is positive definite.
This means that there cannot exist any state that satisfies \eqref{eq:for-a-contra}, and as a consequence \eqref{eq:ch-dis}, and this concludes the proof.

From the above discussion, we can conclude that the zero-error quantum reading capacity of the memory cell $\msc{B}_{\msc{X}}$ is bounded from below by $\frac{1}{2}$ whereas the zero-error non-adaptive quantum reading capacity is equal to zero.

\section{Conclusion}
\label{sec:conclusion-qr}
In this chapter, we have introduced the most general and natural definitions for quantum reading protocols and quantum reading capacities. We have defined environment-parametrized memory cells for quantum reading, which are sets of quantum channels obeying certain symmetries. We have determined upper bounds on the quantum reading capacity and the non-adaptive quantum reading capacity of an arbitrary memory cell. We have also derived strong converse and second-order  bounds on quantum reading capacities of environment-parametrized memory cells. We have calculated quantum reading capacities for a thermal memory cell, a qudit erasure memory cell, and a qudit depolarizing memory cell. Finally, we have shown the advantage of an adaptive strategy over a non-adaptive strategy in the context of zero-error quantum reading capacity of a memory cell.

We note that it is possible to use the methods developed here to obtain bounds on the quantum reading capacities of memory cells based on amplifying bosonic channels, in the same spirit as the results of a thermal memory cell (the argument follows from \cite{TW16}).

A natural question following from the developments in this chapter is whether there exists a memory cell for which the quantum reading capacity is larger than what we could achieve by using a non-adaptive strategy. As discussed above, we have found a positive answer to this question in the setting of zero error. However, the question remains open for the case of Shannon-theoretic capacity (i.e., with arbitrarily small error). We may suspect that this question will have a positive answer, and we may strongly suspect it will be the case in the setting of non-asymptotic capacity, our latter suspicion being due to the fact that feedback is known to help in non-asymptotic settings for communication (see, e.g., \cite{PPV11feedback}). We leave the investigation of this question for future work.

%% file: pr.tex
\chapter{Private Reading of Memory Devices}\label{ch:priv-read}
Devising a communication or information processing protocol that is secure against an eavesdropper is an area of primary interest and concern in information science and technology. In this chapter, we introduce the task of private reading of information stored in a memory device\blfootnote{Most of this chapter is based on \cite{DBW17}, a joint work with Stefan B\"auml and Mark M.~Wilde.}. A secret message can either be  encrypted in a computer program with circuit gates or in a physical storage device, such as a CD-ROM, DVD, etc. Here we limit the discussion to the case in which these computer programs or physical storage devices are used for read-only tasks; for simplicity, we refer to such media as memory devices. 

In a reading protocol (see Chapter~\ref{ch:read} for precise description), it is assumed that the reader has a description of a memory cell, which is a set of quantum channels. The memory cell is used to encode a classical message in a memory device. The memory device containing the encoded message is then delivered to the interested reader, whose task is to read out the message stored in it.
To decode the message, the reader can transmit a quantum state to the memory device and perform a quantum measurement on the output state. In general, since quantum channels are noisy, there is a loss of information to the environment, and there is a limitation on how well information can be read out from the memory device.

To motivate the task of private reading, consider that the computational and information processing capability of an adversary is limited only by the laws of quantum theory. A memory device is to be read using computer. There could be a circumstance in which an individual (reader) would have to access a computer in a public library under the surveillance of a librarian or other adversarial party, who supposedly is a passive eavesdropper, Eve. At a fundamental level, any reading mechanism involves transmitting of a probe system through a sequence of quantum channels, which are noisy in general. In such a situation, the reader would want information in a memory device not to be leaked to Eve, who has complete access to the environment, for security and privacy reasons. This naturally gives rise to the question of whether there exists a protocol for reading out a classical message that is secure from a passive eavesdropper.

In what follows, we introduce the details of private reading \cite{DBW17}: briefly, it is the task of reading out a classical message (key) stored in a memory device, encoded with a memory cell, by the reader such that the message is not leaked to Eve. We note here that private reading can be understood as a particular kind of secret-key-agreement protocol that employs a particular kind of bipartite interaction, and thus, there is a strong link between the developments in Section~\ref{sec:priv-key} and what follows. In Section~\ref{sec:p-r}, we present formal description of a private reading protocol, whose goal is to generate a secret key between an encoder and a reader. In Section~\ref{sec:n-a-priv-read-coherent}, we present purified (coherent) version of the private reading protocol. In both of the aforementioned sections, we derive both lower and upper bounds on the private reading capacities. In Section~\ref{sec:coh-read}, we discuss a protocol whose goal is to generate entanglement between two parties who have coherent access to a memory cell, and we derive a lower bound on the entanglement generation capacity in this setting.

\section{Private reading protocol}\label{sec:p-r}
In a private reading protocol, we consider an encoder and a reader (transceiver: receiver and decoder). Alice, an encoder, is one who encodes a secret classical message onto a read-only memory device that is delivered to Bob, a receiver, whose task is to read the message. Bob is also referred as the reader. The private reading task comprises the estimation of the secret message encoded in the form of a sequence of quantum wiretap channels chosen from a given set $\{\mc{M}^x_{B'\to BE}\}_{x\in\msc{X}}$ of quantum wiretap channels (called a \textit{wiretap memory cell}), where $\msc{X}$ is an alphabet of finite size $|X|$, such that there is negligible leakage of information to Eve, who has access to the system $E$. A special case of this is when each wiretap channel $\mc{M}^x_{B'\to BE}$ is an isometric channel. In the most natural and general setting, the reader can use an adaptive strategy when decoding, as considered in the reading protocol described in Chapter~\ref{ch:read}.  

Consider a set $\{\mc{M}^x_{B'\to BE}\}_{x\in\msc{X}}$ of wiretap quantum channels, where the size of $B'$, $B$, and $E$ are fixed and independent of $x$. The memory cell from the encoder Alice to the reader Bob is as follows: $\overline{\mathcal{M}}_{\msc{X}}=\{\mathcal{M}^x_{B'\to B}\}_{x}$, where
\begin{equation}
\forall x\in\msc{X}:\ \mathcal{M}^x_{B'\to B}(\cdot)\coloneqq \Tr_E\{\mc{M}^x_{B'\to BE}(\cdot)\},
\end{equation}
 which may also be known to Eve, before executing the reading protocol. It is assumed that only the systems $E$ are accessible to Eve for all channels $\mc{M}^{x}$ in a memory cell. Thus, Eve is a passive eavesdropper in the sense that  all she can do is to access the output of the channels 
 \begin{equation}
\forall x\in\msc{X}:\  \mc{M}^x_{B'\to E}(\cdot)=\Tr_{B}\left\{\mc{M}^x_{B'\to B E}(\cdot)\right\}.
 \end{equation}
 
Consider a finite classical message set $\msc{K}$ of size $|K|$, and let $K_A$ be an associated system denoting a classical register for the secret message. In general, Alice encodes a message $k\in\msc{K}$ using a codeword
$x^n(k)=x_1(k)x_2(k)\cdots x_n(k)$
of length $n$, where $x_i(k)\in\msc{X}$ for all $i\in[n]$. Each codeword identifies with a corresponding sequence of quantum channels chosen from the wiretap memory cell $\overline{\mathcal{M}}_{\msc{X}}$:
\begin{equation} %\mathcal{M}^{x^n(k)}_{{B'}^n\to B^n}=
\(\mathcal{M}^{x_1(k)}_{{B}_1'\to B_1 E_1}, \mathcal{M}^{x_2(k)}_{B_2'\to B_2 E_2},\ldots,\mathcal{M}^{x_n(k)}_{B_n'\to B_n E_n}\).
\end{equation} 
Each quantum channel in a codeword, each of which represents one part of the stored information, is only read once.

\begin{figure}[h]
		\centering
		\includegraphics[scale=0.7]{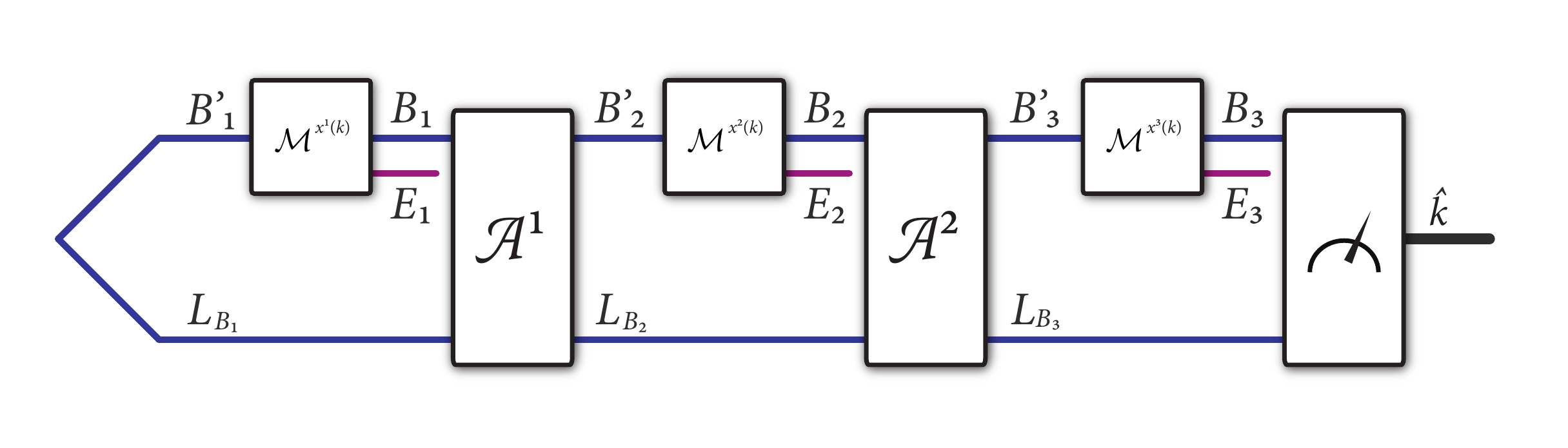}
		\caption{The figure depicts a private reading protocol that calls a memory cell three times to
decode the key $k$ as $\hat{k}$. See the discussion in Section~\ref{sec:p-r} for a detailed description of a private
reading protocol.}\label{fig:priv-read}
	\end{figure}

An adaptive decoding strategy makes $n$ calls to the memory cell, as depicted in Figure~\ref{fig:priv-read}. It is specified in terms of a transmitter state $\rho_{L_{B_1}B_1'}$, a set of adaptive, interleaved channels $\{\mc{A}^i_{L_{B_i}B_i\to L_{B_{i+1}}B'_{i+1}}\}_{i=1}^{n-1}$, and a final quantum measurement $\{\Lambda^{(\hat{k})}_{L_{B_n}B_n}\}_{\hat{k}}$ that outputs an estimate~$\hat{k}$ of the message~$k$. The strategy begins with Bob preparing the input state $\rho_{L_{B_1}B'_1}$ and sending the $B'_1$ system into the channel $\mc{M}^{x_1(k)}_{B'_1\to B_1 E_1}$. The channel outputs the system $B_1$ for Bob. He adjoins the system $B_1$ to the system $L_{B_1}$ and applies the channel $\mc{A}^1_{L_{B_1}B_1\to L_{B_2}B'_2}$. The channel $\mc{A}^i_{L_{B_i}B_i\to L_{B_{i+1}}B'_{i+1}}$ is called adaptive because it can take an action conditioned on the information in the system $B_i$, which itself might contain partial information about the message~$k$. Then, he sends the system $B'_2$ into the  channel $\mc{M}^{x_2(k)}_{B'_2\to B_2 E_2}$, which outputs systems $B_2$ and $E_2$. The process of successively using the channels interleaved by the adaptive channels  continues $n-2$ more times, which results in the final output systems $L_{B_n}$ and $B_n$ with Bob. Next, he performs a measurement $\{\Lambda^{(\hat{k})}_{L_{B_n}B_n}\}_{\hat{k}}$ on the output state $\rho_{L_{B_n} B_n}$, and the measurement outputs an estimate $\hat{k}$ of the original message $k$. It is natural to assume that the outputs of the adaptive channels and their complementary channels are inaccessible to Eve and are instead held securely by Bob. 

It is apparent that a non-adaptive strategy is a special case of an adaptive strategy. In a non-adaptive strategy, the reader does not perform any adaptive channels and instead uses $\rho_{L_B{B'}^n}$ as the transmitter state with each $B'_i$ system passing through the corresponding channel $\mc{M}^{x_i(k)}_{B'_i\to B_i E_i}$ and $L_B$ being a reference system. The final step in such a non-adaptive strategy is to perform a decoding measurement on the  joint system $L_BB^n$. 

As argued in the previous chapter, based on the physical setup of (quantum) reading, in which the reader assumes the role of both a transmitter and receiver, it is natural to consider the use of an adaptive strategy when defining the private reading capacity of a memory cell. 

\begin{definition}[Private reading protocol]\label{def:PR}
An $(n,P,\varepsilon,\delta)$ private reading protocol for a wiretap memory cell $\overline{\mc{M}}_{\msc{X}}$ is defined by an encoding map $\mc{K}_{\tn{enc}}\to \mc{X}^{\otimes n}$, an adaptive strategy with measurement $\{\Lambda_{L_{B_n} B_n}^{(\hat{k})}\}_{\hat{k}}$, such that, the average success probability is at least $1-\varepsilon$ where $\varepsilon\in(0,1)$:
\begin{equation}
1- \varepsilon \leq 1 - p_{\operatorname{err}} :=\frac{1}{|K|} \sum_{k}\Tr\left\{\Lambda^{(k)}_{L_{B_n}B_n}\rho^{(k)}_{L_{B_n}B_n}\right\},
\end{equation} 
where
\begin{equation}\label{eq3.3}
\rho^{(k)}_{L_{B_n}B_nE^n}=\(\mc{M}^{x_n(k)}_{B'_n\to B_nE_n}\circ\mc{A}^{{n-1}}_{L_{B_{n-1}}B_{n-1}\to L_{B_n}B'_n}\circ\cdots\circ\mc{A}^{1}_{L_{B_1}B_1\to L_{B_2}B'_2}\circ\mc{M}^{x_1(k)}_{B'_1\to B_1E_1}\)\(\rho_{L_{B_1}B'_1}\).
\end{equation}
Furthermore, the security condition is that
\begin{equation}
\frac{1}{|K|}\sum_{ k\in\msc{K}} \frac{1}{2}\left\|\rho^{(k)}_{E^n}- \tau_{E^n}\right\|_1\leq\delta,
\end{equation}
where $\rho^{(k)}_{E^n}$ denotes the state accessible to the passive eavesdropper when message $k$ is encoded. Also, $\tau_{E^n}$ is some fixed state. The rate $P\coloneqq\frac{1}{n}\log_2 |K|$ of a given $(n,|K|,\varepsilon,\delta)$ private reading protocol is equal to the number of secret bits read per channel use.
\end{definition}

Based on the discussions in \cite[Appendix~B]{WTB16}, there are connections between the notions of private communication given in Section~\ref{sec:priv-dist-protocol} and Definition~\ref{def:PR}, and we exploit these in what follows. 

To arrive at a definition of the private reading capacity, we demand that there exists a sequence of private reading protocols, indexed by $n$, for which the error probability $p_{\operatorname{err}}\to 0$ and security parameter $\delta\to 0$ as $n\to \infty$ at a fixed rate~$P$.

A rate $P$ is called achievable if for all $\varepsilon,\delta \in (0,1]$, $\delta^\prime >0$, and sufficiently large $n$, there exists an $(n,P-\delta^\prime,\varepsilon,\delta)$ private reading protocol. 
The private reading capacity $P^{\textnormal{read}}(\overline{\mc{M}}_{\msc{X}})$ of a wiretap memory cell $\overline{\mc{M}}_{\msc{X}}$ is defined as the supremum of all achievable rates $P$.

An $(n,P,\varepsilon,\delta)$ private reading protocol for a wiretap memory cell $\overline{\mc{M}}_{\msc{X}}$ is a non-adaptive private reading protocol when the reader abstains from employing any adaptive strategy for decoding. 
The non-adaptive private reading capacity $P^{\textnormal{read}}_{\textnormal{n-a}}(\overline{\mc{M}}_{\msc{X}})$ of a wiretap  memory cell $\overline{\mc{M}}_{\msc{X}}$ is defined as the supremum of all achievable rates $P$ for a private reading protocol that is limited to non-adaptive strategies.

\subsection{Non-adaptive private reading capacity}\label{sec:na-priv-read}

In what follows we restrict our attention to reading protocols that employ a non-adaptive strategy, and we now derive a regularized expression for the non-adaptive private reading capacity of a general wiretap memory cell.
%Furthermore, we assume $\mc{M}^x_{B'\to BE}$ to be an isometric channel denoted as $\mc{U}^{\mc{M}}_{B'\to BE}$ for all $x\in\mc{X}$. 
%This upper bound also holds for private reading capacities of environment-parametrized memory cells for which there is adaptive-to-non-adaptive reduction in the protocol. 

\begin{theorem}\label{thm:n-a-priv-read}
The non-adaptive private reading capacity of a wiretap memory cell $\overline{\mc{M}}_{\msc{X}}$ is given by
\begin{equation}
P^{\textnormal{read}}_{\textnormal{n-a}}\(\overline{\mc{M}}_{\msc{X}}\)=\sup_{n}\max_{p_{X^n},\sigma_{L_B{B'}^n}}\frac{1}{n}\[I(X^n;L_BB^n)_\tau-I(X^n;E^n)_\tau\],
\end{equation}
where
\begin{equation}\label{eq:cq-na-read}
\tau_{X^nL_BB^nE^n}\coloneqq\sum_{x^n}p_{X^n}(x^n)\op{x^n}_{X^n}\otimes \mc{M}^{x^n}_{{B'}^n\to B^nE^n}(\sigma_{L_B{B'}^n}),
\end{equation}
and it suffices for $\sigma_{L_B{B'}^n}$ to be  a pure state such that $L_B \simeq {B'}^n$.
\end{theorem}
\begin{proof}
Let us begin by defining a cq-state corresponding to the task of private reading. Consider a wiretap memory cell $\overline{\mc{M}}_{\msc{X}}=\{\mc{M}^x_{B'\to BE}\}_{x\in\msc{X}}$. The initial  state $\rho_{K_AL_B{B'}^n}$ of a non-adaptive private reading protocol takes  the form
\begin{equation}
\rho_{K_AL_B{B'}^n}\coloneqq\frac{1}{|K|}\sum_{k}\op{k}_{K_A}\otimes\rho_{L_B{B'}^n}.
\end{equation}
Bob then passes the transmitter state $\rho_{L_B{B'}^n}$ through a channel codeword sequence $\mathcal{M}^{x^n(k)}_{{B'}^n\to B^n E^n}\coloneqq \bigotimes_{i=1}^n\mc{M}^{x_i(k)}_{{B'}_i\to B_i E_i}$. 
Then the resulting state is
\begin{equation}
\rho_{K_AL_BB^nE^n}\coloneqq \frac{1}{|K|}\sum_{k}\op{k}_{K_A}\otimes\mc{M}^{x^n(k)}_{{B'}^n\to B^n E^n}\left(\rho_{L_B{B'}^n}\right).
\end{equation}
Let $\rho_{K_AK_B}\coloneqq \mc{D}_{L_BB^n\to K_B}\(\rho_{K_AL_BB^n}\)$ be the output state at the end of the protocol after the decoding channel $\mc{D}_{L_BB^n\to K_B}$ is performed by Bob. The privacy criterion (Definition~\ref{def:PR}) requires that
\begin{equation}
\frac{1}{|K|} \sum_{ k\in\msc{K}}  \frac{1}{2}\Vert \rho^{x^n(k)}_{E^n}-\tau_{E^n}\V_1\leq\delta,
\end{equation}
where 
$\rho^{x^n(k)}_{E^n}\coloneqq \operatorname{Tr}_{L_BB^n} \{{\mc{M}^{x^n(k)}}_{{B'}^n\to B^n E^n}\left(\rho_{L_B{B'}^n}\right)\}$ and 
$\tau_{E^n}$
is some arbitrary constant state. Hence
\begin{align}
\delta & \geq \frac{1}{2}\sum_{k}\frac{1}{|K|}\Vert \rho^{x^n(k)}_{E^n}-\tau_{E^n}\Vert_1 \\ 
& =\frac{1}{2}\Vert \rho_{K_AE^n}-\pi_{K_A}\otimes\tau_{E^n}\Vert_1,
\end{align}
where $\pi_{K_A}$ denotes maximally mixed state, i.e., $\pi_{K_A}\coloneqq\frac{1}{|K|}\sum_{k}\op{k}_{K_A}$. We note that
\begin{align}
I(K_A;E^n)_\rho &= S(K_A)_\rho-S(K_A|E^n)_\rho\\
&= S(K_A|E^n)_{\pi\otimes\tau}-S(K_A|E^n)_\rho\\
&\leq \delta\log_2 |K|+g(\delta),
\end{align}
which follows from an application of Lemma~\ref{thm:AFW}.

We are now ready to derive a weak converse bound on the private reading rate:
\begin{align}
\log_2 |K| &= S(K_A)_\rho=I(K_A;K_B)_\rho+S(K_A|K_B)_{\rho}\\
&\leq I(K_A;K_B)_\rho+\varepsilon\log_2|K|+h_2(\varepsilon)\\
&\leq I(K_A;L_BB^n)_\rho+\varepsilon\log_2|K|+h_2(\varepsilon)\\ 
&\leq I(K_A;L_BB^n)_\rho-I(K_A;E^n)_\rho+\varepsilon\log_2|K|+h_2(\varepsilon)+\delta\log_2|K|+g(\delta)\\ 
&\leq \max_{p_{X^n},\sigma_{L_B{B'}^n}\in\msc{D}(\mc{H}_{L_B{B'}^n})}\[I(X^n;L_BB^n)_\tau-I(X^n;E^n)_\tau\]+\varepsilon\log_2|K|+h_2(\varepsilon)+\delta\log_2|K|+g(\delta),
\end{align}
where $\tau_{X^nL_BB^nE^n}$ is a state of the form  in \eqref{eq:cq-na-read}.
The first inequality follows from Fano's inequality \cite{F08}. The second inequality follows from the monotonicity of mutual information under the action of a local quantum channel by Bob (Holevo bound). The final inequality follows because the maximization is over all possible probability distributions and input states. Then,
\begin{equation}
\frac{\log_2|K|}{n} (1-\varepsilon-\delta)\leq \max_{p_{X^n},\sigma_{L_B{B'}^n}}\frac{1}{n}\[I(X^n;L_BB^n)_\tau-I(X^n;E^n)_\tau\]+\frac{h_2(\varepsilon)+g(\delta)}{n}.
\end{equation}
Now considering a sequence of non-adaptive $(n,P,\varepsilon_n,\delta_n)$ protocols with $\lim_{n \to \infty} \frac{\log_2 K_n}{n} = P$,
$\lim_{n \to \infty} \varepsilon_n = 0$, and $\lim_{n \to \infty} \delta_n = 0$, 
the converse bound on non-adaptive private reading capacity of memory cell $\overline{\mc{M}}_{\msc{X}}$ is given by
\begin{equation}\label{eq:priv-na-rad-up}
P\leq\sup_{n}\max_{p_{X^n},\sigma_{L_B{B'}^n}}\frac{1}{n}\[I(X^n;L_BB^n)_\tau-I(X^n;E^n)_\tau\],
\end{equation} 
which follows by taking the limit as $n\to \infty$.

It follows from the results of \cite{D05,DW05} that right-hand side of \eqref{eq:priv-na-rad-up} is also an achievable rate in the limit $n\to\infty$. Indeed, the encoder and reader can induce the cq-wiretap channel $x \to \mc{M}^{x}_{{B'}\to B E}(\sigma_{L_BB'})$, to which the results of \cite{D05,DW05} apply. A regularized coding strategy then gives the general achievability statement. Therefore, the non-adaptive private reading capacity is given as stated in the theorem.
\end{proof}

\section{Purifying private reading protocols}\label{sec:n-a-priv-read-coherent}

As observed in \cite{HHHO05,HHHO09} and reviewed in Section~\ref{sec:rev-priv-states}, any protocol of the above form (see Section~\ref{sec:na-priv-read}) can be purified in the following sense. In this section, we assume that each wiretap memory cell consists of a set of isometric channels, written as
$\{\mc{U}^{\mc{M}^x}_{B'\to BE}\}_{x\in\msc{X}}$. 
Thus, Eve has access to system $E$, which is the output of a particular isometric extension of the channel $\mc{M}^x_{B'\to B}$, i.e., $\widehat{\mc{M}}^x_{B'\to E}(\cdot) = 
\Tr_B\{\mc{U}^{\mc{M}^x}_{B'\to BE}(\cdot)\}$, for all $x\in\msc{X}$. Such memory cell is to be referred as an \textit{isometric wiretap memory cell}.

We begin by considering non-adaptive private reading protocols. A non-adaptive purified secret-key-agreement protocol that uses an isometric wiretap memory cell begins with Alice preparing a purification of the maximally classically correlated state: 
\begin{equation}
\frac{1}{\sqrt{|K|}}\sum_{k\in\msc{K}}\ket{k}_{K_A}\ket{k}_{\hat{K}}\ket{k}_{C},
\end{equation}
where $\msc{K}$ is a finite classical message set of size $|K|$, and $K_A$, $\hat{K}$, and $C$ are classical registers. Alice coherently encodes the value of the register $C$ using the memory cell, the codebook $\{x^n(k)\}_k$,
and the isometric mapping $\ket{k}_C \to \ket{x^n(k)}_{X^n}$. Alice makes two coherent copies of the codeword $x^n(k)$ and stores them safely in coherent  classical registers $X^n$ and $\hat{X}^n$. At the same time, she acts on Bob's input state $\rho_{L_B{B'}^n}$ with the following isometry:
\begin{equation}\label{eq:iso-v-m}
 \sum_{x^n}\op{x^n}_{X^n}\otimes{U}^{\mc{M}^{x^n}}_{{B'}^n\to B^nE^n}\otimes\ket{x^n}_{\hat{X}^n}.
\end{equation} 
For the task of reading, Bob inputs the state $\rho_{L_B{B'}^n}$ to the channel sequence $\mc{M}^{x^n(k)}$, with the goal of decoding $k$. In the purified setting, the resulting output state is $\psi_{K_A\hat{K}X^nL_B'L_BB^nE^n\hat{X}^n}$, which  includes all concerned coherent classical registers or quantum systems accessible by Alice, Bob and Eve:
\begin{equation}\label{eq:coh-mem-state}
\ket{\psi}_{K_A\hat{K}X^nL_B'L_BB^nE^n\hat{X}^n}\coloneqq \frac{1}{\sqrt{|K|}}\sum_{k}\ket{k}_{K_A}\ket{k}_{\hat{K}}\ket{x^n(k)}_{X^n}U^{\mc{M}^{x^n}}_{{B'}^n\to B^nE^n}\ket{\psi}_{L_B'L_B{B'}^n}\ket{x^n(k)}_{\hat{X}^n},
\end{equation}
where $\psi_{L_B'L_B{B'}^n}$ is a purification of $\rho_{L_B{B'}^n}$ and the systems $L_B'$, $L_B$, and $B^n$ are held by Bob, whereas Eve has access only to $E^n$. The final global state is $\psi_{K_A\hat{K}X^nL_B'K_BE^n\hat{X}^n}$  after Bob applies the decoding channel $\mc{D}_{L_BB^n\to K_B}$, where 
\begin{equation}
\ket{\psi}_{K_A\hat{K}X^nL_B'L_B''K_BE^n\hat{X}^n}\coloneqq U^{\mc{D}}_{L_BB^n\to L_B'' K_B}\ket{\psi}_{K_A\hat{K}X^nL_B'L_BB^nE^n\hat{X}^n},
\end{equation}
$U^{\mc{D}}$ is an isometric extension of the decoding channel $\mc{D}$, and $L_B''$ is part of the shield system of Bob. 

At the end of the purified protocol, Alice possesses the key system $K_A$ and the shield systems $\hat{K}X^n\hat{X}^n$, Bob possesses the key system $K_B$ and the shield systems $L_B'L_B''$, and Eve possesses the environment system $E^n$. The state $\psi_{K_A\hat{K}X^nL_B'L_B''K_B\hat{X}^nE^n}$ at the end of the protocol is a pure state.

For a fixed $n,\ |K|\in\mathbb{N},\ \varepsilon\in[0,1]$, the original protocol is an $(n,P,\sqrt{\varepsilon},\sqrt{\varepsilon})$ private reading protocol if the memory cell is called $n$ times as discussed above, where private reading rate $P\coloneqq \frac{1}{n}\log_2|K|$, and if 
\begin{equation}\label{eq:coh-key-aprox-priv}
F(\psi_{K_A\hat{K}X^nL_B'L_B''K_B\hat{X}^n},\gamma_{S_AK_AK_BS_B})\geq 1-\varepsilon,
\end{equation}
where $\gamma$ is a private state such that $S_A=\hat{K}X^n\hat{X}^n,\ K_A=K_A, \ K_B=K_B,\ S_B=L_B'L_B''$. See \cite[Appendix~B]{WTB16} for further details.

Similarly, it is possible to purify a general adaptive private reading protocol, but we omit the details.

\subsection{Converse bounds on private reading capacities}\label{sec:priv-read-sc}

In this section, we derive different upper bounds on the private reading capacity of an isometric wiretap memory cell. 
The first is a weak converse upper bound on the non-adaptive private reading capacity in terms of the squashed entanglement. The second is a strong converse upper bound on the (adaptive) private reading capacity in terms of the bidirectional max-relative entropy of entanglement. Finally, we evaluate the private reading capacity for an example:  a qudit erasure memory cell.  

We derive the first converse bound on non-adaptive private reading capacity by making the following observation, related to the development in \cite[Appendix~B]{WTB16}: any non-adaptive $(n,P,\varepsilon,\delta)$ private reading protocol of an isometric wiretap memory cell $\overline{\mc{M}}_{\msc{X}}$, for reading out a secret key, can be realized by  an $(n,P,\varepsilon'(2-\varepsilon'))$ non-adaptive purified secret-key-agreement reading protocol, where $\varepsilon'\coloneqq\varepsilon+2\delta$. As such, a converse bound for the latter protocol implies a converse bound for the former.

First, we derive an upper bound on the non-adaptive private reading capacity in terms of the squashed entanglement \cite{CW04}: 
\begin{proposition}
The non-adaptive private reading capacity $P^{\textnormal{read}}_{\textnormal{n-a}}(\overline{\mc{M}}_{\msc{X}})$ of an isometric wiretap memory cell $\overline{\mc{M}}_{\msc{X}}=\{\mc{U}^{\mc{M}^x}_{B'\to BE}\}_{x\in\msc{X}}$ is bounded from above as
\begin{equation}
P^{\textnormal{read}}_{\textnormal{n-a}}(\overline{\mc{M}}_{\msc{X}})\leq \sup_{p_X,\psi_{LB'}}E_{\sq}(XL_B;B)_\omega,
\end{equation} 
where $\omega_{XL_BB}=\Tr_{E}\{\omega_{XL_BBE}\}$, such that $\psi_{L_BB'}$ is a pure state and
\begin{equation}
\vert \omega\rangle_{XLBE}=\sum_{x\in\msc{X}}\sqrt{p_X(x)}|x\>_X\otimes{U}^{\mc{M}^{x}}_{B'\to BE}\ket{\psi}_{L_BB'} \label{eq:opt-form-sq-bnd}.
\end{equation}
\end{proposition}

\begin{proof}
For the discussed purified non-adaptive secret-key-agreement reading protocol, when \eqref{eq:coh-key-aprox-priv} holds, the dimension of the secret key system is upper bounded as \cite[Theorem 2]{Wil16}:
\begin{equation}
\log_2 |K|\leq E_{\sq}(\hat{K}X^n\hat{X}^nK_A;K_BL_BL_B'')_{\psi}+f_1(\sqrt{\varepsilon},|K|),
\end{equation}
where 
\begin{equation}
f_1(\varepsilon,|K|)\coloneqq 2\varepsilon\log_2|K|+2g(\varepsilon).
\end{equation}
We can then proceed as follows:
\begin{align}
\log_2 |K| &\leq E_{\sq}(\hat{K}X^n\hat{X}^nK_A;K_BL_B''L_B')_{\psi}+f_1(\sqrt{\varepsilon},|K|)\\
&= E_{\sq}(\hat{K}X^n\hat{X}^nK_A;B^nL_BL_B')_{\psi}+f_1(\sqrt{\varepsilon},|K|).\label{eq:priv-sq-bound}
\end{align}
where the first equality is due to the invariance of $E_{\sq}$ under isometries. 

For any five-partite pure state $\phi_{B'B_1B_2E_1E_2}$, the following inequality holds \cite[Theorem~7]{TGW14}:
\begin{equation}\label{eq:tri-sq-ent}
E_{\sq}(B';B_1B_2)_\phi\leq E_{\sq}(B'B_2E_2;B_1)_{\phi}+E_{\sq}(B'B_1E_1;B_2)_{\phi}.
\end{equation}
This implies that
\begin{align}
& E_{\sq}(\hat{K}X^n\hat{X}^nK_A;B^nL_BL_B')_{\psi}\nonumber\\
&\leq E_{\sq}(\hat{K}X^n\hat{X}^{n}K_AL_BL_B'B^{n-1}E^{n-1};B_n)_{\psi}+E_{\sq}(\hat{K}X^n\hat{X}^{n}K_AB_nE_n;L_BL_B'B^{n-1})_{\psi}\\
&= E_{\sq}(\hat{K}X^n\hat{X}^nK_AL_BL_B'B^{n-1}E^{n-1};B_n)_{\psi}+E_{\sq}(\hat{K}X^n\hat{X}^{n-1}K_AB'_n;L_BL_B'B^{n-1})_{\psi}.\label{eq:sq-ine-n-pre}
\end{align}
where the equality holds by considering an isometry with the following uncomputing action:
\begin{align}
&\ket{k}_{K_A}\ket{k}_{\hat{K}}\ket{x^n(k)}_{X^n}U^{\mc{M}^{x^n}}_{{B'}^n\to B^nE^n}\ket{\psi}_{L_B'L_B{B'}^n}\ket{x^n(k)}_{\hat{X}^n} \nonumber \\
&\quad\quad\quad \to \ket{k}_{K_A}\ket{k}_{\hat{K}}\ket{x^n(k)}_{X^n}U^{\mc{M}^{x^{n-1}}}_{{B'}^{n-1}\to B^{n-1}E^{n-1}}\ket{\psi}_{L_B'L_B{B'}^n}\ket{x^{n-1}(k)}_{\hat{X}^{n-1}}.
\end{align}

 Applying the inequality in \eqref{eq:tri-sq-ent} and uncomputing isometries like the above repeatedly to \eqref{eq:sq-ine-n-pre}, we get 
\begin{equation}\label{eq:single-let-sq}
E_{\sq}(\hat{K}X^n\hat{X}^nK_A;B^nL_BL_B')_{\psi}\leq \sum_{i=1}^n E_{\sq}(\hat{K}X^n\hat{X}_iK_AL_BL_B'B'^{n\setminus \{i\}};B_i),
\end{equation}
where the notation ${B'}^{n\setminus\{i\}}$ indicates the composite system $B'_1B'_2\cdots B'_{i-1}B'_{i+1}\cdots B'_n$, i.e. all $n-1$\ $B'$-labeled systems except $B'_i$.
Each summand above is equal to the squashed entanglement of some state of the following form: a bipartite state is prepared on some auxiliary system $Z$ and a control system $X$, a bipartite state is prepared on systems $L_B$ and $B'$, a controlled isometry $\sum_x \op{x}_X \otimes {U}^{\mc{M}^{x}}_{B'\to BE}$ is performed from $X$ to $B'$, and then $E$ is traced out. By applying the development in \cite[Appendix~A]{CY16}, we conclude that the auxiliary system $Z$ is not necessary. Thus, the state of systems $X$, $L_B$, $B'$, and $E$ can be taken to have the form in \eqref{eq:opt-form-sq-bnd}. From \eqref{eq:priv-sq-bound} and the above reasoning, since $\lim_{\varepsilon\to 0}\lim_{n\to \infty} \frac{f_1(\sqrt{\varepsilon},|K|)}{n}=0$, we can conclude that 
\begin{equation}
\widetilde{P}^{\textnormal{read}}_{\textnormal{n-a}}(\overline{\mc{M}}_{\msc{X}})\leq \sup_{p_X,\psi_{L_BB'}}E_{\sq}(XL;B)_\omega,
\end{equation} 
where $\omega_{XL_BB}=\Tr_{E}\{\omega_{XL_BBE}\}$, such that $\psi_{L_BB'}$ is a pure state and
\begin{equation}
\vert \omega\rangle_{XL_BBE}=\sum_{x\in\msc{X}}\sqrt{p_X(x)}|x\>_X\otimes{U}^{\mc{M}^{x}}_{B'\to BE}\ket{\psi}_{L_BB'}.
\end{equation}
This concludes the proof.
\end{proof}
\bigskip

We now bound the strong converse private reading capacity of an isometric wiretap memory cell in terms of the bidirectional max-relative entropy (see Chapter~\ref{ch:bqi}).
\begin{theorem}\label{thm:priv-read-strong-converse}
The strong converse private reading capacity $\widetilde{P}^{\textnormal{read}}(\overline{\mc{M}}_{\msc{X}})$ of an isometric wiretap memory cell $\overline{\mc{M}}_{\msc{X}}=\{\mc{U}^{\mc{M}^x}_{B'\to BE}\}_{x\in\msc{X}}$ is  bounded from above by the bidirectional max-relative entropy of entanglement $E^{2\to 2}_{\max}(\mc{N}^{\overline{\mc{M}}_{\msc{X}}}_{X'B'\to XB})$ of the bidirectional channel $\mc{N}^{\overline{\mc{M}}_{\msc{X}}}_{X'B'\to XB}$, i.e.,
\begin{equation}\label{eq:priv-read-strong-converse}
\widetilde{P}^{\textnormal{read}}(\overline{\mc{M}}_{\msc{X}})\leq E^{2\to 2}_{\max}(\mc{N}^{\overline{\mc{M}}_{\msc{X}}}_{XB'\to XB}),
\end{equation}
where 
\begin{equation}\label{eq:memory-cell-bidir}
\mc{N}^{\overline{\mc{M}}_{\msc{X}}}_{XB'\to XB}(\cdot)\coloneqq \Tr_{E}\left\{U^{\overline{\mc{M}}_{\msc{X}}}_{XB'\to XBE}(\cdot)\(U^{\overline{\mc{M}}_{\msc{X}}}_{XB'\to XBE}\)^\dag\right\},
\end{equation}
such that
\begin{equation}
U^{\overline{\mc{M}}_{\msc{X}}}_{XB'\to XBE}\coloneqq \sum_{x\in\msc{X}}\op{x}_{X}\otimes U^{\mc{M}^x}_{B'\to BE}.
\end{equation}
\end{theorem}

\begin{proof}
First we recall, as stated previously,
that a $(n,P,\varepsilon,\delta)$ (adaptive) private reading protocol of a memory cell $\overline{\mc{M}}_{\msc{X}}$, for reading out a secret key, can be realized by an $(n,P,\varepsilon'(2-\varepsilon'))$  purified secret-key-agreement reading protocol, where $\varepsilon'\coloneqq \varepsilon+2 \delta$.
Given that a purified secret-key-agreement reading protocol can be understood as particular case of a bidirectional secret-key-agreement protocol (as discussed in Section~\ref{sec:priv-dist-protocol}), we can conclude that the strong converse private reading capacity is bounded from above by 
\begin{equation}
\widetilde{P}^{\textnormal{read}}_{\textnormal{n-a}}(\overline{\mc{M}}_{\msc{X}}) \leq E^{2\to 2}_{\max}(\mc{N}^{\overline{\mc{M}}_{\msc{X}}}_{XB'\to XB}),
\end{equation}
where the bidirectional channel is
\begin{equation}
\mc{N}^{\overline{\mc{M}}_{\msc{X}}}_{XB'\to XB}(\cdot)=\Tr_{E}\left\{U^{\overline{\mc{M}}_{\msc{X}}}_{XB'\to XBE}(\cdot)\(U^{\overline{\mc{M}}_{\msc{X}}}_{XB'\to XBE}\)^\dag\right\},
\end{equation}
such that
\begin{equation}
U^{\overline{\mc{M}}_{\msc{X}}}_{XB'\to XBE}\coloneqq \sum_{x\in\msc{X}}\op{x}_{X}\otimes U^{\mc{M}^x}_{B'\to BE}. 
\end{equation}
The reading protocol is a particular instance of an LOCC-assisted bidirectional secret-key-agreement protocol in which classical communication between Alice and Bob does not occur. The local operations of Bob in the bidirectional secret-key-agreement protocol are equivalent to adaptive operations by Bob in reading. Therefore, applying Theorem~\ref{thm:emax-ent-dist-strong-converse}, we find that \eqref{eq:priv-read-strong-converse} holds, where the strong converse in this context means that $\varepsilon+2 \delta \to 1$ in the limit as $n\to \infty$ if the reading rate exceeds $E^{2\to 2}_{\max}(\mc{N}^{\overline{\mc{M}}_{\msc{X}}}_{XB'\to XB})$.\footnote{Such a bound might be called a ``pretty strong converse,'' in the sense of \cite{MW13}. However, we could have alternatively defined a private reading protocol to have a single parameter characterizing reliability and security, as in \cite{WTB16}, and with such a definition, we would get a true strong converse.}
\end{proof}

\subsubsection{Qudit erasure wiretap memory cell}

The main goal of this section is to evaluate the private reading capacity of the qudit erasure wiretap memory cell (cf.~Definition~\ref{def:qudit-erasure}). 

\begin{definition}[Qudit erasure wiretap memory cell]\label{def:qudit-erasure-pr}
The qudit erasure wiretap memory cell $\overline{\mc{Q}}^q_{\msc{X}}=\left\{\mc{Q}^{q,x}_{B'\to BE}\right\}_{x\in\msc{X}}$, where $\msc{X}$ is of size $|X|=d^2$, consists of the following qudit channels:
\begin{equation}
\mc{Q}^{q,x}(\cdot)=\mc{Q}^q(\sigma^x(\cdot)\(\sigma^x\)^\dag) ,
\end{equation}
where $\mc{Q}^q$ is an isometric channel extending the qudit erasure channel \cite{GBP97}:
\begin{align}
\mc{Q}^q(\rho_{B'})& =U^{q} \rho_{B'} (U^q)^\dag,\\
U^q \vert \psi\rangle_{B'}
&= \sqrt{1-q}\vert \psi \rangle_B \vert e \>_E + \sqrt{q}|e\>_B \vert \psi\>_E,
\end{align}
such that $q\in [0,1]$, $\dim(\mc{H}_{B'})=d$, $\op{e}$ is an erasure state orthogonal to the support of all possible input states $\rho$, and
$\forall x\in\msc{X}: \sigma^x\in\mathbf{H}$ are the Heisenberg--Weyl operators as reviewed in \eqref{eq:HW-op}. Observe that $\mc{Q}^q_{\msc{X}}$ is jointly covariant with respect to the Heisenberg--Weyl group $\mathbf{H}$ because the qudit erasure channel $\msc{Q}^q$ is covariant with respect to $\mathbf{H}$.
\end{definition}

Now we establish the private reading capacity of the qudit erasure wiretap memory cell. 
 
\begin{proposition}
The private reading capacity and strong converse private reading capacity of the qudit erasure wiretap memory cell $\overline{\mc{Q}}^q_{\msc{X}}$ are given by 
\begin{equation}
P^{\operatorname{read}}(\overline{\mc{Q}}^q_{\msc{X}})=
\widetilde{P}^{\operatorname{read}}(\overline{\mc{Q}}^q_{\msc{X}})=
2(1-q)\log_2 d.
\end{equation}
\end{proposition}

\begin{proof}
To prove the proposition, let us consider that
$\mc{N}^{\overline{\mc{Q}}^q_{\msc{X}}}$ as defined in \eqref{eq:memory-cell-bidir} is bicovariant and $\mc{Q}^q_{B'\to B}$ is covariant. Thus, to get an upper bound on the strong converse private reading capacity, it is sufficient to consider the action of a coherent use of the memory cell on a maximally entangled state (see Corollary~\ref{cor:str-conv-TP-simul}). 
We furthermore apply the development in \cite[Appendix~A]{CY16} to restrict to the following state:
\begin{align}
\phi_{XL_BBE}&\coloneqq\frac{1}{\sqrt{|X|}}\sum_{x\in\msc{X}}\ket{x}_{X}\otimes U^{\mc{Q}^{q,x}}_{B'\to BE}\ket{\Phi}_{L_B{B'}}\nonumber\\
&=\sqrt{\frac{1-q}{d |X|}}\sum_{i=0}^d\sum_{x}\ket{x}_{X}\otimes\sigma^x\ket{i}_B\ket{i}_{L_B}\ket{e}_{E} +\sqrt{\frac{q}{d |X|}}\sum_{i=0}^d\sum_{x}\ket{x}_{X}\otimes\ket{e}_B\ket{i}_{L_B}\otimes \sigma^x\ket{i}_{E}.
\end{align}

Observe that
$\sum_{i=0}^{d-1}\sum_{x}\ket{x}_{X}\otimes\ket{e}_B\ket{i}_{L_B}\otimes \sigma^x\ket{i}_{E}$ and $\sum_{i=0}^{d-1}\sum_{x}\ket{x}_{X}\otimes\sigma^x\ket{i}_B\ket{i}_{L_B}\ket{e}_{E}$ are orthogonal. Also, since, $\ket{e}$ is orthogonal to the input Hilbert space, the only term contributing to the relative entropy of entanglement is $\sqrt{1-q}\frac{1}{d}\sum_{i=0}^d\sum_{x}\ket{x}_{X}\otimes\sigma^x\ket{i}_B\ket{i}_{L_B}$. Let
\begin{equation}
\ket{\psi}_{XL_BB}=\frac{1}{\sqrt{|X|}}\sum_{x=0}^{d^2-1}\ket{x}_{X}\otimes\sigma^x\ket{\Phi}_{BL_B}.
\end{equation}
$\{\sigma^x\ket{\Phi}_{BL_B}\}_{x\in\msc{X}}\in\ONB(\mc{H}_{B}\otimes\mc{H}_{L_B})$ (see Appendix~\ref{app:qudit}), so 
\begin{equation}
\ket{\psi}_{XL_BB}=\ket{\Phi}_{X:BL_B}=\frac{1}{d}\sum_{x=0}^{d^2-1}\ket{x}_{X}\otimes\ket{x}_{BL_B},
\end{equation} 
and $E(X;LB)_\Phi =2\log_2 d$. Applying Corollary~\ref{cor:str-conv-TP-simul} and convexity of relative entropy of entanglement, we can conclude that
\begin{equation}\label{eq:e-cell-up}
\widetilde{P}^{\text{read}}(\overline{\mc{Q}}^{q}_{\msc{X}})\leq 2(1-q)\log_2 d.
\end{equation}
From Theorem~\ref{thm:n-a-priv-read}, the following bound holds
\begin{align}
P^{\text{read}}(\overline{\mc{Q}}^q_{\msc{X}}) &\geq P^{\text{read}}_{\text{n-a}}(\overline{\mc{Q}}^q_{\msc{X}})\label{eq:e-cell-down}\\
&\geq I(X;L_BB)_{\rho}-I(X;E)_{\rho},
\end{align}
where 
\begin{equation}
\rho_{XL_BBE}=\frac{1}{d^2}\sum_{x=0}^{d^2-1}\op{x}_X\otimes \mc{U}^{\mc{Q}^{q,x}}_{B'\to BE}(\Phi_{X:L_BB'}).
\end{equation}
After a calculation, we find that $I(X;E)_{\rho}=0$ and $I(X;L_BB)_{\rho}=2(1-q)\log_2 d$. Therefore, from \eqref{eq:e-cell-up} and the above, the statement of the theorem is concluded.
\end{proof}

\bigskip
From the above and Corollary~\ref{cor:e-cell-r}, we can conclude that there is no difference between the private reading capacity of the qudit erasure memory cell and its reading capacity.

\section{Entanglement generation from a coherent memory cell}\label{sec:coh-read}

In this section, we consider an entanglement distillation task between two parties Alice and Bob holding systems $X$ and $B$, respectively. The set up is similar to purified secret key generation when using a memory cell (see Section~\ref{sec:n-a-priv-read-coherent}). The goal of the protocol is as follows: Alice and Bob, who are spatially separated, try to generate a maximally entangled state between them by making coherent use of an isometric wiretap memory cell $\overline{\mc{M}}_{\msc{X}}=\{\mc{U}^{\mc{M}^x}_{B'\to BE}\}_{x\in\msc{X}}$ known to both parties. That is, Alice and Bob have access to the following controlled isometry:
\begin{equation}
U^{\overline{\mc{M}}_{\msc{X}}}_{XB'\to XBE}\coloneqq \sum_{x\in\msc{X}}\op{x}_{X}\otimes U^{\mc{M}^x}_{B'\to BE},\label{eq:contrl-iso-coherent-mem-cell}
\end{equation}
such that $X$ and $E$ are inaccessible to Bob. Using techniques from \cite{DW05}, we can state an achievable rate of entanglement generation by coherently using the memory cell.  

\begin{theorem}
The following
rate is achievable for entanglement generation when using the controlled isometry in \eqref{eq:contrl-iso-coherent-mem-cell}:
\begin{equation}
I(X\rangle L_BB)_{\omega},
\end{equation}
where $I(X\rangle L_BB)_{\omega}$ is the coherent information of state $\omega_{XL_BB}$ \eqref{eq:coh-info} such that 
\begin{equation}
\ket{\omega}_{XL_BBE}=\sum_{x\in\msc{X}}\sqrt{p_{X}(x)}|x\rangle_{X}\otimes U^{\mc{M}^x}_{B'\rightarrow
BE}|\psi\rangle_{L_BB'}.
\end{equation}
\end{theorem}

\begin{proof}
Let $\{x^{n}(m,k)\}_{m,k}$ denote a codebook for private reading, as discussed in Section~\ref{sec:na-priv-read}, and let
$\psi_{L_BB'}$ denote a pure state that can be fed in to each coherent use of the memory cell. The
codebook is such that for each $m\in\msc{M}$ and $k\in\msc{K}$, the codeword $x^{n}(m,k)$ is
unique. The rate of private reading is given by
\begin{equation}
I(X;L_BB)_{\rho}-I(X;E)_{\rho},
\end{equation}
where%
\begin{equation}
\rho_{XB'BE}=\sum_{x}p_{X}(x)\op{x}_{X}\otimes\mathcal{U}
_{B'\rightarrow BE}^{\mc{M}^x}(\psi_{L_BB'}).
\end{equation}
Note that the following equality holds%
\begin{equation}
I(X;L_BB)_{\rho}-I(X;E)_{\rho}=I(X\rangle L_BB)_{\omega},
\end{equation}
where%
\begin{equation}
\ket{\omega}_{XL_BBE}=\sum_{x}\sqrt{p_{X}(x)}|x\rangle_{X}\otimes U_{B'\rightarrow
BE}^{\mc{M}^x}|\psi\rangle_{L_BB'}. 
\end{equation} 
The code is such that there is a measurement $\Lambda
_{L_B^{n}B^{n}}^{m,k}$ for all $m,k$, for which%
\begin{equation}
\operatorname{Tr}\{\Lambda_{L_B^{n}B^{n}}^{m,k}\mathcal{M}_{{B'}^{n}\rightarrow
B^{n}}^{x^{n}(m,k)}(\psi_{L_BB'}^{\otimes n})\}\geq1-\varepsilon,
\end{equation}
and%
\begin{equation}
\frac{1}{2}\left\Vert \frac{1}{|K|}\sum_{k}\widehat{\mathcal{M}}_{{B'}^{n}
\rightarrow E^{n}}^{x^{n}(m,k)}(\psi_{B'}^{\otimes n})-\sigma_{E^{n}%
}\right\Vert _{1}\leq\delta.\label{eq:security-condition}
\end{equation}

From this private reading code, we construct a coherent reading code as
follows. Alice begins by preparing the state%
\begin{equation}
\frac{1}{\sqrt{|M||K|}}\sum_{m,k}|m\rangle_{M_A}|k\rangle_{K_A}.
\end{equation}
Alice performs a unitary that implements the following mapping:%
\begin{equation}
|m\rangle_{M_A}|k\rangle_{K_A}|0\rangle_{X^{n}}\rightarrow|m\rangle_{M_A}%
|k\rangle_{K_A}|x^{n}(m,k)\rangle_{X^{n}},\label{eq:encoding-unitary}%
\end{equation}
so that the state above becomes%
\begin{equation}
\frac{1}{\sqrt{|M||K|}}\sum_{m,k}|m\rangle_{M_A}|k\rangle_{K_A}|x^{n}(m,k)\rangle
_{X^{n}}.
\end{equation}
Bob prepares the state $|\psi\rangle_{L_BB'}^{\otimes n}$, so that the overall
state is%
\begin{equation}
\frac{1}{\sqrt{|M|K|}}\sum_{m,k}|m\rangle_{M_A}|k\rangle_{K_A}|x^{n}(m,k)\rangle
_{X^{n}}|\psi\rangle_{L_BB'}^{\otimes n}.
\end{equation}
Now Alice and Bob are allowed to access $n$ instances of the controlled
isometry%
\begin{equation}
\sum_{x}|x\rangle\langle x|_{X}\otimes U_{B'\rightarrow BE}^{\mc{M}^x}%
,\label{eq:controlled-isometry-coherent-memory-cell}%
\end{equation}
and the state becomes
\begin{equation}
\frac{1}{\sqrt{|M||K|}}\sum_{m,k}|m\rangle_{M_A}|k\rangle_{K_A}|x^{n}(m,k)\rangle
_{X^{n}}U_{{B'}^{n}\rightarrow B^{n}E^{n}}^{\mc{M}^{x^{n}(m,k)}}|\psi\rangle_{L_BB'}^{\otimes
n}.
\end{equation}

Bob now performs the isometry%
\begin{equation}
\sum_{m,k}\sqrt{\Lambda_{L_B^{n}B^{n}}^{m,k}}\otimes|m\rangle_{M_{1}}%
|k\rangle_{K_{1}},
\end{equation}
and the resulting state is close to%
\begin{equation}
\frac{1}{\sqrt{|M||K|}}\sum_{m,k}|m\rangle_{M_A}|k\rangle_{K_A}|x^{n}(m,k)\rangle
_{X^{n}}U_{{B'}^{n}\rightarrow B^{n}E^{n}}^{x^{n}(m,k)}|\psi\rangle_{L_BB'}^{\otimes
n}|m\rangle_{M_{1}}|k\rangle_{K_{1}}.
\end{equation}
At this point, Alice locally uncomputes the unitary from
\eqref{eq:encoding-unitary} and discards the $X^{n}$ register, leaving the
following state:%
\begin{equation}
\frac{1}{\sqrt{|M||K|}}\sum_{m,k}|m\rangle_{M_A}|k\rangle_{K_A}U_{{B'}^{n}\rightarrow
B^{n}E^{n}}^{\mc{M_A}^{x^{n}(m,k)}}|\psi\rangle_{L_BB'}^{\otimes n}|m\rangle_{M_{1}%
}|k\rangle_{K_{1}}.
\end{equation}
Following the scheme of \cite{DW05} for entanglement distillation, she then performs a Fourier transform on
the register $K_A$ and measures it, obtaining an outcome $k^{\prime}%
\in\{0,\ldots,K-1\}$, leaving the following state:%
\begin{equation}
\frac{1}{\sqrt{|M||K|}}\sum_{m,k}e^{2\pi ik^{\prime}k/K}|m\rangle_{M_A}%
U_{{B'}^{n}\rightarrow B^{n}E^{n}}^{\mc{M_A}^{x^{n}(m,k)}}|\psi\rangle_{L_BB'}^{\otimes
n}|m\rangle_{M_{1}}|k\rangle_{K_{1}}.
\end{equation}
She communicates the outcome to Bob, who can then perform a local unitary on system $K_1$ to
bring the state to%
\begin{equation}
\frac{1}{\sqrt{|M||K|}}\sum_{m,k}|m\rangle_{M_A}U_{{B'}^{n}\rightarrow B^{n}E^{n}%
}^{\mc{M}^{x^{n}(m,k)}}|\psi\rangle_{L_BB'}^{\otimes n}|m\rangle_{M_{1}}|k\rangle_{K_{1}}.
\end{equation}
Now consider that, conditioned on a value $m$ in register $M$, the local state
of Eve's register $E^n$ is given by%
\begin{equation}
\frac{1}{|K|}\sum_{k}\widehat{\mathcal{M}}_{{B'}^{n}\rightarrow E^{n}}^{x^{n}%
(m,k)}(\psi_{B'}^{\otimes n}).
\end{equation}
Thus, by invoking the security condition in \eqref{eq:security-condition} and
Uhlmann's theorem \cite{Uhl76}, there exists a isometry $V_{L_B^{n}B^{n}K_{1}\rightarrow
\widetilde{B}}^{m}$ such that
\begin{equation}
V_{L_B^{n}B^{n}K_{1}\rightarrow\widetilde{B}}^{m}\left[  \frac{1}{\sqrt{|K|}}%
\sum_{k}U_{{B'}^{n}\rightarrow B^{n}E^{n}}^{\mc{M}^{x^{n}(m,k)}}|\psi\rangle_{L_BB'}^{\otimes
n}|k\rangle_{K_{1}}\right]  \approx|\varphi^{\sigma}\rangle_{E^{n}%
\widetilde{B}}.
\end{equation}
Thus, Bob applies the controlled isometry
\begin{equation}
\sum_{m}|m\rangle\langle m|_{M_{1}}\otimes V_{L_B^{n}B^{n}K_{1}\rightarrow
\widetilde{B}}^{m},
\end{equation}
and then the overall state is close to%
\begin{equation}
\frac{1}{\sqrt{|M|}}\sum_{m}|m\rangle_{M_A}|\varphi^{\sigma}\rangle_{E^{n}%
\widetilde{B}}|m\rangle_{M_{1}}.
\end{equation}
Bob now discards the register $\widetilde{B}$ and Alice and Bob are left with
a maximally entangled state that is locally equivalent to approximately $ n[I(X;L_BB)_{\rho
}-I(X;E)_{\rho}] = nI(X\rangle L_BB)_{\omega}$ ebits.
\end{proof}

\section{Conclusion}
In this chapter, we discussed a private communication task called private reading. This task allows for secret key agreement between an encoder and a reader in the presence of a passive eavesdropper. Observing that access to an isometric wiretap memory cell by an encoder and the reader is a particular kind of bipartite quantum interaction, we were able to leverage bounds derived on the LOCC-assisted bidirectional secret-key-agreement capacity (Section~\ref{sec:secret-dist-protocol}) to determine bounds on its private reading capacity. We also determined a regularized expression for the non-adaptive private reading capacity of an arbitrary wiretap memory cell. For particular classes of memory cells obeying certain symmetries, such that there is an adaptive-to-non-adaptive reduction in a reading protocol (see Chapter~\ref{ch:read}), the private reading capacity and the non-adaptive private reading capacity are equal. We derived a single-letter, weak converse upper bound on the non-adaptive private reading capacity of an isometric wiretap memory cell in terms of the squashed entanglement. We also proved a strong converse upper bound on the private reading capacity of an isometric wiretap memory cell in terms of  the bidirectional max-relative entropy of entanglement. We applied discussed  results to show that the private reading capacity and the reading capacity of the qudit erasure memory cell are equal. Finally, we determined an achievable rate at which entanglement can be generated between two parties who have coherent access to a memory cell. 

We note that there is a connection to private reading protocol and floodlight quantum key distribution (FL-QKD) protocol \cite{Sha09,ZZD+16}. In FL-QKD, Alice transmits system (signal) $A'$ in the state $\rho_{LA'}$ to Bob through a channel $\mc{A}_{A'\to B}$ and keeps idler system $L$ with her. Bob performs some unitary channel $\mc{U}^{k}_{B\to B}$, which is noiseless, based on a key $(k)$ that he wants to communicate to Alice. Next, Bob performs another transformation $\mc{B}_{B\to A'}$ on the local output after the action of $\mc{U}^k$.  Now the system $A'$ in the state $\mc{B}_{B\to A'}\circ\mc{U}^k_{B\to B}\circ\mc{A}_{A'\to B}(\rho_{LA'})$ is transmitted to Alice over the channel $\mc{A}_{A'\to B}$. It is assumed that Eve has complete access to complementary $\hat{\mc{A}}_{A'\to E}$ of $\mc{A}$. Bob performs joint measurement on $LB$ in the state $\mc{A}_{A'\to B}\circ\mc{B}_{B\to A'}\circ\mc{U}^k_{B\to B}\circ\mc{A}_{A'\to B}(\rho_{LA'})$ to decode the key. 

Now let us modify the above FL-QKD protocol in the following way. Let us assume $\mc{A}_{A'\to B}$ to be noiseless channel and constraining access of Eve only to losses due to noisy local operations $\mc{B}\circ\mc{U}^k$ for all $k$, where encoding of key is such that the system (loss) accessible to Eve is independent of $k$. Then the modified FL-QKD protocol effectively reduces to a private reading protocol. Hence, framework to derive bounds on the private capacity here may provide some insight for deriving converse bounds on the private capacity of FL-QKD. We leave this question open for future direction.

%% file: sr.tex
\chapter{Hiding Digital Information in Quantum Processes for Security}\label{ch:sr}
%\section{Introduction}
In this era of rapidly-advancing technologies, there is a pertinent need for protocols that allow for secure reading of memory devices under adversarial scrutiny\blfootnote{Part of this chapter is based on an unpublished work done in collaboration with Sumeet Khatri and Mark M.~Wilde.}. Security requirements may vary depending on the situation the reader is in; for example, a person reading a document in a library (or internet caf\'{e}) wants to ensure that the librarian is not eavesdropping (cf.~Chapter~\ref{ch:priv-read}). Similarly, a spy desires to read messages securely from his or her home organization when in the vicinity of a rival organization, without arousing suspicion of the rival organization. 

	By exploiting the laws of quantum mechanics, the capabilities of information processing and computing tasks can be pushed beyond the limitations imposed by classical information theory. This also provides the opportunity to devise new information processing protocols, e.g., quantum key distribution \cite{BB84,RennerThesis} and quantum teleportation \cite{BBC+93,BFK00}. The tasks of imaging, reading, sensing, or spectroscopy of an (unknown) object of interest essentially involves the identification of an interaction process between probe system, which is in known state, and the object of interest \cite{MPS+02,BRV00,Llo08,YA11,Pir11,DRC17,DBW17}. It is known that the most interaction processes that lead to physical transformation of the state can be understood as a quantum channel. This makes it natural to model task of reading any digital memory device as the read-out of classical bits of information encoded as a sequence of quantum channels chosen from a particular set, which is called memory cell. 
	
	In this chapter, we discuss information processing and communication protocols for the secure reading of digital information hidden in quantum processes. Here we consider two kinds of adversaries: a passive adversary who has access only to the environment and a semi-passive adversary who has access to the memory device but cannot alter it. Before introducing secure reading protocols in Section~\ref{sec:sr-pro}, we briefly discuss secure communication protocol with zero-error where message is hidden in noiseless gates, i.e., unitary operations.  Then we formally introduce secure reading protocols where message is encoded in noisy gates, i.e., quantum channels. In Section~\ref{sec:sr-ex}, we illustrate the possibility of secure reading by providing examples where memory device is encoded using a binary memory cell consisting of amplitude damping channels \cite{NC00,Fuj04} or depolarizing channels \cite{BSST99}. Finally, we briefly discuss the application of aforementioned protocols to threat level identification (TLI), which is inspired by IFF: identification, friend or foe \cite{Bow85}.
	
\section{Secure communication using gate circuit}\label{sec:sr-toy}
Let us assume a situation where two distant friends,  Hardy and Ramanujan, share a computer network. Assume that Hardy's computer is also accessible to Littlewood. Ramanujan wants to share a message on network intended only for Hardy. Ramanujan is fine with communicating message on network if chances of getting discovered by Littlewood is low. We refer to a collection of unitary operations as a gate-set. 

Consider that a bipartite computer network between Ramanujan and Hardy has five ports, labeled as $M,K,K',B',B$. $M$ is message register taking values $m\in\msc{M}$ and $K,K'$ are key registers taking values $k\in\msc{K}$, where $K=K'$, and $B',B$ corresponds to probing and reading ports, respectively. It is publicly known that computer performs a unitary operation $\mc{U}^g_{B'\to B}(\cdot)\coloneqq U^g_{B'\to B}(\cdot)(U^g_{B'\to B})^\dag$ on finite-dimensional input system $B'$ to yield finite-dimensional output system $B$. Furthermore, we consider that $\{U^g_{B'\to B}\}_{g\in \msc{G}}$ forms a finite gate-set of unitary one-design. Ramanujan and Hardy share key $k$ that is unknown to Littlewood. Ramanujan has access to $M,K$, Hardy has access to $K',B',B$, but Littlewood has access only to $B',B$. For simplicity, Ramanjuan and Hardy devise communication strategy such that $|M|$ and $|K|$ are both equal to the size $|G|$ of $\msc{G}$.   

Now Ramanujan inputs message $m$ and key $k$ in his message and key registers, respectively, and following computation is set on computer network:
\begin{equation}
\op{m}_{M}\otimes\op{k}_{K}\otimes\op{k}_{K'}\otimes\mc{U}^{x(m,k)}_{B'\to B}(\rho_{LB'}), 
\end{equation}
where 
\begin{equation}
x(m,k)\coloneqq m\oplus k = (m+k)\kern-0.45em\mod|G| 
\end{equation} 
and $\rho_{LB}$ is a probe state\footnote{We note that an entangled input is not necessary for perfect discrimination between two unitary operations \cite{DFY07}} inserted on demand by the reader, Hardy or Littlewood. $L$ is an idler system that is held by the reader. 

Ramanujan can send message $m$ with an apriori probability $p(m)$ (see \cite{BRV00}). However, for now we consider a case when message and key are chosen uniformly at random. Under such scenario, the resulting state of composite system $MB$ is a product state after the computation takes place, 
\begin{equation}
\op{m}_{M}\otimes\frac{1}{|G|}\sum_{m}\mc{U}^{x(m,k)}_{B'\to B}(\rho_{LB'}) =\op{m}_{M}\otimes\rho_{L}\otimes\frac{\mathbbm{1}_{B}}{|B|},
\end{equation}  
which holds for any input state $\rho_{LB'}$. It is clear that $M$, $B$, and $L$ are all uncorrelated in absence of key, and any measurement by Littlewood on output system $LB$ will give random value for $m$. This is equivalent to success probability of Littlewood in guessing $m$ with probability $1/|M|$. 

Now we inspect the ability of Hardy to decode the message. Since Hardy possesses key, situation boils down to the task considered in \cite{BRV00}. If a key $k\in\msc{K}$ in $x(m,k)$ is fixed, then each unitary operation $\mc{U}^{x(m,k)}$ in the given gate-set corresponds to a unique value of message $m\in\msc{M}$. Therefore, identification of unknown unitary operation $\mc{U}^{x(m,k)}$ with absolute certainty will allow Hardy to perfectly decode the message $m$ using key. If the states $\mc{U}^{x}_{B'\to B}(\rho_{LB'})$ are pure and orthogonal for all $x$ then just a single use of the unknown unitary is sufficient for the identification task. It follows from the fact that orthogonal states are perfectly distinguishable. In general, any unknown unitary operation randomly chosen from a finite set of unitary operations can be determined with certainty by sequentially applying only a finite amount of the unknown unitary operation \cite{Aci01,DPP01,DFY07}\footnote{It is in contrast to the fact that two non-orthogonal quantum states cannot be perfectly distinguishable whenever only a finite number of copies are available \cite{Che01,ACM+07}.}. 

\section{Secure identification protocols}\label{sec:sr-pro}
Any secure reading protocol consists of two parties of interest, an encoder and a reader, and an adversary. The encoder, Alice, encodes a secret classical message $m$ from a set of messages $\mathscr{M}$ onto a read-only memory device. Consider that the size of $\msc{M}$ is $|M|$. It is assumed that the memory device is delivered to the reader, Bob, whose task is to decode, i.e., read the message in a secure way despite scrutiny of adversary. 
	
	We now define a secure reading protocol, which consists of an encoding scheme, a reading scheme with reliability criterion, and a security criterion. 

\subsection*{Encoding scheme} 
The digital memory device is defined by a set $\hl{\mc{M}}_{\msc{X}}\coloneqq\{\mc{N}^x_{B'\to BE}\}_{x\in\mathscr{X}}$ of wiretap quantum channels \cite{DBW17}, where $\mathscr{X}$ is an alphabet and the size of $B',B,E'$ are fixed and independent of $x$. We call outputs of channels $\mc{N}^x_{B'\to BE}$ as ports, and output $B$ is accessible in reading port and output $E$ is accessible in bath port. Both encoder and reader agree upon the memory cell. The memory cell from Alice to Bob is then given by $\{\mc{N}^x_{B'\to B}\}_{x\in\msc{X}}$, where
	\begin{equation}
	\forall x\in\msc{X}:\ \mc{N}^x_{B'\to B}(\cdot)\coloneqq \Tr_{E}\{\mc{N}^x_{B'\to BE}(\cdot)\},
	\end{equation}
	which may also be known to adversary. The reader does not have access to bath port. 
	 
	 We also assume that Alice and Bob share a secret key taking values in a set $\msc{K}$, where the size of $\msc{K}$ is $|K|$. Then, for each message $m\in\mathscr{M}$ and key $k\in\msc{K}$, we define an encoding of strings $x^n(m,k)=x_1(m,k)\dotsb x_n(m,k)\in\mathscr{X}^n$ of blocklength $n\in\mathbb{N}$ into codewords:
	\begin{equation}
		x^n(m,k)\mapsto \left(\mc{N}^{x_1(m,k)}_{B_1'\to B_1E_1},\dotsc,\mc{N}^{x_n(m,k)}_{B_n'\to B_nE_n}\right).
	\end{equation}
Each quantum channel in a codeword, each of which represents one part of the stored information is only read once. 
	
\subsection*{Reliability criterion for non-adaptive scheme}
Reading of the memory device is defined by the channel
	\begin{equation}\label{eq:na-read-cw}
		\mc{N}^{x^n(m,k)}_{B^{'n}\to B^nE^n }\coloneqq \mc{N}^{x_1(m,k)}_{B_1'\to B_1E_1}\otimes\dotsb\otimes\mc{N}^{x_n(m,k)}_{B_n\to B_nE_n},
	\end{equation}
	an input state $\rho_{LB^{'n}}$ and a positive operator-valued measure (POVM) $\left\{\Lambda_{LB^{n}}^{(m,k)}\right\}_{m}$ for each $k$, where $L$ is an arbitrary reference system belonging to the reader. Bob transmits the state $\rho_{LB^{'n}}$ through the wiretap channel $\mc{N}^{x^n(m,k)}_{B^{'n}\to B^nE^n}$, measures the output using the POVM $\left\{\Lambda_{LB^{n}}^{(m,k)}\right\}_{m}$, and uses the measurement outcome to guess the message $m$.
	
	For reliable reading of the message, the following reliability criterion should hold for each $k \in \mathscr{K} $ and $m \in \mathscr{M}$,
	\begin{equation}\label{eq:reliability-criteria}
		\Tr\{\Lambda_{LB^n}^{(m,k)}\mathcal{N}^{x^n(m,k)}_{B^{'n}\to B^nE^n}(\rho_{R^nB^{'n}})\}\geq 1-\varepsilon,
	\end{equation}
	which states that Bob's reading error probability $p_e$ is less than $\varepsilon\in(0,1)$.
	
Let us consider following definitions.
	\begin{equation}
		\sigma_{LB^nE^n}^{x^n(m,k)}\coloneqq \mc{N}^{x^n(m,k)}_{B^{'n}\to B^nE^n}(\rho_{LB^{'n}}),\quad \quad
		\hl{\sigma}_{LB^nE^n}\coloneqq\frac{1}{|M||K|}\sum_{m\in\msc{M},k\in\msc{K}}\sigma_{LB^nE^n}^{x^n(m,k)},
	\end{equation}	
and $1$-norm $\norm{\cdot}_1$ of a Hermiticity preserving map $\mc{M}_{B'\to B}$ is given as\footnote{The success probability of discriminating two channels $\mc{A}_{B'\to B}$ and $\mc{B}_{B'\to B}$ when observer is not allowed to use entangled states is 
\begin{equation}
\frac{1}{2}\(1+\frac{1}{2}\norm{\mc{A}_{B'\to B}-\mc{B}_{B'\to B}}_1\),
\end{equation}
where $\norm{\mc{M}_{B'\to B}}_1$ for a Hemiticity preserving map $\mc{M}_{B'\to B}$ is defined as \eqref{eq:1norm}.}
\begin{equation}\label{eq:1norm}
\max_{\rho_{B'}\in\msc{D}(\mc{H}_{B'})}\norm{\mc{M}_{B'\to B}(\rho_{B'})}_1. 
\end{equation}

\subsection*{Security criterion for non-adaptive scheme}
\begin{figure}[h]
		\centering
		\includegraphics[scale=1.0]{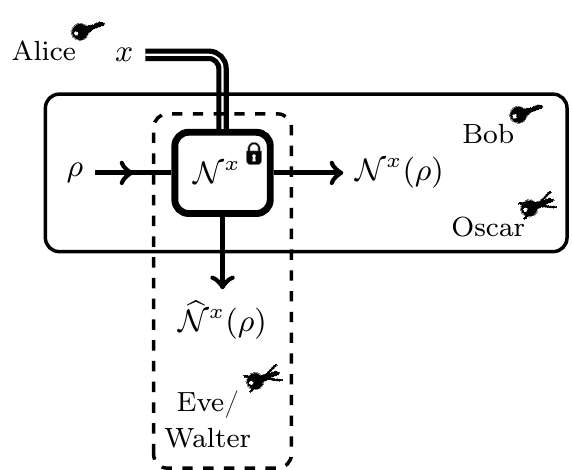}
		\caption{The scenario of secure reading along with the two types of adversaries being considered. Oscar is an adversary who can have direct access to the memory device. Walter and Eve, on the other hand, do not have direct access to the device. Furthermore, while both Alice and Bob possess a key, neither of the three adversaries do.}\label{fig-secure_reading}
	\end{figure}
	
 The criterion for secure reading will depend on the adversarial situation in which the reader, Bob, is. Here we consider three natural types of adversarial conditions, which correspond to three different secure reading protocols. These security criteria are motivated by those presented in \cite{SB11,BGP+15,WWZ16,Blo16,DBW17}. These adversarial conditions are illustrated in Fig.~\ref{fig-secure_reading}. 
	
\begin{enumerate}
\item Incognito reading: the adversary, Oscar, can have access to the memory device. The goal of Alice and Bob is to ensure that Oscar cannot figure out that the memory device contains any useful information intended for Bob. Oscar has access to the reading port but has no access to bath port. To this end, Alice and Bob share a prior secret key, which will give Bob the required advantage over Oscar. Formally, we require that
		\begin{equation}
			\norm{\frac{1}{|M||K|}\sum_{m\in\msc{M}, k\in\msc{K}}\mc{N}^{x^n(m,k)}_{B^{'n}\to B^n}-\left(\mc{N}^0_{B'\to B}\right)^{\otimes n}}_{\diamond}\leq\delta'_{I},
		\end{equation}
		where $\delta'_I\in(0,1)$ and $\Vert \cdot\Vert_{\diamond}$ is the diamond norm \eqref{eq-diamond_norm}. A memory device containing no information is assumed to be encoded with the single-element memory cell $\{\mc{N}^{x=0}_{B'\to BE}\}$. We call $\mc{N}^0_{B'\to BE}$ the innocent channel. This means that, if the reader does not possess the key, the memory device cannot be distinguished from one containing no information. To achieve this security criterion, we make stronger assumption
		\begin{align}\label{eq-incognito_security}
			\max_{\rho_{LB^{'n}}\in\msc{D}(\mc{H}_{LB^{'n}})}D\!\left(\hl{\sigma}_{LB^n}\middle\Vert \sigma^0_{LB^n}\right)\leq\delta_I,
		\end{align}
		where $\delta_I\in(0,1)$ and $\sigma^0_{LB^n}=\Tr_{E^n}\{\mc{N}^{\otimes n}_{B'\to BE}(\rho_{LB^{'n}})\}$. Note that it suffices to take optimization in \eqref{eq-incognito_security} over pure states $\rho_{LB^{'n}}$ such that $L\simeq B^{'n}$.
		
\item Covert reading: the adversary, Walter, has access to the bath port only and no access to any other ports of memory device, reading or transmitter. The goal of Alice and Bob is to ensure that Walter cannot detect that any useful information is being read by Bob. In this case, Bob has an advantage over Walter if the wiretap memory cell consists of degradable channels, i.e., $\mc{N}^x_{B'\to E}=\mc{N}^x_{B\to E}\circ\mc{N}^x_{B'\to B}$ for all $x\in\msc{X}$, where $\mc{N}^x_{B\to E}$ is a quantum channel. In general, however, we assume that Alice and Bob share a prior secret key in order for Bob to have an advantage over Walter. Formally, we require that
		\begin{equation}
			\norm{\frac{1}{|M||K|}\sum_{m\in\msc{M},k\in\msc{K}}\mc{N}^{x^n(m,k)}_{B'\to E}-\(\mc{N}^0_{B'\to E}\)^{\otimes n}}_{1}\leq\delta'_{C},
		\end{equation}
		where $\delta'_C\in(0,1)$ and $\norm{\cdot}_1$ of a Hermiticity preserving map is defined as \eqref{eq:1norm}. To achieve this security criterion, we make stronger assumption
		\begin{align}\label{eq:delta_c}
			\max_{\rho_{B^{'n}}}D\!\left(\hl{\sigma}_{E^n}\middle\Vert \(\sigma_{E^n}^0\)\right)\leq\delta_C,
		\end{align}
		where $\delta_C\in(0,1)$ and $\sigma^0_{E^n}=\Tr_{LB^n}\{\mc{N}^{\otimes n}_{B'\to BE}(\rho_{LB^{'n}})\}$.
		
\item Confidential reading: the adversary, Eve, has complete access to the bath port but no direct access to the reading port. The goal of Alice and Bob is to ensure that no information stored in the device is leaked to Eve. In general, we assume that Alice and Bob share a prior secret key in order for Bob to have an advantage over Eve. We demand that for all $m\in\msc{M}$,
	\begin{equation}
		\norm{\frac{1}{|K|}\sum_{k\in\msc{K}}\mc{N}_{B'\to E}^{x^n(m,k)}-\(\mc{N}_{B'\to E}^0\)^{\otimes n}}_{1}\leq\delta_{P},
	\end{equation}
	where $\delta_P\in(0,1)$. 
\end{enumerate}

It is to be noted that all the above security criteria for non-adaptive secure reading protocols can be generalized straightforwardly to secure reading protocols, in which adaptive strategies are employed as part of the reading protocols (see Chapter~\ref{ch:read}), in the same way as for private reading protocols (see Chapter~\ref{ch:priv-read}).

An $(n,P,\varepsilon,\delta)$ secure reading protocol is called incognito reading, covert reading, or confidential reading when $\delta=\delta_I$, $\delta=\delta_C$, or $\delta=\delta_P$, respectively, where $P\coloneqq \frac{1}{n}\log_2|M|$. The rate of an $(n,P,\varepsilon,\delta)$ secure reading protocol is equal to $\frac{1}{n}\log_2|M|$. To define the capacity of a secure reading protocol, we demand that there exists a sequence of secure reading protocols indexed by $n$ for which $\varepsilon\to 0$ and $\delta\to 0$ as $n\to\infty$ at a fixed rate $P$. A rate $P$ is called achievable if for all $\varepsilon,\delta\in(0,1]$, $\xi>0$, and sufficiently large $n$, there exists an $(n,P-\xi,\varepsilon,\delta)$ secure reading protocol. The secure reading capacity $P^{\text{secure}}(\hl{\mc{M}}_{\msc{X}})$ of a wiretap memory cell is defined as the supremum of all achievable rates $P$.
	
It can be concluded that there exists a reduction from secure reading protocols to non-adaptive secure reading protocols for jointly covariant memory cells, in the sense that the former can simulated by the latter (see Chapter~\ref{ch:read}).

\section{Illustration of secure reading}\label{sec:sr-ex}
Here we provide examples of non-adaptive secure reading protocols for memory devices encoded with a binary memory cell $\{\mc{N}^x_{B'\to BE}\}_{x\in\{0,1\}}$ consisting of depolarizing channels or generalized amplitude damping channels $\mc{N}^x_{B'\to B}$ for the reader and complementary of these channels $\widehat{\mc{N}}^x_{B'\to E}$ for Walter. Goal of this section is to determine number of non-innocent symbols that can be securely transmitted from Alice to Bob. 
	\begin{figure}[h]
		\centering
		\includegraphics[scale=1.0]{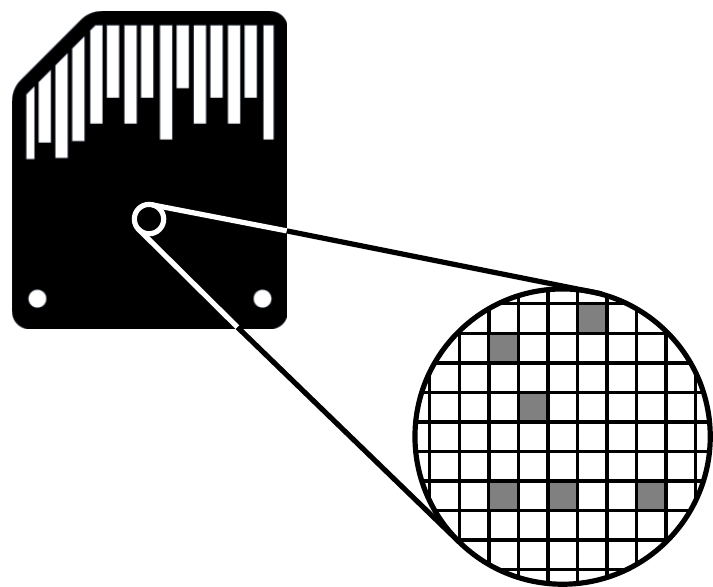}
		\caption{The encoding of a message into a digital memory device. The message is encoded into certain domains (indicated by the shaded squares) based on the key shared by the encoder and the reader.}\label{fig-reading}
	\end{figure}
	The encoding of the message onto the memory device, e.g., a CD-ROM or a flash memory drive, corresponds to either the innocent channel $\mc{N}^0_{B'\to B}$ or $\mc{N}^1_{B'\to B}$ at each of the sites (domains) of the underlying physical components comprising the device. We assume that Alice uses $N$ of these sites to encode the message. The length of the codewords is thus $N$, which is spread over the entire reading space; see Figure \ref{fig-reading}. For secure reading, Alice and Bob share a set of secret keys. The keys correspond to a particular choice of sites used to encode the message. With probability $q\ll 1$, Alice encodes each of these sites with $\mc{N}^1_{B'\to B}$, and with probability $1-q$ she encodes it with $\mc{N}^0_{B'\to B}$. The rest of $N$ sites are left blank, i.e., encoded with $\mc{N}^0_{B'\to B}$. This implies Alice sending on average $Nq$ non-innocent channels that corresponds to meaningful secure signals.

	Suppose that each $\mc{N}^x_{B'\to B}=\mc{D}^{\theta}_{B'\to B,\eta_x}$, i.e., the memory cell consists of a two-parameter family of channels. For the sites on which Alice encodes her message, the effective channel is $q\mc{D}^{\theta}_{B'\to B,\eta_1}+(1-q)\mc{D}^{\theta}_{B'\to B,\eta_0}$, and for the empty sites the channel is $\mc{D}^{\theta}_{B'\to B,\eta_0}$, whereas the effective channels to adversary are $q\widehat{\mc{D}}^{\theta}_{B'\to E,\eta_1}+(1-q)\widehat{\mc{D}}^{\theta}_{B'\to E,\eta_0}$ and $\widehat{\mc{D}}^{\theta}_{B'\to E,\eta_0}$, respectively. We assume that any reader who has access to the memory device (reading port) transmits a tensor product $(\Phi^+)^{\otimes N}$ of $N$ maximally entangled states $\Phi_{RB'}^+$. During useful reading, the reader's output state is $\omega_{R^NB^N}=(q\omega_{RB,\eta_1}^\theta+(1-q)\omega_{RB,\eta_0}^\theta)^{\otimes N}$, where $\omega_{RBE,\eta_x}^\theta\coloneqq U^{\mc{D}_{\eta_x}^\theta}_{RB'\to RBE}(\Phi^+_{RB'})(U^{\mc{D}_{\eta_x}^\theta}_{RB'\to RBE})^\dagger$. When the memory device is blank, the output state is $\omega_{R^NB^NE^N}^0=(\omega_{RBE,\eta_0}^\theta)^{\otimes N}$, where $R^NB^N$ is accessible only at the reading port for Bob or Oscar and $E^N$ is accessible only to Walter. 
	
	The security criterion for non-adaptive incognito reading reduces to $D_I\coloneqq D(\omega_{RB}\Vert\omega_{RB}^0)\leq\frac{\delta_I}{N}$, and for non-adaptive covert reading it reduces to $D_C\coloneqq D(\omega_{E}\Vert\omega_{E}^0)\leq\frac{\delta_C}{N}$. For given strategies, it turns out that $D_I=D_C$. 
	
	\begin{figure}[h]
		\centering
		\includegraphics[scale=1.0]{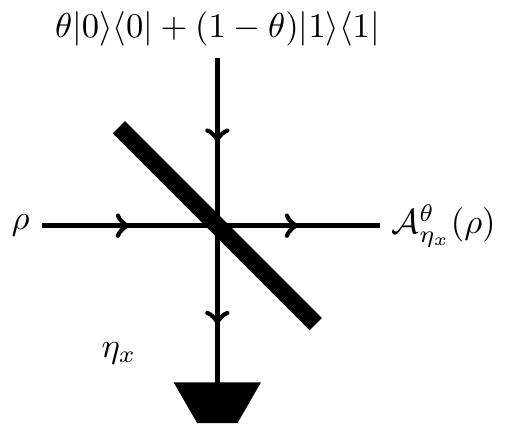}
		\caption{The generalized amplitude damping channel as an interaction of the input signal $\rho$ with a beamsplitter of transmissivity $\eta_x$ followed by discarding the state of the environment. The parameter $\theta$ quantifies the noise of the reading environment.}\label{fig-GADC}
	\end{figure}
	
	\textbf{Example 1}. Consider a binary memory cell $\{\mc{A}_{\eta_x}^{\theta}\}_{x\in\{0,1\}}$ consisting of two generalized damping channels, where $\mc{A}_{B'\to B,\eta_x}^{\theta}$ acts on qubits and is defined as $\mc{A}_{B'\to B,\eta_x}^\theta(\cdot)=\sum_{i=1}^4 E_i(\cdot)E_i^\dagger$, where $\theta,\eta\in[0,1]$ and
	\begin{equation*}
		 E_1=\sqrt{\theta}\begin{pmatrix}1&0\\0&\sqrt{\eta_x}\end{pmatrix},\quad E_2=\sqrt{\theta}\begin{pmatrix}0&\sqrt{1-\eta_x}\\0&0\end{pmatrix},
		 \end{equation*}
		 \begin{equation*}
		 E_3=\sqrt{1-\theta}\begin{pmatrix}\sqrt{\eta_x}&0\\0&1\end{pmatrix},\quad E_4=\sqrt{1-\theta}\begin{pmatrix}0&0\\\sqrt{1-\eta_x}&0\end{pmatrix}.
	\end{equation*}
	
	The generalized amplitude damping channel can be viewed as an interaction of the input signal with a qubit environment by means of a beamsplitter, followed by discarding the state of the environment. The parameter $\theta$ corresponds to the noise injected by the environment when the memory is being read and may be intrinsic to the memory reading device. %Since the authentic reader and the adversary can use reading device of different quality, the noise parameter $\theta$ may be differ for concerned parties. 
	
	If we let $N=1000$, $\eta_0=0.45$, $\eta_1=0.4$, and $\theta=0.5$, then we can send on average 5 non-innocent channels corresponding to secure information can be encoded with security parameter $\delta_I, \delta_C=7\times 10^{-5}$; see Fig.~\ref{fig-GADC_DEP_inc_cov}.
	
	\begin{figure}
		\centering
		\subcaptionbox{ Memory cell $\{\mc{A}_{B'\to B,\eta_x}^{\theta}\}_x$.}{\centering\includegraphics[scale=1.65]{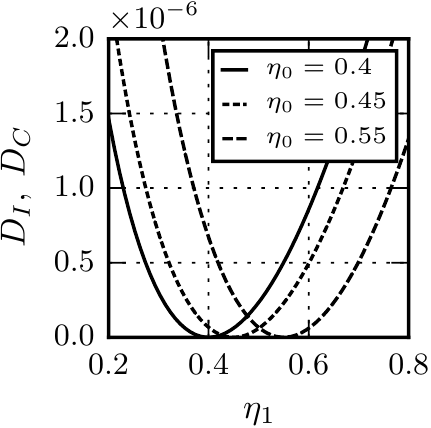}\label{subfig-GADC_incog}}
		\subcaptionbox{ Memory cell $\{\mc{D}_{B'\to B,\eta_x}^{2}\}_x$.}{\centering\includegraphics[scale=1.65]{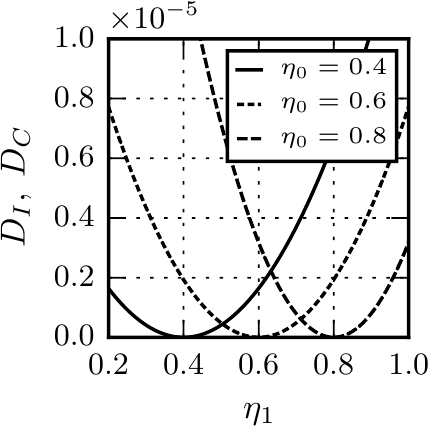}\label{subfig-DEP_incog}}
		\caption{Values of the relative entropy on the left-hand side of \eqref{eq-incognito_security},\eqref{eq:delta_c} corresponding to the security parameter $\delta_I, \delta_C$ for incognito and covert reading, respectively, when the input is restricted to the maximally entangled state. (a) Memory cell $\{\mc{A}_{B'\to B,\eta_x}^{\theta}\}_x$ with $q=0.005$, $\theta=0.5$. (b) Memory cell $\{\mc{D}_{B'\to B,\eta_x}^{2}\}_x$ with $q=0.005$ and $d=2$.}\label{fig-GADC_DEP_inc_cov}
	\end{figure}
	
	\textbf{Example 2}. Consider memory cell $\{\mc{D}^d_{B'\to B,\eta_x}\}_{x\in\{0,1\}}$ consisting of two qudit depolarizing channels, where $\mc{D}_{B'\to B,\eta_x}^d(\rho)=\eta_x\rho+(1-\eta_x)\frac{\mathbbm{1}}{d}$, and $\eta_x\in\left[0,\frac{d^2}{d^2-1}\right]$. Parameter $1-\eta_x$ depends on the deviation of the channel $\mc{D}_{B'\to B,\eta_x}$ from any unitary evolution \cite[Proposition 11]{DBW17}. In this case, the security criterion for non-adaptive incognito reading is
	\begin{equation}
		\lambda_1^{\omega}\log_2\left(\frac{\lambda_1^{\omega}}{\lambda_1^{\omega^0}}\right)+(d^2-1)\lambda_2^{\omega}\log_2\left(\frac{\lambda_2^{\omega}}{\lambda_2^{\omega^0}}\right)\leq\frac{\delta_I}{N},
	\end{equation}
	where $\{\lambda_i^\rho\}_i$ denotes the spectrum of the state $\rho$ and
	\begin{align*}
		\lambda_1^{\omega^0}&=\eta_0+\frac{1-\eta_0}{d^2},& \quad \lambda_1^{\omega}&=(q\eta_1+(1-q)\eta_0) +\frac{1}{d^2}(q(1-\eta_1)+(1-q)(1-\eta_0)),\\
 \lambda_2^{\omega^0}&=\frac{1-\eta_0}{d^2},&\quad \lambda_2^{\omega}&=\frac{1}{d^2}(q(1-\eta_1)+(1-q)(1-\eta_0)).
	\end{align*}
	
	If we let $d=2$, $N=1000$, $q=0.005$, $\eta_0=0.8$, and $\eta_1=0.7$ or $\eta_1=0.9$, then we can send on average 5 non-innocent channels corresponding to secure information with security parameter $\delta_I,\delta_C=8.5\times 10^{-4}$; see Fig.~\ref{fig-GADC_DEP_inc_cov}.
	
\section{Threat level identification}\label{sec:tfi}
Our protocol for non-adaptive incognito reading can be applied to threat level identification (TLI), in which the messages $m\in\msc{M}=\{1,2,\dotsc,|\msc{M}|\}$ correspond to the threat level posed by an adversary, Oscar. A friendly aircraft, to be used as a spy for stealthy surveillance, can be embedded with a memory device and share a secret key with the friendly base. Since Oscar does not have the key, it will not be able to identify the aircraft as being a spy. The aircraft can thus collect information about Oscar's base and report back to its headquarter with a message indicating the threat level. Non-adaptive incognito reading protocols are natural in this context since the memory device and the reader are at distant locations and scout situations, which are time sensitive, may not allow enough time to execute adaptive protocols.	

\section*{Open problems}
For future work, it would be interesting to explore the task of secure reading when the memory cell consists of channels acting on continuous variable systems \cite{book1991cover,H13book}. Since reading protocols are based on channel discrimination and there are connections between parmameter estimation in metrology, process tomography, and channel discrimination, another future direction is to study the possibility of some new secure parameter estimation protocols in the context of metrology and sensing \cite{K93book,GLM06,DRC17} (see also \cite{CMP07,GBGD17,HMM17} for the literature on secure parameter estimation).

%% file: app-hw.tex
\chapter{Qudit Systems and Heisenberg--Weyl Group}\label{app:qudit}
Here we introduce some basic notations and definitions related to qudit systems. A system represented with a $d$-dimensional Hilbert space is called a qu$d$it system. Let $J_{B'}=\{|j\>_{B'}\}_{j\in \{0,\ldots,d-1\}}$ be a computational orthonormal basis of $\mc{H}_{B'}$ such that $\dim(\mc{H}_{B'})=d$. There exists a unitary operator called \textit{cyclic shift operator} $X(k)$ that acts on the orthonormal states as follows:
\begin{equation}\label{eq:pauli-x-d}
\forall |j\>_{B'}\in J_{B'}:\ \ X(k)|j\>=|k\oplus j\>,
\end{equation}
where $\oplus$ is a cyclic addition operator, i.e., $k\oplus j:= (k+j)\ \textnormal{mod}\ d$. There also exists another unitary operator called the \textit{phase operator} $Z(l)$ that acts on the qudit computational basis states as
\begin{equation}
\forall |j\>_{B'}\in J_{B'}:\ \ Z(l)|j\>=\exp\(\frac{\iota 2\pi lj}{d}\)|j\>. 
\end{equation}
The $d^2$ operators $\{X(k)Z(l)\}_{k,l\in\{0,\ldots,d-1\}}$ are known as the Heisenberg--Weyl operators. Let $\sigma(k,l):=X(k)Z(l)$. 
The maximally entangled state $\Phi_{R:B'}$ of qudit systems $RB'$ is given as
\begin{equation}
|\Phi\>_{RB'}:=\frac{1}{\sqrt{d}}\sum_{j=0}^{d-1}|j\>_R|j\>_{B'},
\end{equation}
and we define 
\begin{equation}
|\Phi^{k,l}\>_{RB'}:=(\bm{1}_R\otimes\sigma^{k,l}_{B'})|\Phi\>_{R:B'}.
\end{equation}
The $d^2$ states $\{|\Phi^{k,l}\>_{RB'}\}_{k,l\in\{0,\ldots,d-1\}}$ form a complete, orthonormal basis:
\begin{align}
\<\Phi^{k_1,l_1}|\Phi^{k_2,l_2}\>&=\delta_{k_1,k_2}\delta_{l_1,l_2},\\
\sum_{k,l=0}^{d-1}|\Phi^{k,l}\>\<\Phi^{k,l}|_{RB'}&=\bm{1}_{RB'}.
\end{align}

Let $\msc{W}$ be a discrete set of size $|W|=d^2$. There exists one-to-one mapping $\{(k,l)\}_{k,l\in\{0,d-1\}}\leftrightarrow \{w\}_{w\in\msc{W}}$. For example, we can use the following map: $w=k+d\cdot l$ for $\msc{W}=\{0,\ldots,d^2-1\}$. This allows us to define $\sigma^w:=\sigma(k,l)$ and $\Phi^w_{RB'}:=\Phi^{k,l}
_{RB'}$.  Let the set of $d^2$ Heisenberg--Weyl operators be denoted as
\begin{equation}\label{eq:HW-op}
\mathbf{H}:=\{ \sigma^w\}_{w\in\msc{W}}=\{X(k)Z(l)\}_{k,l\in\{0,\ldots,d-1\}},
\end{equation}
and we refer to $\mathbf{H}$  as the Heisenberg--Weyl group.

%% file: app-Rmax.tex
\chapter{Bidirectional Max-Rains Information: Examples}\label{app:rmax}
This appendix contains results discussed in \cite{BDW18}.

Here we restrict $d=2$ in Appendix~\ref{app:qudit} to consider qubit systems. 

\begin{figure}[h]
		\centering
		\includegraphics[width=0.4\linewidth, trim={1cm 5cm 1cm 5cm}]{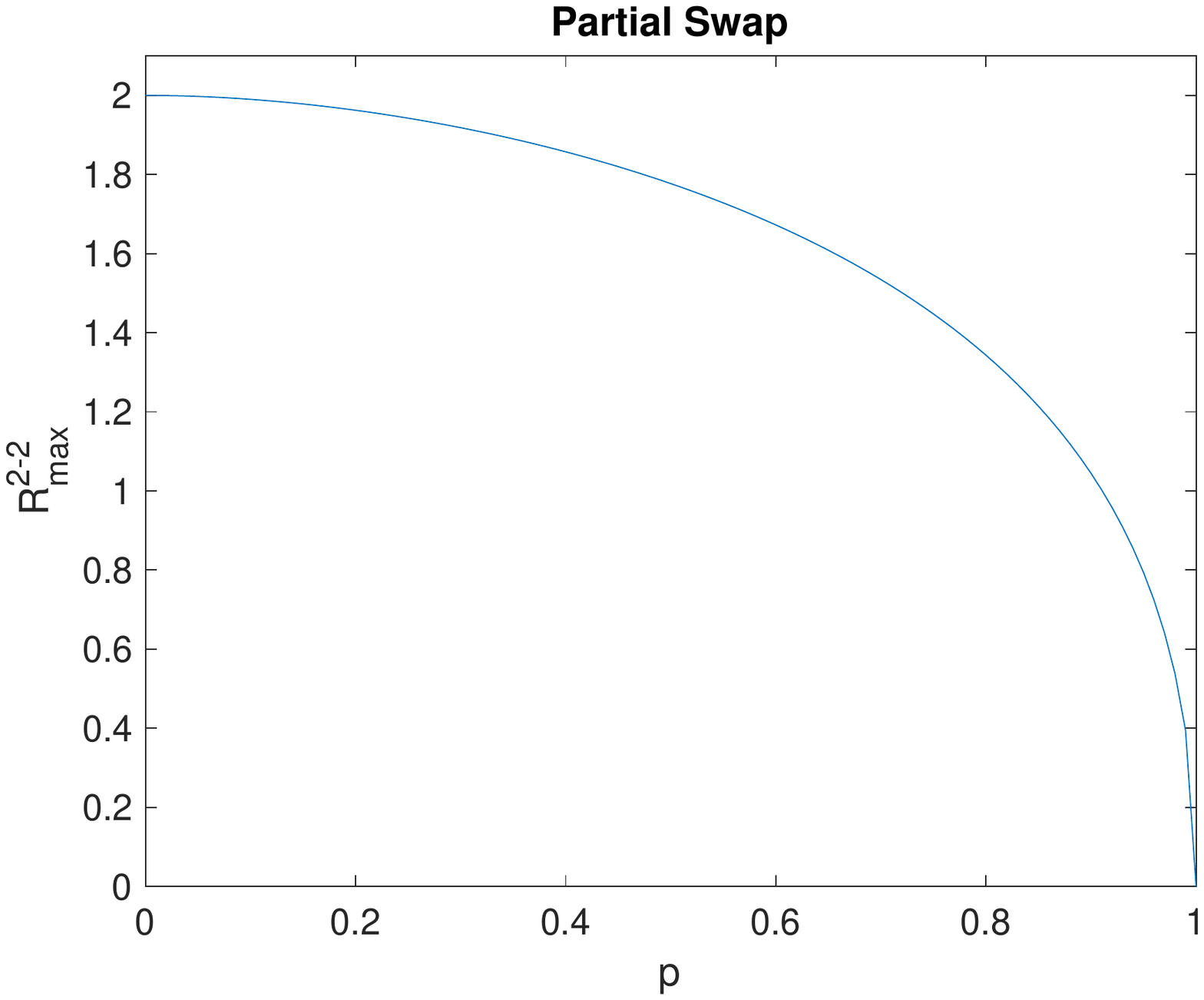}
		\includegraphics[width=0.4\linewidth, trim={1cm 5cm 1cm 5cm}]{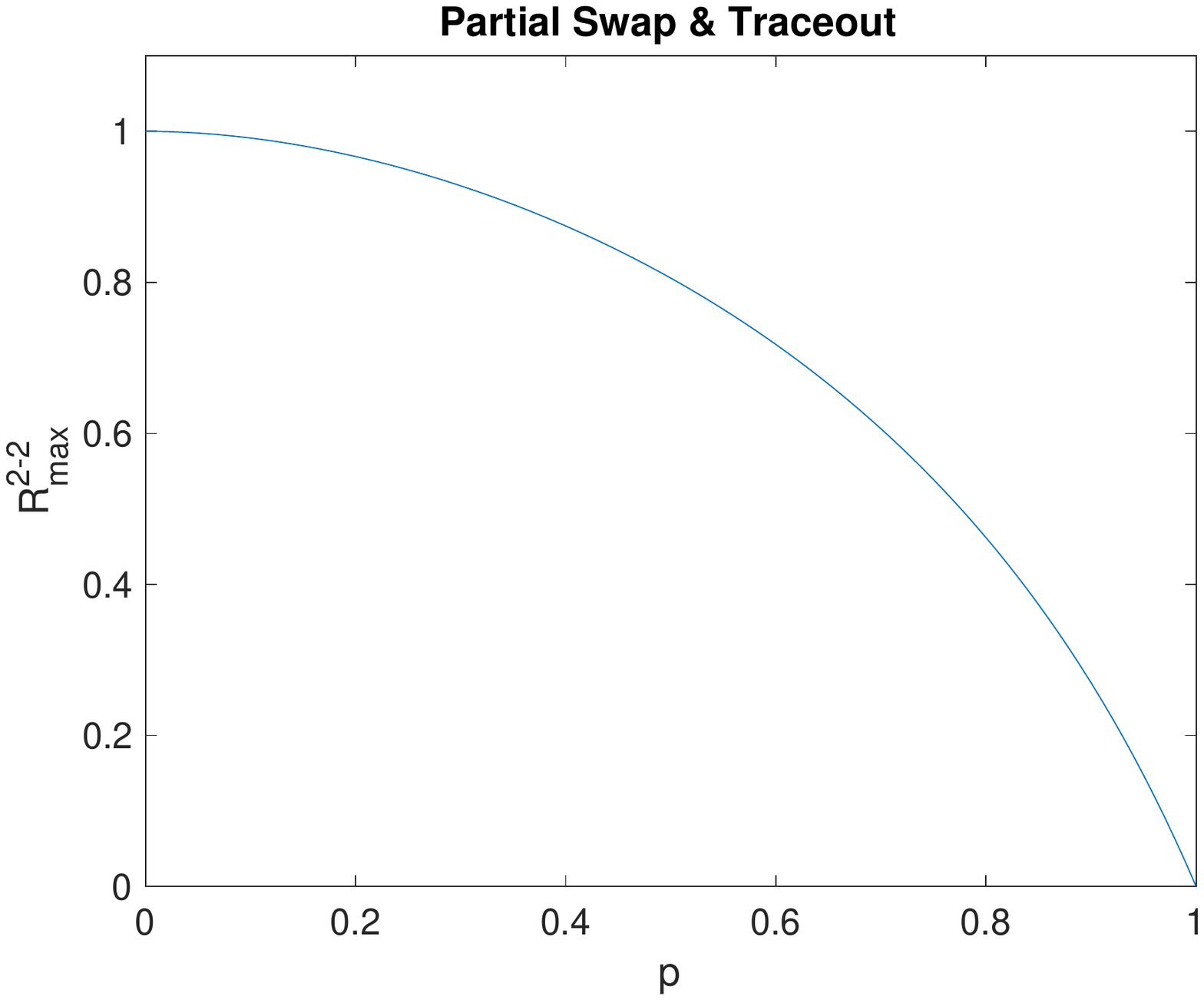}
		\includegraphics[width=0.4\linewidth, trim={1cm 5cm 1cm 5cm}]{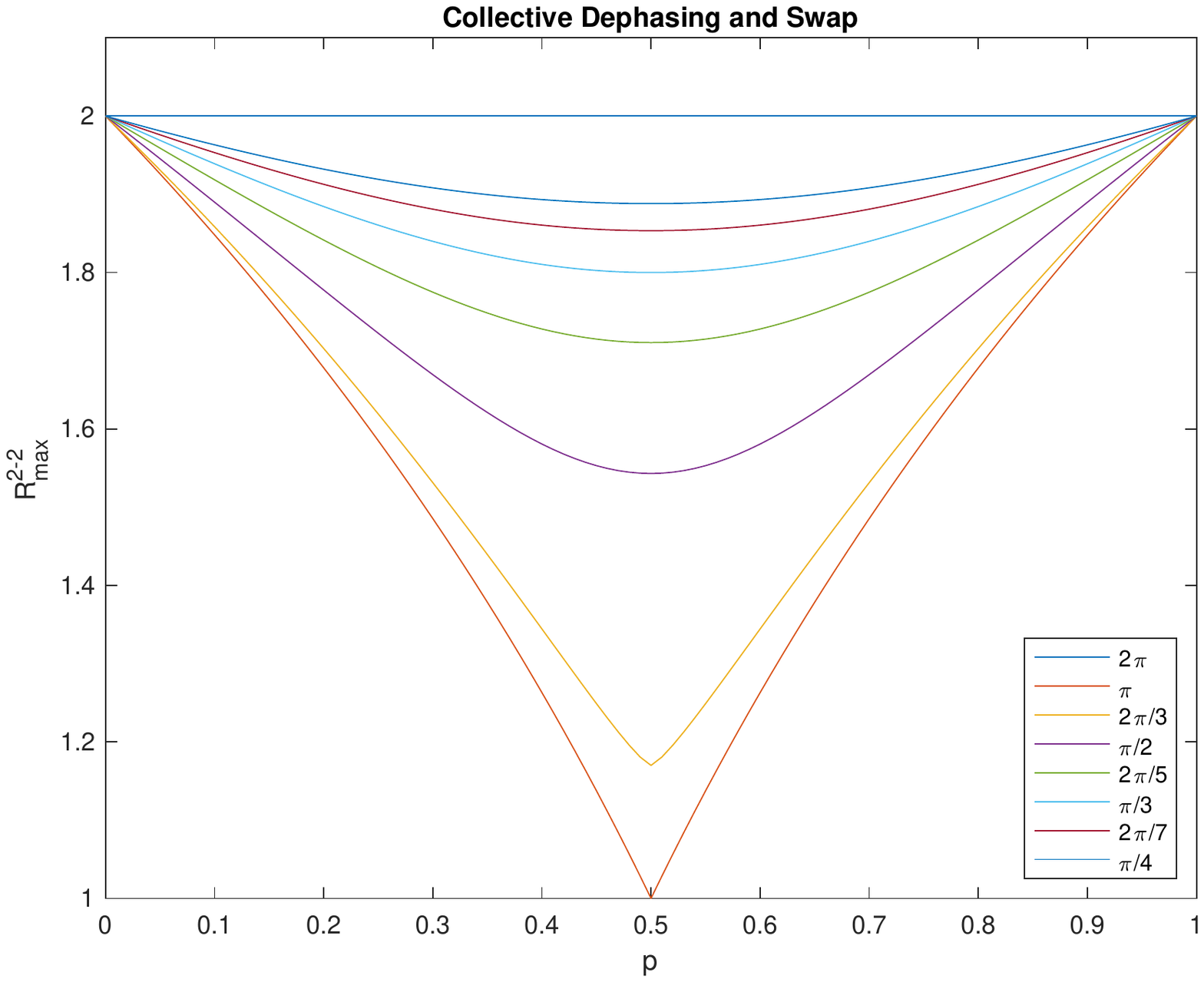}
		\caption{Our bounds plotted versus channel parameter $p$. From top to bottom they belong to (i) the qubit partial swap operation, (ii) the qubit partial swap operation followed by traceout of Alice's output and (iii) a qubit swap operation followed by collective dephasing with various phases $\phi$.}
		\label{fig2}
\end{figure}

As an example, we have numerically computed $R^{2\to 2}_{\max}$ for the qubit partial swap operation \cite{FHSSW11,audenaert2016entropy}, which is performed by application of the unitary $U_p=\sqrt{p}\mathbbm{I}+\iota\sqrt{1-p}S$, where $S=\sum_{ij}|ij\>\<ji|$ is the swap operator. Such an operation, which can be followed by a traceout of Alice's subsystem,  can be compared to a beamsplitter \cite{konig2013limits}. As a second example, we have computed  $R^{2\to 2}_{\max}$ for a qubit swap operator, followed by collective dephasing \cite{palma1996quantum}, which is a typical model for noise in a quantum computer. In the qubit case, collective dephasing acts as $|0\>\to|0\>$, $|1\>\to e^{\iota\phi}|1\>$ for some phase $\phi$. Hence $|00\>\to|00\>$,  $|01\>\to e^{\iota\phi}|01\>$, $|10\>\to e^{\iota\phi}|10\>$, and $|11\>\to e^{2\iota\phi}|11\>$. The collective dephasing occurs with probability $1-p$. 

Our results are plotted in Figure \ref{fig2}. For the partial swap, the top plot shows the expected decline from two ebits to zero, as the channel tends towards total depolarization. For the partial swap and traceout, the decline is from one ebit to zero. In the example of collective dephasing, as expected, the performance is the worst at $p=1/2$, where there is the most uncertainty about whether the collective dephasing has taken place. For phase $\phi=\pi$, we can have a reduction of a factor of $1/2$. Let us remark that this bound can actually be achieved. To do so, Alice and Bob both locally create two Bell states $\Phi^+_{L_AA'}$ and $\Phi^+_{B'L_B}$. After the swap operation and the collective dephasing they end up in a state $\frac{1}{2}\Phi^{+}_{AL_B}\otimes\Phi^{+}_{BL_A}+\frac{1}{2}\Phi^{-}_{AL_B}\otimes\Phi^{-}_{BL_A}$. To find out the phase, Alice and Bob can locally measure either $A$ and $ L_B$ or $L_A$ and $B$ in the $X$-basis, thus sacrificing one ebit. If their results agree, they have $\Phi^{+}$, and otherwise $\Phi^{-}$, which can be locally rotated to~$\Phi^{+}$.

%% file: app-eg.tex
\chapter{Rate of Entropy Change: Examples}\label{app-rate_ent_change}

	Here, we review \cite[Appendix B]{DKSW18} to discuss the subtleties involved in determining the rate of entropy change using the formula \eqref{eq-pi_ent_change_rate} (Theorem \ref{thm:rate-oqs}) by considering some examples of dynamical processes. 
	
	Let us first consider a system in a pure state $\psi_t$ undergoing a unitary time evolution. In this case, the entropy is zero for all time, and thus the rate of entropy change is also zero for all time. Note that even though the rank of the state remains the same for all time, the support changes. This implies that the kernel changes with time. However, $\dot{\psi}_t$ is well defined. This allows us to invoke Theorem \ref{thm:rate-oqs}, so the formula \eqref{eq-pi_ent_change_rate} is applicable. 
	
	Formula \eqref{eq-pi_ent_change_rate} is also applicable to states with higher rank whose kernel changes in time and have non-zero entropy. For example, consider the density operator $\rho_t\in\mc{D}(\mc{H})$ with the following time-dependence:
	\begin{equation}\label{eq:rho-t-i}
		\forall ~t\geq 0:\quad \rho_t=\sum_{i\in\mc{I}} \lambda_i(t) U_i(t)\Pi_i(0)U_i^\dagger(t),
	\end{equation}
	where $\mc{I}=\{i:1\leq i\leq d,~d<\dim(\mc{H})\}$, $\sum_{i\in\mc{I}}\lambda_i(t)=1$, $\lambda_i(t)\geq 0$ and the time-derivative $\dot{\lambda}_i(t)$ of $\lambda_i(t)$ is well defined for all $i\in\mathcal{I}$. The operators $U_i(t)$ are time-dependent unitary operators associated with the eigenvalues $\lambda_i(t)$ such that the time-derivative $\dot{U}_i(t)$ of $U_i(t)$ is well defined and $[U_i(0),\Pi_i(0)]=0$ for all $i\in \mc{I}$. The operators $\Pi_i(0)$ are projection operators associated with the eigenvalues $\lambda_i(0)$ such that the spectral decomposition of $\rho_t$ at $t=0$ is
	\begin{align}\label{eq:rho-0-i}
		\rho_0=\sum_{i\in\mc{I}} \lambda_i(0)\Pi_i(0),
	\end{align}
	where $1<\text{rank}(\rho_0)<\dim(\mc{H})$. The evolution of the system is such that $\text{rank}(\rho_t)=\text{rank}(\rho_0)$ for all $t\geq 0$. It is clear from \eqref{eq:rho-t-i} and \eqref{eq:rho-0-i} that the projection $\Pi_{t}$ onto the support of $\rho_t$ depends on time:
	\begin{equation}
		\Pi_{t}=\sum_{i\in \mc{I}}U_i(t)\Pi_i(0)U_i^\dagger(t),
	\end{equation}
	and the time-derivative $\dot{\Pi}_{t}$ of $\Pi_t$ is well defined. The entropy of the system is zero if and only if the state is pure.  %because each $\dot{U}_t(i)$ exists and $\Pi_0(i)$ is time-independent for all $i \in \mc{I}$. 
	
	Let us consider a qubit system $A$ undergoing a damping process such that its state $\rho_t$ at any time $t\geq 0$ is as follows:  
	\begin{equation}
		\rho_t=(1-e^{-t})\op{0}+e^{-t}\op{1},
	\end{equation}
	where $\{\ket{0},\ket{1}\}\in\ONB(\mc{H}_A)$. The entropy $S(\rho_t)$ of the system at time $t$ is
	\begin{equation}\label{eq-rec1}
		S(\rho_t)=-(1-e^{-t})\log(1-e^{-t})-e^{-t}\log(e^{-t}),
	\end{equation} 
	which is continuously differentiable for all $t>0$ and not differentiable at $t=0$. At $t=0$, $\Pi_0=\op{1}$ and $\text{rank}(\rho_0)=1$. At $t=0^+$, there is a jump in the rank from 1 to 2, and the rank and the support remains the same for all $t\in(0,\infty)$. In this case, the formula \eqref{eq-pi_ent_change_rate} agrees with the derivative of \eqref{eq-rec1}.
	
	Now, suppose that the system $A$ undergoes an oscillatory process such that for any time $t\geq 0$ the state $\rho_t$ of the system is given by
	\begin{equation}
		\rho_t=\cos^2(\pi t)\op{0}+\sin^2(\pi t)\op{1}.
	\end{equation}
	In this case, for all $t\geq 0$, the entropy $S(\rho_t)$ is
	\begin{equation}
		S(\rho_t)=-\cos^2(\pi t)\log\cos^2(\pi t)-\sin^2(\pi t)\log\sin^2(\pi t),
	\end{equation}
	and its derivative is
	\begin{equation}\label{eq-example2_derivative}
		\frac{\d}{\d t}S(\rho_t)=\pi\sin(2\pi t)\left[\log\cos^2(\pi t)-\log\sin^2(\pi t)\right],
	\end{equation}
	which exists for all $t\geq 0$. At $t=\frac{n}{2}$ for all $n\in\mathbb{Z}^+\cup\{0\}$, there is a jump in the rank from 1 to 2 and the support changes discontinuously at these instants. One can check that \eqref{eq-pi_ent_change_rate} and \eqref{eq-example2_derivative} are in agreement for all $t\geq 0$.

%% file: vita.tex
\chapter*{Vita}
\label{c:vita}

Siddhartha Das was born in the year 1991. He passed SLC examination ($\text{X}^{\text{th}}$ board) from Koshi Vidya Mandir, Biratnagar. He passed AISSCE ($\text{XII}^{\text{th}}$ board) from Millia Convent School, Purnea. He enrolled in BS-MS Dual Degree Programme (August 2009--May 2014) at Indian Institute of Science Education and Research Pune. He received his BS and MS degree on $15^\textnormal{th}$ June, 2014. Later, he joined the Department of Physics and Astronomy at Louisiana State University, Baton Rouge for PhD starting Fall semester in 2014. He plans to graduate in December 2018, by the end of Fall semester. He has accepted postdoctoral position in the group of Nicolas Cerf at Universit\'e libre de Bruxelles, Brussels. He plans to start his postdoctoral position in early 2019.